\def\tit{Generalized Convolution and Efficient Language Recognition
(Extended version)
}

\newcommand\auth{Conal Elliott}

\documentclass[hidelinks,twoside]{article}  

\usepackage[margin=1in]{geometry}  

\usepackage[square]{natbib}
\author{\auth \\[1.5ex]
Target\\[1.5ex]conal@conal.net
}

\usepackage{fancyhdr}
\tit
\makeatletter
\@ifundefined{lhs2tex.lhs2tex.sty.read}%
  {\@namedef{lhs2tex.lhs2tex.sty.read}{}%
   \newcommand\SkipToFmtEnd{}%
   \newcommand\EndFmtInput{}%
   \long\def\SkipToFmtEnd#1\EndFmtInput{}%
  }\SkipToFmtEnd

\newcommand\ReadOnlyOnce[1]{\@ifundefined{#1}{\@namedef{#1}{}}\SkipToFmtEnd}
\usepackage{amstext}
\usepackage{amssymb}
\usepackage{stmaryrd}
\DeclareFontFamily{OT1}{cmtex}{}
\DeclareFontShape{OT1}{cmtex}{m}{n}
  {<5><6><7><8>cmtex8
   <9>cmtex9
   <10><10.95><12><14.4><17.28><20.74><24.88>cmtex10}{}
\DeclareFontShape{OT1}{cmtex}{m}{it}
  {<-> ssub * cmtt/m/it}{}

\DeclareFontShape{OT1}{cmtt}{bx}{n}
  {<5><6><7><8>cmtt8
   <9>cmbtt9
   <10><10.95><12><14.4><17.28><20.74><24.88>cmbtt10}{}
\DeclareFontShape{OT1}{cmtex}{bx}{n}
  {<-> ssub * cmtt/bx/n}{}

\newcommand{\Conid}[1]{\mathit{#1}}
\newcommand{\Varid}[1]{\mathit{#1}}
\newcommand{\anonymous}{\kern0.06em \vbox{\hrule\@width.5em}}
\newcommand{\plus}{\mathbin{+\!\!\!+}}
\newcommand{\bind}{\mathbin{>\!\!\!>\mkern-6.7mu=}}

\usepackage{polytable}

\@ifundefined{mathindent}%
  {\newdimen\mathindent\mathindent\leftmargini}%
  {}%

\def\resethooks{%
  \global\let\SaveRestoreHook\empty
  \global\let\ColumnHook\empty}
\newcommand*{\savecolumns}[1][default]%
  {\g@addto@macro\SaveRestoreHook{\savecolumns[#1]}}
\newcommand*{\restorecolumns}[1][default]%
  {\g@addto@macro\SaveRestoreHook{\restorecolumns[#1]}}
\newcommand*{\aligncolumn}[2]%
  {\g@addto@macro\ColumnHook{\column{#1}{#2}}}

\resethooks

\newcommand{\onelinecommentchars}{\quad-{}- }
\newcommand{\commentbeginchars}{\enskip\{-}
\newcommand{\commentendchars}{-\}\enskip}

\newcommand{\visiblecomments}{%
  \let\onelinecomment=\onelinecommentchars
  \let\commentbegin=\commentbeginchars
  \let\commentend=\commentendchars}

\newcommand{\invisiblecomments}{%
  \let\onelinecomment=\empty
  \let\commentbegin=\empty
  \let\commentend=\empty}

\visiblecomments

\newlength{\blanklineskip}
\newcommand{\hsindent}[1]{\quad}
\let\hspre\empty
\let\hspost\empty

\EndFmtInput
\makeatother
\ReadOnlyOnce{polycode.fmt}%
\makeatletter

\newcommand{\hsnewpar}[1]%
  {{\parskip=0pt\parindent=0pt\par\vskip #1\noindent}}

\newcommand{\hscodestyle}{}

\newcommand{\sethscode}[1]%
  {\expandafter\let\expandafter\hscode\csname #1\endcsname
   \expandafter\let\expandafter\endhscode\csname end#1\endcsname}

  {\par\noindent
   \advance\leftskip\mathindent
   \hscodestyle
   \let\\=\@normalcr
   \let\hspre\(\let\hspost\)%
   \pboxed}%
  {\endpboxed\)%
   \par\noindent
   \ignorespacesafterend}

  {\hsnewpar\abovedisplayskip
   \advance\leftskip\mathindent
   \hscodestyle
   \let\hspre\(\let\hspost\)%
   \pboxed}%
  {\endpboxed%
   \hsnewpar\belowdisplayskip
   \ignorespacesafterend}

  {\hsnewpar\abovedisplayskip
   \advance\leftskip\mathindent
   \hscodestyle
   \let\\=\@normalcr
   \(\pboxed}%
  {\endpboxed\)%
   \hsnewpar\belowdisplayskip
   \ignorespacesafterend}

\newcommand{\plainhs}{\sethscode{plainhscode}}

\plainhs

  {\hsnewpar\abovedisplayskip
   \advance\leftskip\mathindent
   \hscodestyle
   \let\\=\@normalcr
   \(\parray}%
  {\endparray\)%
   \hsnewpar\belowdisplayskip
   \ignorespacesafterend}

  {\parray}{\endparray}

  {\(\parray}{\endparray\)}

\def\codeframewidth{\arrayrulewidth}
\RequirePackage{calc}

  {\parskip=\abovedisplayskip\par\noindent
   \hscodestyle
   \arrayrulewidth=\codeframewidth
   \tabular{@{}|p{\linewidth-2\arraycolsep-2\arrayrulewidth-2pt}|@{}}%
   \hline\framedhslinecorrect\\{-1.5ex}%
   \let\endoflinesave=\\
   \let\\=\@normalcr
   \(\pboxed}%
  {\endpboxed\)%
   \framedhslinecorrect\endoflinesave{.5ex}\hline
   \endtabular
   \parskip=\belowdisplayskip\par\noindent
   \ignorespacesafterend}

\newcommand{\framedhslinecorrect}[2]%
  {#1[#2]}

  {\(\def\column##1##2{}%
   \let\>\undefined\let\<\undefined\let\\\undefined
   \newcommand\>[1][]{}\newcommand\<[1][]{}\newcommand\\[1][]{}%
   \def\fromto##1##2##3{##3}%
   }{\) }%

  {\let\orighscode=\hscode
   \let\origendhscode=\endhscode
   \def\endhscode{\def\hscode{\endgroup\def\@currenvir{hscode}\\}\begingroup}
   \orighscode\def\hscode{\endgroup\def\@currenvir{hscode}}}%
  {\origendhscode
   \global\let\hscode=\orighscode
   \global\let\endhscode=\origendhscode}%

\makeatother
\EndFmtInput
\ReadOnlyOnce{forall.fmt}%
\makeatletter

\let\HaskellResetHook\empty
\newcommand*{\AtHaskellReset}[1]{%
  \g@addto@macro\HaskellResetHook{#1}}
\newcommand*{\HaskellReset}{\HaskellResetHook}

\newcommand\hsforall{\global\let\hsdot=\hsperiodonce}
\newcommand*\hsperiodonce[2]{#2\global\let\hsdot=\hscompose}
\newcommand*\hscompose[2]{#1}

\AtHaskellReset{\global\let\hsdot=\hscompose}

\HaskellReset

\makeatother
\EndFmtInput
\newcommand{\calculationcomments}{%
  \let\onelinecomment=\onelinecommentchars
  \def\commentbegin{\ \{ }%
  \def\commentend{\}}%
}
\calculationcomments

\newcommand\nc\newcommand
\nc\rnc\renewcommand

\usepackage[utf8]{inputenc} 
\usepackage[T1]{fontenc}
\usepackage{microtype}

\usepackage{epsfig}
\usepackage{latexsym}
\usepackage{amsmath}
\usepackage{amssymb}
\usepackage{color}

\usepackage{subcaption}

\usepackage[us,12hr]{datetime}
\usepackage{setspace}

\usepackage[bottom,hang,flushmargin]{footmisc}

\nc\out[1]{}

\nc\indraft[1]{#1}

\nc\note[1]{\indraft{\textcolor{red}{#1}}}

\nc\notefoot[1]{\note{\footnote{\note{#1}}}}

\nc\todo[1]{\note{To do: #1}}

\nc\eqnlabel[1]{\label{equation:#1}}
\nc\eqnref[1]{Equation~\ref{equation:#1}}
\nc\eqnreftwo[2]{Equations~\ref{equation:#1} and \ref{equation:#2}}

\nc\figlabel[1]{\label{fig:#1}}
\nc\figref[1]{Figure~\ref{fig:#1}}
\nc\figreftwo[2]{Figures~\ref{fig:#1} and \ref{fig:#2}}

\nc\seclabel[1]{\label{sec:#1}}
\nc\secref[1]{Section~\ref{sec:#1}}
\nc\secreftwo[2]{Sections~\ref{sec:#1} and~\ref{sec:#2}}
\nc\secrefs[2]{Sections \ref{sec:#1} through \ref{sec:#2}}

\nc\appref[1]{Appendix~\ref{sec:#1}}

\nc\sectiondef[1]{\section{#1}\seclabel{#1}}
\nc\subsectiondef[1]{\subsection{#1}\seclabel{#1}}
\nc\subsubsectiondef[1]{\subsubsection{#1}\seclabel{#1}}

\nc\needcite{\note{[ref]}}

\nc\figoneW[4]{
\fbox{%
\begin{minipage}{#1\linewidth}
  \centering
  \setlength\mathindent{0ex}
  #4
  \vspace*{-5ex}
  \captionof{figure}{#3}
  \label{fig:#2}
\end{minipage}
}
}
\nc\figone{\figoneW{\stdWidth}}

\nc\figo[1]{
\begin{figure}
\centering
#1
\end{figure}
}

\nc\figp[2]{\begin{figure}\centering #1 \hspace{-2ex} #2\end{figure}}

\nc\figdefG[4]{\begin{#1}[tbp]
\begin{center}
#4
\end{center}
\caption{#3}
\figlabel{#2}
\end{#1}}

\nc\figdef{\figdefG{figure}}
\nc\figdefwide{\figdefG{figure*}}

\nc\figrefdef[3]{\figref{#1}\figdef{#1}{#2}{#3}}

\nc\figrefdefwide[3]{\figref{#1}\figdefwide{#1}{#2}{#3}}

\nc\stdWidth{0.46}

\nc\figpairW[8]{
\begin{figure}
\centering
\figoneW{#1}{#3}{#4}{#5}
\figoneW{#2}{#6}{#7}{#8}
\end{figure}
}
\nc\figpair{\figpairW{\stdWidth}{\stdWidth}}

\nc\incpic[1]{\includegraphics[width=\linewidth]{figures/#1}}

\nc\incpicW[2]{\includegraphics[width=#1\linewidth]{figures/#2}}

\nc\db[1]{\llbracket#1\rrbracket}

\nc\smalltriangleup{\triangle}
\nc\smalltriangledown{\triangledown}

\nc\dq{\text{\tt\char34}}
\nc\hquoted[1]{\dq\!#1\!\dq}

\nc\sectionl[1]{\section{#1}\seclabel{#1}}
\nc\subsectionl[1]{\subsection{#1}\seclabel{#1}}

\nc\workingHere{
\vspace{1ex}
\begin{center}
\setlength{\fboxsep}{3ex}
\setlength{\fboxrule}{4pt}
\huge\textcolor{red}{\framebox{Working here}}
\end{center}
\vspace{1ex}
}

\usepackage{amsthm}
\theoremstyle{definition} 

\newtheoremstyle{plainstyle}
  {\topsep} 
  {\topsep} 
  {} 
  {} 
  {\bfseries} 
  {.} 
  {.5em} 
  {} 

\newtheorem{definition}{Definition}
\nc\deflabel[1]{\label{definition:#1}}
\nc\defref[1]{Definition \ref{definition:#1}}
\nc\defreftwo[2]{Definitions \ref{definition:#1} and \ref{definition:#2}}
\nc\defrefs[2]{Definitions \ref{definition:#1} through \ref{definition:#2}}

\newtheorem{theorem}{Theorem}
\nc\thmlabel[1]{\label{theorem:#1}}
\nc\thmref[1]{Theorem \ref{theorem:#1}}
\nc\thmreftwo[2]{Theorems \ref{theorem:#1} and \ref{theorem:#2}}
\nc\thmrefs[2]{Theorems \ref{theorem:#1} through \ref{theorem:#2}}

\nc\corlabel[1]{\label{corollary:#1}}
\nc\corref[1]{Corollary \ref{corollary:#1}}
\nc\correftwo[2]{Corollaries \ref{corollary:#1} and \ref{corollary:#2}}
\nc\correfs[2]{Corollaries \ref{corollary:#1} through \ref{corollary:#2}}

\newtheorem{lemma}[theorem]{Lemma}
\nc\lemlabel[1]{\label{lemma:#1}}
\nc\lemref[1]{Lemma \ref{lemma:#1}}
\nc\lemreftwo[2]{Lemmas \ref{lemma:#1} and \ref{lemma:#2}}
\nc\lemrefthree[3]{Lemmas \ref{lemma:#1}, \ref{lemma:#2}, and \ref{lemma:#3}}
\nc\lemrefs[2]{Lemmas \ref{lemma:#1} through \ref{lemma:#2}}

\nc\exclabel[1]{\label{exercise:#1}}
\nc\excref[1]{Exercise \ref{exercise:#1}}
\nc\excreftwo[2]{Exercises \ref{exercise:#1} and \ref{exercise:#2}}
\nc\excrefs[2]{Exercises \ref{exercise:#1} through \ref{exercise:#2}}

\definecolor{codesep}{gray}{0.85}
\nc\codesep[1]{
\begin{minipage}[b]{0ex}
\color{codesep}{\rule[1ex]{0.8pt}{#1}}
\end{minipage}}

\nc\twocol[4]{
\\
\begin{minipage}[c]{#1\textwidth}
#2
\vspace{-2ex}
\end{minipage}
\begin{minipage}[c]{#3\textwidth} 
#4
\vspace{-2ex}
\end{minipage}
\\
}

\let\oldFootnote\footnote
\nc\nextToken\relax
\rnc\footnote[1]{%
    \oldFootnote{#1}\futurelet\nextToken\isFootnote}
\nc\footcomma[1]{\ifx#1\nextToken\textsuperscript{,}\fi}
\nc\isFootnote{%
    \footcomma\footnote
    \footcomma\notefoot
}

\DeclareFontFamily{U}{mathx}{\hyphenchar\font45}
\DeclareFontShape{U}{mathx}{m}{n}{
      <5> <6> <7> <8> <9> <10>
      <10.95> <12> <14.4> <17.28> <20.74> <24.88>
      mathx10
      }{}
\DeclareSymbolFont{mathx}{U}{mathx}{m}{n}
\DeclareFontSubstitution{U}{mathx}{m}{n}
\DeclareMathAccent{\widebar}{0}{mathx}{"73} 

\usepackage{footnotebackref}
\usepackage{hyperref}

\nc\iflong[1]{#1}

\rnc\indraft[1]{}

\calculationcomments

\let\cite=\citep

\title\tit

\nc\prooflabel[1]{\label{proof:#1}}
\nc\proofref[1]{Appendix \ref{proof:#1}}
\nc\seeproof[1]{(details in \proofref{#1})}
\nc\provedIn[1]{\textnormal{Proved in \proofref{#1}}}

\nc\set[1]{\{\,#1\,\}}
\nc\Pow{\mathcal{P}}
\nc\mempty{\varepsilon}
\nc\closure[1]{#1^{\ast}}
\nc\mappend{\diamond}
\nc\cat{\mathop{}}
\nc\single\overline
\nc\union{\cup}
\nc\bigunion{\bigcup\limits}
\nc\has[1]{\mathop{\delta_{#1}}}
\nc\derivOp{\mathcal{D}}
\nc\conv{*}
\nc\hasEps{\mathop{\Varid{has}_{\mempty}}}
\nc\id{\mathop{\Varid{id}}}
\nc\ite[3]{\text{if}\ #1\ \text{then}\ #2\ \text{else}\ #3}

\nc\lis{\mathop{\Varid{list}}}
\nc\liftA{\mathop{\Varid{liftA}}}
\nc\cons{\mathit{:}}

\nc\bigOp[3]{{\displaystyle \hspace{-#3ex}#1\limits_{\substack{#2}}\hspace{-#3ex}}}

\begin{document}

\maketitle


\begin{abstract}

\emph{Convolution} is a broadly useful operation with applications including signal processing, machine learning, probability, optics, polynomial multiplication, and efficient parsing.
Usually, however, this operation is understood and implemented in more specialized forms, hiding commonalities and limiting usefulness.
This paper formulates convolution in the common algebraic framework of semirings and semimodules and populates that framework with various representation types.
One of those types is the grand abstract template and itself generalizes to the free semimodule monad.
Other representations serve varied uses and performance trade-offs, with implementations calculated from simple and regular specifications.

Of particular interest is Brzozowski's method for regular expression matching.
Uncovering the method's essence frees it from syntactic manipulations, while generalizing from boolean to weighted membership (such as multisets and probability distributions) and from sets to \emph{n}-ary relations.
The classic \emph{trie} data structure then provides an elegant and efficient alternative to syntax.

Pleasantly, polynomial arithmetic requires no additional implementation effort, works correctly with a variety of representations, and handles multivariate polynomials and power series with ease.
Image convolution also falls out as a special case.

\end{abstract}








\sectionl{Introduction}

The mathematical operation of \emph{convolution} combines two functions into a third---often written ``\ensuremath{\Varid{h}\mathrel{=}\Varid{f}\mathbin{*}\Varid{g}}''---with each \ensuremath{\Varid{h}} value resulting from summing or integrating over the products of several pairs of \ensuremath{\Varid{f}} and \ensuremath{\Varid{g}} values according to a simple rule.
This operation is at the heart of many important and interesting applications in a variety of fields \citep{SnehaHL2017Conv}.
\begin{itemize}
\item In image processing, convolution provides operations like blurring, sharpening, and edge detection \citep{Young95FIP}.
  \note{Add something about more general signal processing \citep[Chapter 2]{Yarlagadda2010ADSS}.}
\item In machine learning convolutional neural networks (CNNs) allowed recognition of translation-independent image features \citep{Fukushima1988Neo, LeCun1998GBDR, Schmidhuber2015DL}.
\item In probability, the convolution of the distributions of two independent random variables yields the distribution of their sum \citep{Grinstead2003IP}.
\item In acoustics, reverberation results from convolving sounds and their echos \citep{Pishdadian2017FRC}.
      Musical uses are known as ``convolution reverb'' \citep[Chapter 4]{HassICM1}.
\item The coefficients of the product of polynomials is the convolution of their coefficients \citep{Dolan2013FunSemi}.
\item In formal languages, (generalized) convolution is language concatenation \citep{Dongol2016CUC}.
\end{itemize}
Usually, however, convolution is taught, applied, and implemented in more specialized forms, obscuring the underlying commonalities and unnecessarily limiting its usefulness.
For instance,
\begin{itemize}
\item
  Standard definitions rely on subtraction (which is unavailable in many useful settings) and are dimension-specific, while the more general form applies to any monoid \citep{Golan2005RecentSemi,Wilding2015LAS}.
\item
  Brzozowski's method of regular expression matching \citep{Brzozowski64} appears quite unlike other applications and is limited to \emph{sets} of strings (i.e., languages), leaving unclear how to generalize to variations like weighted membership (multisets and probability distributions) as well as \emph{n}-ary \emph{relations} between strings.
\item
  Image convolution is usually tied to arrays and involves somewhat arbitrary semantic choices at image boundaries, including replication, zero-padding, and mirroring.
\notefoot{Not really a good example.}
\end{itemize}

\note{Rework the rest of this section for better coherence.}

This paper formulates general convolution in the algebraic framework of semirings and semimodules, including a collection of types for which semiring multiplication is convolution.
One of those types is the grand abstract template, namely the \emph{monoid semiring}, i.e., functions from any monoid to any semiring.
Furthermore, convolution reveals itself as a special case of an even more general notion---the \emph{free semimodule monad}.
The other types are specific representations for various uses and performance trade-offs, relating to the monoid semiring by simple denotation functions (interpretations).
The corresponding semiring implementations are calculated from the requirement that these denotations be semiring homomorphisms, thus guaranteeing that the computationally efficient representations are consistent with their mathematically simple and general template.

An application of central interest in this paper is language specification and recognition, in which convolution specializes to language concatenation.
Here, we examine a method by \citet{Brzozowski64} for flexible and efficient regular expression matching, later extended to parsing context-free languages \citep{Might2010YaccID}.
We will see that the essential technique is much more general, namely functions from lists to an arbitrary semiring.
While Brzozowski's method involves repeated manipulation of syntactic representations (regular expressions or grammars), uncovering the method's essence frees us from such representations.
Thue's tries provide a compelling alternative in simplicity and efficiency, as well as a satisfying confluence of classic techniques from the second and seventh decades of the twentieth century, as well as a modern functional programming notion: the cofree comonad.

Concretely, this paper makes the following contributions:
\notefoot{Maybe add section references.}
\begin{itemize}
\item Generalization of Brzozowski's algorithm from regular expressions representing sets of strings, to various representations of \ensuremath{[\mskip1.5mu \Varid{c}\mskip1.5mu]\to \Varid{b}} where \ensuremath{\Varid{c}} is any type and \ensuremath{\Varid{b}} is any semiring, including $n$-ary functions and relations on lists (via currying).
\item Demonstration that the subtle aspect of Brzozowski's algorithm (matching of concatenated languages) is an instance of generalized convolution.
\item Specialization of the generalized algorithm to tries (rather than regular expressions), yielding a simple and apparently quite efficient implementation, requiring no construction or manipulation of syntactic representations.
\item Observation that Brzozowski's key operations on languages generalize to the comonad operations of the standard function-from-monoid comonad and its various representations (including generalized regular expressions).
      The trie representation is the cofree comonad, which memoizes functions from the free monoid (lists).
\item Application and evaluation of a simple memoization strategy encapsulated in a familiar functor, resulting in dramatic speed improvement.
\end{itemize}

\sectionl{Monoids, Semirings and Semimodules}

The ideas in this paper revolve around a small collection of closely related algebraic abstractions, so let's begin by introducing these abstractions along with examples.

\subsectionl{Monoids}

The simplest abstraction we'll use is the monoid, expressed in Haskell as follows:
\begin{hscode}\SaveRestoreHook
\column{B}{@{}>{\hspre}l<{\hspost}@{}}%
\column{3}{@{}>{\hspre}l<{\hspost}@{}}%
\column{11}{@{}>{\hspre}l<{\hspost}@{}}%
\column{E}{@{}>{\hspre}l<{\hspost}@{}}%
\>[B]{}\mathbf{class}\;\Conid{Monoid}\;\Varid{a}\;\mathbf{where}{}\<[E]%
\\
\>[B]{}\hsindent{3}{}\<[3]%
\>[3]{}\varepsilon{}\<[11]%
\>[11]{}\mathbin{::}\Varid{a}{}\<[E]%
\\
\>[B]{}\hsindent{3}{}\<[3]%
\>[3]{}( \diamond ){}\<[11]%
\>[11]{}\mathbin{::}\Varid{a}\to \Varid{a}\to \Varid{a}{}\<[E]%
\\
\>[B]{}\hsindent{3}{}\<[3]%
\>[3]{}\mathbf{infixr}\;\mathrm{6} \diamond {}\<[E]%
\ColumnHook
\end{hscode}\resethooks
The monoid laws require that \ensuremath{( \diamond )}\iflong{ (sometimes pronounced ``mappend'')} be associative and that \ensuremath{\varepsilon}\iflong{ (``mempty'')} is its left and right identity, i.e.,
\begin{hscode}\SaveRestoreHook
\column{B}{@{}>{\hspre}l<{\hspost}@{}}%
\column{E}{@{}>{\hspre}l<{\hspost}@{}}%
\>[B]{}(\Varid{u} \diamond \Varid{v}) \diamond \Varid{w}\mathrel{=}\Varid{u} \diamond (\Varid{v} \diamond \Varid{w}){}\<[E]%
\\
\>[B]{}\varepsilon \diamond \Varid{v}\mathrel{=}\Varid{v}{}\<[E]%
\\
\>[B]{}\Varid{u} \diamond \varepsilon\mathrel{=}\Varid{u}{}\<[E]%
\ColumnHook
\end{hscode}\resethooks
One monoid especially familiar to functional programmers is lists with append:
\begin{hscode}\SaveRestoreHook
\column{B}{@{}>{\hspre}l<{\hspost}@{}}%
\column{3}{@{}>{\hspre}l<{\hspost}@{}}%
\column{11}{@{}>{\hspre}l<{\hspost}@{}}%
\column{E}{@{}>{\hspre}l<{\hspost}@{}}%
\>[B]{}\mathbf{instance}\;\Conid{Monoid}\;[\mskip1.5mu \Varid{a}\mskip1.5mu]\;\mathbf{where}{}\<[E]%
\\
\>[B]{}\hsindent{3}{}\<[3]%
\>[3]{}\varepsilon{}\<[11]%
\>[11]{}\mathrel{=}[\mskip1.5mu \mskip1.5mu]{}\<[E]%
\\
\>[B]{}\hsindent{3}{}\<[3]%
\>[3]{}( \diamond ){}\<[11]%
\>[11]{}\mathrel{=}(\plus ){}\<[E]%
\ColumnHook
\end{hscode}\resethooks
Natural numbers form a monoid under addition and zero.
These two monoids are related via the function \ensuremath{\Varid{length}\mathbin{::}[\mskip1.5mu \Varid{a}\mskip1.5mu]\to \mathbb N}, which not only maps lists to natural numbers, but does so in a way that preserves monoid structure%
:
\begin{hscode}\SaveRestoreHook
\column{B}{@{}>{\hspre}c<{\hspost}@{}}%
\column{BE}{@{}l@{}}%
\column{5}{@{}>{\hspre}l<{\hspost}@{}}%
\column{27}{@{}>{\hspre}l<{\hspost}@{}}%
\column{E}{@{}>{\hspre}l<{\hspost}@{}}%
\>[5]{}\Varid{length}\;\varepsilon{}\<[E]%
\\
\>[B]{}\mathrel{=}{}\<[BE]%
\>[5]{}\Varid{length}\;[\mskip1.5mu \mskip1.5mu]{}\<[27]%
\>[27]{}\mbox{\onelinecomment  \ensuremath{\varepsilon} on \ensuremath{[\mskip1.5mu \Varid{a}\mskip1.5mu]}}{}\<[E]%
\\
\>[B]{}\mathrel{=}{}\<[BE]%
\>[5]{}\mathrm{0}{}\<[27]%
\>[27]{}\mbox{\onelinecomment  \ensuremath{\Varid{length}} definition}{}\<[E]%
\\
\>[B]{}\mathrel{=}{}\<[BE]%
\>[5]{}\varepsilon{}\<[27]%
\>[27]{}\mbox{\onelinecomment  \ensuremath{\mathrm{0}} on \ensuremath{\mathbb N}}{}\<[E]%
\\[\blanklineskip]%
\>[5]{}\Varid{length}\;(\Varid{u} \diamond \Varid{v}){}\<[E]%
\\
\>[B]{}\mathrel{=}{}\<[BE]%
\>[5]{}\Varid{length}\;(\Varid{u}\plus \Varid{v}){}\<[27]%
\>[27]{}\mbox{\onelinecomment  \ensuremath{( \diamond )} on \ensuremath{[\mskip1.5mu \Varid{a}\mskip1.5mu]}}{}\<[E]%
\\
\>[B]{}\mathrel{=}{}\<[BE]%
\>[5]{}\Varid{length}\;\Varid{u}\mathbin{+}\Varid{length}\;\Varid{v}{}\<[27]%
\>[27]{}\mbox{\onelinecomment  \ensuremath{\Varid{length}} definition and induction}{}\<[E]%
\\
\>[B]{}\mathrel{=}{}\<[BE]%
\>[5]{}\Varid{length}\;\Varid{u} \diamond \Varid{length}\;\Varid{v}{}\<[27]%
\>[27]{}\mbox{\onelinecomment  \ensuremath{( \diamond )} on \ensuremath{\mathbb N}}{}\<[E]%
\ColumnHook
\end{hscode}\resethooks
This pattern is common and useful enough to have a name \citep{Yorgey2012Monoids}:
\begin{definition}\deflabel{monoid homomorphism}
A function \ensuremath{\Varid{h}} from one monoid to another is called a \emph{monoid homomorphism} when it satisfies the following properties:
\begin{hscode}\SaveRestoreHook
\column{B}{@{}>{\hspre}l<{\hspost}@{}}%
\column{E}{@{}>{\hspre}l<{\hspost}@{}}%
\>[B]{}\Varid{h}\;\varepsilon\mathrel{=}\varepsilon{}\<[E]%
\\
\>[B]{}\Varid{h}\;(\Varid{u} \diamond \Varid{v})\mathrel{=}\Varid{h}\;\Varid{u} \diamond \Varid{h}\;\Varid{v}{}\<[E]%
\ColumnHook
\end{hscode}\resethooks
\end{definition}

A fancier monoid example is functions from a type to itself, also known as \emph{endofunctions}, for which \ensuremath{\varepsilon} is the identity function, and \ensuremath{( \diamond )} is composition:
\begin{hscode}\SaveRestoreHook
\column{B}{@{}>{\hspre}l<{\hspost}@{}}%
\column{3}{@{}>{\hspre}l<{\hspost}@{}}%
\column{E}{@{}>{\hspre}l<{\hspost}@{}}%
\>[B]{}\mathbf{newtype}\;\Conid{Endo}\;\Varid{a}\mathrel{=}\Conid{Endo}\;(\Varid{a}\to \Varid{a}){}\<[E]%
\\[\blanklineskip]%
\>[B]{}\mathbf{instance}\;\Conid{Monoid}\;(\Conid{Endo}\;\Varid{a})\;\mathbf{where}{}\<[E]%
\\
\>[B]{}\hsindent{3}{}\<[3]%
\>[3]{}\varepsilon\mathrel{=}\Conid{Endo}\;\Varid{id}{}\<[E]%
\\
\>[B]{}\hsindent{3}{}\<[3]%
\>[3]{}\Conid{Endo}\;\Varid{g} \diamond \Conid{Endo}\;\Varid{f}\mathrel{=}\Conid{Endo}\;(\Varid{g}\hsdot{\circ }{.\:}\Varid{f}){}\<[E]%
\ColumnHook
\end{hscode}\resethooks
The identity and associativity monoid laws follow from the identity and associativity category laws, so we can generalize to endomorphisms, i.e., morphisms from an object to itself in any category.
A modest generalization of Cayley's theorem states that every monoid is isomorphic to a monoid of endofunctions \citep{Boisseau2018YNK}.
This embedding is useful for turning quadratic-time algorithms linear \citep{Hughes1986NRL,Voigtlander2008AIC}.
\twocol{0.45}{
\begin{hscode}\SaveRestoreHook
\column{B}{@{}>{\hspre}l<{\hspost}@{}}%
\column{E}{@{}>{\hspre}l<{\hspost}@{}}%
\>[B]{}\Varid{toEndo}\mathbin{::}\Conid{Monoid}\;\Varid{a}\Rightarrow \Varid{a}\to \Conid{Endo}\;\Varid{a}{}\<[E]%
\\
\>[B]{}\Varid{toEndo}\;\Varid{a}\mathrel{=}\Conid{Endo}\;(\lambda\, \Varid{z}\to \Varid{a} \diamond \Varid{z}){}\<[E]%
\ColumnHook
\end{hscode}\resethooks
}{0.45}{
\begin{hscode}\SaveRestoreHook
\column{B}{@{}>{\hspre}l<{\hspost}@{}}%
\column{E}{@{}>{\hspre}l<{\hspost}@{}}%
\>[B]{}\Varid{fromEndo}\mathbin{::}\Conid{Monoid}\;\Varid{a}\Rightarrow \Conid{Endo}\;\Varid{a}\to \Varid{a}{}\<[E]%
\\
\>[B]{}\Varid{fromEndo}\;(\Conid{Endo}\;\Varid{f})\mathrel{=}\Varid{f}\;\varepsilon{}\<[E]%
\ColumnHook
\end{hscode}\resethooks
}
The \ensuremath{\Varid{toEndo}} embedding provides another example of a monoid homomorphism%
:
\begin{hscode}\SaveRestoreHook
\column{B}{@{}>{\hspre}c<{\hspost}@{}}%
\column{BE}{@{}l@{}}%
\column{5}{@{}>{\hspre}l<{\hspost}@{}}%
\column{51}{@{}>{\hspre}l<{\hspost}@{}}%
\column{E}{@{}>{\hspre}l<{\hspost}@{}}%
\>[5]{}\Varid{toEndo}\;\varepsilon{}\<[E]%
\\
\>[B]{}\mathrel{=}{}\<[BE]%
\>[5]{}\Conid{Endo}\;(\lambda\, \Varid{z}\to \varepsilon \diamond \Varid{z}){}\<[51]%
\>[51]{}\mbox{\onelinecomment  \ensuremath{\Varid{toEndo}} definition}{}\<[E]%
\\
\>[B]{}\mathrel{=}{}\<[BE]%
\>[5]{}\Conid{Endo}\;(\lambda\, \Varid{z}\to \Varid{z}){}\<[51]%
\>[51]{}\mbox{\onelinecomment  monoid law}{}\<[E]%
\\
\>[B]{}\mathrel{=}{}\<[BE]%
\>[5]{}\varepsilon{}\<[51]%
\>[51]{}\mbox{\onelinecomment  \ensuremath{\Varid{id}} on \ensuremath{\Conid{Endo}\;\Varid{a}}}{}\<[E]%
\\[\blanklineskip]%
\>[5]{}\Varid{toEndo}\;(\Varid{a} \diamond \Varid{b}){}\<[E]%
\\
\>[B]{}\mathrel{=}{}\<[BE]%
\>[5]{}\Conid{Endo}\;(\lambda\, \Varid{z}\to (\Varid{a} \diamond \Varid{b}) \diamond \Varid{z}){}\<[51]%
\>[51]{}\mbox{\onelinecomment  \ensuremath{\Varid{toEndo}} definition}{}\<[E]%
\\
\>[B]{}\mathrel{=}{}\<[BE]%
\>[5]{}\Conid{Endo}\;(\lambda\, \Varid{z}\to \Varid{a} \diamond (\Varid{b} \diamond \Varid{z})){}\<[51]%
\>[51]{}\mbox{\onelinecomment  monoid law}{}\<[E]%
\\
\>[B]{}\mathrel{=}{}\<[BE]%
\>[5]{}\Conid{Endo}\;((\lambda\, \Varid{z}\to \Varid{a} \diamond \Varid{z})\hsdot{\circ }{.\:}(\lambda\, \Varid{z}\to \Varid{b} \diamond \Varid{z})){}\<[51]%
\>[51]{}\mbox{\onelinecomment  \ensuremath{(\hsdot{\circ }{.\:})} definition}{}\<[E]%
\\
\>[B]{}\mathrel{=}{}\<[BE]%
\>[5]{}\Conid{Endo}\;(\lambda\, \Varid{z}\to \Varid{a} \diamond \Varid{z}) \diamond \Conid{Endo}\;(\lambda\, \Varid{z}\to \Varid{b} \diamond \Varid{z}){}\<[51]%
\>[51]{}\mbox{\onelinecomment  \ensuremath{( \diamond )} on \ensuremath{\Conid{Endo}\;\Varid{a}}}{}\<[E]%
\\
\>[B]{}\mathrel{=}{}\<[BE]%
\>[5]{}\Varid{toEndo}\;\Varid{a} \diamond \Varid{toEndo}\;\Varid{b}{}\<[51]%
\>[51]{}\mbox{\onelinecomment  \ensuremath{\Varid{toEndo}} definition (twice)}{}\<[E]%
\ColumnHook
\end{hscode}\resethooks

\subsectionl{Additive Monoids}

While \ensuremath{( \diamond )} must be associative, it needn't be commutative.
Commutative monoids, however, will play an important role in this paper as well.
For clarity and familiarity, it will be convenient to use the name ``\ensuremath{(\mathbin{+})}'' instead of ``\ensuremath{( \diamond )}'' and refer to such monoids as ``additive'':
\begin{hscode}\SaveRestoreHook
\column{B}{@{}>{\hspre}l<{\hspost}@{}}%
\column{3}{@{}>{\hspre}l<{\hspost}@{}}%
\column{9}{@{}>{\hspre}l<{\hspost}@{}}%
\column{E}{@{}>{\hspre}l<{\hspost}@{}}%
\>[B]{}\mathbf{class}\;\Conid{Additive}\;\Varid{b}\;\mathbf{where}{}\<[E]%
\\
\>[B]{}\hsindent{3}{}\<[3]%
\>[3]{}\mathrm{0}{}\<[9]%
\>[9]{}\mathbin{::}\Varid{b}{}\<[E]%
\\
\>[B]{}\hsindent{3}{}\<[3]%
\>[3]{}(\mathbin{+}){}\<[9]%
\>[9]{}\mathbin{::}\Varid{b}\to \Varid{b}\to \Varid{b}{}\<[E]%
\\
\>[B]{}\hsindent{3}{}\<[3]%
\>[3]{}\mathbf{infixl}\;\mathrm{6}\mathbin{+}{}\<[E]%
\ColumnHook
\end{hscode}\resethooks
The \ensuremath{\Conid{Additive}} laws are the same as for \ensuremath{\Conid{Monoid}} (translating \ensuremath{\varepsilon} and \ensuremath{( \diamond )} to \ensuremath{\mathrm{0}} and \ensuremath{(\mathbin{+})}), together with commutativity:
\begin{hscode}\SaveRestoreHook
\column{B}{@{}>{\hspre}l<{\hspost}@{}}%
\column{E}{@{}>{\hspre}l<{\hspost}@{}}%
\>[B]{}(\Varid{u}\mathbin{+}\Varid{v})\mathbin{+}\Varid{w}\mathrel{=}\Varid{u}\mathbin{+}(\Varid{v}\mathbin{+}\Varid{w}){}\<[E]%
\\
\>[B]{}\mathrm{0}\mathbin{+}\Varid{v}\mathrel{=}\Varid{v}{}\<[E]%
\\
\>[B]{}\Varid{u}\mathbin{+}\mathrm{0}\mathrel{=}\Varid{u}{}\<[E]%
\\
\>[B]{}\Varid{u}\mathbin{+}\Varid{v}\mathrel{=}\Varid{v}\mathbin{+}\Varid{u}{}\<[E]%
\ColumnHook
\end{hscode}\resethooks
Unlike lists with append, natural numbers form a \emph{additive} monoid.
Another example is functions with pointwise addition, with any domain and with any \emph{additive} codomain:
\begin{hscode}\SaveRestoreHook
\column{B}{@{}>{\hspre}l<{\hspost}@{}}%
\column{3}{@{}>{\hspre}l<{\hspost}@{}}%
\column{10}{@{}>{\hspre}l<{\hspost}@{}}%
\column{E}{@{}>{\hspre}l<{\hspost}@{}}%
\>[B]{}\mathbf{instance}\;\Conid{Additive}\;\Varid{b}\Rightarrow \Conid{Additive}\;(\Varid{a}\to \Varid{b})\;\mathbf{where}{}\<[E]%
\\
\>[B]{}\hsindent{3}{}\<[3]%
\>[3]{}\mathrm{0}\mathrel{=}\lambda\, \Varid{a}\to \mathrm{0}{}\<[E]%
\\
\>[B]{}\hsindent{3}{}\<[3]%
\>[3]{}\Varid{f}\mathbin{+}\Varid{g}{}\<[10]%
\>[10]{}\mathrel{=}\lambda\, \Varid{a}\to \Varid{f}\;\Varid{a}\mathbin{+}\Varid{g}\;\Varid{a}{}\<[E]%
\ColumnHook
\end{hscode}\resethooks

\noindent
Additive monoids have their form of homomorphism:
\begin{definition}\deflabel{additive monoid homomorphism}
A function \ensuremath{\Varid{h}} from one additive monoid to another is an \emph{additive monoid homomorphism} if it satisfies the following properties:
\begin{hscode}\SaveRestoreHook
\column{B}{@{}>{\hspre}l<{\hspost}@{}}%
\column{E}{@{}>{\hspre}l<{\hspost}@{}}%
\>[B]{}\Varid{h}\;\mathrm{0}\mathrel{=}\mathrm{0}{}\<[E]%
\\
\>[B]{}\Varid{h}\;(\Varid{u}\mathbin{+}\Varid{v})\mathrel{=}\Varid{h}\;\Varid{u}\mathbin{+}\Varid{h}\;\Varid{v}{}\<[E]%
\ColumnHook
\end{hscode}\resethooks
\end{definition}
\noindent
Curried function types of \emph{any number} of arguments (and additive result type) are additive, thanks to repeated application of this instance.
In fact,
\begin{theorem}[\provedIn{theorem:curry additive}]\thmlabel{curry additive}
Currying and uncurrying are additive monoid homomorphisms.
\end{theorem}

\subsectionl{Semirings}

The natural numbers form a monoid in two familiar ways: addition and zero, and multiplication and one.
Moreover, these monoids interact usefully in two ways: multiplication distributes over addition, and multiplication by zero (the additive identity) yields zero (i.e., ``annihilates'').
Similarly, \emph{linear} endofunctions and their various representations (e.g., square matrices) forms a monoid via addition and via composition, with composition distributing over addition, and composition with zero yielding zero.
In both examples, addition commutes; but while natural number multiplication commutes, composition does not.
The vocabulary and laws these examples share is called a \emph{semiring} (distinguished from a ring by dropping the requirement of additive inverses):
\begin{hscode}\SaveRestoreHook
\column{B}{@{}>{\hspre}l<{\hspost}@{}}%
\column{3}{@{}>{\hspre}l<{\hspost}@{}}%
\column{10}{@{}>{\hspre}l<{\hspost}@{}}%
\column{E}{@{}>{\hspre}l<{\hspost}@{}}%
\>[B]{}\mathbf{class}\;\Conid{Additive}\;\Varid{b}\Rightarrow \Conid{Semiring}\;\Varid{b}\;\mathbf{where}{}\<[E]%
\\
\>[B]{}\hsindent{3}{}\<[3]%
\>[3]{}\mathrm{1}{}\<[10]%
\>[10]{}\mathbin{::}\Varid{b}{}\<[E]%
\\
\>[B]{}\hsindent{3}{}\<[3]%
\>[3]{}(\mathbin{*}){}\<[10]%
\>[10]{}\mathbin{::}\Varid{b}\to \Varid{b}\to \Varid{b}{}\<[E]%
\\
\>[B]{}\hsindent{3}{}\<[3]%
\>[3]{}\mathbf{infixl}\;\mathrm{7}\mathbin{*}{}\<[E]%
\ColumnHook
\end{hscode}\resethooks
The laws, in addition to those for \ensuremath{\Conid{Additive}} above, include multiplicative monoid, distribution, and annihilation:
\twocol{0.25}{
\begin{hscode}\SaveRestoreHook
\column{B}{@{}>{\hspre}l<{\hspost}@{}}%
\column{E}{@{}>{\hspre}l<{\hspost}@{}}%
\>[B]{}\Varid{u}\mathbin{*}\mathrm{0}\mathrel{=}\mathrm{0}{}\<[E]%
\\
\>[B]{}\mathrm{0}\mathbin{*}\Varid{v}\mathrel{=}\mathrm{0}{}\<[E]%
\\
\>[B]{}\mathrm{1}\mathbin{*}\Varid{v}\mathrel{=}\Varid{v}{}\<[E]%
\\
\>[B]{}\Varid{u}\mathbin{*}\mathrm{1}\mathrel{=}\Varid{u}{}\<[E]%
\ColumnHook
\end{hscode}\resethooks
}{0.35}{
\begin{hscode}\SaveRestoreHook
\column{B}{@{}>{\hspre}l<{\hspost}@{}}%
\column{E}{@{}>{\hspre}l<{\hspost}@{}}%
\>[B]{}(\Varid{u}\mathbin{*}\Varid{v})\mathbin{*}\Varid{w}\mathrel{=}\Varid{u}\mathbin{*}(\Varid{v}\mathbin{*}\Varid{w}){}\<[E]%
\\[\blanklineskip]%
\>[B]{}\Varid{p}\mathbin{*}(\Varid{q}\mathbin{+}\Varid{r})\mathrel{=}\Varid{p}\mathbin{*}\Varid{q}\mathbin{+}\Varid{p}\mathbin{*}\Varid{r}{}\<[E]%
\\
\>[B]{}(\Varid{p}\mathbin{+}\Varid{q})\mathbin{*}\Varid{r}\mathrel{=}\Varid{p}\mathbin{*}\Varid{r}\mathbin{+}\Varid{q}\mathbin{*}\Varid{r}{}\<[E]%
\ColumnHook
\end{hscode}\resethooks
}

\noindent
\begin{definition} \deflabel{semiring homomorphism}
A function \ensuremath{\Varid{h}} from one semiring to another is a \emph{semiring homomorphism} if it is an additive monoid homomorphism (\defref{additive monoid homomorphism}) and satisfies the following additional properties:
\begin{hscode}\SaveRestoreHook
\column{B}{@{}>{\hspre}l<{\hspost}@{}}%
\column{E}{@{}>{\hspre}l<{\hspost}@{}}%
\>[B]{}\Varid{h}\;\mathrm{1}\mathrel{=}\mathrm{1}{}\<[E]%
\\
\>[B]{}\Varid{h}\;(\Varid{u}\mathbin{*}\Varid{v})\mathrel{=}\Varid{h}\;\Varid{u}\mathbin{*}\Varid{h}\;\Varid{v}{}\<[E]%
\ColumnHook
\end{hscode}\resethooks
\end{definition}

\noindent
As mentioned, numbers and various linear endofunction representations form semirings.
A simpler example is the semiring of booleans, with disjunction as addition and conjunction as multiplication (though we could reverse roles):
\twocol{0.4}{
\begin{hscode}\SaveRestoreHook
\column{B}{@{}>{\hspre}l<{\hspost}@{}}%
\column{3}{@{}>{\hspre}l<{\hspost}@{}}%
\column{10}{@{}>{\hspre}l<{\hspost}@{}}%
\column{E}{@{}>{\hspre}l<{\hspost}@{}}%
\>[B]{}\mathbf{instance}\;\Conid{Additive}\;\Conid{Bool}\;\mathbf{where}{}\<[E]%
\\
\>[B]{}\hsindent{3}{}\<[3]%
\>[3]{}\mathrm{0}{}\<[10]%
\>[10]{}\mathrel{=}\Conid{False}{}\<[E]%
\\
\>[B]{}\hsindent{3}{}\<[3]%
\>[3]{}(\mathbin{+}){}\<[10]%
\>[10]{}\mathrel{=}(\mathrel{\vee}){}\<[E]%
\ColumnHook
\end{hscode}\resethooks
}{0.4}{
\begin{hscode}\SaveRestoreHook
\column{B}{@{}>{\hspre}l<{\hspost}@{}}%
\column{3}{@{}>{\hspre}l<{\hspost}@{}}%
\column{10}{@{}>{\hspre}l<{\hspost}@{}}%
\column{E}{@{}>{\hspre}l<{\hspost}@{}}%
\>[B]{}\mathbf{instance}\;\Conid{Semiring}\;\Conid{Bool}\;\mathbf{where}{}\<[E]%
\\
\>[B]{}\hsindent{3}{}\<[3]%
\>[3]{}\mathrm{1}{}\<[10]%
\>[10]{}\mathrel{=}\Conid{True}{}\<[E]%
\\
\>[B]{}\hsindent{3}{}\<[3]%
\>[3]{}(\mathbin{*}){}\<[10]%
\>[10]{}\mathrel{=}(\mathrel{\wedge}){}\<[E]%
\ColumnHook
\end{hscode}\resethooks
}
An example of a semiring homomorphism is testing natural numbers for positivity%
:
\begin{hscode}\SaveRestoreHook
\column{B}{@{}>{\hspre}l<{\hspost}@{}}%
\column{E}{@{}>{\hspre}l<{\hspost}@{}}%
\>[B]{}\Varid{positive}\mathbin{::}\mathbb N\to \Conid{Bool}{}\<[E]%
\\
\>[B]{}\Varid{positive}\;\Varid{n}\mathrel{=}\Varid{n}\mathbin{>}\mathrm{0}{}\<[E]%
\ColumnHook
\end{hscode}\resethooks
As required, the following properties hold for \ensuremath{\Varid{m},\Varid{n}\mathbin{::}\mathbb N}:%
\footnote{\emph{Exercise:} What goes wrong if we replace natural numbers by integers?}
\begin{spacing}{1.2}
\begin{hscode}\SaveRestoreHook
\column{B}{@{}>{\hspre}l<{\hspost}@{}}%
\column{14}{@{}>{\hspre}c<{\hspost}@{}}%
\column{14E}{@{}l@{}}%
\column{16}{@{}>{\hspre}c<{\hspost}@{}}%
\column{16E}{@{}l@{}}%
\column{17}{@{}>{\hspre}l<{\hspost}@{}}%
\column{20}{@{}>{\hspre}l<{\hspost}@{}}%
\column{27}{@{}>{\hspre}l<{\hspost}@{}}%
\column{35}{@{}>{\hspre}c<{\hspost}@{}}%
\column{35E}{@{}l@{}}%
\column{39}{@{}>{\hspre}l<{\hspost}@{}}%
\column{65}{@{}>{\hspre}c<{\hspost}@{}}%
\column{65E}{@{}l@{}}%
\column{68}{@{}>{\hspre}l<{\hspost}@{}}%
\column{E}{@{}>{\hspre}l<{\hspost}@{}}%
\>[B]{}\Varid{positive}\;\mathrm{0}{}\<[16]%
\>[16]{}\mathrel{=}{}\<[16E]%
\>[20]{}\Conid{False}{}\<[27]%
\>[27]{}\mathrel{=}\mathrm{0}{}\<[E]%
\\
\>[B]{}\Varid{positive}\;\mathrm{1}{}\<[16]%
\>[16]{}\mathrel{=}{}\<[16E]%
\>[20]{}\Conid{True}{}\<[27]%
\>[27]{}\mathrel{=}\mathrm{1}{}\<[E]%
\\
\>[B]{}\Varid{positive}\;(\Varid{m}{}\<[14]%
\>[14]{}\mathbin{+}{}\<[14E]%
\>[17]{}\Varid{n})\mathrel{=}\Varid{positive}\;\Varid{m}{}\<[35]%
\>[35]{}\mathrel{\vee}{}\<[35E]%
\>[39]{}\Varid{positive}\;\Varid{n}\mathrel{=}\Varid{positive}\;\Varid{m}{}\<[65]%
\>[65]{}\mathbin{+}{}\<[65E]%
\>[68]{}\Varid{positive}\;\Varid{n}{}\<[E]%
\\
\>[B]{}\Varid{positive}\;(\Varid{m}{}\<[14]%
\>[14]{}\mathbin{*}{}\<[14E]%
\>[17]{}\Varid{n})\mathrel{=}\Varid{positive}\;\Varid{m}{}\<[35]%
\>[35]{}\mathrel{\wedge}{}\<[35E]%
\>[39]{}\Varid{positive}\;\Varid{n}\mathrel{=}\Varid{positive}\;\Varid{m}{}\<[65]%
\>[65]{}\mathbin{*}{}\<[65E]%
\>[68]{}\Varid{positive}\;\Varid{n}{}\<[E]%
\ColumnHook
\end{hscode}\resethooks
\end{spacing}

\noindent
There is a more fundamental example we will have use for later:
\begin{theorem}[\provedIn{theorem:curry semiring}]\thmlabel{curry semiring}
Currying and uncurrying are semiring homomorphisms.
\end{theorem}

\subsectionl{Star Semirings}

The semiring operations allow all \emph{finite} combinations of addition, zero, multiplication, and one.
It's often useful, however, to form infinite combinations, particularly in the form of Kleene's ``star'' (or ``closure'') operation:
\begin{hscode}\SaveRestoreHook
\column{B}{@{}>{\hspre}l<{\hspost}@{}}%
\column{E}{@{}>{\hspre}l<{\hspost}@{}}%
\>[B]{}\closure{\Varid{p}}\mathrel{=}\bigOp\sum{\Varid{i}}{0}{\,}\Varid{p}^\Varid{i}\mbox{\onelinecomment  where \ensuremath{\Varid{p}^\mathrm{0}\mathrel{=}\mathrm{1}}, and \ensuremath{\Varid{p}^{\Varid{n}\mathbin{+}\mathrm{1}}\mathrel{=}\Varid{p}\mathbin{*}\Varid{p}^\Varid{n}}.}{}\<[E]%
\ColumnHook
\end{hscode}\resethooks
Another characterization is as a solution to either of the following semiring equations:
\twocol{0.35}{
\begin{hscode}\SaveRestoreHook
\column{B}{@{}>{\hspre}l<{\hspost}@{}}%
\column{E}{@{}>{\hspre}l<{\hspost}@{}}%
\>[B]{}\closure{\Varid{p}}\mathrel{=}\mathrm{1}\mathbin{+}\Varid{p}\mathbin{*}\closure{\Varid{p}}{}\<[E]%
\ColumnHook
\end{hscode}\resethooks
}{0.35}{
\begin{hscode}\SaveRestoreHook
\column{B}{@{}>{\hspre}l<{\hspost}@{}}%
\column{E}{@{}>{\hspre}l<{\hspost}@{}}%
\>[B]{}\closure{\Varid{p}}\mathrel{=}\mathrm{1}\mathbin{+}\closure{\Varid{p}}\mathbin{*}\Varid{p}{}\<[E]%
\ColumnHook
\end{hscode}\resethooks
}
which we will take as a laws for a new abstraction, as well as a default recursive implementation:
\begin{hscode}\SaveRestoreHook
\column{B}{@{}>{\hspre}l<{\hspost}@{}}%
\column{3}{@{}>{\hspre}l<{\hspost}@{}}%
\column{37}{@{}>{\hspre}l<{\hspost}@{}}%
\column{E}{@{}>{\hspre}l<{\hspost}@{}}%
\>[B]{}\mathbf{class}\;\Conid{Semiring}\;\Varid{b}\Rightarrow \Conid{StarSemiring}\;\Varid{b}\;{}\<[37]%
\>[37]{}\mathbf{where}{}\<[E]%
\\
\>[B]{}\hsindent{3}{}\<[3]%
\>[3]{}\closure{\cdot }\mathbin{::}\Varid{b}\to \Varid{b}{}\<[E]%
\\
\>[B]{}\hsindent{3}{}\<[3]%
\>[3]{}\closure{\Varid{p}}\mathrel{=}\mathrm{1}\mathbin{+}\Varid{p}\mathbin{*}\closure{\Varid{p}}{}\<[E]%
\ColumnHook
\end{hscode}\resethooks
Sometimes there are more appealing alternative implementations.
For instance, when subtraction and division are available, we can instead define \ensuremath{\closure{\Varid{p}}\mathrel{=}(\mathrm{1}\mathbin{-}\Varid{p})^{-1}} \citep{Dolan2013FunSemi}.

Predictably, there is a notion of homomorphisms for star semirings:
\begin{definition} \deflabel{star semiring homomorphism}
A function \ensuremath{\Varid{h}} from one star semiring to another is a \emph{star semiring homomorphism} if it is a semiring homomorphism (\defref{semiring homomorphism}) and satisfies the additional property \ensuremath{\Varid{h}\;(\closure{\Varid{p}})\mathrel{=}\closure{(\Varid{h}\;\Varid{p})}}.
\end{definition}

\noindent
One simple example of a star semiring (also known as a ``closed semiring'' \citep{Lehmann77,Dolan2013FunSemi}) is booleans:
\begin{hscode}\SaveRestoreHook
\column{B}{@{}>{\hspre}l<{\hspost}@{}}%
\column{42}{@{}>{\hspre}l<{\hspost}@{}}%
\column{E}{@{}>{\hspre}l<{\hspost}@{}}%
\>[B]{}\mathbf{instance}\;\Conid{StarSemiring}\;\Conid{Bool}\;\mathbf{where}\;\closure{\Varid{b}}{}\<[42]%
\>[42]{}\mathrel{=}\mathrm{1}\mbox{\onelinecomment  \ensuremath{\mathrel{=}\mathrm{1}\mathrel{\vee}(\Varid{b}\mathrel{\wedge}\closure{\Varid{b}})}}{}\<[E]%
\ColumnHook
\end{hscode}\resethooks

A useful property of star semirings is that recursive affine equations have solutions
\citep{Dolan2013FunSemi}:
\begin{lemma}\lemlabel{affine over semiring}
In a star semiring, the affine equation \ensuremath{\Varid{p}\mathrel{=}\Varid{b}\mathbin{+}\Varid{m}\mathbin{*}\Varid{p}} has solution \ensuremath{\Varid{p}\mathrel{=}\closure{\Varid{m}}\mathbin{*}\Varid{b}}.
\end{lemma}
\begin{proof}~
\begin{hscode}\SaveRestoreHook
\column{B}{@{}>{\hspre}c<{\hspost}@{}}%
\column{BE}{@{}l@{}}%
\column{5}{@{}>{\hspre}l<{\hspost}@{}}%
\column{31}{@{}>{\hspre}l<{\hspost}@{}}%
\column{E}{@{}>{\hspre}l<{\hspost}@{}}%
\>[5]{}\Varid{b}\mathbin{+}\Varid{m}\mathbin{*}(\closure{\Varid{m}}\mathbin{*}\Varid{b}){}\<[E]%
\\
\>[B]{}\mathrel{=}{}\<[BE]%
\>[5]{}\Varid{b}\mathbin{+}(\Varid{m}\mathbin{*}\closure{\Varid{m}})\mathbin{*}\Varid{b}{}\<[31]%
\>[31]{}\mbox{\onelinecomment  associativity of \ensuremath{(\mathbin{*})}}{}\<[E]%
\\
\>[B]{}\mathrel{=}{}\<[BE]%
\>[5]{}\mathrm{1}\mathbin{*}\Varid{b}\mathbin{+}\Varid{m}\mathbin{*}\closure{\Varid{m}}\mathbin{*}\Varid{b}{}\<[31]%
\>[31]{}\mbox{\onelinecomment  identity for \ensuremath{(\mathbin{*})}}{}\<[E]%
\\
\>[B]{}\mathrel{=}{}\<[BE]%
\>[5]{}(\mathrm{1}\mathbin{+}\Varid{m}\mathbin{*}\closure{\Varid{m}})\mathbin{*}\Varid{b}{}\<[31]%
\>[31]{}\mbox{\onelinecomment  distributivity}{}\<[E]%
\\
\>[B]{}\mathrel{=}{}\<[BE]%
\>[5]{}\closure{\Varid{m}}\mathbin{*}\Varid{b}{}\<[31]%
\>[31]{}\mbox{\onelinecomment  star semiring law}{}\<[E]%
\ColumnHook
\end{hscode}\resethooks
\end{proof}

\note{Mention tropical semirings, schedule algebra (max-plus), and optimization algebra (min-plus) \citep{Golan2005RecentSemi}. Maybe also polynomials.}

\subsectionl{Semimodules}

\note{I think I can remove semimodules from the discussion and use \ensuremath{\Varid{fmap}\;(\Varid{s}\;{}\mathbin{*})} in place of \ensuremath{(\Varid{scale}\;\Varid{s})}.
If so, do it.
One serious catch, however: when I introduce \ensuremath{\Varid{b}\leftarrow\Varid{a}}, I'll no longer have a functor in \ensuremath{\Varid{a}}.}


As fields are to vector spaces, rings are to modules, and semirings are to \emph{semimodules}.
For any semiring \ensuremath{\Varid{s}}, a \emph{left \ensuremath{\Varid{s}}-semimodule} \ensuremath{\Varid{b}} is a additive monoid whose values can be multiplied by \ensuremath{\Varid{s}} values on the left.
Here, \ensuremath{\Varid{s}} plays the role of ``scalars'', while \ensuremath{\Varid{b}} plays the role of ``vectors''.
\notefoot{Perhaps say just ``semimodule'', and add a remark that I really mean ``left semimodule'' throughout.
Or start out with ``left'', then make the remark, and then perhaps add an occasional ``(left)''.}
\begin{hscode}\SaveRestoreHook
\column{B}{@{}>{\hspre}l<{\hspost}@{}}%
\column{3}{@{}>{\hspre}l<{\hspost}@{}}%
\column{E}{@{}>{\hspre}l<{\hspost}@{}}%
\>[B]{}\mathbf{class}\;(\Conid{Semiring}\;\Varid{s},\Conid{Additive}\;\Varid{b})\Rightarrow \Conid{LeftSemimodule}\;\Varid{s}\;\Varid{b}\mid \Varid{b}\to \Varid{s}\;\mathbf{where}{}\<[E]%
\\
\>[B]{}\hsindent{3}{}\<[3]%
\>[3]{}(\cdot)\mathbin{::}\Varid{s}\to \Varid{b}\to \Varid{b}{}\<[E]%
\ColumnHook
\end{hscode}\resethooks
In addition to the laws for the additive monoid \ensuremath{\Varid{b}} and the semiring \ensuremath{\Varid{s}}, we have the following laws specific to left semimodules: \citep{Golan2005RecentSemi}:
\twocol{0.35}{
\begin{hscode}\SaveRestoreHook
\column{B}{@{}>{\hspre}l<{\hspost}@{}}%
\column{E}{@{}>{\hspre}l<{\hspost}@{}}%
\>[B]{}(\Varid{s}\mathbin{*}\Varid{t})\cdot\Varid{b}\mathrel{=}\Varid{s}\cdot(\Varid{t}\cdot\Varid{b}){}\<[E]%
\\
\>[B]{}(\Varid{s}\mathbin{+}\Varid{t})\cdot\Varid{b}\mathrel{=}\Varid{s}\cdot\Varid{b}\mathbin{+}\Varid{t}\cdot\Varid{b}{}\<[E]%
\\
\>[B]{}\Varid{s}\cdot(\Varid{b}\mathbin{+}\Varid{c})\mathrel{=}\Varid{s}\cdot\Varid{b}\mathbin{+}\Varid{s}\cdot\Varid{c}{}\<[E]%
\ColumnHook
\end{hscode}\resethooks
}{0.25}{
\begin{hscode}\SaveRestoreHook
\column{B}{@{}>{\hspre}l<{\hspost}@{}}%
\column{7}{@{}>{\hspre}l<{\hspost}@{}}%
\column{E}{@{}>{\hspre}l<{\hspost}@{}}%
\>[B]{}\mathrm{1}{}\<[7]%
\>[7]{}\cdot\Varid{b}\mathrel{=}\Varid{b}{}\<[E]%
\\
\>[B]{}\mathrm{0}{}\<[7]%
\>[7]{}\cdot\Varid{b}\mathrel{=}\mathrm{0}{}\<[E]%
\ColumnHook
\end{hscode}\resethooks
}
There is also a corresponding notion of \emph{right} \ensuremath{\Varid{s}}-semimodule (with multiplication on the right by \ensuremath{\Varid{s}} values), which we will not need in this paper.
(Rings also have left- and right-modules, and in \emph{commutative} rings and semirings (including vector spaces), the left and right variants coincide.)

As usual, we have a corresponding notion of homomorphism, which is more commonly referred to as ``linearity'':
\begin{definition} \deflabel{left semimodule homomorphism}
A function \ensuremath{\Varid{h}} from one left \ensuremath{\Varid{s}}-semimodule to another is a \emph{left \ensuremath{\Varid{s}}-semimodule homomorphism} if it is an additive monoid homomorphism (\defref{additive monoid homomorphism}) and satisfies the additional property \ensuremath{\Varid{h}\;(\Varid{s}\cdot\Varid{b})\mathrel{=}\Varid{s}\cdot\Varid{h}\;\Varid{b}}.
\end{definition}

Familiar \ensuremath{\Varid{s}}-semimodule examples include various containers of \ensuremath{\Varid{s}} values, including single- or multi-dimensional arrays, lists, infinite streams, sets, multisets, and trees.
Another, of particular interest in this paper, is functions from any type to any semiring:
\begin{hscode}\SaveRestoreHook
\column{B}{@{}>{\hspre}l<{\hspost}@{}}%
\column{E}{@{}>{\hspre}l<{\hspost}@{}}%
\>[B]{}\mathbf{instance}\;\Conid{LeftSemimodule}\;\Varid{s}\;(\Varid{a}\to \Varid{s})\;\mathbf{where}\;\Varid{s}\cdot\Varid{f}\mathrel{=}\lambda\, \Varid{a}\to \Varid{s}\mathbin{*}\Varid{f}\;\Varid{a}{}\<[E]%
\ColumnHook
\end{hscode}\resethooks
If we think of \ensuremath{\Varid{a}\to \Varid{s}} as a ``vector'' of \ensuremath{\Varid{s}} values, indexed by \ensuremath{\Varid{a}}, then \ensuremath{\Varid{s}\cdot\Varid{f}} scales each component of the vector \ensuremath{\Varid{f}} by \ensuremath{\Varid{s}}.

There is an important optimization to be made for scaling.
When \ensuremath{\Varid{s}\mathrel{=}\mathrm{0}}, \ensuremath{\Varid{s}\cdot\Varid{p}\mathrel{=}\mathrm{0}}, so we can discard \ensuremath{\Varid{p}} entirely.
This optimization applies quite often in practice, for instance with languages, which tend to be sparse.
Another optimization (though less dramatically helpful) is \ensuremath{\mathrm{1}\cdot\Varid{p}\mathrel{=}\Varid{p}}.
Rather than burden each \ensuremath{\Conid{LeftSemimodule}} instance with these two optimizations, let's define \ensuremath{(\cdot)} via a more primitive \ensuremath{(\mathbin{\hat{\cdot}})} method:
\begin{hscode}\SaveRestoreHook
\column{B}{@{}>{\hspre}l<{\hspost}@{}}%
\column{3}{@{}>{\hspre}l<{\hspost}@{}}%
\column{9}{@{}>{\hspre}l<{\hspost}@{}}%
\column{18}{@{}>{\hspre}l<{\hspost}@{}}%
\column{22}{@{}>{\hspre}l<{\hspost}@{}}%
\column{E}{@{}>{\hspre}l<{\hspost}@{}}%
\>[B]{}\mathbf{class}\;(\Conid{Semiring}\;\Varid{s},\Conid{Additive}\;\Varid{b})\Rightarrow \Conid{LeftSemimodule}\;\Varid{s}\;\Varid{b}\mid \Varid{b}\to \Varid{s}\;\mathbf{where}{}\<[E]%
\\
\>[B]{}\hsindent{3}{}\<[3]%
\>[3]{}(\mathbin{\hat{\cdot}})\mathbin{::}\Varid{s}\to \Varid{b}\to \Varid{b}{}\<[E]%
\\[\blanklineskip]%
\>[B]{}\mathbf{infixr}\;\mathrm{7}\cdot{}\<[E]%
\\
\>[B]{}(\cdot)\mathbin{::}(\Conid{Additive}\;\Varid{b},\Conid{LeftSemimodule}\;\Varid{s}\;\Varid{b},\Conid{IsZero}\;\Varid{s},\Conid{IsOne}\;\Varid{s})\Rightarrow \Varid{s}\to \Varid{b}\to \Varid{b}{}\<[E]%
\\
\>[B]{}\Varid{s}\cdot\Varid{b}{}\<[9]%
\>[9]{}\mid \Varid{isZero}\;\Varid{s}{}\<[22]%
\>[22]{}\mathrel{=}\mathrm{0}{}\<[E]%
\\
\>[9]{}\mid \Varid{isOne}\;{}\<[18]%
\>[18]{}\Varid{s}{}\<[22]%
\>[22]{}\mathrel{=}\Varid{b}{}\<[E]%
\\
\>[9]{}\mid \Varid{otherwise}{}\<[22]%
\>[22]{}\mathrel{=}\Varid{s}\mathbin{\hat{\cdot}}\Varid{b}{}\<[E]%
\ColumnHook
\end{hscode}\resethooks
The \ensuremath{\Conid{IsZero}} and \ensuremath{\Conid{IsOne}} classes:
\notefoot{Maybe use semiring-num again.}
\begin{hscode}\SaveRestoreHook
\column{B}{@{}>{\hspre}l<{\hspost}@{}}%
\column{17}{@{}>{\hspre}l<{\hspost}@{}}%
\column{38}{@{}>{\hspre}l<{\hspost}@{}}%
\column{54}{@{}>{\hspre}l<{\hspost}@{}}%
\column{E}{@{}>{\hspre}l<{\hspost}@{}}%
\>[B]{}\mathbf{class}\;\Conid{Additive}\;{}\<[17]%
\>[17]{}\Varid{b}\Rightarrow \Conid{IsZero}\;{}\<[38]%
\>[38]{}\Varid{b}\;\mathbf{where}\;\Varid{isZero}{}\<[54]%
\>[54]{}\mathbin{::}\Varid{b}\to \Conid{Bool}{}\<[E]%
\\
\>[B]{}\mathbf{class}\;\Conid{Semiring}\;{}\<[17]%
\>[17]{}\Varid{b}\Rightarrow \Conid{IsOne}\;{}\<[38]%
\>[38]{}\Varid{b}\;\mathbf{where}\;\Varid{isOne}{}\<[54]%
\>[54]{}\mathbin{::}\Varid{b}\to \Conid{Bool}{}\<[E]%
\ColumnHook
\end{hscode}\resethooks

As with star semirings (\lemref{affine over semiring}), recursive affine equations in semimodules \emph{over} star semirings also have solutions%
:
\begin{lemma}\lemlabel{affine over semimodule}
In a left semimodule over a star semiring, the affine equation \ensuremath{\Varid{p}\mathrel{=}\Varid{b}\mathbin{+}\Varid{m}\cdot\Varid{p}} has solution \ensuremath{\Varid{p}\mathrel{=}\closure{\Varid{m}}\cdot\Varid{b}}%
\end{lemma}
The proof closely resembles that of \lemref{affine over semiring}, using the left semimodule laws above:
\begin{proof}~
\begin{hscode}\SaveRestoreHook
\column{B}{@{}>{\hspre}c<{\hspost}@{}}%
\column{BE}{@{}l@{}}%
\column{5}{@{}>{\hspre}l<{\hspost}@{}}%
\column{41}{@{}>{\hspre}l<{\hspost}@{}}%
\column{E}{@{}>{\hspre}l<{\hspost}@{}}%
\>[5]{}\closure{\Varid{s}}\cdot\Varid{r}{}\<[E]%
\\
\>[B]{}\mathrel{=}{}\<[BE]%
\>[5]{}(\mathrm{1}\mathbin{+}\Varid{s}\mathbin{*}\closure{\Varid{s}})\cdot\Varid{r}{}\<[41]%
\>[41]{}\mbox{\onelinecomment  star semiring law}{}\<[E]%
\\
\>[B]{}\mathrel{=}{}\<[BE]%
\>[5]{}\mathrm{1}\cdot\Varid{r}\mathbin{+}(\Varid{s}\mathbin{*}\closure{\Varid{s}})\cdot\Varid{r}{}\<[41]%
\>[41]{}\mbox{\onelinecomment  distributivity}{}\<[E]%
\\
\>[B]{}\mathrel{=}{}\<[BE]%
\>[5]{}\Varid{r}\mathbin{+}\Varid{s}\cdot(\closure{\Varid{s}}\cdot\Varid{r}){}\<[41]%
\>[41]{}\mbox{\onelinecomment  multiplicative identity and associativity}{}\<[E]%
\ColumnHook
\end{hscode}\resethooks
\vspace{-4ex}
\end{proof}

\subsectionl{Function-like Types and Singletons}

Most of the representations used in this paper behave like functions, and it will be useful to use a standard vocabulary.
An ``indexable'' type \ensuremath{\Varid{x}} with domain \ensuremath{\Varid{a}} and codomain \ensuremath{\Varid{b}} represents \ensuremath{\Varid{a}\to \Varid{b}}:
Sometimes we will need to restrict \ensuremath{\Varid{a}} or \ensuremath{\Varid{b}}.
\begin{hscode}\SaveRestoreHook
\column{B}{@{}>{\hspre}l<{\hspost}@{}}%
\column{3}{@{}>{\hspre}l<{\hspost}@{}}%
\column{E}{@{}>{\hspre}l<{\hspost}@{}}%
\>[B]{}\mathbf{class}\;\Conid{Indexable}\;\Varid{a}\;\Varid{b}\;\Varid{x}\mid \Varid{x}\to \Varid{a}\;\Varid{b}\;\mathbf{where}{}\<[E]%
\\
\>[B]{}\hsindent{3}{}\<[3]%
\>[3]{}\mathbf{infixl}\;\mathrm{9}\mathbin{!}{}\<[E]%
\\
\>[B]{}\hsindent{3}{}\<[3]%
\>[3]{}(\mathbin{!})\mathbin{::}\Varid{x}\to \Varid{a}\to \Varid{b}{}\<[E]%
\\[\blanklineskip]%
\>[B]{}\mathbf{instance}\;\Conid{Indexable}\;\Varid{a}\;\Varid{b}\;(\Varid{a}\to \Varid{b})\;\mathbf{where}{}\<[E]%
\\
\>[B]{}\hsindent{3}{}\<[3]%
\>[3]{}\Varid{f}\mathbin{!}\Varid{k}\mathrel{=}\Varid{f}\;\Varid{k}{}\<[E]%
\ColumnHook
\end{hscode}\resethooks

\note{Add a law for \ensuremath{\Conid{Indexable}}: \ensuremath{(\mathbin{!})} must be natural?
Probably also that \ensuremath{\Varid{h}} maps \ensuremath{\Conid{Additive}} to \ensuremath{\Conid{Additive}} and that \ensuremath{(\mathbin{!})} is an \ensuremath{\Conid{Additive}} homomorphism.
Hm. It seems I can't even express those laws now that there's no functor.
}

\secrefs{Monoids}{Semimodules} provides a fair amount of vocabulary for combining values.
We'll also want an operation that constructs a ``vector'' (e.g., language or function) with a single nonzero component:
\begin{hscode}\SaveRestoreHook
\column{B}{@{}>{\hspre}l<{\hspost}@{}}%
\column{3}{@{}>{\hspre}l<{\hspost}@{}}%
\column{E}{@{}>{\hspre}l<{\hspost}@{}}%
\>[B]{}\mathbf{class}\;\Conid{Indexable}\;\Varid{a}\;\Varid{b}\;\Varid{x}\Rightarrow \Conid{HasSingle}\;\Varid{a}\;\Varid{b}\;\Varid{x}\;\mathbf{where}{}\<[E]%
\\
\>[B]{}\hsindent{3}{}\<[3]%
\>[3]{}\mathbf{infixr}\;\mathrm{2}\mapsto{}\<[E]%
\\
\>[B]{}\hsindent{3}{}\<[3]%
\>[3]{}(\mapsto)\mathbin{::}\Varid{a}\to \Varid{b}\to \Varid{x}{}\<[E]%
\\[\blanklineskip]%
\>[B]{}\mathbf{instance}\;(\Conid{Eq}\;\Varid{a},\Conid{Additive}\;\Varid{b})\Rightarrow \Conid{HasSingle}\;\Varid{a}\;\Varid{b}\;(\Varid{a}\to \Varid{b})\;\mathbf{where}{}\<[E]%
\\
\>[B]{}\hsindent{3}{}\<[3]%
\>[3]{}\Varid{a}\mapsto\Varid{b}\mathrel{=}\lambda\, \Varid{a'}\to \mathbf{if}\;\Varid{a'}\mathrel{=}\Varid{a}\;\mathbf{then}\;\Varid{b}\;\mathbf{else}\;\mathrm{0}{}\<[E]%
\ColumnHook
\end{hscode}\resethooks
Two specializations of \ensuremath{\Varid{a}\mapsto\Varid{b}} will come in handy: one for \ensuremath{\Varid{a}\mathrel{=}\varepsilon}, and the other for \ensuremath{\Varid{b}\mathrel{=}\mathrm{1}}.
\begin{hscode}\SaveRestoreHook
\column{B}{@{}>{\hspre}l<{\hspost}@{}}%
\column{E}{@{}>{\hspre}l<{\hspost}@{}}%
\>[B]{}\Varid{single}\mathbin{::}(\Conid{HasSingle}\;\Varid{a}\;\Varid{b}\;\Varid{x},\Conid{Semiring}\;\Varid{b})\Rightarrow \Varid{a}\to \Varid{x}{}\<[E]%
\\
\>[B]{}\Varid{single}\;\Varid{a}\mathrel{=}\Varid{a}\mapsto\mathrm{1}{}\<[E]%
\\[\blanklineskip]%
\>[B]{}\Varid{value}\mathbin{::}(\Conid{HasSingle}\;\Varid{a}\;\Varid{b}\;\Varid{x},\Conid{Monoid}\;\Varid{a})\Rightarrow \Varid{b}\to \Varid{x}{}\<[E]%
\\
\>[B]{}\Varid{value}\;\Varid{b}\mathrel{=}\varepsilon\mapsto\Varid{b}{}\<[E]%
\ColumnHook
\end{hscode}\resethooks
In particular, \ensuremath{\varepsilon\mapsto\mathrm{1}\mathrel{=}\Varid{single}\;\varepsilon\mathrel{=}\Varid{value}\;\mathrm{1}}.

The \ensuremath{(\mapsto)} operation gives a way to decompose arbitrary functions:
\begin{lemma}[\provedIn{lemma:decomp +->}]\lemlabel{decomp +->}
For all \ensuremath{\Varid{f}\mathbin{::}\Varid{a}\to \Varid{b}} where \ensuremath{\Varid{b}} is an additive monoid,
\begin{hscode}\SaveRestoreHook
\column{B}{@{}>{\hspre}l<{\hspost}@{}}%
\column{E}{@{}>{\hspre}l<{\hspost}@{}}%
\>[B]{}\Varid{f}\mathrel{=}\bigOp\sum{\Varid{a}}{0}\;\Varid{a}\mapsto\Varid{f}\;\Varid{a}{}\<[E]%
\ColumnHook
\end{hscode}\resethooks
\vspace{-3ex}
\end{lemma}
\noindent
For the uses in this paper, \ensuremath{\Varid{f}} is often ``sparse'', i.e., nonzero on a relatively small (e.g., finite or at least countable) subset of its domain.

Singletons also curry handily and provide another useful homomorphism:
\begin{lemma}[\provedIn{lemma:curry +->}]\lemlabel{curry +->}~
\begin{hscode}\SaveRestoreHook
\column{B}{@{}>{\hspre}l<{\hspost}@{}}%
\column{E}{@{}>{\hspre}l<{\hspost}@{}}%
\>[B]{}(\Varid{a}\mapsto\Varid{b}\mapsto\Varid{c})\mathrel{=}\Varid{curry}\;((\Varid{a},\Varid{b})\mapsto\Varid{c}){}\<[E]%
\ColumnHook
\end{hscode}\resethooks
\vspace{-4ex}
\end{lemma}
\begin{lemma} \lemlabel{+-> homomorphism}
For \ensuremath{(\to )\;\Varid{a}}, partial applications \ensuremath{(\Varid{a}\mapsto)} are left semi-module (and hence additive) homomorphisms.
Moreover, \ensuremath{\Varid{single}\mathrel{=}(\varepsilon\mapsto)} is a semiring homomorphism.
\end{lemma}
\begin{proof}
Straightforward from the definition of \ensuremath{(\mapsto)}.
\end{proof}

\sectionl{Calculating Instances from Homomorphisms}

So far, we've started with instance definitions and then noted and proved homomorphisms where they arise.
We can instead invert the process, taking homomorphisms as specifications and \emph{calculating} instance definitions that satisfy them.
This process of calculating instances from homomorphisms is the central guiding principle of this paper, so let's see how it works.

Consider a type ``\ensuremath{\Pow\;\Varid{a}}'' of mathematical \emph{sets} of values of some type \ensuremath{\Varid{a}}.
Are there useful instances of the abstractions from \secref{Monoids, Semirings and Semimodules} for sets?
Rather than guessing at such instances and then trying to prove the required laws, let's consider how sets are related to a type for which we already know instances, namely functions.

Sets are closely related to functions-to-booleans (``predicates''):
\twocol{0.4}{
\begin{hscode}\SaveRestoreHook
\column{B}{@{}>{\hspre}l<{\hspost}@{}}%
\column{E}{@{}>{\hspre}l<{\hspost}@{}}%
\>[B]{}\Varid{pred}\mathbin{::}\Pow\;\Varid{a}\to (\Varid{a}\to \Conid{Bool}){}\<[E]%
\\
\>[B]{}\Varid{pred}\;\Varid{as}\mathrel{=}\lambda\, \Varid{a}\to \Varid{a}\mathbin{\in}\Varid{as}{}\<[E]%
\ColumnHook
\end{hscode}\resethooks
}{0.4}{
\begin{hscode}\SaveRestoreHook
\column{B}{@{}>{\hspre}l<{\hspost}@{}}%
\column{E}{@{}>{\hspre}l<{\hspost}@{}}%
\>[B]{}\Varid{pred}^{-1}\mathbin{::}(\Varid{a}\to \Conid{Bool})\to \Pow\;\Varid{a}{}\<[E]%
\\
\>[B]{}\Varid{pred}^{-1}\;\Varid{f}\mathrel{=}\set{\Varid{a}\mid \Varid{f}\;\Varid{a}}{}\<[E]%
\ColumnHook
\end{hscode}\resethooks
}
This pair of functions forms an isomorphism, i.e., \ensuremath{\Varid{pred}^{-1}\hsdot{\circ }{.\:}\Varid{pred}\mathrel{=}\Varid{id}} and \ensuremath{\Varid{pred}\hsdot{\circ }{.\:}\Varid{pred}^{-1}\mathrel{=}\Varid{id}}, as can be checked by inlining definitions and simplifying.
Moreover, for sets \ensuremath{\Varid{p}} and \ensuremath{\Varid{q}}, \ensuremath{\Varid{p}\mathrel{=}\Varid{q}\Longleftrightarrow\Varid{pred}\;\Varid{p}\mathrel{=}\Varid{pred}\;\Varid{q}}, by the \emph{extensionality} axiom of sets and of functions.
Now let's also require that \ensuremath{\Varid{pred}} be an \emph{additive monoid homomorphism}.
The required homomorphism properties:
\begin{spacing}{1.2}
\begin{hscode}\SaveRestoreHook
\column{B}{@{}>{\hspre}l<{\hspost}@{}}%
\column{E}{@{}>{\hspre}l<{\hspost}@{}}%
\>[B]{}\Varid{pred}\;\mathrm{0}\mathrel{=}\mathrm{0}{}\<[E]%
\\
\>[B]{}\Varid{pred}\;(\Varid{p}\mathbin{+}\Varid{q})\mathrel{=}\Varid{pred}\;\Varid{p}\mathbin{+}\Varid{pred}\;\Varid{q}{}\<[E]%
\ColumnHook
\end{hscode}\resethooks
\end{spacing}\noindent
We already know definitions of \ensuremath{\Varid{pred}} as well as the function versions of \ensuremath{\mathrm{0}} and \ensuremath{(\mathbin{+})} (on the RHS) but not yet the set versions of \ensuremath{\mathrm{0}} and \ensuremath{(\mathbin{+})} (on the LHS).
We thus have two algebra problems in two unknowns.
Since only one unknown appears in each homomorphism equation, we can solve them independently.
The \ensuremath{\Varid{pred}}/\ensuremath{\Varid{pred}^{-1}} isomorphism makes it easy to solve these equations, and removes all semantic choice, allowing only varying implementations of the same meaning.
\begin{hscode}\SaveRestoreHook
\column{B}{@{}>{\hspre}c<{\hspost}@{}}%
\column{BE}{@{}l@{}}%
\column{6}{@{}>{\hspre}l<{\hspost}@{}}%
\column{68}{@{}>{\hspre}l<{\hspost}@{}}%
\column{E}{@{}>{\hspre}l<{\hspost}@{}}%
\>[6]{}\Varid{pred}\;\mathrm{0}\mathrel{=}\mathrm{0}{}\<[E]%
\\
\>[B]{}\Longleftrightarrow{}\<[BE]%
\>[6]{}\Varid{pred}^{-1}\;(\Varid{pred}\;\mathrm{0})\mathrel{=}\Varid{pred}^{-1}\;\mathrm{0}{}\<[68]%
\>[68]{}\mbox{\onelinecomment  \ensuremath{\Varid{pred}^{-1}} injectivity}{}\<[E]%
\\
\>[B]{}\Longleftrightarrow{}\<[BE]%
\>[6]{}\mathrm{0}\mathrel{=}\Varid{pred}^{-1}\;\mathrm{0}{}\<[68]%
\>[68]{}\mbox{\onelinecomment  \ensuremath{\Varid{pred}^{-1}\hsdot{\circ }{.\:}\Varid{pred}\mathrel{=}\Varid{id}}}{}\<[E]%
\\[\blanklineskip]%
\>[6]{}\Varid{pred}\;(\Varid{p}\mathbin{+}\Varid{q})\mathrel{=}\Varid{pred}\;\Varid{p}\mathbin{+}\Varid{pred}\;\Varid{q}{}\<[E]%
\\
\>[B]{}\Longleftrightarrow{}\<[BE]%
\>[6]{}\Varid{pred}^{-1}\;(\Varid{pred}\;(\Varid{p}\mathbin{+}\Varid{q}))\mathrel{=}\Varid{pred}^{-1}\;(\Varid{pred}\;\Varid{p}\mathbin{+}\Varid{pred}\;\Varid{q}){}\<[68]%
\>[68]{}\mbox{\onelinecomment  \ensuremath{\Varid{pred}^{-1}} injectivity}{}\<[E]%
\\
\>[B]{}\Longleftrightarrow{}\<[BE]%
\>[6]{}\Varid{p}\mathbin{+}\Varid{q}\mathrel{=}\Varid{pred}^{-1}\;(\Varid{pred}\;\Varid{p}\mathbin{+}\Varid{pred}\;\Varid{q}){}\<[68]%
\>[68]{}\mbox{\onelinecomment  \ensuremath{\Varid{pred}^{-1}\hsdot{\circ }{.\:}\Varid{pred}\mathrel{=}\Varid{id}}}{}\<[E]%
\ColumnHook
\end{hscode}\resethooks
We thus have sufficient (and in this case semantically necessary) definitions for \ensuremath{\mathrm{0}} and \ensuremath{(\mathbin{+})} on sets.
Now let's simplify to get more direct definitions:
\begin{hscode}\SaveRestoreHook
\column{B}{@{}>{\hspre}c<{\hspost}@{}}%
\column{BE}{@{}l@{}}%
\column{5}{@{}>{\hspre}l<{\hspost}@{}}%
\column{50}{@{}>{\hspre}l<{\hspost}@{}}%
\column{E}{@{}>{\hspre}l<{\hspost}@{}}%
\>[5]{}\Varid{pred}^{-1}\;\mathrm{0}{}\<[E]%
\\
\>[B]{}\mathrel{=}{}\<[BE]%
\>[5]{}\Varid{pred}^{-1}\;(\lambda\, \Varid{a}\to \mathrm{0}){}\<[50]%
\>[50]{}\mbox{\onelinecomment  \ensuremath{\mathrm{0}} on functions}{}\<[E]%
\\
\>[B]{}\mathrel{=}{}\<[BE]%
\>[5]{}\Varid{pred}^{-1}\;(\lambda\, \Varid{a}\to \Conid{False}){}\<[50]%
\>[50]{}\mbox{\onelinecomment  \ensuremath{\mathrm{0}} on \ensuremath{\Conid{Bool}}}{}\<[E]%
\\
\>[B]{}\mathrel{=}{}\<[BE]%
\>[5]{}\set{\Varid{a}\mid\Conid{False}}{}\<[50]%
\>[50]{}\mbox{\onelinecomment  \ensuremath{\Varid{pred}^{-1}} definition}{}\<[E]%
\\
\>[B]{}\mathrel{=}{}\<[BE]%
\>[5]{}\emptyset{}\<[E]%
\\[\blanklineskip]%
\>[5]{}\Varid{pred}^{-1}\;(\Varid{pred}\;\Varid{p}\mathbin{+}\Varid{pred}\;\Varid{q}){}\<[E]%
\\
\>[B]{}\mathrel{=}{}\<[BE]%
\>[5]{}\Varid{pred}^{-1}\;((\lambda\, \Varid{a}\to \Varid{a}\mathbin{\in}\Varid{p})\mathbin{+}(\lambda\, \Varid{a}\to \Varid{a}\mathbin{\in}\Varid{q})){}\<[50]%
\>[50]{}\mbox{\onelinecomment  \ensuremath{\Varid{pred}} definition (twice)}{}\<[E]%
\\
\>[B]{}\mathrel{=}{}\<[BE]%
\>[5]{}\Varid{pred}^{-1}\;(\lambda\, \Varid{a}\to (\Varid{a}\mathbin{\in}\Varid{p})\mathbin{+}(\Varid{a}\mathbin{\in}\Varid{q})){}\<[50]%
\>[50]{}\mbox{\onelinecomment  \ensuremath{(\mathbin{+})} on functions}{}\<[E]%
\\
\>[B]{}\mathrel{=}{}\<[BE]%
\>[5]{}\Varid{pred}^{-1}\;(\lambda\, \Varid{a}\to \Varid{a}\mathbin{\in}\Varid{p}\mathrel{\vee}\Varid{a}\mathbin{\in}\Varid{q}){}\<[50]%
\>[50]{}\mbox{\onelinecomment  \ensuremath{(\mathbin{+})} on \ensuremath{\Conid{Bool}}}{}\<[E]%
\\
\>[B]{}\mathrel{=}{}\<[BE]%
\>[5]{}\set{\Varid{a}\mid\Varid{a}\mathbin{\in}\Varid{p}\mathrel{\vee}\Varid{a}\mathbin{\in}\Varid{q}}{}\<[50]%
\>[50]{}\mbox{\onelinecomment  \ensuremath{\Varid{pred}^{-1}} definition}{}\<[E]%
\\
\>[B]{}\mathrel{=}{}\<[BE]%
\>[5]{}\Varid{p}\cup\Varid{q}{}\<[50]%
\>[50]{}\mbox{\onelinecomment  \ensuremath{(\cup)} definition}{}\<[E]%
\ColumnHook
\end{hscode}\resethooks
Without applying any real creativity, we have discovered the desired \ensuremath{\Conid{Semiring}} instance for sets:
\begin{hscode}\SaveRestoreHook
\column{B}{@{}>{\hspre}l<{\hspost}@{}}%
\column{3}{@{}>{\hspre}l<{\hspost}@{}}%
\column{9}{@{}>{\hspre}l<{\hspost}@{}}%
\column{E}{@{}>{\hspre}l<{\hspost}@{}}%
\>[B]{}\mathbf{instance}\;\Conid{Additive}\;(\Pow\;\Varid{a})\;\mathbf{where}{}\<[E]%
\\
\>[B]{}\hsindent{3}{}\<[3]%
\>[3]{}\mathrm{0}{}\<[9]%
\>[9]{}\mathrel{=}\emptyset{}\<[E]%
\\
\>[B]{}\hsindent{3}{}\<[3]%
\>[3]{}(\mathbin{+}){}\<[9]%
\>[9]{}\mathrel{=}(\cup){}\<[E]%
\ColumnHook
\end{hscode}\resethooks

Next consider a \ensuremath{\Conid{LeftSemimodule}} instance for sets.
We might be tempted to define \ensuremath{\Varid{s}\cdot\Varid{p}} to multiply \ensuremath{\Varid{s}} by each value in \ensuremath{\Varid{p}}, i.e.,
\begin{hscode}\SaveRestoreHook
\column{B}{@{}>{\hspre}l<{\hspost}@{}}%
\column{79}{@{}>{\hspre}l<{\hspost}@{}}%
\column{E}{@{}>{\hspre}l<{\hspost}@{}}%
\>[B]{}\mathbf{instance}\;\Conid{LeftSemimodule}\;\Varid{s}\;(\Pow\;\Varid{s})\;\mathbf{where}\;\Varid{s}\mathbin{\hat{\cdot}}\Varid{p}\mathrel{=}\set{\Varid{s}\mathbin{*}\Varid{x}\mid \Varid{x}\mathbin{\in}\Varid{p}}{}\<[79]%
\>[79]{}\mbox{\onelinecomment  \emph{wrong}}{}\<[E]%
\ColumnHook
\end{hscode}\resethooks
This definition, however, would violate the semimodule law that \ensuremath{\mathrm{0}\cdot\Varid{p}\mathrel{=}\mathrm{0}}, since \ensuremath{\mathrm{0}\cdot\Varid{p}} would be \ensuremath{\set{\mathrm{0}}}, but \ensuremath{\mathrm{0}} for sets is \ensuremath{\emptyset}.
Both semimodule distributive laws fail as well.
There is an alternative choice, necessitated by requiring that \ensuremath{\Varid{pred}^{-1}} be a left \ensuremath{\Conid{Bool}}-semimodule homomorphism.
The choice of \ensuremath{\Conid{Bool}} is inevitable from the type of \ensuremath{\Varid{pred}^{-1}} and the fact that \ensuremath{\Varid{a}\to \Varid{b}} is a \ensuremath{\Varid{b}}-semimodule for all semirings \ensuremath{\Varid{b}}, so \ensuremath{\Varid{a}\to \Conid{Bool}} is a \ensuremath{\Conid{Bool}}-semimodule.
The necessary homomorphism property:
\begin{hscode}\SaveRestoreHook
\column{B}{@{}>{\hspre}l<{\hspost}@{}}%
\column{E}{@{}>{\hspre}l<{\hspost}@{}}%
\>[B]{}\Varid{pred}\;(\Varid{s}\cdot\Varid{p})\mathrel{=}\Varid{s}\cdot\Varid{pred}\;\Varid{p}{}\<[E]%
\ColumnHook
\end{hscode}\resethooks
Equivalently,
\begin{spacing}{1.2}
\vspace{-1ex}
\begin{hscode}\SaveRestoreHook
\column{B}{@{}>{\hspre}c<{\hspost}@{}}%
\column{BE}{@{}l@{}}%
\column{5}{@{}>{\hspre}l<{\hspost}@{}}%
\column{65}{@{}>{\hspre}l<{\hspost}@{}}%
\column{E}{@{}>{\hspre}l<{\hspost}@{}}%
\>[5]{}\Varid{s}\cdot\Varid{p}{}\<[E]%
\\
\>[B]{}\mathrel{=}{}\<[BE]%
\>[5]{}\Varid{pred}^{-1}\;(\Varid{s}\cdot\Varid{pred}\;\Varid{p}){}\<[65]%
\>[65]{}\mbox{\onelinecomment  \ensuremath{\Varid{pred}^{-1}} injectivity}{}\<[E]%
\\
\>[B]{}\mathrel{=}{}\<[BE]%
\>[5]{}\Varid{pred}^{-1}\;(\Varid{s}\cdot(\lambda\, \Varid{a}\to \Varid{a}\mathbin{\in}\Varid{p})){}\<[65]%
\>[65]{}\mbox{\onelinecomment  \ensuremath{\Varid{pred}} definition}{}\<[E]%
\\
\>[B]{}\mathrel{=}{}\<[BE]%
\>[5]{}\Varid{pred}^{-1}\;(\lambda\, \Varid{a}\to \Varid{s}\mathbin{*}(\Varid{a}\mathbin{\in}\Varid{p})){}\<[65]%
\>[65]{}\mbox{\onelinecomment  \ensuremath{(\cdot)} on functions}{}\<[E]%
\\
\>[B]{}\mathrel{=}{}\<[BE]%
\>[5]{}\Varid{pred}^{-1}\;(\lambda\, \Varid{a}\to \Varid{s}\mathrel{\wedge}\Varid{a}\mathbin{\in}\Varid{p}){}\<[65]%
\>[65]{}\mbox{\onelinecomment  \ensuremath{(\mathbin{*})} on \ensuremath{\Conid{Bool}}}{}\<[E]%
\\
\>[B]{}\mathrel{=}{}\<[BE]%
\>[5]{}\set{\Varid{a}\mid\Varid{s}\mathrel{\wedge}\Varid{a}\mathbin{\in}\Varid{p}}{}\<[65]%
\>[65]{}\mbox{\onelinecomment  \ensuremath{\Varid{pred}^{-1}} definition}{}\<[E]%
\\
\>[B]{}\mathrel{=}{}\<[BE]%
\>[5]{}\mathbf{if}\;\Varid{s}\;\mathbf{then}\;\set{\Varid{a}\mid\Varid{s}\mathrel{\wedge}\Varid{a}\mathbin{\in}\Varid{p}}\;\mathbf{else}\;\set{\Varid{a}\mid\Varid{s}\mathrel{\wedge}\Varid{a}\mathbin{\in}\Varid{p}}{}\<[65]%
\>[65]{}\mbox{\onelinecomment  property of \ensuremath{\mathbf{if}}}{}\<[E]%
\\
\>[B]{}\mathrel{=}{}\<[BE]%
\>[5]{}\mathbf{if}\;\Varid{s}\;\mathbf{then}\;\set{\Varid{a}\mid\Varid{a}\mathbin{\in}\Varid{p}}\;\mathbf{else}\;\emptyset{}\<[65]%
\>[65]{}\mbox{\onelinecomment  simplify conditional branches}{}\<[E]%
\\
\>[B]{}\mathrel{=}{}\<[BE]%
\>[5]{}\mathbf{if}\;\Varid{s}\;\mathbf{then}\;\Varid{p}\;\mathbf{else}\;\emptyset{}\<[65]%
\>[65]{}\mbox{\onelinecomment  \ensuremath{\Varid{pred}^{-1}\hsdot{\circ }{.\:}\Varid{pred}\mathrel{=}\Varid{id}}}{}\<[E]%
\\
\>[B]{}\mathrel{=}{}\<[BE]%
\>[5]{}\mathbf{if}\;\Varid{s}\;\mathbf{then}\;\Varid{p}\;\mathbf{else}\;\mathrm{0}{}\<[65]%
\>[65]{}\mbox{\onelinecomment  \ensuremath{\mathrm{0}} for sets}{}\<[E]%
\ColumnHook
\end{hscode}\resethooks
\end{spacing}
\noindent
Summarizing,
\begin{hscode}\SaveRestoreHook
\column{B}{@{}>{\hspre}l<{\hspost}@{}}%
\column{3}{@{}>{\hspre}l<{\hspost}@{}}%
\column{E}{@{}>{\hspre}l<{\hspost}@{}}%
\>[B]{}\mathbf{instance}\;\Conid{LeftSemimodule}\;\Conid{Bool}\;(\Pow\;\Varid{a})\;\mathbf{where}{}\<[E]%
\\
\>[B]{}\hsindent{3}{}\<[3]%
\>[3]{}\Varid{s}\mathbin{\hat{\cdot}}\Varid{p}\mathrel{=}\mathbf{if}\;\Varid{s}\;\mathbf{then}\;\Varid{p}\;\mathbf{else}\;\mathrm{0}{}\<[E]%
\ColumnHook
\end{hscode}\resethooks
While perhaps obscure at first, this alternative will prove useful later on.

Note that the left \ensuremath{\Varid{s}}-semimodule laws specialized to \ensuremath{\Varid{s}\mathrel{=}\Conid{Bool}} require \ensuremath{\Conid{True}} (\ensuremath{\mathrm{1}}) to preserve and \ensuremath{\Conid{False}} (\ensuremath{\mathrm{0}}) to annihilate the second \ensuremath{(\cdot)} argument.
\emph{Every} left \ensuremath{\Conid{Bool}}-semimodule instance must therefore agree with this definition.\out{ Also note that \ensuremath{\forall \Varid{a}\hsforall \hsdot{\circ }{.\:}(\Varid{a}\mathbin{\in}\Varid{s}\cdot\Varid{p})\Longleftrightarrow(\Varid{s}\mathrel{\wedge}\Varid{a}\mathbin{\in}\Varid{p})}, which resembles the \ensuremath{\Conid{LeftSemimodule}\;(\Varid{a}\to \Varid{b})} instance given above.}

\note{Demonstrate that homomorphic specifications also guarantee that laws hold, assuming that equality is consistent with homomorphism.}

\sectionl{Languages and the Monoid Semiring}

A \emph{language} is a set of strings over some alphabet, so the \ensuremath{\Conid{Additive}} and \ensuremath{\Conid{LeftSemimodule}} instances for sets given above apply directly.
Conspicuously missing, however, are the usual notions of language concatenation and closure (Kleene star), which are defined as follows for languages \ensuremath{\Conid{U}} and \ensuremath{\Conid{V}}:
\begin{hscode}\SaveRestoreHook
\column{B}{@{}>{\hspre}l<{\hspost}@{}}%
\column{E}{@{}>{\hspre}l<{\hspost}@{}}%
\>[B]{}\Conid{U}\;\Conid{V}\mathrel{=}\set{\Varid{u} \diamond \Varid{v}\mid \Varid{u}\mathbin{\in}\Conid{U}\mathrel{\wedge}\Varid{v}\mathbin{\in}\Conid{V}}{}\<[E]%
\\[\blanklineskip]%
\>[B]{}\closure{\Conid{U}}\mathrel{=}\bigOp\bigcup{\Varid{i}}{0}{\,}\Conid{U}^\Varid{i}\mbox{\onelinecomment  where \ensuremath{\Conid{U}^\mathrm{0}\mathrel{=}\mathrm{1}}, and \ensuremath{\Conid{U}^{\Varid{n}\mathbin{+}\mathrm{1}}\mathrel{=}\Conid{U}{\,}\Conid{U}^\Varid{n}}.}{}\<[E]%
\ColumnHook
\end{hscode}\resethooks
Intriguingly, this \ensuremath{\closure{\Conid{U}}} definition would satisfy the \ensuremath{\Conid{StarSemiring}} laws if \ensuremath{(\mathbin{*})} were language concatenation.
A bit of reasoning shows that all of the semiring laws would hold as well:
\begin{itemize}
\item Concatenation is associative and has as identity the language \ensuremath{\set{\varepsilon}}.
\item Concatenation distributes over union, both from the left and from the right.
\item The \ensuremath{\mathrm{0}} (empty) language annihilates (yields \ensuremath{\mathrm{0}}) under concatenation, both from the left and from the right.
\end{itemize}
All we needed from strings is that they form a monoid, so we may as well generalize:
\begin{hscode}\SaveRestoreHook
\column{B}{@{}>{\hspre}l<{\hspost}@{}}%
\column{3}{@{}>{\hspre}l<{\hspost}@{}}%
\column{E}{@{}>{\hspre}l<{\hspost}@{}}%
\>[B]{}\mathbf{instance}\;\Conid{Monoid}\;\Varid{a}\Rightarrow \Conid{Semiring}\;(\Conid{P}\;\Varid{a})\;\mathbf{where}{}\<[E]%
\\
\>[B]{}\hsindent{3}{}\<[3]%
\>[3]{}\mathrm{1}\mathrel{=}\set{\varepsilon}\mbox{\onelinecomment  \ensuremath{\mathrel{=}\varepsilon\mapsto\mathrm{1}\mathrel{=}\Varid{single}\;\varepsilon\mathrel{=}\Varid{value}\;\mathrm{1}} (\secref{Function-like Types and Singletons})}{}\<[E]%
\\
\>[B]{}\hsindent{3}{}\<[3]%
\>[3]{}\Varid{p}\mathbin{*}\Varid{q}\mathrel{=}\set{\Varid{u} \diamond \Varid{v}\mid\Varid{u}\mathbin{\in}\Varid{p}\mathrel{\wedge}\Varid{v}\mathbin{\in}\Varid{q}}{}\<[E]%
\\[\blanklineskip]%
\>[B]{}\mathbf{instance}\;\Conid{StarSemiring}\;(\Pow\;\Varid{a})\mbox{\onelinecomment  use default \ensuremath{\closure{\cdot }} definition (\secref{Star Semirings}).}{}\<[E]%
\ColumnHook
\end{hscode}\resethooks

\noindent
These new instances indeed satisfy the laws for additive monoids, semimodules, semirings, and star semirings.
They seem to spring from nothing, however, which is disappointing compared with the way the \ensuremath{\Conid{Additive}} and \ensuremath{\Conid{LeftSemimodule}} instances for sets follow inevitably from the requirement that \ensuremath{\Varid{pred}} be a homomorphism for those classes (\secref{Calculating Instances from Homomorphisms}).
Let's not give up yet, however.
Perhaps there's a \ensuremath{\Conid{Semiring}} instance for \ensuremath{\Varid{a}\to \Varid{b}} that specializes with \ensuremath{\Varid{b}\mathrel{=}\Conid{Bool}} to bear the same relationship to \ensuremath{\Pow\;\Varid{a}} that the \ensuremath{\Conid{Additive}} and \ensuremath{\Conid{LeftSemimodule}} instances bear.
The least imaginative thing we can try is to require that \ensuremath{\Varid{pred}} be a \emph{semiring} homomorphism.
If we apply the same sort of reasoning as in \secref{Calculating Instances from Homomorphisms} and then generalize from \ensuremath{\Conid{Bool}} to an arbitrary semiring, we get the definitions in \figrefdef{monoid semiring}{The monoid semiring}{
\begin{hscode}\SaveRestoreHook
\column{B}{@{}>{\hspre}l<{\hspost}@{}}%
\column{3}{@{}>{\hspre}l<{\hspost}@{}}%
\column{10}{@{}>{\hspre}l<{\hspost}@{}}%
\column{34}{@{}>{\hspre}l<{\hspost}@{}}%
\column{60}{@{}>{\hspre}l<{\hspost}@{}}%
\column{E}{@{}>{\hspre}l<{\hspost}@{}}%
\>[B]{}\mathbf{instance}\;(\Conid{Semiring}\;\Varid{b},\Conid{Monoid}\;\Varid{a})\Rightarrow \Conid{Semiring}\;(\Varid{a}\to \Varid{b})\;\mathbf{where}{}\<[E]%
\\
\>[B]{}\hsindent{3}{}\<[3]%
\>[3]{}\mathrm{1}\mathrel{=}\Varid{single}\;\varepsilon{}\<[E]%
\\
\>[B]{}\hsindent{3}{}\<[3]%
\>[3]{}\Varid{f}\mathbin{*}\Varid{g}{}\<[10]%
\>[10]{}\mathrel{=}\bigOp\sum{\Varid{u},\Varid{v}}{0}\;\Varid{u} \diamond \Varid{v}\mapsto\Varid{f}\;\Varid{u}\mathbin{*}\Varid{g}\;\Varid{v}{}\<[E]%
\\
\>[10]{}\mathrel{=}\lambda\, \Varid{w}\to \bigOp\sum{\Varid{u},\Varid{v}\;\!\!\\\!\!\;\Varid{u} \diamond \Varid{v}\mathrel{=}\Varid{w}}{1}\;\Varid{f}\;\Varid{u}\mathbin{*}\Varid{g}\;\Varid{v}{}\<[E]%
\\[\blanklineskip]%
\>[B]{}\mathbf{instance}\;(\Conid{Semiring}\;\Varid{b},\Conid{Monoid}\;\Varid{a}){}\<[34]%
\>[34]{}\Rightarrow \Conid{StarSemiring}\;(\Varid{a}\to \Varid{b}){}\<[60]%
\>[60]{}\mbox{\onelinecomment  default \ensuremath{\closure{\cdot }}}{}\<[E]%
\ColumnHook
\end{hscode}\resethooks
\vspace{-4ex}
}.
With this instance, \ensuremath{\Varid{a}\to \Varid{b}} type is known as the \emph{monoid semiring}, and its \ensuremath{(\mathbin{*})} operation as \emph{convolution} \citep{Golan2005RecentSemi,Wilding2015LAS}.

\begin{theorem}[\provedIn{theorem:semiring hom ->}]\thmlabel{semiring hom ->}
Given the instance definitions in \figref{monoid semiring}, \ensuremath{\Varid{pred}} is a star semiring homomorphism.
\end{theorem}

For some monoids, we can also express the product operation in a more clearly computable form via \emph{splittings}:
\begin{hscode}\SaveRestoreHook
\column{B}{@{}>{\hspre}l<{\hspost}@{}}%
\column{3}{@{}>{\hspre}l<{\hspost}@{}}%
\column{E}{@{}>{\hspre}l<{\hspost}@{}}%
\>[3]{}\Varid{f}\mathbin{*}\Varid{g}\mathrel{=}\lambda\, \Varid{w}\to \bigOp\sum{(\Varid{u},\Varid{v})\mathbin{\in}\Varid{splits}\;\Varid{w}}{2.5}\;\Varid{f}\;\Varid{u}\mathbin{*}\Varid{g}\;\Varid{v}{}\<[E]%
\ColumnHook
\end{hscode}\resethooks
where \ensuremath{\Varid{splits}\;\Varid{w}} yields all pairs \ensuremath{(\Varid{u},\Varid{v})} such that \ensuremath{\Varid{u} \diamond \Varid{v}\mathrel{=}\Varid{w}}:
\begin{hscode}\SaveRestoreHook
\column{B}{@{}>{\hspre}l<{\hspost}@{}}%
\column{3}{@{}>{\hspre}l<{\hspost}@{}}%
\column{12}{@{}>{\hspre}l<{\hspost}@{}}%
\column{29}{@{}>{\hspre}l<{\hspost}@{}}%
\column{E}{@{}>{\hspre}l<{\hspost}@{}}%
\>[B]{}\mathbf{class}\;\Conid{Monoid}\;\Varid{t}\Rightarrow \Conid{Splittable}\;\Varid{t}\;\mathbf{where}{}\<[E]%
\\
\>[B]{}\hsindent{3}{}\<[3]%
\>[3]{}\Varid{splits}{}\<[12]%
\>[12]{}\mathbin{::}\Varid{t}\to [\mskip1.5mu (\Varid{t},\Varid{t})\mskip1.5mu]{}\<[29]%
\>[29]{}\mbox{\onelinecomment  multi-valued inverse of \ensuremath{(\diamond)}}{}\<[E]%
\ColumnHook
\end{hscode}\resethooks
Examples of splittable monoids include natural numbers and lists%
:
\begin{hscode}\SaveRestoreHook
\column{B}{@{}>{\hspre}l<{\hspost}@{}}%
\column{3}{@{}>{\hspre}l<{\hspost}@{}}%
\column{18}{@{}>{\hspre}l<{\hspost}@{}}%
\column{E}{@{}>{\hspre}l<{\hspost}@{}}%
\>[B]{}\mathbf{instance}\;\Conid{Splittable}\;\mathbb N\;\mathbf{where}{}\<[E]%
\\
\>[B]{}\hsindent{3}{}\<[3]%
\>[3]{}\Varid{splits}\;\Varid{n}\mathrel{=}[\mskip1.5mu (\Varid{i},\Varid{n}\mathbin{-}\Varid{i})\mid \Varid{i}\leftarrow [\mskip1.5mu \mathrm{0}\mathinner{\ldotp\ldotp}\Varid{n}\mskip1.5mu]\mskip1.5mu]{}\<[E]%
\\[\blanklineskip]%
\>[B]{}\mathbf{instance}\;\Conid{Splittable}\;[\mskip1.5mu \Varid{c}\mskip1.5mu]\;\mathbf{where}{}\<[E]%
\\
\>[B]{}\hsindent{3}{}\<[3]%
\>[3]{}\Varid{splits}\;[\mskip1.5mu \mskip1.5mu]{}\<[18]%
\>[18]{}\mathrel{=}[\mskip1.5mu ([\mskip1.5mu \mskip1.5mu],[\mskip1.5mu \mskip1.5mu])\mskip1.5mu]{}\<[E]%
\\
\>[B]{}\hsindent{3}{}\<[3]%
\>[3]{}\Varid{splits}\;(\Varid{c}\mathbin{:}\Varid{cs}){}\<[18]%
\>[18]{}\mathrel{=}([\mskip1.5mu \mskip1.5mu],\Varid{c}\mathbin{:}\Varid{cs})\mathbin{:}[\mskip1.5mu ((\Varid{c}\mathbin{:}\Varid{l}),\Varid{r})\mid (\Varid{l},\Varid{r})\leftarrow \Varid{splits}\;\Varid{cs}\mskip1.5mu]{}\<[E]%
\ColumnHook
\end{hscode}\resethooks

While simple, general, and (assuming \ensuremath{\Conid{Splittable}} domain) computable, the definitions of \ensuremath{(\mathbin{+})} and \ensuremath{(\mathbin{*})} above for the monoid semiring make for quite inefficient implementations, primarily due to naive backtracking.
As a simple example, consider the language \ensuremath{\Varid{single}\;\text{\tt \char34 pickles\char34}\mathbin{+}\Varid{single}\;\text{\tt \char34 pickled\char34}}, and suppose that we want to test the word ``pickling'' for membership.
The \ensuremath{(\mathbin{+})} definition from \secref{Additive Monoids} will first try ``pickles'', fail near the end, and then backtrack all the way to the beginning to try ``pickled''.
The second attempt redundantly discovers that the prefix ``pickl'' is also a prefix of the candidate word and that ``pickle'' is not.
Next consider the language \ensuremath{\Varid{single}\;\text{\tt \char34 ca\char34}\mathbin{*}\Varid{single}\;\text{\tt \char34 ts\char34}\mathbin{*}\Varid{single}\;\text{\tt \char34 up\char34}}, and suppose we want to test ``catsup'' for membership.
The \ensuremath{(\mathbin{*})} implementation above will try all possible three-way splittings of the test string.

\sectionl{Finite maps}


\note{I don't think finite maps need their own section. Look for another home. Maybe with \ensuremath{\Conid{Cofree}} as a suggested functor.}

One representation of \emph{partial} functions is the type of finite maps, \ensuremath{\Conid{Map}\;\Varid{a}\;\Varid{b}} from keys of type \ensuremath{\Varid{a}} to values of type \ensuremath{\Varid{b}}, represented is a key-ordered balanced tree \citep{Adams1993Sets,Straka2012ATR,Nievergelt1973BST}.
To model \emph{total} functions instead, we can treat unassigned keys as denoting zero.
Conversely, merging two finite maps can yield a key collision, which can be resolved by addition.
Both interpretations require \ensuremath{\Varid{b}} to be an additive monoid.
Given the definitions in \figrefdef{Map}{Finite maps}{
\begin{hscode}\SaveRestoreHook
\column{B}{@{}>{\hspre}l<{\hspost}@{}}%
\column{3}{@{}>{\hspre}l<{\hspost}@{}}%
\column{E}{@{}>{\hspre}l<{\hspost}@{}}%
\>[B]{}\mathbf{instance}\;(\Conid{Ord}\;\Varid{a},\Conid{Additive}\;\Varid{b})\Rightarrow \Conid{Indexable}\;\Varid{a}\;\Varid{b}\;(\Conid{Map}\;\Varid{a}\;\Varid{b})\;\mathbf{where}{}\<[E]%
\\
\>[B]{}\hsindent{3}{}\<[3]%
\>[3]{}\Varid{m}\mathbin{!}\Varid{a}\mathrel{=}\Conid{M}\!.\Varid{findWithDefault}\;\mathrm{0}\;\Varid{a}\;\Varid{m}{}\<[E]%
\\[\blanklineskip]%
\>[B]{}\mathbf{instance}\;(\Conid{Ord}\;\Varid{a},\Conid{Additive}\;\Varid{b})\Rightarrow \Conid{HasSingle}\;\Varid{a}\;\Varid{b}\;(\Conid{Map}\;\Varid{a}\;\Varid{b})\;\mathbf{where}{}\<[E]%
\\
\>[B]{}\hsindent{3}{}\<[3]%
\>[3]{}(\mapsto)\mathrel{=}\Conid{M}\!.\Varid{singleton}{}\<[E]%
\\[\blanklineskip]%
\>[B]{}\mathbf{instance}\;(\Conid{Ord}\;\Varid{a},\Conid{Additive}\;\Varid{b})\Rightarrow \Conid{Additive}\;(\Conid{Map}\;\Varid{a}\;\Varid{b})\;\mathbf{where}{}\<[E]%
\\
\>[B]{}\hsindent{3}{}\<[3]%
\>[3]{}\mathrm{0}\mathrel{=}\Conid{M}\!.\Varid{empty}{}\<[E]%
\\
\>[B]{}\hsindent{3}{}\<[3]%
\>[3]{}(\mathbin{+})\mathrel{=}\Conid{M}\!.\Varid{unionWith}\;(\mathbin{+}){}\<[E]%
\\[\blanklineskip]%
\>[B]{}\mathbf{instance}\;(\Conid{Ord}\;\Varid{a},\Conid{Additive}\;\Varid{b})\Rightarrow \Conid{IsZero}\;(\Conid{Map}\;\Varid{a}\;\Varid{b})\;\mathbf{where}\;\Varid{isZero}\mathrel{=}\Conid{M}\!.\Varid{null}{}\<[E]%
\\[\blanklineskip]%
\>[B]{}\mathbf{instance}\;\Conid{Semiring}\;\Varid{b}\Rightarrow \Conid{LeftSemimodule}\;\Varid{b}\;(\Conid{Map}\;\Varid{a}\;\Varid{b})\;\mathbf{where}{}\<[E]%
\\
\>[B]{}\hsindent{3}{}\<[3]%
\>[3]{}(\mathbin{\hat{\cdot}})\;\Varid{b}\mathrel{=}\Varid{fmap}\;(\Varid{b}\;{}\mathbin{*}){}\<[E]%
\\[\blanklineskip]%
\>[B]{}\mathbf{instance}\;(\Conid{Ord}\;\Varid{a},\Conid{Monoid}\;\Varid{a},\Conid{Semiring}\;\Varid{b})\Rightarrow \Conid{Semiring}\;(\Conid{Map}\;\Varid{a}\;\Varid{b})\;\mathbf{where}{}\<[E]%
\\
\>[B]{}\hsindent{3}{}\<[3]%
\>[3]{}\mathrm{1}\mathrel{=}\varepsilon\mapsto\mathrm{1}{}\<[E]%
\\
\>[B]{}\hsindent{3}{}\<[3]%
\>[3]{}\Varid{p}\mathbin{*}\Varid{q}\mathrel{=}\Varid{sum}\;[\mskip1.5mu \Varid{u} \diamond \Varid{v}\mapsto\Varid{p}\mathbin{!}\Varid{u}\mathbin{*}\Varid{q}\mathbin{!}\Varid{v}\mid \Varid{u}\leftarrow \Conid{M}\!.\Varid{keys}\;\Varid{p},\Varid{v}\leftarrow \Conid{M}\!.\Varid{keys}\;\Varid{q}\mskip1.5mu]{}\<[E]%
\ColumnHook
\end{hscode}\resethooks
\vspace{-4ex}
}, \ensuremath{(\mathbin{!})} is a homomorphism with respect to each instantiated class.
(The ``\ensuremath{\Conid{M}}.'' module qualifier indicates names coming from the finite map library \citep{Data.Map}.)
\notefoot{Do I want a theorem and proof here?
I think so, though I'll have to make a few assumptions about finite maps.
On the other hand, those assumptions don't differ much from the homomorphism properties I'm claiming to hold.}
The finiteness of finite maps prevents giving a useful \ensuremath{\Conid{StarSemiring}} instance.

\sectionl{Decomposing Functions from Lists}

For functions from \emph{lists} specifically, we can decompose in a way that lays the groundwork for more efficient implementations than the ones in previous sections.
\begin{lemma}[\provedIn{lemma:decomp ([c] -> b)}]\lemlabel{decomp ([c] -> b)}
Any \ensuremath{\Varid{f}\mathbin{::}[\mskip1.5mu \Varid{c}\mskip1.5mu]\to \Varid{b}} can be decomposed as follows:
\begin{hscode}\SaveRestoreHook
\column{B}{@{}>{\hspre}l<{\hspost}@{}}%
\column{E}{@{}>{\hspre}l<{\hspost}@{}}%
\>[B]{}\Varid{f}\mathrel{=}\Varid{at}_\epsilon\;\Varid{f}\mathrel\triangleleft\derivOp\;\Varid{f}{}\<[E]%
\ColumnHook
\end{hscode}\resethooks
Moreover, for all \ensuremath{\Varid{b}} and \ensuremath{\Varid{h}},
\begin{hscode}\SaveRestoreHook
\column{B}{@{}>{\hspre}l<{\hspost}@{}}%
\column{8}{@{}>{\hspre}l<{\hspost}@{}}%
\column{E}{@{}>{\hspre}l<{\hspost}@{}}%
\>[B]{}\Varid{at}_\epsilon\;{}\<[8]%
\>[8]{}(\Varid{b}\mathrel\triangleleft\Varid{h})\mathrel{=}\Varid{b}{}\<[E]%
\\
\>[B]{}\derivOp\;{}\<[8]%
\>[8]{}(\Varid{b}\mathrel\triangleleft\Varid{h})\mathrel{=}\Varid{h}{}\<[E]%
\ColumnHook
\end{hscode}\resethooks
where
\begin{hscode}\SaveRestoreHook
\column{B}{@{}>{\hspre}l<{\hspost}@{}}%
\column{48}{@{}>{\hspre}l<{\hspost}@{}}%
\column{E}{@{}>{\hspre}l<{\hspost}@{}}%
\>[B]{}\Varid{at}_\epsilon\mathbin{::}([\mskip1.5mu \Varid{c}\mskip1.5mu]\to \Varid{b})\to \Varid{b}{}\<[E]%
\\
\>[B]{}\Varid{at}_\epsilon\;\Varid{f}\mathrel{=}\Varid{f}\;\varepsilon{}\<[E]%
\\[\blanklineskip]%
\>[B]{}\derivOp\mathbin{::}([\mskip1.5mu \Varid{c}\mskip1.5mu]\to \Varid{b})\to \Varid{c}\to ([\mskip1.5mu \Varid{c}\mskip1.5mu]\to \Varid{b}){}\<[E]%
\\
\>[B]{}\derivOp\;\Varid{f}\mathrel{=}\lambda\, \Varid{c}\;\Varid{cs}\to \Varid{f}\;(\Varid{c}\mathbin{:}\Varid{cs}){}\<[E]%
\\[\blanklineskip]%
\>[B]{}\mathbf{infix}\;\mathrm{1}\mathrel\triangleleft{}\<[E]%
\\
\>[B]{}(\mathrel\triangleleft)\mathbin{::}\Varid{b}\to (\Varid{c}\to ([\mskip1.5mu \Varid{c}\mskip1.5mu]\to \Varid{b}))\to ([\mskip1.5mu \Varid{c}\mskip1.5mu]\to \Varid{b}){}\<[E]%
\\
\>[B]{}\Varid{b}\mathrel\triangleleft\Varid{h}\mathrel{=}\lambda\, {}\;\mathbf{case}\;\{\mskip1.5mu {}\;[\mskip1.5mu \mskip1.5mu]\to \Varid{b}\;{};{}\;\Varid{c}\mathbin{:}\Varid{cs}{}\<[48]%
\>[48]{}\to \Varid{h}\;\Varid{c}\;\Varid{cs}\;{}\mskip1.5mu\}{}\<[E]%
\ColumnHook
\end{hscode}\resethooks
\vspace{-4ex}
\end{lemma}
\noindent
Considering the isomorphism \ensuremath{\Pow\;[\mskip1.5mu \Varid{c}\mskip1.5mu] \simeq [\mskip1.5mu \Varid{c}\mskip1.5mu]\to \Conid{Bool}}, this decomposition generalizes the \ensuremath{\delta } and \ensuremath{\derivOp} operations used by \citet{Brzozowski64} mapping languages to languages (as sets of strings), the latter of which he referred to as the ``derivative''.\footnote{Brzozowski wrote ``\ensuremath{\derivOp_{\Varid{c}}\;\Varid{p}}'' instead of ``\ensuremath{\derivOp\;\Varid{p}\;\Varid{c}}'', but the latter will prove more convenient below.}
Brzozowski used differentiation with respect to single symbols to implement a more general form of language differentiation with respect to a \emph{string} of symbols, where the \emph{derivative} \ensuremath{\derivOp^{\ast}\;\Varid{u}\;\Varid{p}} of a language \ensuremath{\Varid{p}} with respect to a prefix string \ensuremath{\Varid{u}} is the set of \ensuremath{\Varid{u}}-suffixes of strings in \ensuremath{\Varid{p}}, i.e.,
\begin{hscode}\SaveRestoreHook
\column{B}{@{}>{\hspre}l<{\hspost}@{}}%
\column{3}{@{}>{\hspre}l<{\hspost}@{}}%
\column{E}{@{}>{\hspre}l<{\hspost}@{}}%
\>[3]{}\derivOp^{\ast}\;\Varid{p}\;\Varid{u}\mathrel{=}\set{\Varid{v}\mid \Varid{u} \diamond \Varid{v}\mathbin{\in}\Varid{p}}{}\<[E]%
\ColumnHook
\end{hscode}\resethooks
so that
\begin{hscode}\SaveRestoreHook
\column{B}{@{}>{\hspre}l<{\hspost}@{}}%
\column{3}{@{}>{\hspre}l<{\hspost}@{}}%
\column{E}{@{}>{\hspre}l<{\hspost}@{}}%
\>[3]{}\Varid{u}\mathbin{\in}\Varid{p}\Longleftrightarrow\varepsilon\mathbin{\in}\derivOp^{\ast}\;\Varid{p}\;\Varid{u}{}\<[E]%
\ColumnHook
\end{hscode}\resethooks
Further, he noted that%
\footnote{Here, Brzozowski's notation makes for a prettier formulation:
\begin{hscode}\SaveRestoreHook
\column{B}{@{}>{\hspre}l<{\hspost}@{}}%
\column{21}{@{}>{\hspre}l<{\hspost}@{}}%
\column{E}{@{}>{\hspre}l<{\hspost}@{}}%
\>[B]{}\derivOp^{\ast}_{\varepsilon}\;\Varid{p}{}\<[21]%
\>[21]{}\mathrel{=}\Varid{p}{}\<[E]%
\\
\>[B]{}\derivOp^{\ast}_{\Varid{u} \diamond \Varid{v}}\;\Varid{p}{}\<[21]%
\>[21]{}\mathrel{=}\derivOp^{\ast}_{\Varid{v}}\;(\derivOp^{\ast}_{\Varid{u}}\;\Varid{p}){}\<[E]%
\ColumnHook
\end{hscode}\resethooks
Equivalently,
\begin{hscode}\SaveRestoreHook
\column{B}{@{}>{\hspre}l<{\hspost}@{}}%
\column{19}{@{}>{\hspre}l<{\hspost}@{}}%
\column{E}{@{}>{\hspre}l<{\hspost}@{}}%
\>[B]{}\derivOp^{\ast}_{\varepsilon}{}\<[19]%
\>[19]{}\mathrel{=}\Varid{id}{}\<[E]%
\\
\>[B]{}\derivOp^{\ast}_{\Varid{u} \diamond \Varid{v}}{}\<[19]%
\>[19]{}\mathrel{=}\derivOp^{\ast}_{\Varid{v}}\hsdot{\circ }{.\:}\derivOp^{\ast}_{\Varid{u}}{}\<[E]%
\ColumnHook
\end{hscode}\resethooks
where \ensuremath{\Varid{id}} is the identity function.
In other words, \ensuremath{\derivOp^{\ast}_{\cdot }} is a contravariant monoid homomorphism (targeting the monoid of endofunctions).}
\begin{hscode}\SaveRestoreHook
\column{B}{@{}>{\hspre}l<{\hspost}@{}}%
\column{20}{@{}>{\hspre}l<{\hspost}@{}}%
\column{E}{@{}>{\hspre}l<{\hspost}@{}}%
\>[B]{}\derivOp^{\ast}\;\Varid{p}\;\varepsilon{}\<[20]%
\>[20]{}\mathrel{=}\Varid{p}{}\<[E]%
\\
\>[B]{}\derivOp^{\ast}\;\Varid{p}\;(\Varid{u} \diamond \Varid{v}){}\<[20]%
\>[20]{}\mathrel{=}\derivOp^{\ast}\;(\derivOp^{\ast}\;\Varid{p}\;\Varid{u})\;\Varid{v}{}\<[E]%
\ColumnHook
\end{hscode}\resethooks
Thanks to this decomposition property and the fact that \ensuremath{\derivOp\;\Varid{p}\;\Varid{c}\mathrel{=}\derivOp^{\ast}\;\Varid{p}\;[\mskip1.5mu \Varid{c}\mskip1.5mu]}, one can successively differentiate with respect to single symbols.

Generalizing from sets to functions,
\begin{hscode}\SaveRestoreHook
\column{B}{@{}>{\hspre}l<{\hspost}@{}}%
\column{E}{@{}>{\hspre}l<{\hspost}@{}}%
\>[B]{}\derivOp^{\ast}\;\Varid{f}\;\Varid{u}\mathrel{=}\lambda\, \Varid{v}\to \Varid{f}\;(\Varid{u} \diamond \Varid{v}){}\<[E]%
\ColumnHook
\end{hscode}\resethooks
so that
\begin{hscode}\SaveRestoreHook
\column{B}{@{}>{\hspre}l<{\hspost}@{}}%
\column{4}{@{}>{\hspre}l<{\hspost}@{}}%
\column{E}{@{}>{\hspre}l<{\hspost}@{}}%
\>[B]{}\Varid{f}{}\<[4]%
\>[4]{}\mathrel{=}\lambda\, \Varid{u}\to \derivOp^{\ast}\;\Varid{f}\;\Varid{u}\;\varepsilon{}\<[E]%
\\
\>[4]{}\mathrel{=}\lambda\, \Varid{u}\to \Varid{at}_\epsilon\;(\derivOp^{\ast}\;\Varid{f}\;\Varid{u}){}\<[E]%
\\
\>[4]{}\mathrel{=}\Varid{at}_\epsilon\hsdot{\circ }{.\:}\derivOp^{\ast}\;\Varid{f}{}\<[E]%
\\
\>[4]{}\mathrel{=}\Varid{at}_\epsilon\hsdot{\circ }{.\:}\Varid{foldl}\;\derivOp\;\Varid{f}{}\<[E]%
\ColumnHook
\end{hscode}\resethooks
where \ensuremath{\Varid{foldl}} is the usual left fold on lists:
\notefoot{The choice of \emph{left} fold bothers me. Could we define \ensuremath{\derivOp^{\ast}} via a \emph{right} fold (the natural fold for lists), perhaps leading to a right-to-left string matching algorithm? Or maybe left vs right doesn't matter?}
\begin{hscode}\SaveRestoreHook
\column{B}{@{}>{\hspre}l<{\hspost}@{}}%
\column{19}{@{}>{\hspre}l<{\hspost}@{}}%
\column{E}{@{}>{\hspre}l<{\hspost}@{}}%
\>[B]{}\Varid{foldl}\mathbin{::}(\Varid{c}\to \Varid{b}\to \Varid{b})\to \Varid{b}\to [\mskip1.5mu \Varid{c}\mskip1.5mu]\to \Varid{b}{}\<[E]%
\\
\>[B]{}\Varid{foldl}\;\Varid{h}\;\Varid{e}\;[\mskip1.5mu \mskip1.5mu]{}\<[19]%
\>[19]{}\mathrel{=}\Varid{e}{}\<[E]%
\\
\>[B]{}\Varid{foldl}\;\Varid{h}\;\Varid{e}\;(\Varid{c}\mathbin{:}\Varid{cs}){}\<[19]%
\>[19]{}\mathrel{=}\Varid{foldl}\;\Varid{h}\;(\Varid{h}\;\Varid{e}\;\Varid{c})\;\Varid{cs}{}\<[E]%
\ColumnHook
\end{hscode}\resethooks
Intriguingly, \ensuremath{\Varid{at}_\epsilon} and \ensuremath{\derivOp^{\ast}} correspond to \ensuremath{\Varid{coreturn}} and \ensuremath{\Varid{cojoin}} for the function-from-monoid comonad, also called the ``exponent comonad'' \citep{Uustalu2008CNC}.

Understanding how \ensuremath{\Varid{at}_\epsilon} and \ensuremath{\derivOp} relate to the semiring vocabulary will help us develop efficient implementations in later sections.

\begin{lemma}[\provedIn{lemma:atEps [c] -> b}]\lemlabel{atEps [c] -> b}
The \ensuremath{\Varid{at}_\epsilon} function is a star semiring and left semimodule homomorphism, i.e.,
\begin{spacing}{1.3}
\begin{hscode}\SaveRestoreHook
\column{B}{@{}>{\hspre}l<{\hspost}@{}}%
\column{11}{@{}>{\hspre}c<{\hspost}@{}}%
\column{11E}{@{}l@{}}%
\column{16}{@{}>{\hspre}l<{\hspost}@{}}%
\column{20}{@{}>{\hspre}l<{\hspost}@{}}%
\column{32}{@{}>{\hspre}c<{\hspost}@{}}%
\column{32E}{@{}l@{}}%
\column{37}{@{}>{\hspre}l<{\hspost}@{}}%
\column{E}{@{}>{\hspre}l<{\hspost}@{}}%
\>[B]{}\Varid{at}_\epsilon\;\mathrm{0}{}\<[20]%
\>[20]{}\mathrel{=}\mathrm{0}{}\<[E]%
\\
\>[B]{}\Varid{at}_\epsilon\;\mathrm{1}{}\<[20]%
\>[20]{}\mathrel{=}\mathrm{1}{}\<[E]%
\\
\>[B]{}\Varid{at}_\epsilon\;(\Varid{p}{}\<[11]%
\>[11]{}\mathbin{+}{}\<[11E]%
\>[16]{}\Varid{q}){}\<[20]%
\>[20]{}\mathrel{=}\Varid{at}_\epsilon\;\Varid{p}{}\<[32]%
\>[32]{}\mathbin{+}{}\<[32E]%
\>[37]{}\Varid{at}_\epsilon\;\Varid{q}{}\<[E]%
\\
\>[B]{}\Varid{at}_\epsilon\;(\Varid{p}{}\<[11]%
\>[11]{}\mathbin{*}{}\<[11E]%
\>[16]{}\Varid{q}){}\<[20]%
\>[20]{}\mathrel{=}\Varid{at}_\epsilon\;\Varid{p}{}\<[32]%
\>[32]{}\mathbin{*}{}\<[32E]%
\>[37]{}\Varid{at}_\epsilon\;\Varid{q}{}\<[E]%
\\
\>[B]{}\Varid{at}_\epsilon\;(\closure{\Varid{p}}){}\<[20]%
\>[20]{}\mathrel{=}\closure{(\Varid{at}_\epsilon\;\Varid{p})}{}\<[E]%
\ColumnHook
\end{hscode}\resethooks
\end{spacing}
\noindent
Moreover,\footnote{Mathematically, the \ensuremath{(\cdot)} equation says that \ensuremath{\Varid{at}_\epsilon} is a left \ensuremath{\Varid{b}}-semiring homomorphism as well, since every semiring is a (left and right) semimodule over itself.
Likewise, the \ensuremath{(\mapsto)} equation might be written as ``\ensuremath{\Varid{null}\;\Varid{w}\mapsto\Varid{b}}'' or even ``\ensuremath{\Varid{at}_\epsilon\;\Varid{w}\mapsto\Varid{b}}''.
Unfortunately, these prettier formulations would lead to ambiguity during Haskell type inference.}
\begin{spacing}{1.3}
\begin{hscode}\SaveRestoreHook
\column{B}{@{}>{\hspre}l<{\hspost}@{}}%
\column{11}{@{}>{\hspre}l<{\hspost}@{}}%
\column{14}{@{}>{\hspre}l<{\hspost}@{}}%
\column{18}{@{}>{\hspre}l<{\hspost}@{}}%
\column{E}{@{}>{\hspre}l<{\hspost}@{}}%
\>[B]{}\Varid{at}_\epsilon\;(\Varid{s}\cdot\Varid{p}){}\<[18]%
\>[18]{}\mathrel{=}\Varid{s}\mathbin{*}\Varid{at}_\epsilon\;\Varid{p}{}\<[E]%
\\[\blanklineskip]%
\>[B]{}\Varid{at}_\epsilon\;({}\<[11]%
\>[11]{}[\mskip1.5mu \mskip1.5mu]{}\<[18]%
\>[18]{}\mapsto\Varid{b})\mathrel{=}\Varid{b}{}\<[E]%
\\
\>[B]{}\Varid{at}_\epsilon\;(\Varid{c}{}\<[11]%
\>[11]{}\mathbin{:}{}\<[14]%
\>[14]{}\Varid{cs}{}\<[18]%
\>[18]{}\mapsto\Varid{b})\mathrel{=}\mathrm{0}{}\<[E]%
\ColumnHook
\end{hscode}\resethooks
\end{spacing}
\vspace{-2ex}
\noindent
\end{lemma}
\begin{lemma}[\provedIn{lemma:deriv [c] -> b}, generalizing Lemma 3.1 of \citet{Brzozowski64}]\lemlabel{deriv [c] -> b}
Differentiation has the following properties:
\notefoot{If I replace application to \ensuremath{\Varid{c}} by indexing by \ensuremath{\Varid{c}} (i.e., \ensuremath{(\mathbin{!}{}\;\Varid{c})}), will this lemma hold for all of the representations? I suspect so. Idea: Define $\derivOp_c\,p = \derivOp\,p\:!\:c$.}
\begin{spacing}{1.3}
\begin{hscode}\SaveRestoreHook
\column{B}{@{}>{\hspre}l<{\hspost}@{}}%
\column{11}{@{}>{\hspre}c<{\hspost}@{}}%
\column{11E}{@{}l@{}}%
\column{13}{@{}>{\hspre}l<{\hspost}@{}}%
\column{14}{@{}>{\hspre}l<{\hspost}@{}}%
\column{16}{@{}>{\hspre}l<{\hspost}@{}}%
\column{21}{@{}>{\hspre}l<{\hspost}@{}}%
\column{22}{@{}>{\hspre}l<{\hspost}@{}}%
\column{E}{@{}>{\hspre}l<{\hspost}@{}}%
\>[B]{}\derivOp\;\mathrm{0}\;{}\<[14]%
\>[14]{}\Varid{c}{}\<[22]%
\>[22]{}\mathrel{=}\mathrm{0}{}\<[E]%
\\
\>[B]{}\derivOp\;\mathrm{1}\;{}\<[14]%
\>[14]{}\Varid{c}{}\<[22]%
\>[22]{}\mathrel{=}\mathrm{0}{}\<[E]%
\\
\>[B]{}\derivOp\;(\Varid{p}{}\<[11]%
\>[11]{}\mathbin{+}{}\<[11E]%
\>[16]{}\Varid{q})\;\Varid{c}{}\<[22]%
\>[22]{}\mathrel{=}\derivOp\;\Varid{p}\;\Varid{c}\mathbin{+}\derivOp\;\Varid{q}\;\Varid{c}{}\<[E]%
\\
\>[B]{}\derivOp\;(\Varid{p}{}\<[11]%
\>[11]{}\mathbin{*}{}\<[11E]%
\>[16]{}\Varid{q})\;\Varid{c}{}\<[22]%
\>[22]{}\mathrel{=}\Varid{at}_\epsilon\;\Varid{p}\cdot\derivOp\;\Varid{q}\;\Varid{c}\mathbin{+}\derivOp\;\Varid{p}\;\Varid{c}\mathbin{*}\Varid{q}{}\<[E]%
\\
\>[B]{}\derivOp\;(\closure{\Varid{p}})\;\Varid{c}{}\<[22]%
\>[22]{}\mathrel{=}\closure{(\Varid{at}_\epsilon\;\Varid{p})}\cdot\derivOp\;\Varid{p}\;\Varid{c}\mathbin{*}\closure{\Varid{p}}{}\<[E]%
\\
\>[B]{}\derivOp\;(\Varid{s}\cdot\Varid{p})\;\Varid{c}{}\<[22]%
\>[22]{}\mathrel{=}\Varid{s}\cdot\derivOp\;\Varid{p}\;\Varid{c}{}\<[E]%
\\[\blanklineskip]%
\>[B]{}\derivOp\;({}\<[13]%
\>[13]{}[\mskip1.5mu \mskip1.5mu]{}\<[21]%
\>[21]{}\mapsto\Varid{b})\mathrel{=}\lambda\, \Varid{c}\to \mathrm{0}{}\<[E]%
\\
\>[B]{}\derivOp\;(\Varid{c'}{}\<[13]%
\>[13]{}\mathbin{:}{}\<[16]%
\>[16]{}\Varid{cs'}{}\<[21]%
\>[21]{}\mapsto\Varid{b})\mathrel{=}\Varid{c'}\mapsto\Varid{cs'}\mapsto\Varid{b}{}\<[E]%
\ColumnHook
\end{hscode}\resethooks
\end{spacing}
\vspace{-2ex}
\end{lemma}
Although \ensuremath{\derivOp\;\Varid{p}} is defined as a \emph{function} from leading symbols, it could instead be another representation with function-like semantics, such as as \ensuremath{\Varid{h}\;\Varid{b}} for an appropriate functor \ensuremath{\Varid{h}}.
To relate \ensuremath{\Varid{h}} to the choice of alphabet \ensuremath{\Varid{c}}, introduce a type family:
\begin{hscode}\SaveRestoreHook
\column{B}{@{}>{\hspre}l<{\hspost}@{}}%
\column{26}{@{}>{\hspre}l<{\hspost}@{}}%
\column{E}{@{}>{\hspre}l<{\hspost}@{}}%
\>[B]{}\mathbf{type}\;\mathbf{family}\;\Conid{Key}\;(\Varid{h}\mathbin{::}\Conid{Type}\to \Conid{Type})\mathbin{::}\Conid{Type}{}\<[E]%
\\[\blanklineskip]%
\>[B]{}\mathbf{type}\;\mathbf{instance}\;\Conid{Key}\;((\to )\;{}\<[26]%
\>[26]{}\Varid{a})\mathrel{=}\Varid{a}{}\<[E]%
\\
\>[B]{}\mathbf{type}\;\mathbf{instance}\;\Conid{Key}\;(\Conid{Map}\;{}\<[26]%
\>[26]{}\Varid{a})\mathrel{=}\Varid{a}{}\<[E]%
\ColumnHook
\end{hscode}\resethooks
Generalizing in this way (with functions as a special case) enables convenient memoization, which has been found to be quite useful in practice for derivative-based parsing \citep{Might2010YaccID}.
A few generalizations to the equations in \lemref{deriv [c] -> b} suffice to generalize from \ensuremath{\Varid{c}\to ([\mskip1.5mu \Varid{c}\mskip1.5mu]\to \Varid{b})} to \ensuremath{\Varid{h}\;([\mskip1.5mu \Varid{c}\mskip1.5mu]\to \Varid{b})} \seeproof{lemma:deriv [c] -> b}.
We must assume that \ensuremath{\Conid{Key}\;\Varid{h}\mathrel{=}\Varid{c}} and that \ensuremath{\Varid{h}} is an ``additive functor'', i.e., \ensuremath{\forall \Varid{b}\hsforall \hsdot{\circ }{.\:}\Conid{Additive}\;\Varid{b}\Rightarrow \Conid{Additive}\;(\Varid{h}\;\Varid{b})} with \ensuremath{(\mathbin{!})} for \ensuremath{\Varid{h}} being an additive monoid homomorphism.
\begin{spacing}{1.3}
\begin{hscode}\SaveRestoreHook
\column{B}{@{}>{\hspre}l<{\hspost}@{}}%
\column{11}{@{}>{\hspre}c<{\hspost}@{}}%
\column{11E}{@{}l@{}}%
\column{16}{@{}>{\hspre}l<{\hspost}@{}}%
\column{20}{@{}>{\hspre}l<{\hspost}@{}}%
\column{E}{@{}>{\hspre}l<{\hspost}@{}}%
\>[B]{}\derivOp\;\mathrm{0}{}\<[20]%
\>[20]{}\mathrel{=}\mathrm{0}{}\<[E]%
\\
\>[B]{}\derivOp\;\mathrm{1}{}\<[20]%
\>[20]{}\mathrel{=}\mathrm{0}{}\<[E]%
\\
\>[B]{}\derivOp\;(\Varid{p}{}\<[11]%
\>[11]{}\mathbin{+}{}\<[11E]%
\>[16]{}\Varid{q}){}\<[20]%
\>[20]{}\mathrel{=}\derivOp\;\Varid{p}\mathbin{+}\derivOp\;\Varid{q}{}\<[E]%
\\
\>[B]{}\derivOp\;(\Varid{p}{}\<[11]%
\>[11]{}\mathbin{*}{}\<[11E]%
\>[16]{}\Varid{q}){}\<[20]%
\>[20]{}\mathrel{=}\Varid{fmap}\;(\Varid{at}_\epsilon\;\Varid{p}\;{}\cdot)\;(\derivOp\;\Varid{q})\mathbin{+}\Varid{fmap}\;(\mathbin{*}\Varid{q})\;(\derivOp\;\Varid{p}){}\<[E]%
\\
\>[B]{}\derivOp\;(\closure{\Varid{p}}){}\<[20]%
\>[20]{}\mathrel{=}\Varid{fmap}\;(\lambda\, \Varid{d}\to \closure{(\Varid{at}_\epsilon\;\Varid{p})}\cdot\Varid{d}\mathbin{*}\Conid{Star}\;\Varid{p})\;(\derivOp\;\Varid{p}){}\<[E]%
\\
\>[B]{}\derivOp\;(\Varid{s}\cdot\Varid{p}){}\<[20]%
\>[20]{}\mathrel{=}\Varid{fmap}\;(\Varid{s}\;{}\cdot)\;(\derivOp\;\Varid{p}){}\<[E]%
\ColumnHook
\end{hscode}\resethooks
\vspace{-4ex}
\end{spacing}

\note{Consider reexpressing \lemref{deriv [c] -> b} in terms of \ensuremath{(\mathbin{!})}. Maybe even generalize \ensuremath{(\mathrel\triangleleft)} to indexable functors.}

\begin{theorem}[\provedIn{theorem:semiring decomp [c] -> b}]\thmlabel{semiring decomp [c] -> b}
The following properties hold (in the generalized setting of a functor \ensuremath{\Varid{h}} with \ensuremath{\Conid{Key}\;\Varid{h}\mathrel{=}\Varid{c}}):
\begin{spacing}{1.4}
\begin{hscode}\SaveRestoreHook
\column{B}{@{}>{\hspre}l<{\hspost}@{}}%
\column{5}{@{}>{\hspre}c<{\hspost}@{}}%
\column{5E}{@{}l@{}}%
\column{7}{@{}>{\hspre}l<{\hspost}@{}}%
\column{9}{@{}>{\hspre}l<{\hspost}@{}}%
\column{16}{@{}>{\hspre}l<{\hspost}@{}}%
\column{28}{@{}>{\hspre}l<{\hspost}@{}}%
\column{34}{@{}>{\hspre}c<{\hspost}@{}}%
\column{34E}{@{}l@{}}%
\column{39}{@{}>{\hspre}l<{\hspost}@{}}%
\column{E}{@{}>{\hspre}l<{\hspost}@{}}%
\>[B]{}\mathrm{0}{}\<[7]%
\>[7]{}\mathrel{=}\mathrm{0}{}\<[16]%
\>[16]{}\mathrel\triangleleft\mathrm{0}{}\<[E]%
\\
\>[B]{}\mathrm{1}{}\<[7]%
\>[7]{}\mathrel{=}\mathrm{1}{}\<[16]%
\>[16]{}\mathrel\triangleleft\mathrm{0}{}\<[E]%
\\
\>[B]{}(\Varid{a}{}\<[5]%
\>[5]{}\mathrel\triangleleft{}\<[5E]%
\>[9]{}\Varid{dp})\mathbin{+}(\Varid{b}\mathrel\triangleleft\Varid{dq}){}\<[28]%
\>[28]{}\mathrel{=}\Varid{a}{}\<[34]%
\>[34]{}\mathbin{+}{}\<[34E]%
\>[39]{}\Varid{b}\mathrel\triangleleft\Varid{dp}\mathbin{+}\Varid{dq}{}\<[E]%
\\
\>[B]{}(\Varid{a}{}\<[5]%
\>[5]{}\mathrel\triangleleft{}\<[5E]%
\>[9]{}\Varid{dp})\mathbin{*}\Varid{q}\mathrel{=}\Varid{a}\cdot\Varid{q}\mathbin{+}(\mathrm{0}\mathrel\triangleleft\Varid{fmap}\;(\mathbin{*}{}\;\Varid{q})\;\Varid{dp}){}\<[E]%
\\
\>[B]{}\closure{(\Varid{a}\mathrel\triangleleft\Varid{dp})}\mathrel{=}\Varid{q}\;\mathbf{where}\;\Varid{q}\mathrel{=}\closure{\Varid{a}}\cdot(\mathrm{1}\mathrel\triangleleft\Varid{fmap}\;(\mathbin{*}{}\;\Varid{q})\;\Varid{dp}){}\<[E]%
\\
\>[B]{}\Varid{s}\cdot(\Varid{a}\mathrel\triangleleft\Varid{dp})\mathrel{=}\Varid{s}\mathbin{*}\Varid{a}\mathrel\triangleleft\Varid{fmap}\;(\Varid{s}\;{}\cdot)\;\Varid{dp}{}\<[E]%
\\
\>[B]{}\Varid{w}\mapsto\Varid{b}\mathrel{=}\Varid{foldr}\;(\lambda\, \Varid{c}\;\Varid{t}\to \mathrm{0}\mathrel\triangleleft\Varid{c}\mapsto\Varid{t})\;(\Varid{b}\mathrel\triangleleft\mathrm{0})\;\Varid{w}{}\<[E]%
\ColumnHook
\end{hscode}\resethooks
\vspace{-6ex}
\end{spacing}
\end{theorem}

\sectionl{Regular Expressions}

\lemreftwo{atEps [c] -> b}{deriv [c] -> b} generalize and were inspired by a technique of \citet{Brzozowski64} for recognizing regular languages.
\figrefdef{RegExp}{Semiring-generalized regular expressions denoting \ensuremath{[\mskip1.5mu \Varid{c}\mskip1.5mu]\to \Varid{b}}}{
\begin{hscode}\SaveRestoreHook
\column{B}{@{}>{\hspre}l<{\hspost}@{}}%
\column{3}{@{}>{\hspre}l<{\hspost}@{}}%
\column{8}{@{}>{\hspre}l<{\hspost}@{}}%
\column{11}{@{}>{\hspre}c<{\hspost}@{}}%
\column{11E}{@{}l@{}}%
\column{12}{@{}>{\hspre}c<{\hspost}@{}}%
\column{12E}{@{}l@{}}%
\column{16}{@{}>{\hspre}c<{\hspost}@{}}%
\column{16E}{@{}l@{}}%
\column{17}{@{}>{\hspre}l<{\hspost}@{}}%
\column{18}{@{}>{\hspre}l<{\hspost}@{}}%
\column{19}{@{}>{\hspre}l<{\hspost}@{}}%
\column{22}{@{}>{\hspre}l<{\hspost}@{}}%
\column{23}{@{}>{\hspre}l<{\hspost}@{}}%
\column{27}{@{}>{\hspre}c<{\hspost}@{}}%
\column{27E}{@{}l@{}}%
\column{30}{@{}>{\hspre}l<{\hspost}@{}}%
\column{42}{@{}>{\hspre}c<{\hspost}@{}}%
\column{42E}{@{}l@{}}%
\column{48}{@{}>{\hspre}l<{\hspost}@{}}%
\column{E}{@{}>{\hspre}l<{\hspost}@{}}%
\>[B]{}\mathbf{data}\;\Conid{RegExp}\;\Varid{h}\;\Varid{b}{}\<[27]%
\>[27]{}\mathrel{=}{}\<[27E]%
\>[30]{}\Conid{Char}\;(\Conid{Key}\;\Varid{h}){}\<[E]%
\\
\>[27]{}\mid {}\<[27E]%
\>[30]{}\Conid{Value}\;\Varid{b}{}\<[E]%
\\
\>[27]{}\mid {}\<[27E]%
\>[30]{}\Conid{RegExp}\;\Varid{h}\;\Varid{b}{}\<[42]%
\>[42]{}\mathbin{:\!\!+}{}\<[42E]%
\>[48]{}\Conid{RegExp}\;\Varid{h}\;\Varid{b}{}\<[E]%
\\
\>[27]{}\mid {}\<[27E]%
\>[30]{}\Conid{RegExp}\;\Varid{h}\;\Varid{b}{}\<[42]%
\>[42]{}\mathbin{:\!\!\conv}{}\<[42E]%
\>[48]{}\Conid{RegExp}\;\Varid{h}\;\Varid{b}{}\<[E]%
\\
\>[27]{}\mid {}\<[27E]%
\>[30]{}\Conid{Star}\;(\Conid{RegExp}\;\Varid{h}\;\Varid{b}){}\<[E]%
\\
\>[B]{}\hsindent{3}{}\<[3]%
\>[3]{}\mathbf{deriving}\;\Conid{Functor}{}\<[E]%
\\[\blanklineskip]%
\>[B]{}\mathbf{instance}\;\Conid{Additive}\;\Varid{b}\Rightarrow \Conid{Additive}\;(\Conid{RegExp}\;\Varid{h}\;\Varid{b})\;\mathbf{where}{}\<[E]%
\\
\>[B]{}\hsindent{3}{}\<[3]%
\>[3]{}\mathrm{0}\mathrel{=}\Conid{Value}\;\mathrm{0}{}\<[E]%
\\
\>[B]{}\hsindent{3}{}\<[3]%
\>[3]{}(\mathbin{+})\mathrel{=}(\mathbin{:\!\!+}){}\<[E]%
\\[\blanklineskip]%
\>[B]{}\mathbf{instance}\;\Conid{Semiring}\;\Varid{b}\Rightarrow \Conid{Semiring}\;(\Conid{RegExp}\;\Varid{h}\;\Varid{b})\;\mathbf{where}{}\<[E]%
\\
\>[B]{}\hsindent{3}{}\<[3]%
\>[3]{}(\mathbin{\hat{\cdot}})\;\Varid{b}\mathrel{=}\Varid{fmap}\;(\Varid{b}\;{}\mathbin{*}){}\<[E]%
\\[\blanklineskip]%
\>[B]{}\mathbf{instance}\;\Conid{Semiring}\;\Varid{b}\Rightarrow \Conid{Semiring}\;(\Conid{RegExp}\;\Varid{h}\;\Varid{b})\;\mathbf{where}{}\<[E]%
\\
\>[B]{}\hsindent{3}{}\<[3]%
\>[3]{}\mathrm{1}\mathrel{=}\Conid{Value}\;\mathrm{1}{}\<[E]%
\\
\>[B]{}\hsindent{3}{}\<[3]%
\>[3]{}(\mathbin{*})\mathrel{=}(\mathbin{:\!\!\conv}){}\<[E]%
\\[\blanklineskip]%
\>[B]{}\mathbf{instance}\;\Conid{Semiring}\;\Varid{b}\Rightarrow \Conid{StarSemiring}\;(\Conid{RegExp}\;\Varid{h}\;\Varid{b})\;\mathbf{where}{}\<[E]%
\\
\>[B]{}\hsindent{3}{}\<[3]%
\>[3]{}\closure{\Varid{e}}\mathrel{=}\Conid{Star}\;\Varid{e}{}\<[E]%
\\[\blanklineskip]%
\>[B]{}\mathbf{type}\;\Conid{FR}\;\Varid{h}\;\Varid{b}\mathrel{=}{}\<[16]%
\>[16]{}({}\<[16E]%
\>[19]{}\Conid{HasSingle}\;(\Conid{Key}\;\Varid{h})\;(\Conid{RegExp}\;\Varid{h}\;\Varid{b})\;(\Varid{h}\;(\Conid{RegExp}\;\Varid{h}\;\Varid{b})){}\<[E]%
\\
\>[16]{},{}\<[16E]%
\>[19]{}\Conid{Additive}\;(\Varid{h}\;(\Conid{RegExp}\;\Varid{h}\;\Varid{b})),\Conid{Functor}\;\Varid{h},\Conid{IsZero}\;\Varid{b},\Conid{IsOne}\;\Varid{b}){}\<[E]%
\\[\blanklineskip]%
\>[B]{}\mathbf{instance}\;(\Conid{FR}\;\Varid{h}\;\Varid{b},\Conid{StarSemiring}\;\Varid{b},\Varid{c}\mathbin{\sim}\Conid{Key}\;\Varid{h},\Conid{Eq}\;\Varid{c})\Rightarrow \Conid{Indexable}\;[\mskip1.5mu \Varid{c}\mskip1.5mu]\;\Varid{b}\;(\Conid{RegExp}\;\Varid{h}\;\Varid{b})\;\mathbf{where}{}\<[E]%
\\
\>[B]{}\hsindent{3}{}\<[3]%
\>[3]{}\Varid{e}\mathbin{!}\Varid{w}\mathrel{=}\Varid{at}_\epsilon\;(\Varid{foldl}\;((\mathbin{!})\hsdot{\circ }{.\:}\derivOp)\;\Varid{e}\;\Varid{w}){}\<[E]%
\\[\blanklineskip]%
\>[B]{}\mathbf{instance}\;(\Conid{FR}\;\Varid{h}\;\Varid{b},\Conid{StarSemiring}\;\Varid{b},\Varid{c}\mathbin{\sim}\Conid{Key}\;\Varid{h},\Conid{Eq}\;\Varid{c})\Rightarrow \Conid{HasSingle}\;[\mskip1.5mu \Varid{c}\mskip1.5mu]\;\Varid{b}\;(\Conid{RegExp}\;\Varid{h}\;\Varid{b})\;\mathbf{where}{}\<[E]%
\\
\>[B]{}\hsindent{3}{}\<[3]%
\>[3]{}\Varid{w}\mapsto\Varid{b}\mathrel{=}\Varid{b}\cdot\Varid{product}\;(\Varid{map}\;\Conid{Char}\;\Varid{w}){}\<[E]%
\\[\blanklineskip]%
\>[B]{}\Varid{at}_\epsilon\mathbin{::}\Conid{StarSemiring}\;\Varid{b}\Rightarrow \Conid{RegExp}\;\Varid{h}\;\Varid{b}\to \Varid{b}{}\<[E]%
\\
\>[B]{}\Varid{at}_\epsilon\;(\Conid{Char}\;\anonymous ){}\<[22]%
\>[22]{}\mathrel{=}\mathrm{0}{}\<[E]%
\\
\>[B]{}\Varid{at}_\epsilon\;(\Conid{Value}\;\Varid{b}){}\<[22]%
\>[22]{}\mathrel{=}\Varid{b}{}\<[E]%
\\
\>[B]{}\Varid{at}_\epsilon\;(\Varid{p}{}\<[11]%
\>[11]{}\mathbin{:\!\!+}{}\<[11E]%
\>[17]{}\Varid{q}){}\<[22]%
\>[22]{}\mathrel{=}\Varid{at}_\epsilon\;\Varid{p}\mathbin{+}\Varid{at}_\epsilon\;\Varid{q}{}\<[E]%
\\
\>[B]{}\Varid{at}_\epsilon\;(\Varid{p}{}\<[11]%
\>[11]{}\mathbin{:\!\!\conv}{}\<[11E]%
\>[17]{}\Varid{q}){}\<[22]%
\>[22]{}\mathrel{=}\Varid{at}_\epsilon\;\Varid{p}\mathbin{*}\Varid{at}_\epsilon\;\Varid{q}{}\<[E]%
\\
\>[B]{}\Varid{at}_\epsilon\;(\Conid{Star}\;\Varid{p}){}\<[22]%
\>[22]{}\mathrel{=}\closure{(\Varid{at}_\epsilon\;\Varid{p})}{}\<[E]%
\\[\blanklineskip]%
\>[B]{}\derivOp\mathbin{::}(\Conid{FR}\;\Varid{h}\;\Varid{b},\Conid{StarSemiring}\;\Varid{b})\Rightarrow \Conid{RegExp}\;\Varid{h}\;\Varid{b}\to \Varid{h}\;(\Conid{RegExp}\;\Varid{h}\;\Varid{b}){}\<[E]%
\\
\>[B]{}\derivOp\;{}\<[8]%
\>[8]{}(\Conid{Char}\;\Varid{c}){}\<[23]%
\>[23]{}\mathrel{=}\Varid{single}\;\Varid{c}{}\<[E]%
\\
\>[B]{}\derivOp\;{}\<[8]%
\>[8]{}(\Conid{Value}\;\anonymous ){}\<[23]%
\>[23]{}\mathrel{=}\mathrm{0}{}\<[E]%
\\
\>[B]{}\derivOp\;{}\<[8]%
\>[8]{}(\Varid{p}{}\<[12]%
\>[12]{}\mathbin{:\!\!+}{}\<[12E]%
\>[18]{}\Varid{q}){}\<[23]%
\>[23]{}\mathrel{=}\derivOp\;\Varid{p}\mathbin{+}\derivOp\;\Varid{q}{}\<[E]%
\\
\>[B]{}\derivOp\;{}\<[8]%
\>[8]{}(\Varid{p}\mathbin{:\!\!\conv}\Varid{q}){}\<[23]%
\>[23]{}\mathrel{=}\Varid{fmap}\;(\Varid{at}_\epsilon\;\Varid{p}\;{}\cdot)\;(\derivOp\;\Varid{q})\mathbin{+}\Varid{fmap}\;(\mathbin{*}{}\;\Varid{q})\;(\derivOp\;\Varid{p}){}\<[E]%
\\
\>[B]{}\derivOp\;{}\<[8]%
\>[8]{}(\Conid{Star}\;\Varid{p}){}\<[23]%
\>[23]{}\mathrel{=}\Varid{fmap}\;(\lambda\, \Varid{d}\to \closure{(\Varid{at}_\epsilon\;\Varid{p})}\cdot\Varid{d}\mathbin{*}\Conid{Star}\;\Varid{p})\;(\derivOp\;\Varid{p}){}\<[E]%
\ColumnHook
\end{hscode}\resethooks
\vspace{-4ex}
} generalizes regular expressions in the same way that \ensuremath{\Varid{a}\to \Varid{b}} generalizes \ensuremath{\Pow\;\Varid{a}}, to yield a value of type \ensuremath{\Varid{b}} (a star semiring).
The constructor \ensuremath{\Conid{Value}\;\Varid{b}} generalizes \ensuremath{\mathrm{0}} and \ensuremath{\mathrm{1}} to yield a semiring value.
\begin{theorem}\thmlabel{RegExp}
Given the definitions in \figref{RegExp}, \ensuremath{(\mathbin{!})} is a homomorphism with respect to each instantiated class.
\end{theorem}
The implementation in \figref{RegExp} generalizes the regular expression matching algorithm of \citet{Brzozowski64}, adding customizable memoization, depending on choice of the indexable functor \ensuremath{\Varid{h}}.
Note that the definition of \ensuremath{\Varid{e}\mathbin{!}\Varid{w}} is exactly \ensuremath{\Varid{at}_\epsilon\;(\derivOp^{\ast}\;\Varid{e}\;\Varid{w})} generalized to indexable \ensuremath{\Varid{h}}, performing syntactic differentiation with respect to successive characters in \ensuremath{\Varid{w}} and applying \ensuremath{\Varid{at}_\epsilon} to the final resulting regular expression.

For efficiency, and sometimes even termination (with recursively defined languages), we will need to add some optimizations to the \ensuremath{\Conid{Additive}} and \ensuremath{\Conid{Semiring}} instances for \ensuremath{\Conid{RegExp}} in \figref{RegExp}:
\twocol{0.4}{
\begin{hscode}\SaveRestoreHook
\column{B}{@{}>{\hspre}l<{\hspost}@{}}%
\column{3}{@{}>{\hspre}l<{\hspost}@{}}%
\column{12}{@{}>{\hspre}l<{\hspost}@{}}%
\column{25}{@{}>{\hspre}l<{\hspost}@{}}%
\column{E}{@{}>{\hspre}l<{\hspost}@{}}%
\>[3]{}\Varid{p}\mathbin{+}\Varid{q}{}\<[12]%
\>[12]{}\mid \Varid{isZero}\;\Varid{p}{}\<[25]%
\>[25]{}\mathrel{=}\Varid{q}{}\<[E]%
\\
\>[12]{}\mid \Varid{isZero}\;\Varid{q}{}\<[25]%
\>[25]{}\mathrel{=}\Varid{p}{}\<[E]%
\\
\>[12]{}\mid \Varid{otherwise}{}\<[25]%
\>[25]{}\mathrel{=}\Varid{p}\mathbin{:\!\!+}\Varid{q}{}\<[E]%
\ColumnHook
\end{hscode}\resethooks
}{0.45}{
\begin{hscode}\SaveRestoreHook
\column{B}{@{}>{\hspre}l<{\hspost}@{}}%
\column{3}{@{}>{\hspre}l<{\hspost}@{}}%
\column{12}{@{}>{\hspre}l<{\hspost}@{}}%
\column{21}{@{}>{\hspre}l<{\hspost}@{}}%
\column{25}{@{}>{\hspre}l<{\hspost}@{}}%
\column{E}{@{}>{\hspre}l<{\hspost}@{}}%
\>[3]{}\Varid{p}\mathbin{*}\Varid{q}{}\<[12]%
\>[12]{}\mid \Varid{isZero}\;\Varid{p}{}\<[25]%
\>[25]{}\mathrel{=}\mathrm{0}{}\<[E]%
\\
\>[12]{}\mid \Varid{isOne}\;{}\<[21]%
\>[21]{}\Varid{p}{}\<[25]%
\>[25]{}\mathrel{=}\Varid{q}{}\<[E]%
\\
\>[12]{}\mid \Varid{otherwise}{}\<[25]%
\>[25]{}\mathrel{=}\Varid{p}\mathbin{:\!\!\conv}\Varid{q}{}\<[E]%
\ColumnHook
\end{hscode}\resethooks
}
For \ensuremath{\Varid{p}\mathbin{*}\Varid{q}}, we might also check whether \ensuremath{\Varid{q}} is \ensuremath{\mathrm{0}} or \ensuremath{\mathrm{1}}, but doing so itself leads to non-termination in right-recursive grammars.

As an alternative to repeated syntactic differentiation, we can reinterpret the original (syntactic) regular expression in another semiring as follows:
\begin{hscode}\SaveRestoreHook
\column{B}{@{}>{\hspre}l<{\hspost}@{}}%
\column{12}{@{}>{\hspre}c<{\hspost}@{}}%
\column{12E}{@{}l@{}}%
\column{18}{@{}>{\hspre}l<{\hspost}@{}}%
\column{22}{@{}>{\hspre}l<{\hspost}@{}}%
\column{34}{@{}>{\hspre}c<{\hspost}@{}}%
\column{34E}{@{}l@{}}%
\column{39}{@{}>{\hspre}l<{\hspost}@{}}%
\column{E}{@{}>{\hspre}l<{\hspost}@{}}%
\>[B]{}\Varid{regexp}\mathbin{::}(\Conid{StarSemiring}\;\Varid{x},\Conid{HasSingle}\;[\mskip1.5mu \Conid{Key}\;\Varid{h}\mskip1.5mu]\;\Varid{b}\;\Varid{x},\Conid{Semiring}\;\Varid{b})\Rightarrow \Conid{RegExp}\;\Varid{h}\;\Varid{b}\to \Varid{x}{}\<[E]%
\\
\>[B]{}\Varid{regexp}\;(\Conid{Char}\;\Varid{c}){}\<[22]%
\>[22]{}\mathrel{=}\Varid{single}\;[\mskip1.5mu \Varid{c}\mskip1.5mu]{}\<[E]%
\\
\>[B]{}\Varid{regexp}\;(\Conid{Value}\;\Varid{b}){}\<[22]%
\>[22]{}\mathrel{=}\Varid{value}\;\Varid{b}{}\<[E]%
\\
\>[B]{}\Varid{regexp}\;(\Varid{u}{}\<[12]%
\>[12]{}\mathbin{:\!\!+}{}\<[12E]%
\>[18]{}\Varid{v}){}\<[22]%
\>[22]{}\mathrel{=}\Varid{regexp}\;\Varid{u}{}\<[34]%
\>[34]{}\mathbin{+}{}\<[34E]%
\>[39]{}\Varid{regexp}\;\Varid{v}{}\<[E]%
\\
\>[B]{}\Varid{regexp}\;(\Varid{u}{}\<[12]%
\>[12]{}\mathbin{:\!\!\conv}{}\<[12E]%
\>[18]{}\Varid{v}){}\<[22]%
\>[22]{}\mathrel{=}\Varid{regexp}\;\Varid{u}{}\<[34]%
\>[34]{}\mathbin{*}{}\<[34E]%
\>[39]{}\Varid{regexp}\;\Varid{v}{}\<[E]%
\\
\>[B]{}\Varid{regexp}\;(\Conid{Star}\;\Varid{u}){}\<[22]%
\>[22]{}\mathrel{=}\closure{(\Varid{regexp}\;\Varid{u})}{}\<[E]%
\ColumnHook
\end{hscode}\resethooks
Next, we will see a choice of \ensuremath{\Varid{f}} that eliminates the syntactic overhead.

\sectionl{Tries}

\secref{Languages and the Monoid Semiring} provided an implementation of language recognition and its generalization to the monoid semiring (\ensuremath{\Varid{a}\to \Varid{b}} for monoid \ensuremath{\Varid{a}} and semiring \ensuremath{\Varid{b}}), packaged as instances of a few common algebraic abstractions (\ensuremath{\Conid{Additive}}, \ensuremath{\Conid{Semiring}} etc).
While simple and correct, these implementations are quite inefficient, primarily due to naive backtracking and redundant comparison.
\secref{Decomposing Functions from Lists} explored the nature of functions on lists, identifying a decomposition principle and its relationship to the vocabulary of semirings and related algebraic abstractions.
Applying this principle to a generalized form of regular expressions led to Brzozowski's algorithm, generalized from sets to functions in \secref{Regular Expressions}, providing an alternative to naive backtracking but still involving repeated syntactic manipulation as each candidate string is matched.
Nevertheless, with some syntactic optimizations and memoization, recognition speed with this technique can be fairly good \citep{Might2010YaccID,Adams2016CPP}.

As an alternative to regular expression differentiation, note that the problem of redundant comparison is solved elegantly by the classic trie (``prefix tree'') data structure introduced by Thue in 1912 \citep[Section 6.3]{Knuth1998ACP3}.
This data structure was later generalized to arbitrary (regular) algebraic data types \citep{Connelly1995GenTrie} and then from sets to functions \citep{Hinze2000GGT}.\out{ We'll explore the data type generalization later.\notefoot{Add a forward pointer, or remove the promise.}}
Restricting our attention to functions of \emph{lists} (``strings'' over some alphabet), we can formulate a simple trie data type along the lines of \ensuremath{(\mathrel\triangleleft)} from \secref{Decomposing Functions from Lists}, with an entry for \ensuremath{\varepsilon} and a sub-trie for each possible character:
\begin{hscode}\SaveRestoreHook
\column{B}{@{}>{\hspre}l<{\hspost}@{}}%
\column{39}{@{}>{\hspre}l<{\hspost}@{}}%
\column{E}{@{}>{\hspre}l<{\hspost}@{}}%
\>[B]{}\mathbf{data}\;\Conid{LTrie}\;\Varid{c}\;\Varid{b}\mathrel{=}\Varid{b}\mathrel{\Varid{:\!\!\triangleleft}}\Varid{c}\to \Conid{LTrie}\;\Varid{c}\;\Varid{b}{}\<[39]%
\>[39]{}\mbox{\onelinecomment  first guess}{}\<[E]%
\ColumnHook
\end{hscode}\resethooks
While this definition would work, we can get much better efficiency if we memoize the functions of \ensuremath{\Varid{c}}, e.g., as a generalized trie or a finite map.
Rather than commit to a particular representation for subtrie collections, let's replace the type parameter \ensuremath{\Varid{c}} with a functor \ensuremath{\Varid{h}} whose associated key type is \ensuremath{\Varid{c}}.
The functor-parametrized list trie is also known as the ``cofree comonad'' \citep{Uustalu2005EDP,Uustalu2008CNC,Uustalu2011RS,Hinze2013USR,Kmett2015MfL,Penner2017RSTT}.
\begin{hscode}\SaveRestoreHook
\column{B}{@{}>{\hspre}l<{\hspost}@{}}%
\column{E}{@{}>{\hspre}l<{\hspost}@{}}%
\>[B]{}\mathbf{data}\;\Conid{Cofree}\;\Varid{h}\;\Varid{b}\mathrel{=}\Varid{b}\mathrel{\Varid{:\!\!\triangleleft}}\Varid{h}\;(\Conid{Cofree}\;\Varid{h}\;\Varid{b}){}\<[E]%
\ColumnHook
\end{hscode}\resethooks

The similarity between \ensuremath{\Conid{Cofree}\;\Varid{h}\;\Varid{b}} and the function decomposition from \secref{Decomposing Functions from Lists} (motivating the constructor name ``\ensuremath{\mathrel{\Varid{:\!\!\triangleleft}}}'') makes for easy instance calculation.
As with \ensuremath{\Pow\;\Varid{a}} and \ensuremath{\Conid{Map}\;\Varid{a}\;\Varid{b}}, we can define a trie counterpart to the free monoid semiring \ensuremath{[\mskip1.5mu \Varid{c}\mskip1.5mu]\to \Varid{b}}.
\begin{theorem}[\provedIn{theorem:Cofree}]\thmlabel{Cofree}
Given the definitions in \figrefdef{Cofree}{List tries denoting \ensuremath{[\mskip1.5mu \Varid{c}\mskip1.5mu]\to \Varid{b}}}{
\begin{hscode}\SaveRestoreHook
\column{B}{@{}>{\hspre}l<{\hspost}@{}}%
\column{3}{@{}>{\hspre}l<{\hspost}@{}}%
\column{11}{@{}>{\hspre}l<{\hspost}@{}}%
\column{E}{@{}>{\hspre}l<{\hspost}@{}}%
\>[B]{}\mathbf{infix}\;\mathrm{1}\mathrel{\Varid{:\!\!\triangleleft}}{}\<[E]%
\\
\>[B]{}\mathbf{data}\;\Conid{Cofree}\;\Varid{h}\;\Varid{b}\mathrel{=}\Varid{b}\mathrel{\Varid{:\!\!\triangleleft}}\Varid{h}\;(\Conid{Cofree}\;\Varid{h}\;\Varid{b})\;\mathbf{deriving}\;\Conid{Functor}{}\<[E]%
\\[\blanklineskip]%
\>[B]{}\mathbf{instance}\;\Conid{Indexable}\;\Varid{c}\;(\Conid{Cofree}\;\Varid{h}\;\Varid{b})\;(\Varid{h}\;(\Conid{Cofree}\;\Varid{h}\;\Varid{b}))\Rightarrow \Conid{Indexable}\;[\mskip1.5mu \Varid{c}\mskip1.5mu]\;\Varid{b}\;(\Conid{Cofree}\;\Varid{h}\;\Varid{b})\;\mathbf{where}{}\<[E]%
\\
\>[B]{}\hsindent{3}{}\<[3]%
\>[3]{}(\mathbin{!})\;(\Varid{b}\mathrel{\Varid{:\!\!\triangleleft}}\Varid{dp})\mathrel{=}\Varid{b}\mathrel\triangleleft(\mathbin{!})\hsdot{\circ }{.\:}(\mathbin{!})\;\Varid{dp}\mbox{\onelinecomment  \ensuremath{(\Varid{b}\mathrel{\Varid{:\!\!\triangleleft}}\Varid{dp})\mathbin{!}\Varid{w}\mathrel{=}\mathbf{case}\;\Varid{w}\;\mathbf{of}\;\{\mskip1.5mu [\mskip1.5mu \mskip1.5mu]\to \Varid{b};\Varid{c}\mathbin{:}\Varid{cs}\to \Varid{dp}\mathbin{!}\Varid{c}\mathbin{!}\Varid{cs}\mskip1.5mu\}}}{}\<[E]%
\\[\blanklineskip]%
\>[B]{}\mathbf{instance}\;(\Conid{Additive}\;(\Varid{h}\;(\Conid{Cofree}\;\Varid{h}\;\Varid{b})),\Conid{Additive}\;\Varid{b})\Rightarrow \Conid{Additive}\;(\Conid{Cofree}\;\Varid{h}\;\Varid{b})\;\mathbf{where}{}\<[E]%
\\
\>[B]{}\hsindent{3}{}\<[3]%
\>[3]{}\mathrm{0}\mathrel{=}\mathrm{0}\mathrel{\Varid{:\!\!\triangleleft}}\mathrm{0}{}\<[E]%
\\
\>[B]{}\hsindent{3}{}\<[3]%
\>[3]{}(\Varid{a}\mathrel{\Varid{:\!\!\triangleleft}}\Varid{dp})\mathbin{+}(\Varid{b}\mathrel{\Varid{:\!\!\triangleleft}}\Varid{dq})\mathrel{=}\Varid{a}\mathbin{+}\Varid{b}\mathrel{\Varid{:\!\!\triangleleft}}\Varid{dp}\mathbin{+}\Varid{dq}{}\<[E]%
\\[\blanklineskip]%
\>[B]{}\mathbf{instance}\;(\Conid{Functor}\;\Varid{h},\Conid{Semiring}\;\Varid{b})\Rightarrow \Conid{LeftSemimodule}\;\Varid{b}\;(\Conid{Cofree}\;\Varid{h}\;\Varid{b})\;\mathbf{where}{}\<[E]%
\\
\>[B]{}\hsindent{3}{}\<[3]%
\>[3]{}(\mathbin{\hat{\cdot}})\;\Varid{s}\mathrel{=}\Varid{fmap}\;(\Varid{s}\;{}\mathbin{*}){}\<[E]%
\\[\blanklineskip]%
\>[B]{}\mathbf{instance}\;{}\<[11]%
\>[11]{}(\Conid{Functor}\;\Varid{h},\Conid{Additive}\;(\Varid{h}\;(\Conid{Cofree}\;\Varid{h}\;\Varid{b})),\Conid{Semiring}\;\Varid{b},\Conid{IsZero}\;\Varid{b})\Rightarrow {}\<[E]%
\\
\>[11]{}\Conid{Semiring}\;(\Conid{Cofree}\;\Varid{h}\;\Varid{b})\;\mathbf{where}{}\<[E]%
\\
\>[B]{}\hsindent{3}{}\<[3]%
\>[3]{}\mathrm{1}\mathrel{=}\mathrm{1}\mathrel{\Varid{:\!\!\triangleleft}}\mathrm{0}{}\<[E]%
\\
\>[B]{}\hsindent{3}{}\<[3]%
\>[3]{}(\Varid{a}\mathrel{\Varid{:\!\!\triangleleft}}\Varid{dp})\mathbin{*}\Varid{q}\mathrel{=}\Varid{a}\cdot\Varid{q}\mathbin{+}(\mathrm{0}\mathrel{\Varid{:\!\!\triangleleft}}\Varid{fmap}\;(\mathbin{*}{}\;\Varid{q})\;\Varid{dp}){}\<[E]%
\\[\blanklineskip]%
\>[B]{}\mathbf{instance}\;{}\<[11]%
\>[11]{}(\Conid{Functor}\;\Varid{h},\Conid{Additive}\;(\Varid{h}\;(\Conid{Cofree}\;\Varid{h}\;\Varid{b})),\Conid{StarSemiring}\;\Varid{b},\Conid{IsZero}\;\Varid{b})\Rightarrow {}\<[E]%
\\
\>[11]{}\Conid{StarSemiring}\;(\Conid{Cofree}\;\Varid{h}\;\Varid{b})\;\mathbf{where}{}\<[E]%
\\
\>[B]{}\hsindent{3}{}\<[3]%
\>[3]{}\closure{(\Varid{a}\mathrel{\Varid{:\!\!\triangleleft}}\Varid{dp})}\mathrel{=}\Varid{q}\;\mathbf{where}\;\Varid{q}\mathrel{=}\closure{\Varid{a}}\cdot(\mathrm{1}\mathrel{\Varid{:\!\!\triangleleft}}\Varid{fmap}\;(\mathbin{*}{}\;\Varid{q})\;\Varid{dp}){}\<[E]%
\\[\blanklineskip]%
\>[B]{}\mathbf{instance}\;{}\<[11]%
\>[11]{}(\Conid{HasSingle}\;(\Conid{Cofree}\;\Varid{h}\;\Varid{b})\;\Varid{h},\Conid{Additive}\;(\Varid{h}\;(\Conid{Cofree}\;\Varid{h}\;\Varid{b})),\Conid{Additive}\;\Varid{b})\Rightarrow {}\<[E]%
\\
\>[11]{}\Conid{HasSingle}\;\Varid{b}\;(\Conid{Cofree}\;\Varid{h})\;\mathbf{where}{}\<[E]%
\\
\>[B]{}\hsindent{3}{}\<[3]%
\>[3]{}\Varid{w}\mapsto\Varid{b}\mathrel{=}\Varid{foldr}\;(\lambda\, \Varid{c}\;\Varid{t}\to \mathrm{0}\mathrel{\Varid{:\!\!\triangleleft}}\Varid{c}\mapsto\Varid{t})\;(\Varid{b}\mathrel{\Varid{:\!\!\triangleleft}}\mathrm{0})\;\Varid{w}{}\<[E]%
\\[\blanklineskip]%
\>[B]{}\mathbf{instance}\;{}\<[11]%
\>[11]{}(\Conid{Additive}\;(\Varid{h}\;(\Conid{Cofree}\;\Varid{h}\;\Varid{b})),\Conid{IsZero}\;(\Varid{h}\;(\Conid{Cofree}\;\Varid{h}\;\Varid{b})),\Conid{IsZero}\;\Varid{b})\Rightarrow {}\<[E]%
\\
\>[11]{}\Conid{IsZero}\;(\Conid{Cofree}\;\Varid{h}\;\Varid{b})\;\mathbf{where}{}\<[E]%
\\
\>[B]{}\hsindent{3}{}\<[3]%
\>[3]{}\Varid{isZero}\;(\Varid{a}\mathrel{\Varid{:\!\!\triangleleft}}\Varid{dp})\mathrel{=}\Varid{isZero}\;\Varid{a}\mathrel{\wedge}\Varid{isZero}\;\Varid{dp}{}\<[E]%
\\[\blanklineskip]%
\>[B]{}\mathbf{instance}\;{}\<[11]%
\>[11]{}(\Conid{Functor}\;\Varid{h},\Conid{Additive}\;(\Varid{h}\;(\Conid{Cofree}\;\Varid{h}\;\Varid{b})),\Conid{IsZero}\;\Varid{b},\Conid{IsZero}\;(\Varid{h}\;(\Conid{Cofree}\;\Varid{h}\;\Varid{b})),\Conid{IsOne}\;\Varid{b})\Rightarrow {}\<[E]%
\\
\>[11]{}\Conid{IsOne}\;(\Conid{Cofree}\;\Varid{h}\;\Varid{b})\;\mathbf{where}{}\<[E]%
\\
\>[B]{}\hsindent{3}{}\<[3]%
\>[3]{}\Varid{isOne}\;(\Varid{a}\mathrel{\Varid{:\!\!\triangleleft}}\Varid{dp})\mathrel{=}\Varid{isOne}\;\Varid{a}\mathrel{\wedge}\Varid{isZero}\;\Varid{dp}{}\<[E]%
\ColumnHook
\end{hscode}\resethooks
\vspace{-4ex}
}, \ensuremath{(\mathbin{!})} is a homomorphism with respect to each instantiated class.
\end{theorem}

Although the \ensuremath{(\mathrel\triangleleft)} decomposition in \secref{Decomposing Functions from Lists} was inspired by wanting to understand the essence of regular expression derivatives, the application to tries is in retrospect more straightforward, since the representation directly mirrors the decomposition.\out{ Pleasantly, this trie data structure is a classic, though perhaps not in its lazy infinite form for use as a language representation.}
Applying the \ensuremath{(\mathrel\triangleleft)} decomposition to tries also appears to be more streamlined than the application to regular expressions.
During matching, the next character in the candidate string is used to directly index to the relevant derivative (sub-trie), efficiently bypassing all other paths.
As one might hope,
\ensuremath{(\mathbin{!})} on \ensuremath{\Conid{Cofree}\;\Varid{h}} is another homomorphism:
\begin{theorem}[\provedIn{theorem:Cofree hom}]\thmlabel{Cofree hom}
Given the definitions in \figreftwo{Cofree}{Cofree hom}, if \ensuremath{(\mathbin{!})} on \ensuremath{\Varid{h}} behaves like\notefoot{Come up with a better phrasing of this condition, and use it consistently} \ensuremath{(\to )\;(\Conid{Key}\;\Varid{h})}, then \ensuremath{\Conid{Cofree}\;\Varid{h}} is a comonad homomorphism from \ensuremath{\Conid{Cofree}\;\Varid{h}} to \ensuremath{(\to )\;(\Conid{Key}\;\Varid{h})}.
\end{theorem}
\figdef{Cofree hom}{\ensuremath{\Conid{Comonad}} and instances}{
\begin{hscode}\SaveRestoreHook
\column{B}{@{}>{\hspre}l<{\hspost}@{}}%
\column{3}{@{}>{\hspre}l<{\hspost}@{}}%
\column{13}{@{}>{\hspre}l<{\hspost}@{}}%
\column{14}{@{}>{\hspre}l<{\hspost}@{}}%
\column{25}{@{}>{\hspre}l<{\hspost}@{}}%
\column{E}{@{}>{\hspre}l<{\hspost}@{}}%
\>[B]{}\mathbf{instance}\;\Conid{Functor}\;\Varid{w}\Rightarrow \Conid{Comonad}\;\Varid{w}\;\mathbf{where}{}\<[E]%
\\
\>[B]{}\hsindent{3}{}\<[3]%
\>[3]{}\Varid{coreturn}{}\<[13]%
\>[13]{}\mathbin{::}\Varid{w}\;\Varid{b}\to \Varid{b}{}\<[E]%
\\
\>[B]{}\hsindent{3}{}\<[3]%
\>[3]{}\Varid{cojoin}{}\<[13]%
\>[13]{}\mathbin{::}\Varid{w}\;\Varid{b}\to \Varid{w}\;(\Varid{w}\;\Varid{b}){}\<[E]%
\\[\blanklineskip]%
\>[B]{}\mathbf{instance}\;\Conid{Monoid}\;\Varid{a}\Rightarrow \Conid{Comonad}\;((\to )\;\Varid{a})\;\mathbf{where}{}\<[E]%
\\
\>[B]{}\hsindent{3}{}\<[3]%
\>[3]{}\Varid{coreturn}\;{}\<[13]%
\>[13]{}\Varid{f}\mathrel{=}\Varid{f}\;\varepsilon{}\<[E]%
\\
\>[B]{}\hsindent{3}{}\<[3]%
\>[3]{}\Varid{cojoin}\;{}\<[13]%
\>[13]{}\Varid{f}\mathrel{=}\lambda\, \Varid{u}\to \lambda\, \Varid{v}\to \Varid{f}\;(\Varid{u} \diamond \Varid{v}){}\<[E]%
\\[\blanklineskip]%
\>[B]{}\mathbf{instance}\;\Conid{Functor}\;\Varid{h}\Rightarrow \Conid{Functor}\;(\Conid{Cofree}\;\Varid{h})\;\mathbf{where}{}\<[E]%
\\
\>[B]{}\hsindent{3}{}\<[3]%
\>[3]{}\Varid{fmap}\;\Varid{f}\;(\Varid{a}\mathrel{\Varid{:\!\!\triangleleft}}\Varid{ds})\mathrel{=}\Varid{f}\;\Varid{a}\mathrel{\Varid{:\!\!\triangleleft}}\Varid{fmap}\;(\Varid{fmap}\;\Varid{f})\;\Varid{ds}{}\<[E]%
\\[\blanklineskip]%
\>[B]{}\mathbf{instance}\;\Conid{Functor}\;\Varid{h}\Rightarrow \Conid{Comonad}\;(\Conid{Cofree}\;\Varid{h})\;\mathbf{where}{}\<[E]%
\\
\>[B]{}\hsindent{3}{}\<[3]%
\>[3]{}\Varid{coreturn}\;{}\<[14]%
\>[14]{}(\Varid{a}\mathrel{\Varid{:\!\!\triangleleft}}\anonymous ){}\<[25]%
\>[25]{}\mathrel{=}\Varid{a}{}\<[E]%
\\
\>[B]{}\hsindent{3}{}\<[3]%
\>[3]{}\Varid{cojoin}\;\Varid{t}\mathord{@}{}\<[14]%
\>[14]{}(\anonymous \mathrel{\Varid{:\!\!\triangleleft}}\Varid{ds}){}\<[25]%
\>[25]{}\mathrel{=}\Varid{t}\mathrel{\Varid{:\!\!\triangleleft}}\Varid{fmap}\;\Varid{cojoin}\;\Varid{ds}{}\<[E]%
\ColumnHook
\end{hscode}\resethooks
}

\sectionl{Performance}

\nc\stat[6]{#1 & #2 & #3 & #4 & #5 \\ \hline}
\nc\incstat[1]{\input{test/stats/#1.txt}}

\nc\hang{\ensuremath{\infty}}

While the implementation has had no performance tuning and only rudimentary benchmarking, we can at least get a sanity check on performance and functionality.
\figrefdef{examples}{Examples}{
\begin{hscode}\SaveRestoreHook
\column{B}{@{}>{\hspre}l<{\hspost}@{}}%
\column{5}{@{}>{\hspre}l<{\hspost}@{}}%
\column{E}{@{}>{\hspre}l<{\hspost}@{}}%
\>[B]{}\Varid{a}{}\<[5]%
\>[5]{}\mathrel{=}\Varid{single}\;\text{\tt \char34 a\char34}{}\<[E]%
\\
\>[B]{}\Varid{b}{}\<[5]%
\>[5]{}\mathrel{=}\Varid{single}\;\text{\tt \char34 b\char34}{}\<[E]%
\\[\blanklineskip]%
\>[B]{}\Varid{atoz}\mathrel{=}\Varid{sum}\;[\mskip1.5mu \Varid{single}\;[\mskip1.5mu \Varid{c}\mskip1.5mu]\mid \Varid{c}\leftarrow [\mskip1.5mu \text{\tt 'a'}\mathinner{\ldotp\ldotp}\text{\tt 'z'}\mskip1.5mu]\mskip1.5mu]{}\<[E]%
\\[\blanklineskip]%
\>[B]{}\Varid{fishy}\mathrel{=}\closure{\Varid{atoz}}\mathbin{*}\Varid{single}\;\text{\tt \char34 fish\char34}\mathbin{*}\closure{\Varid{atoz}}{}\<[E]%
\\[\blanklineskip]%
\>[B]{}\Varid{anbn}\mathrel{=}\mathrm{1}\mathbin{+}\Varid{a}\mathbin{*}\Varid{anbn}\mathbin{*}\Varid{b}{}\<[E]%
\\[\blanklineskip]%
\>[B]{}\Varid{dyck}\mathrel{=}\closure{(\Varid{single}\;\text{\tt \char34 [\char34}\mathbin{*}\Varid{dyck}\mathbin{*}\Varid{single}\;\text{\tt \char34 ]\char34})}{}\<[E]%
\ColumnHook
\end{hscode}\resethooks
} shows the source code for a collection of examples, all polymorphic in the choice of semiring.
The \ensuremath{\Varid{atoz}} language contains single letters from `a' to `z'.
The examples \ensuremath{\Varid{anbn}} and \ensuremath{\Varid{dyck}} are two classic, non-regular, context-free languages: \ensuremath{\set{\Varid{a}^\Varid{n}{\,}\Varid{b}^\Varid{n}\mid\Varid{n}\mathbin{\in}\mathbb N}} and the Dyck language of balanced brackets.

\figrefdef{stats}{Running times for examples in \figref{examples}}{
\begin{center}
\begin{tabular}{|c|c|c|c|c|}
\hline
\stat{Example}{\ensuremath{\Conid{RegExp}_\to }}{\ensuremath{\Conid{RegExp}_\Conid{Map}}}{\ensuremath{\Conid{Cofree}_\to }}{\ensuremath{\Conid{Cofree}_\Conid{Map}}}{} \hline
\hline
\stat{\ensuremath{\closure{\Varid{a}}}}{30.56 $\mu{}$s}{22.45 $\mu{}$s}{5.258 ms}{2.624 $\mu{}$s}{}
\stat{\ensuremath{\closure{\Varid{atoz}}}}{690.4 $\mu{}$s}{690.9 $\mu{}$s}{10.89 ms}{3.574 $\mu{}$s}{}
\stat{\ensuremath{\closure{\Varid{a}}\mathbin{*}\closure{\Varid{a}}}}{2.818 ms}{1.274 ms}{601.6 ms}{2.619 $\mu{}$s}{}
\stat{\ensuremath{\closure{\Varid{a}}\mathbin{*}\closure{\Varid{b}}}}{52.26 $\mu{}$s}{36.59 $\mu{}$s}{14.40 ms}{2.789 $\mu{}$s}{}
\stat{\ensuremath{\closure{\Varid{a}}\mathbin{*}\Varid{b}\mathbin{*}\closure{\Varid{a}}}}{56.53 $\mu{}$s}{49.21 $\mu{}$s}{14.58 ms}{2.798 $\mu{}$s}{}
\stat{\ensuremath{\Varid{fishy}}}{1.276 ms}{2.528 ms}{29.73 ms}{4.233 $\mu{}$s}{}
\stat{\ensuremath{\Varid{anbn}}}{1.293 ms}{\hang}{12.12 ms}{2.770 $\mu{}$s}{}
\stat{\ensuremath{\Varid{dyck}}}{254.9 $\mu{}$s}{\hang}{24.77 ms}{3.062 $\mu{}$s}{}
\end{tabular}
\end{center}
} gives some execution times for these examples measured with the \emph{criterion} library \citep{criterion}, compiled with GHC 8.6.3, and running on a late 2013 MacBook Pro.
(Note milliseconds vs microseconds---``ms'' vs ``$\mu{}$s''.)
Each example is interpreted in four semirings: \ensuremath{\Conid{RegExp}\;((\to )\;\Conid{Char})\;\mathbb N}, \ensuremath{\Conid{RegExp}\;(\Conid{Map}\;\Conid{Char})\;\mathbb N}, \ensuremath{\Conid{Cofree}\;((\to )\;\Conid{Char}\;\mathbb N)}, and \ensuremath{\Conid{Cofree}\;(\Conid{Map}\;\Conid{Char})\;\mathbb N}.
Each interpretation of each language is given a matching input string of length 100; and matches are counted, thanks to use of the \ensuremath{\mathbb N} semiring.
(The \ensuremath{\closure{\Varid{a}}\mathbin{*}\closure{\Varid{a}}} example matches in 101 ways, while the others match uniquely.)
As the figure shows, memoization (via \ensuremath{\Conid{Map}}) is only moderately helpful (and occasionally harmful) for \ensuremath{\Conid{RegExp}}.
\ensuremath{\Conid{Cofree}}, on the other hand, performs terribly without memoization and (in these examples) 2K to 230K times faster with memoization.
Here, memoized \ensuremath{\Conid{Cofree}} performs between 8.5 and 485 times faster than memoized \ensuremath{\Conid{RegExp}} and between 11.5 and 1075 times faster than nonmemoized \ensuremath{\Conid{RegExp}}.
The two recursively defined examples fail to terminate with \ensuremath{\Conid{RegExp}\;\Conid{Map}}, perhaps because the implementation (\secref{Regular Expressions}) lacks one more crucial tricks \citep{Might2010YaccID}.
Other \ensuremath{\Conid{RegExp}} improvements \citep{Might2010YaccID,Adams2016CPP} might narrow the gap further, and careful study and optimization of the \ensuremath{\Conid{Cofree}} implementation (\figref{Cofree}) might widen it.


\sectionl{Convolution}

Consider again the definition of multiplication in the monoid semiring, on \ensuremath{\Varid{f},\Varid{g}\mathbin{::}\Varid{a}\to \Varid{b}} from \figref{monoid semiring}.
\begin{hscode}\SaveRestoreHook
\column{B}{@{}>{\hspre}l<{\hspost}@{}}%
\column{8}{@{}>{\hspre}l<{\hspost}@{}}%
\column{E}{@{}>{\hspre}l<{\hspost}@{}}%
\>[B]{}\Varid{f}\mathbin{*}\Varid{g}{}\<[8]%
\>[8]{}\mathrel{=}\bigOp\sum{\Varid{u},\Varid{v}}{0}\;\Varid{u} \diamond \Varid{v}\mapsto\Varid{f}\;\Varid{u}\mathbin{*}\Varid{g}\;\Varid{v}{}\<[E]%
\ColumnHook
\end{hscode}\resethooks
As in \secref{Languages and the Monoid Semiring}, specializing the \emph{codomain} to \ensuremath{\Conid{Bool}}%
, we get
\begin{hscode}\SaveRestoreHook
\column{B}{@{}>{\hspre}l<{\hspost}@{}}%
\column{E}{@{}>{\hspre}l<{\hspost}@{}}%
\>[B]{}\Varid{f}\mathbin{*}\Varid{g}\mathrel{=}\bigOp\bigvee{\Varid{u},\Varid{v}}{0}\;\Varid{u} \diamond \Varid{v}\mapsto\Varid{f}\;\Varid{u}\mathrel{\wedge}\Varid{g}\;\Varid{v}{}\<[E]%
\ColumnHook
\end{hscode}\resethooks
Using
the set/predicate isomorphism from \secref{Calculating Instances from Homomorphisms}, we can translate this definition from predicates to ``languages'' (sets of values in some monoid):
\begin{hscode}\SaveRestoreHook
\column{B}{@{}>{\hspre}l<{\hspost}@{}}%
\column{E}{@{}>{\hspre}l<{\hspost}@{}}%
\>[B]{}\Varid{f}\mathbin{*}\Varid{g}\mathrel{=}\set{\Varid{u} \diamond \Varid{v}\mid \Varid{u}\mathbin{\in}\Varid{f}\mathrel{\wedge}\Varid{v}\mathbin{\in}\Varid{g}}{}\<[E]%
\ColumnHook
\end{hscode}\resethooks
which is the definition of the concatenation of two languages from  \secref{Languages and the Monoid Semiring}.
Likewise, by specializing the \emph{domain} of the functions to sequences (from general monoids), we got efficient matching of semiring-generalized ``languages'', as in \secreftwo{Decomposing Functions from Lists}{Tries}, which translated to regular expressions (\secref{Regular Expressions}), generalizing work of \citet{Brzozowski64}.

Let's now consider specializing the functions' domains to \emph{integers} rather than sequences, recalling that integers (and numeric types in general) form a monoid under addition.
\vspace{-2ex}
\begin{spacing}{1.5}
\begin{hscode}\SaveRestoreHook
\column{B}{@{}>{\hspre}l<{\hspost}@{}}%
\column{10}{@{}>{\hspre}l<{\hspost}@{}}%
\column{77}{@{}>{\hspre}l<{\hspost}@{}}%
\column{E}{@{}>{\hspre}l<{\hspost}@{}}%
\>[B]{}\Varid{f}\mathbin{*}\Varid{g}{}\<[10]%
\>[10]{}\mathrel{=}\bigOp\sum{\Varid{u},\Varid{v}}{0}\;\Varid{u}\mathbin{+}\Varid{v}\mapsto\Varid{f}\;\Varid{u}\mathbin{*}\Varid{g}\;\Varid{v}{}\<[77]%
\>[77]{}\mbox{\onelinecomment  \figref{monoid semiring} with \ensuremath{( \diamond )\mathrel{=}(\mathbin{+})}}{}\<[E]%
\\
\>[10]{}\mathrel{=}\lambda\, \Varid{w}\to \bigOp\sum{\Varid{u},\Varid{v}\;\!\!\\\!\!\;\Varid{u}\mathbin{+}\Varid{v}\mathrel{=}\Varid{w}}{1.4}\;\Varid{f}\;\Varid{u}\mathbin{*}\Varid{g}\;\Varid{v}{}\<[77]%
\>[77]{}\mbox{\onelinecomment  equivalent definition}{}\<[E]%
\\
\>[10]{}\mathrel{=}\lambda\, \Varid{w}\to \bigOp\sum{\Varid{u},\Varid{v}\;\!\!\\\!\!\;\Varid{v}\mathrel{=}\Varid{w}\mathbin{-}\Varid{u}}{1.4}\;\Varid{f}\;\Varid{u}\mathbin{*}\Varid{g}\;\Varid{v}{}\<[77]%
\>[77]{}\mbox{\onelinecomment  solve \ensuremath{\Varid{u}\mathbin{+}\Varid{v}\mathrel{=}\Varid{w}} for \ensuremath{\Varid{v}}}{}\<[E]%
\\
\>[10]{}\mathrel{=}\lambda\, \Varid{w}\to \bigOp\sum{\Varid{u}}{0}\;\Varid{f}\;\Varid{u}\mathbin{*}\Varid{g}\;(\Varid{w}\mathbin{-}\Varid{u}){}\<[77]%
\>[77]{}\mbox{\onelinecomment  substitute \ensuremath{\Varid{w}\mathbin{-}\Varid{u}} for \ensuremath{\Varid{v}}}{}\<[E]%
\ColumnHook
\end{hscode}\resethooks
\end{spacing}
\vspace{-3ex}
\noindent
This last form is the standard definition of one-dimensional, discrete \emph{convolution} \citep[Chapter 6]{Smith1997SEG}.\footnote{Note that this reasoning applies to \emph{any} group (monoid with inverses).}
Therefore, just as monoid semiring multiplication generalizes language concatenation (via the predicate/set isomorphism), it also generalizes the usual notion of discrete convolution.
Moreover, if the domain is a continuous type such as \ensuremath{\mathbb R} or \ensuremath{\mathbb C}, we can reinterpret summation as integration, resulting in \emph{continuous} convolution.
Additionally, for multi-dimensional (discrete or continuous) convolution, we can simply use tuples of scalar indices for \ensuremath{\Varid{w}} and \ensuremath{\Varid{u}}, defining tuple addition and subtraction componentwise.
\notefoot{More generally, cartesian products of monoids are also monoids.
Consider multi-dimensional convolution in which different dimensions have different types, even mixing discrete and continuous, and maybe even sequences and numbers.
At the least, it's useful to combine finite dimensions of different sizes.}
Alternatively, curry, convolve, and uncurry, exploiting the fact that \ensuremath{\Varid{curry}} is a semiring homomorphism (\thmref{curry semiring}).
\notefoot{Mention the connection between generalized tries and currying.}

\note{Maybe give some convolution examples.}

What if we use functions from \ensuremath{\mathbb N} rather than from \ensuremath{\mathbb Z}?
Because \ensuremath{\mathbb N \simeq [\mskip1.5mu ()\mskip1.5mu]} (essentially, Peano numbers), we can directly use the definitions in \secref{Decomposing Functions from Lists} for domain \ensuremath{[\mskip1.5mu \Varid{c}\mskip1.5mu]}, specialized to \ensuremath{\Varid{c}\mathrel{=}()}.
As a suitable indexable functor, we can simply use the identity functor:
\begin{hscode}\SaveRestoreHook
\column{B}{@{}>{\hspre}l<{\hspost}@{}}%
\column{3}{@{}>{\hspre}l<{\hspost}@{}}%
\column{21}{@{}>{\hspre}l<{\hspost}@{}}%
\column{E}{@{}>{\hspre}l<{\hspost}@{}}%
\>[B]{}\mathbf{newtype}\;\Conid{Identity}\;\Varid{b}\mathrel{=}\Conid{Identity}\;\Varid{b}\;\mathbf{deriving}{}\<[E]%
\\
\>[B]{}\hsindent{3}{}\<[3]%
\>[3]{}(\Conid{Functor},\Conid{Additive},\Conid{IsZero},\Conid{IsOne},\Conid{LeftSemimodule}\;\Varid{s},\Conid{Semiring}){}\<[E]%
\\[\blanklineskip]%
\>[B]{}\mathbf{instance}\;\Conid{Indexable}\;{}\<[21]%
\>[21]{}()\;\Varid{b}\;(\Conid{Identity}\;\Varid{b})\;\mathbf{where}\;\Conid{Identity}\;\Varid{a}\mathbin{!}()\mathrel{=}\Varid{a}{}\<[E]%
\\
\>[B]{}\mathbf{instance}\;\Conid{HasSingle}\;{}\<[21]%
\>[21]{}()\;\Varid{b}\;(\Conid{Identity}\;\Varid{b})\;\mathbf{where}\;()\mapsto\Varid{b}\mathrel{=}\Conid{Identity}\;\Varid{b}{}\<[E]%
\ColumnHook
\end{hscode}\resethooks
The type \ensuremath{\Conid{Cofree}\;\Conid{Identity}} is isomorphic to \emph{streams} (infinite-only lists).
Inlining and simplification during compilation might eliminate all of the run-time overhead of introducing the identity functor.

Just as \ensuremath{\Conid{Cofree}\;\Conid{Identity}} gives (necessarily infinite) streams, \ensuremath{\Conid{Cofree}\;\Conid{Maybe}} gives (possibly finite) \emph{nonempty lists} \citep{Uustalu2008CNC, Maguire2016}.
As with finite maps, we can interpret absence (\ensuremath{\Conid{Nothing}}) as \ensuremath{\mathrm{0}}%
:
\begin{hscode}\SaveRestoreHook
\column{B}{@{}>{\hspre}l<{\hspost}@{}}%
\column{3}{@{}>{\hspre}l<{\hspost}@{}}%
\column{12}{@{}>{\hspre}l<{\hspost}@{}}%
\column{13}{@{}>{\hspre}l<{\hspost}@{}}%
\column{26}{@{}>{\hspre}l<{\hspost}@{}}%
\column{E}{@{}>{\hspre}l<{\hspost}@{}}%
\>[B]{}\mathbf{instance}\;\Conid{Additive}\;\Varid{b}\Rightarrow \Conid{Indexable}\;()\;\Varid{b}\;(\Conid{Maybe}\;\Varid{b})\;\mathbf{where}{}\<[E]%
\\
\>[B]{}\hsindent{3}{}\<[3]%
\>[3]{}\Conid{Nothing}{}\<[12]%
\>[12]{}\mathbin{!}()\mathrel{=}\mathrm{0}{}\<[E]%
\\
\>[B]{}\hsindent{3}{}\<[3]%
\>[3]{}\Conid{Just}\;\Varid{b}{}\<[12]%
\>[12]{}\mathbin{!}()\mathrel{=}\Varid{b}{}\<[E]%
\\[\blanklineskip]%
\>[B]{}\mathbf{instance}\;(\Conid{IsZero}\;\Varid{b},\Conid{Additive}\;\Varid{b})\Rightarrow \Conid{HasSingle}\;()\;\Varid{b}\;(\Conid{Maybe}\;\Varid{b})\;\mathbf{where}{}\<[E]%
\\
\>[B]{}\hsindent{3}{}\<[3]%
\>[3]{}()\mapsto\Varid{b}{}\<[13]%
\>[13]{}\mid \Varid{isZero}\;\Varid{b}{}\<[26]%
\>[26]{}\mathrel{=}\Conid{Nothing}{}\<[E]%
\\
\>[13]{}\mid \Varid{otherwise}{}\<[26]%
\>[26]{}\mathrel{=}\Conid{Just}\;\Varid{b}{}\<[E]%
\ColumnHook
\end{hscode}\resethooks

Alternatively, define instances directly for lists, specified by a denotation of \ensuremath{[\mskip1.5mu \Varid{b}\mskip1.5mu]} as \ensuremath{\mathbb N\to \Varid{b}}.
The instances resemble those in \figref{Cofree}, but have an extra case for the empty list and no \ensuremath{\Varid{fmap}}:
\begin{hscode}\SaveRestoreHook
\column{B}{@{}>{\hspre}l<{\hspost}@{}}%
\column{3}{@{}>{\hspre}l<{\hspost}@{}}%
\column{7}{@{}>{\hspre}c<{\hspost}@{}}%
\column{7E}{@{}l@{}}%
\column{10}{@{}>{\hspre}l<{\hspost}@{}}%
\column{13}{@{}>{\hspre}l<{\hspost}@{}}%
\column{14}{@{}>{\hspre}l<{\hspost}@{}}%
\column{20}{@{}>{\hspre}l<{\hspost}@{}}%
\column{21}{@{}>{\hspre}l<{\hspost}@{}}%
\column{E}{@{}>{\hspre}l<{\hspost}@{}}%
\>[B]{}\mathbf{instance}\;\Conid{Additive}\;\Varid{b}\Rightarrow \Conid{Indexable}\;\mathbb N\;\Varid{b}\;[\mskip1.5mu \Varid{b}\mskip1.5mu]\;\mathbf{where}{}\<[E]%
\\
\>[B]{}\hsindent{3}{}\<[3]%
\>[3]{}[\mskip1.5mu \mskip1.5mu]\mathbin{!}\anonymous \mathrel{=}\mathrm{0}{}\<[E]%
\\
\>[B]{}\hsindent{3}{}\<[3]%
\>[3]{}(\Varid{b}{}\<[7]%
\>[7]{}\mathbin{:}{}\<[7E]%
\>[10]{}\anonymous {}\<[14]%
\>[14]{})\mathbin{!}\mathrm{0}{}\<[21]%
\>[21]{}\mathrel{=}\Varid{b}{}\<[E]%
\\
\>[B]{}\hsindent{3}{}\<[3]%
\>[3]{}(\anonymous {}\<[7]%
\>[7]{}\mathbin{:}{}\<[7E]%
\>[10]{}\Varid{bs}{}\<[14]%
\>[14]{})\mathbin{!}\Varid{n}{}\<[21]%
\>[21]{}\mathrel{=}\Varid{bs}\mathbin{!}(\Varid{n}\mathbin{-}\mathrm{1}){}\<[E]%
\\[\blanklineskip]%
\>[B]{}\mathbf{instance}\;\Conid{Additive}\;\Varid{b}\Rightarrow \Conid{Additive}\;[\mskip1.5mu \Varid{b}\mskip1.5mu]\;\mathbf{where}{}\<[E]%
\\
\>[B]{}\hsindent{3}{}\<[3]%
\>[3]{}\mathrm{0}\mathrel{=}[\mskip1.5mu \mskip1.5mu]{}\<[E]%
\\
\>[B]{}\hsindent{3}{}\<[3]%
\>[3]{}[\mskip1.5mu \mskip1.5mu]\mathbin{+}\Varid{bs}\mathrel{=}\Varid{bs}{}\<[E]%
\\
\>[B]{}\hsindent{3}{}\<[3]%
\>[3]{}\Varid{as}\mathbin{+}[\mskip1.5mu \mskip1.5mu]\mathrel{=}\Varid{as}{}\<[E]%
\\
\>[B]{}\hsindent{3}{}\<[3]%
\>[3]{}(\Varid{a}\mathbin{:}\Varid{as})\mathbin{+}(\Varid{b}\mathbin{:}\Varid{bs})\mathrel{=}\Varid{a}\mathbin{+}\Varid{b}\mathbin{:}\Varid{as}\mathbin{+}\Varid{bs}{}\<[E]%
\\[\blanklineskip]%
\>[B]{}\mathbf{instance}\;(\Conid{Semiring}\;\Varid{b},\Conid{IsZero}\;\Varid{b},\Conid{IsOne}\;\Varid{b})\Rightarrow \Conid{Semiring}\;[\mskip1.5mu \Varid{b}\mskip1.5mu]\;\mathbf{where}{}\<[E]%
\\
\>[B]{}\hsindent{3}{}\<[3]%
\>[3]{}\mathrm{1}\mathrel{=}\mathrm{1}\mathbin{:}\mathrm{0}{}\<[E]%
\\
\>[B]{}\hsindent{3}{}\<[3]%
\>[3]{}[\mskip1.5mu \mskip1.5mu]{}\<[13]%
\>[13]{}\mathbin{*}\anonymous {}\<[20]%
\>[20]{}\mathrel{=}[\mskip1.5mu \mskip1.5mu]\;{}\mbox{\onelinecomment  \ensuremath{\mathrm{0}\mathbin{*}\Varid{q}\mathrel{=}\mathrm{0}}}{}\<[E]%
\\
\>[B]{}\hsindent{3}{}\<[3]%
\>[3]{}(\Varid{a}\mathbin{:}\Varid{dp}){}\<[13]%
\>[13]{}\mathbin{*}\Varid{q}{}\<[20]%
\>[20]{}\mathrel{=}\Varid{a}\cdot\Varid{q}\mathbin{+}(\mathrm{0}\mathbin{:}\Varid{dp}\mathbin{*}\Varid{q}){}\<[E]%
\ColumnHook
\end{hscode}\resethooks
This last definition is reminiscent of long multiplication, which is convolution in disguise.

\sectionl{Beyond Convolution}

Many uses of discrete convolution (including convolutional neural networks \citep[Chapter 9]{LecunBengioHinton2015DLNature}) involve functions having finite support, i.e., nonzero on only a finite subset of their domains.
\notefoot{First suggest finite maps, using instances from \figref{Map}. Then intervals/arrays.}
In many cases, these domain subsets may be defined by finite \emph{intervals}.
For instance, such a 2D operation would be given by intervals in each dimension, together specifying lower left and upper right corners of a 2D interval (rectangle) outside of which the functions are guaranteed to be zero.
The two input intervals needn't have the same size, and the result's interval of support is typically larger than both inputs, with size equaling the sum of the sizes in each dimension (minus one for the discrete case).
\notefoot{Show an example as a picture.}
Since the result's support size is entirely predictable and based only on the arguments' sizes, it is appealing to track sizes statically via types.
For instance, a 1D convolution might have the following type:
\notefoot{To do: More clearly distinguish between functions with finite support vs functions with finite domains. I think I started this paragraph with the former mindset but switched to the latter.}
\begin{hscode}\SaveRestoreHook
\column{B}{@{}>{\hspre}l<{\hspost}@{}}%
\column{E}{@{}>{\hspre}l<{\hspost}@{}}%
\>[B]{}(\mathbin{*})\mathbin{::}\Conid{Semiring}\;\Varid{s}\Rightarrow \Conid{Array}_{\Varid{m}\mathbin{+}\mathrm{1}}\;\Varid{s}\to \Conid{Array}_{\Varid{n}\mathbin{+}\mathrm{1}}\;\Varid{s}\to \Conid{Array}_{\Varid{m}\mathbin{+}\Varid{n}\mathbin{+}\mathrm{1}}\;\Varid{s}{}\<[E]%
\ColumnHook
\end{hscode}\resethooks
Unfortunately, this signature is incompatible with semiring multiplication, in which arguments and result all have the same type.

From the perspective of functions, an array of size \ensuremath{\Varid{n}} is a memoized function from \ensuremath{\Conid{Fin}_{\Varid{n}}}, a type representing the finite set \ensuremath{\set{\mathrm{0},\ldots,\Varid{n}\mathbin{-}\mathrm{1}}}.
We can still define convolution in the customary sense in terms of index addition:
\begin{hscode}\SaveRestoreHook
\column{B}{@{}>{\hspre}l<{\hspost}@{}}%
\column{E}{@{}>{\hspre}l<{\hspost}@{}}%
\>[B]{}\Varid{f}\mathbin{*}\Varid{g}\mathrel{=}\bigOp\sum{\Varid{u},\Varid{v}}{0}\;\Varid{u}\mathbin{+}\Varid{v}\mapsto\Varid{f}\;\Varid{u}\mathbin{*}\Varid{g}\;\Varid{v}{}\<[E]%
\ColumnHook
\end{hscode}\resethooks
where now
\begin{hscode}\SaveRestoreHook
\column{B}{@{}>{\hspre}l<{\hspost}@{}}%
\column{E}{@{}>{\hspre}l<{\hspost}@{}}%
\>[B]{}(\mathbin{+})\mathbin{::}\Conid{Fin}_{\Varid{m}\mathbin{+}\mathrm{1}}\to \Conid{Fin}_{\Varid{n}\mathbin{+}\mathrm{1}}\to \Conid{Fin}_{\Varid{m}\mathbin{+}\Varid{n}\mathbin{+}\mathrm{1}}{}\<[E]%
\ColumnHook
\end{hscode}\resethooks
Indices can no longer form a monoid under addition, however, due to the nonuniformity of types.

The inability to support convolution on statically sized arrays (or other memoized forms of functions over finite domains) as semiring multiplication came from the expectation that indices/arguments combine via a monoid.
Fortunately, this expectation can be dropped by generalizing from monoidal combination to an \emph{arbitrary} binary operation \ensuremath{\Varid{h}\mathbin{::}\Varid{a}\to \Varid{b}\to \Varid{c}}.
For now, let's call this more general operation ``\ensuremath{\Varid{lift}_{2}\;\Varid{h}}''.
\begin{hscode}\SaveRestoreHook
\column{B}{@{}>{\hspre}l<{\hspost}@{}}%
\column{E}{@{}>{\hspre}l<{\hspost}@{}}%
\>[B]{}\Varid{lift}_{2}\mathbin{::}\Conid{Semiring}\;\Varid{s}\Rightarrow (\Varid{a}\to \Varid{b}\to \Varid{c})\to (\Varid{a}\to \Varid{s})\to (\Varid{b}\to \Varid{s})\to (\Varid{c}\to \Varid{s}){}\<[E]%
\\
\>[B]{}\Varid{lift}_{2}\;\Varid{h}\;\Varid{f}\;\Varid{g}\mathrel{=}\bigOp\sum{\Varid{u},\Varid{v}}{0}\;\Varid{h}\;\Varid{u}\;\Varid{v}\mapsto\Varid{f}\;\Varid{u}\mathbin{*}\Varid{g}\;\Varid{v}{}\<[E]%
\ColumnHook
\end{hscode}\resethooks
We can similarly lift functions of \emph{any} arity:
\begin{hscode}\SaveRestoreHook
\column{B}{@{}>{\hspre}l<{\hspost}@{}}%
\column{E}{@{}>{\hspre}l<{\hspost}@{}}%
\>[B]{}\Varid{lift}_n\mathbin{::}\Conid{Semiring}\;\Varid{s}\Rightarrow (\Varid{a}_{1}\to \cdots\to \Varid{a}_n\to \Varid{b})\to (\Varid{a}_{1}\to \Varid{s})\to \cdots\to (\Varid{a}_n\to \Varid{s})\to (\Varid{b}\to \Varid{s}){}\<[E]%
\\
\>[B]{}\Varid{lift}_n\;\Varid{h}\;\Varid{f}_{1}\cdots\Varid{f}_n\mathrel{=}\bigOp\sum{\Varid{u}_{1},\ldots,\Varid{u}_n}{1}\;\Varid{h}\;\Varid{u}_{1}\cdots\Varid{u}_n\mapsto\Varid{f}_{1}\;\Varid{u}_{1}\mathbin{*}\cdots\mathbin{*}\Varid{f}_n\;\Varid{u}_n{}\<[E]%
\ColumnHook
\end{hscode}\resethooks
Here we are summing over the set-valued \emph{preimage} of \ensuremath{\Varid{w}} under \ensuremath{\Varid{h}}.
Now consider two specific instances of \ensuremath{\Varid{lift}_n}:
\begin{hscode}\SaveRestoreHook
\column{B}{@{}>{\hspre}l<{\hspost}@{}}%
\column{10}{@{}>{\hspre}l<{\hspost}@{}}%
\column{E}{@{}>{\hspre}l<{\hspost}@{}}%
\>[B]{}\Varid{lift}_{1}\mathbin{::}\Conid{Semiring}\;\Varid{s}\Rightarrow (\Varid{a}\to \Varid{b})\to (\Varid{a}\to \Varid{s})\to (\Varid{b}\to \Varid{s}){}\<[E]%
\\
\>[B]{}\Varid{lift}_{1}\;\Varid{h}\;\Varid{f}\mathrel{=}\bigOp\sum{\Varid{u}}{0}\;\Varid{h}\;\Varid{u}\mapsto\Varid{f}\;\Varid{u}{}\<[E]%
\\[\blanklineskip]%
\>[B]{}\Varid{lift}_{0}\mathbin{::}\Conid{Semiring}\;\Varid{s}\Rightarrow \Varid{b}\to (\Varid{b}\to \Varid{s}){}\<[E]%
\\
\>[B]{}\Varid{lift}_{0}\;\Varid{b}{}\<[10]%
\>[10]{}\mathrel{=}\Varid{b}\mapsto\mathrm{1}{}\<[E]%
\\
\>[10]{}\mathrel{=}\Varid{single}\;\Varid{b}{}\<[E]%
\ColumnHook
\end{hscode}\resethooks

\noindent
The signatures of \ensuremath{\Varid{lift}_{2}}, \ensuremath{\Varid{lift}_{1}}, and \ensuremath{\Varid{lift}_{0}} \emph{almost} generalize to those of \ensuremath{\Varid{liftA}_{2}}, \ensuremath{\Varid{fmap}}, and \ensuremath{\Varid{pure}} from the \ensuremath{\Conid{Functor}} and \ensuremath{\Conid{Applicative}} type classes \citep{McBride2008APE,Typeclassopedia}.
In type systems like Haskell's, however, \ensuremath{\Varid{a}\to \Varid{s}} is the functor \ensuremath{(\Varid{a}\to )} applied to \ensuremath{\Varid{s}}, while we would need it to be \ensuremath{(\to \Varid{s})} applied to \ensuremath{\Varid{a}}.
To fix this problem, define a type wrapper that swaps domain and codomain type parameters:
\begin{hscode}\SaveRestoreHook
\column{B}{@{}>{\hspre}l<{\hspost}@{}}%
\column{E}{@{}>{\hspre}l<{\hspost}@{}}%
\>[B]{}\mathbf{newtype}\;\Varid{s}\leftarrow\Varid{a}\mathrel{=}\Conid{F}\;(\Varid{a}\to \Varid{s}){}\<[E]%
\ColumnHook
\end{hscode}\resethooks
The use of \ensuremath{\Varid{s}\leftarrow\Varid{a}} as an alternative to \ensuremath{\Varid{a}\to \Varid{s}} allows us to give instances for both and to stay within Haskell's type system (and ability to infer types via first-order unification).

With this change, we can replace the specialized \ensuremath{\Varid{lift}_n} operations with standard ones.
An enhanced version of the \ensuremath{\Conid{Functor}}, \ensuremath{\Conid{Applicative}}, and \ensuremath{\Conid{Monad}} classes (similar to those by \citet{Kidney2017CA}) appear in \figrefdef{FunApp}{\ensuremath{\Conid{Functor}} and \ensuremath{\Conid{Applicative}} classes and some instances}{
\begin{hscode}\SaveRestoreHook
\column{B}{@{}>{\hspre}l<{\hspost}@{}}%
\column{3}{@{}>{\hspre}l<{\hspost}@{}}%
\column{21}{@{}>{\hspre}l<{\hspost}@{}}%
\column{32}{@{}>{\hspre}l<{\hspost}@{}}%
\column{41}{@{}>{\hspre}l<{\hspost}@{}}%
\column{48}{@{}>{\hspre}l<{\hspost}@{}}%
\column{50}{@{}>{\hspre}l<{\hspost}@{}}%
\column{64}{@{}>{\hspre}l<{\hspost}@{}}%
\column{72}{@{}>{\hspre}l<{\hspost}@{}}%
\column{E}{@{}>{\hspre}l<{\hspost}@{}}%
\>[B]{}\mathbf{class}\;\Conid{Functor}\;\Varid{f}\;\mathbf{where}{}\<[E]%
\\
\>[B]{}\hsindent{3}{}\<[3]%
\>[3]{}\mathbf{type}\;\Conid{Ok}\;\Varid{f}\;\Varid{a}\mathbin{::}\Conid{Constraint}{}\<[E]%
\\
\>[B]{}\hsindent{3}{}\<[3]%
\>[3]{}\mathbf{type}\;\Conid{Ok}\;\Varid{f}\;\Varid{a}\mathrel{=}(){}\<[21]%
\>[21]{}\mbox{\onelinecomment  default}{}\<[E]%
\\
\>[B]{}\hsindent{3}{}\<[3]%
\>[3]{}\Varid{fmap}\mathbin{::}(\Conid{Ok}\;\Varid{f}\;\Varid{a},\Conid{Ok}\;\Varid{f}\;\Varid{b})\Rightarrow (\Varid{a}\to \Varid{b})\to \Varid{f}\;\Varid{a}\to \Varid{f}\;\Varid{b}{}\<[E]%
\\[\blanklineskip]%
\>[B]{}\mathbf{class}\;\Conid{Functor}\;\Varid{f}\Rightarrow \Conid{Applicative}\;\Varid{f}\;\mathbf{where}{}\<[E]%
\\
\>[B]{}\hsindent{3}{}\<[3]%
\>[3]{}\Varid{pure}\mathbin{::}\Conid{Ok}\;\Varid{f}\;\Varid{a}\Rightarrow \Varid{a}\to \Varid{f}\;\Varid{a}{}\<[E]%
\\
\>[B]{}\hsindent{3}{}\<[3]%
\>[3]{}\Varid{liftA}_{2}\mathbin{::}(\Conid{Ok}\;\Varid{f}\;\Varid{a},\Conid{Ok}\;\Varid{f}\;\Varid{b},\Conid{Ok}\;\Varid{f}\;\Varid{c})\Rightarrow (\Varid{a}\to \Varid{b}\to \Varid{c})\to \Varid{f}\;\Varid{a}\to \Varid{f}\;\Varid{b}\to \Varid{f}\;\Varid{c}{}\<[E]%
\\[\blanklineskip]%
\>[B]{}\mathbf{infixl}\;\mathrm{1}\bind {}\<[E]%
\\
\>[B]{}\mathbf{class}\;\Conid{Applicative}\;\Varid{f}\Rightarrow \Conid{Monad}\;\Varid{f}\;\mathbf{where}{}\<[E]%
\\
\>[B]{}\hsindent{3}{}\<[3]%
\>[3]{}(\bind )\mathbin{::}(\Conid{Ok}\;\Varid{f}\;\Varid{a},\Conid{Ok}\;\Varid{f}\;\Varid{b})\Rightarrow \Varid{f}\;\Varid{a}\to (\Varid{a}\to \Varid{f}\;\Varid{b})\to \Varid{f}\;\Varid{b}{}\<[E]%
\\
\>[B]{}{}{}\<[E]%
\\[\blanklineskip]%
\>[B]{}\mathbf{instance}\;\Conid{Functor}\;((\to )\;\Varid{a})\;\mathbf{where}\;{}\<[48]%
\>[48]{}\mathbf{instance}\;\Conid{Semiring}\;\Varid{b}\Rightarrow \Conid{Functor}\;((\leftarrow)\;\Varid{b})\;\mathbf{where}{}\<[E]%
\\
\>[48]{}\hsindent{2}{}\<[50]%
\>[50]{}\mathbf{type}\;\Conid{Ok}\;((\leftarrow)\;\Varid{b})\;\Varid{a}\mathrel{=}\Conid{Eq}\;\Varid{a}{}\<[E]%
\\
\>[B]{}\hsindent{3}{}\<[3]%
\>[3]{}\Varid{fmap}\;\Varid{h}\;\Varid{f}\mathrel{=}\lambda\, \Varid{a}\to \Varid{h}\;(\Varid{f}\;\Varid{a})\;{}\<[50]%
\>[50]{}\Varid{fmap}\;\Varid{h}\;(\Conid{F}\;\Varid{f}){}\<[64]%
\>[64]{}\mathrel{=}\bigOp\sum{\Varid{u}}{0}\;\Varid{h}\;\Varid{u}\mapsto\Varid{f}\;\Varid{u}{}\<[E]%
\\
\>[64]{}\mathrel{=}\Conid{F}\;(\lambda\, \Varid{z}\to \bigOp\sum{\Varid{u}\;\!\!\\\!\!\;\Varid{h}\;\Varid{u}\mathrel{=}\Varid{z}}{1}\;\Varid{f}\;\Varid{u}){}\<[E]%
\\[\blanklineskip]%
\>[B]{}\mathbf{instance}\;\Conid{Applicative}\;((\to )\;\Varid{a})\;\mathbf{where}\;{}\<[41]%
\>[41]{}\ \ \ \ \ \ \ \ \;{}\<[48]%
\>[48]{}\mathbf{instance}\;\Conid{Semiring}\;\Varid{b}\Rightarrow \Conid{Applicative}\;((\leftarrow)\;\Varid{b})\;\mathbf{where}{}\<[E]%
\\
\>[B]{}\hsindent{3}{}\<[3]%
\>[3]{}\Varid{pure}\;\Varid{b}\mathrel{=}\lambda\, \Varid{a}\to \Varid{b}\;{}\<[50]%
\>[50]{}\Varid{pure}\;\Varid{a}\mathrel{=}\Varid{single}\;\Varid{a}{}\<[E]%
\\
\>[B]{}\hsindent{3}{}\<[3]%
\>[3]{}\Varid{liftA}_{2}\;\Varid{h}\;\Varid{f}\;\Varid{g}\mathrel{=}\lambda\, \Varid{a}\to \Varid{h}\;(\Varid{f}\;\Varid{a})\;(\Varid{g}\;\Varid{a})\;{}\<[50]%
\>[50]{}\Varid{liftA}_{2}\;\Varid{h}\;(\Conid{F}\;\Varid{f})\;(\Conid{F}\;\Varid{g}){}\<[72]%
\>[72]{}\mathrel{=}\bigOp\sum{\Varid{u},\Varid{v}}{0}\;\Varid{h}\;\Varid{u}\;\Varid{v}\mapsto\Varid{f}\;\Varid{u}\mathbin{*}\Varid{g}\;\Varid{v}{}\<[E]%
\\
\>[72]{}\mathrel{=}\Conid{F}\;(\lambda\, \Varid{z}\to \bigOp\sum{\Varid{u},\Varid{v}\;\!\!\\\!\!\;\Varid{h}\;\Varid{u}\;\Varid{v}\mathrel{=}\Varid{z}}{1}\;\Varid{f}\;\Varid{u}\mathbin{*}\Varid{g}\;\Varid{v}){}\<[E]%
\\[\blanklineskip]%
\>[B]{}\mathbf{instance}\;\Conid{Monad}\;((\to )\;\Varid{a})\;\mathbf{where}{}\<[E]%
\\
\>[B]{}\hsindent{3}{}\<[3]%
\>[3]{}\Varid{m}\bind \Varid{f}\mathrel{=}\lambda\, \Varid{a}\to \Varid{f}\;(\Varid{m}\;\Varid{a})\;\Varid{a}{}\<[E]%
\\[\blanklineskip]%
\>[B]{}\mathbf{instance}\;\Conid{Ord}\;\Varid{a}\Rightarrow \Conid{Functor}\;{}\<[32]%
\>[32]{}(\Conid{Map}\;\Varid{a})\;\mathbf{where}\ldots{}\<[E]%
\\
\>[B]{}\mathbf{instance}\;\Conid{Ord}\;\Varid{a}\Rightarrow \Conid{Applicative}\;{}\<[32]%
\>[32]{}(\Conid{Map}\;\Varid{a})\;\mathbf{where}\ldots{}\<[E]%
\\[\blanklineskip]%
\>[B]{}\mathbf{newtype}\;\Conid{Map'}\;\Varid{b}\;\Varid{a}\mathrel{=}\Conid{M}\;(\Conid{Map}\;\Varid{a}\;\Varid{b}){}\<[E]%
\\[\blanklineskip]%
\>[B]{}\mathbf{instance}\;\Conid{IsZero}\;\Varid{b}\Rightarrow \Conid{Functor}\;(\Conid{Map'}\;\Varid{b})\;\mathbf{where}{}\<[E]%
\\
\>[B]{}\hsindent{3}{}\<[3]%
\>[3]{}\mathbf{type}\;\Conid{Ok}\;(\Conid{Map'}\;\Varid{b})\;\Varid{a}\mathrel{=}\Conid{Ord}\;\Varid{a}{}\<[E]%
\\
\>[B]{}\hsindent{3}{}\<[3]%
\>[3]{}\Varid{fmap}\;\Varid{h}\;(\Conid{M}\;\Varid{p})\mathrel{=}\bigOp\sum{\Varid{a}\mathbin{\in}\Conid{M}\!.\Varid{keys}\;\Varid{p}}{2.5}\;\Varid{h}\;\Varid{a}\mapsto\Varid{p}\mathbin{!}\Varid{a}{}\<[E]%
\\[\blanklineskip]%
\>[B]{}\mathbf{instance}\;\Conid{IsZero}\;\Varid{b}\Rightarrow \Conid{Applicative}\;(\Conid{Map'}\;\Varid{b})\;\mathbf{where}{}\<[E]%
\\
\>[B]{}\hsindent{3}{}\<[3]%
\>[3]{}\Varid{pure}\;\Varid{a}\mathrel{=}\Varid{single}\;\Varid{a}{}\<[E]%
\\
\>[B]{}\hsindent{3}{}\<[3]%
\>[3]{}\Varid{liftA}_{2}\;\Varid{h}\;(\Conid{M}\;\Varid{p})\;(\Conid{M}\;\Varid{q})\mathrel{=}\bigOp\sum{\Varid{a}\mathbin{\in}\Conid{M}\!.\Varid{keys}\;\Varid{p}\;\!\!\\\!\!\;\Varid{b}\mathbin{\in}\Conid{M}\!.\Varid{keys}\;\Varid{q}}{2.5}\;\Varid{h}\;\Varid{a}\;\Varid{b}\mapsto(\Varid{p}\mathbin{!}\Varid{a})\mathbin{*}(\Varid{q}\mathbin{!}\Varid{b}){}\<[E]%
\ColumnHook
\end{hscode}\resethooks

\vspace{-3ex}
}, along with instances for functions and finite maps.
Other representations would need similar reversal of type arguments.
\footnote{The enhancement is the associated constraint \citep{Bolingbroke2011CK} \ensuremath{\Conid{Ok}}, limiting the types that the class methods must support. The line ``\ensuremath{\mathbf{type}\;\Conid{Ok}\;\Varid{f}\;\Varid{a}\mathrel{=}()}'' means that the constraint on \ensuremath{\Varid{a}} defaults to \ensuremath{()}, which holds vacuously for all \ensuremath{\Varid{a}}.}%
\footnote{Originally, \ensuremath{\Conid{Applicative}} had a \ensuremath{(\mathbin{<\!\!\!*\!\!\!>})} method from which one can easily define \ensuremath{\Varid{liftA}_{2}}. Since the base library version 4.10, \ensuremath{\Varid{liftA}_{2}} was added as a method (along with a default definition of \ensuremath{(\mathbin{<\!\!\!*\!\!\!>})}) to allow for more efficient implementation \citep[Section 3.2.2]{GHC821}.}
Higher-arity liftings can be defined via these three.\out{ (Exercise.)}
For \ensuremath{\Varid{b}\leftarrow\Varid{a}}, these definitions are not really executable code, since they involve potentially infinite summations, but they serve as specifications for other representations such as finite maps, regular expressions, and tries.
\begin{theorem}
For each instance defined in \figref{FunApp}, \ensuremath{\mathrm{1}\mathrel{=}\Varid{pure}\;\varepsilon}, and \ensuremath{(\mathbin{*})\mathrel{=}\Varid{liftA}_{2}\;( \diamond )}.
\notefoot{Probe more into the pattern of \ensuremath{\Varid{single}\mathrel{=}\Varid{pure}}, \ensuremath{\mathrm{1}\mathrel{=}\Varid{single}\;\varepsilon} and \ensuremath{(\mathbin{*})\mathrel{=}\Varid{liftA}_{2}\;( \diamond )}.
Also the relationship between forward and reverse functions and maps.}
\end{theorem}
\begin{proof}
Immediate from the instance definitions.
\end{proof}

Given the type distinction between \ensuremath{\Varid{a}\to \Varid{b}} and \ensuremath{\Varid{b}\leftarrow\Varid{a}}, let's now reconsider the \ensuremath{\Conid{Semiring}} instances for functions in \figref{monoid semiring} and for sets in \secref{Languages and the Monoid Semiring}.
Each has an alternative choice that is in some ways more compelling, as shown in \figrefdef{-> and <-- semirings}{The \ensuremath{\Varid{a}\to \Varid{b}} and \ensuremath{\Varid{b}\leftarrow\Varid{a}} semirings}{
\begin{hscode}\SaveRestoreHook
\column{B}{@{}>{\hspre}l<{\hspost}@{}}%
\column{3}{@{}>{\hspre}l<{\hspost}@{}}%
\column{21}{@{}>{\hspre}l<{\hspost}@{}}%
\column{22}{@{}>{\hspre}l<{\hspost}@{}}%
\column{E}{@{}>{\hspre}l<{\hspost}@{}}%
\>[B]{}\mathbf{instance}\;\Conid{Semiring}\;\Varid{b}\Rightarrow \Conid{Semiring}\;(\Varid{a}\to \Varid{b})\;\mathbf{where}{}\<[E]%
\\
\>[B]{}\hsindent{3}{}\<[3]%
\>[3]{}\mathrm{1}\mathrel{=}\Varid{pure}\;\mathrm{1}{}\<[21]%
\>[21]{}\mbox{\onelinecomment  i.e., \ensuremath{\mathrm{1}\mathrel{=}\lambda\, \Varid{a}\to \mathrm{1}}}{}\<[E]%
\\
\>[B]{}\hsindent{3}{}\<[3]%
\>[3]{}(\mathbin{*})\mathrel{=}\Varid{liftA}_{2}\;(\mathbin{*}){}\<[21]%
\>[21]{}\mbox{\onelinecomment  i.e., \ensuremath{\Varid{f}\mathbin{*}\Varid{g}\mathrel{=}\lambda\, \Varid{a}\to \Varid{f}\;\Varid{a}\mathbin{*}\Varid{g}\;\Varid{a}}}{}\<[E]%
\\[\blanklineskip]%
\>[B]{}\mathbf{newtype}\;\Varid{b}\leftarrow\Varid{a}\mathrel{=}\Conid{F}\;(\Varid{a}\to \Varid{b})\;\mathbf{deriving}\;(\Conid{Additive},\Conid{HasSingle}\;\Varid{b},\Conid{LeftSemimodule}\;\Varid{b},\Conid{Indexable}\;\Varid{a}\;\Varid{b}){}\<[E]%
\\[\blanklineskip]%
\>[B]{}\mathbf{instance}\;(\Conid{Semiring}\;\Varid{b},\Conid{Monoid}\;\Varid{a})\Rightarrow \Conid{Semiring}\;(\Varid{b}\leftarrow\Varid{a})\;\mathbf{where}{}\<[E]%
\\
\>[B]{}\hsindent{3}{}\<[3]%
\>[3]{}\mathrm{1}\mathrel{=}\Varid{pure}\;\varepsilon{}\<[E]%
\\
\>[B]{}\hsindent{3}{}\<[3]%
\>[3]{}(\mathbin{*})\mathrel{=}\Varid{liftA}_{2}\;( \diamond ){}\<[E]%
\\[\blanklineskip]%
\>[B]{}\mathbf{instance}\;\Conid{Semiring}\;(\Pow\;\Varid{a})\;\mathbf{where}{}\<[E]%
\\
\>[B]{}\hsindent{3}{}\<[3]%
\>[3]{}\mathrm{1}\mathrel{=}\set{\Varid{a}\mid \Conid{True}}{}\<[E]%
\\
\>[B]{}\hsindent{3}{}\<[3]%
\>[3]{}(\mathbin{*})\mathrel{=}\cap{}\<[E]%
\\[\blanklineskip]%
\>[B]{}\mathbf{newtype}\;\Pow'\;\Varid{a}\mathrel{=}\Conid{P}\;(\Pow\;\Varid{a})\;\mathbf{deriving}\;(\Conid{Additive},\Conid{HasSingle}\;\Varid{b},\Conid{LeftSemimodule}\;\Varid{b},\Conid{Indexable}\;\Varid{a}\;\Conid{Bool}){}\<[E]%
\\[\blanklineskip]%
\>[B]{}\mathbf{instance}\;\Conid{Semiring}\;(\Pow'\;\Varid{a})\;\mathbf{where}{}\<[E]%
\\
\>[B]{}\hsindent{3}{}\<[3]%
\>[3]{}\mathrm{1}\mathrel{=}\Varid{pure}\;\varepsilon{}\<[22]%
\>[22]{}\mbox{\onelinecomment  \ensuremath{\mathrm{1}\mathrel{=}\set{\varepsilon}\mathrel{=}\Varid{single}\;\Varid{empty}\mathrel{=}\Varid{value}\;\mathrm{1}}}{}\<[E]%
\\
\>[B]{}\hsindent{3}{}\<[3]%
\>[3]{}(\mathbin{*})\mathrel{=}\Varid{liftA}_{2}\;( \diamond ){}\<[22]%
\>[22]{}\mbox{\onelinecomment  \ensuremath{\Varid{p}\mathbin{*}\Varid{q}\mathrel{=}\set{\Varid{u} \diamond \Varid{v}\mid\Varid{u}\mathbin{\in}\Varid{p}\mathrel{\wedge}\Varid{v}\mathbin{\in}\Varid{q}}}}{}\<[E]%
\ColumnHook
\end{hscode}\resethooks
\vspace{-4ex}
}, along with a the old \ensuremath{\Varid{a}\to \Varid{b}} instance reexpressed and reassigned to \ensuremath{\Varid{b}\leftarrow\Varid{a}}.
Just as the \ensuremath{\Conid{Additive}} and \ensuremath{\Conid{Semiring}} instances for \ensuremath{\Conid{Bool}\leftarrow\Varid{a}} give us four important languages operations (union, concatenation and their identities), now the \ensuremath{\Conid{Semiring}\;(\Varid{a}\to \Conid{Bool})} gives us two more: the \emph{intersection} of languages and its identity (the set of all ``strings'').
These two semirings share several instances in common, expressed in \figref{-> and <-- semirings} via GHC-Haskell's \text{\tt GeneralizedNewtypeDeriving} language extension (present since GHC 6.8.1 and later made safe by \citet{Breitner2016SZC}).
All six of these operations are also useful in their generalized form (i.e., for \ensuremath{\Varid{a}\to \Varid{b}} and \ensuremath{\Varid{b}\leftarrow\Varid{a}} for semirings \ensuremath{\Varid{b}}).
As with \ensuremath{\Conid{Additive}}, this \ensuremath{\Conid{Semiring}\;(\Varid{a}\to \Varid{b})} instance implies that curried functions (of any number and type of arguments and with semiring result type) are semirings, with \ensuremath{\Varid{curry}} and \ensuremath{\Varid{uncurry}} being semiring homomorphisms.
(The proof is very similar to that of \thmref{curry additive}.)

The \ensuremath{\Varid{a}\to \Varid{b}} and \ensuremath{\Varid{b}\leftarrow\Varid{a}} semirings have another deep relationship:
\begin{theorem}\thmlabel{Fourier}
The Fourier transform is a semiring and left semimodule homomorphism from \ensuremath{\Varid{b}\leftarrow \Varid{a}} to \ensuremath{\Varid{a}\to \Varid{b}}.
\end{theorem}
This theorem is more often expressed by saying that (a) the Fourier transform is linear (i.e., an additive-monoid and left-semimodule homomorphism), and (b) the Fourier transform of a convolution (i.e., \ensuremath{(\mathbin{*})} on \ensuremath{\Varid{b}\leftarrow\Varid{a}}) of two functions is the pointwise product (i.e., \ensuremath{(\mathbin{*})} on \ensuremath{\Varid{a}\to \Varid{b}}) of the Fourier transforms of the two functions.
The latter property is known as ``the convolution theorem'' \citep[Chapter 6]{Bracewell2000Fourier}.

There is also an important relationship between \ensuremath{\Varid{a}\to \Varid{b}} and \ensuremath{\Pow\;\Varid{a}\leftarrow\Varid{b}}.
Given a function \ensuremath{\Varid{f}\mathbin{::}\Varid{a}\to \Varid{b}}, the \emph{preimage} under \emph{f} of a codomain value \ensuremath{\Varid{b}} is the set of all values that get mapped to \ensuremath{\Varid{b}}:
\begin{hscode}\SaveRestoreHook
\column{B}{@{}>{\hspre}l<{\hspost}@{}}%
\column{E}{@{}>{\hspre}l<{\hspost}@{}}%
\>[B]{}\Varid{pre}\mathbin{::}(\Varid{a}\to \Varid{b})\to (\Pow\;\Varid{a}\leftarrow\Varid{b}){}\<[E]%
\\
\>[B]{}\Varid{pre}\;\Varid{f}\mathrel{=}\Conid{F}\;(\lambda\, \Varid{b}\to \set{\Varid{a}\mid \Varid{f}\;\Varid{a}\mathrel{=}\Varid{b}}){}\<[E]%
\ColumnHook
\end{hscode}\resethooks
\begin{theorem}[\provedIn{theorem:pre hom}]\thmlabel{pre hom}
The \ensuremath{\Varid{pre}} operation is a \ensuremath{\Conid{Functor}} and \ensuremath{\Conid{Applicative}} homomorphism.
\end{theorem}

\sectionl{The Free Semimodule Monad}

Where there's an applicative, there's often a compatible monad.
For \ensuremath{\Varid{b}\leftarrow\Varid{a}}, the monad is known as the ``free semimodule monad'' (or sometimes the ``free \emph{vector space} monad'') \citep{Piponi2007MonadVS,Kmett2011FreeModules,Gehrke2017Q}.
The semimodule's dimension is the cardinality of \ensuremath{\Varid{a}}.
Basis vectors have the form \ensuremath{\Varid{single}\;\Varid{u}\mathrel{=}\Varid{u}\mapsto\mathrm{1}} for \ensuremath{\Varid{u}\mathbin{::}\Varid{a}} (mapping \ensuremath{\Varid{u}} to \ensuremath{\mathrm{1}} and every other value to \ensuremath{\mathrm{0}} as in \figref{monoid semiring}).

The monad instances for \ensuremath{(\leftarrow)\;\Varid{b}} and \ensuremath{\Conid{Map'}\;\Varid{b}} are defined as follows:\footnote{The \ensuremath{\Varid{return}} method does not appear here, since it is equivalent to \ensuremath{\Varid{pure}} from \ensuremath{\Conid{Applicative}}.}
\begin{hscode}\SaveRestoreHook
\column{B}{@{}>{\hspre}l<{\hspost}@{}}%
\column{3}{@{}>{\hspre}l<{\hspost}@{}}%
\column{E}{@{}>{\hspre}l<{\hspost}@{}}%
\>[B]{}\mathbf{instance}\;\Conid{Semiring}\;\Varid{s}\Rightarrow \Conid{Monad}\;((\leftarrow)\;\Varid{s})\;\mathbf{where}{}\<[E]%
\\
\>[B]{}\hsindent{3}{}\<[3]%
\>[3]{}(\bind )\mathbin{::}(\Varid{s}\leftarrow\Varid{a})\to (\Varid{a}\to (\Varid{s}\leftarrow\Varid{b})))\to (\Varid{s}\leftarrow\Varid{b}){}\<[E]%
\\
\>[B]{}\hsindent{3}{}\<[3]%
\>[3]{}\Conid{F}\;\Varid{f}\bind \Varid{h}\mathrel{=}\bigOp\sum{\Varid{a}}{0}\;\Varid{f}\;\Varid{a}\cdot\Varid{h}\;\Varid{a}{}\<[E]%
\\[\blanklineskip]%
\>[B]{}\mathbf{instance}\;(\Conid{Semiring}\;\Varid{b},\Conid{IsZero}\;\Varid{b})\Rightarrow \Conid{Monad}\;(\Conid{Map'}\;\Varid{b})\;\mathbf{where}{}\<[E]%
\\
\>[B]{}\hsindent{3}{}\<[3]%
\>[3]{}\Conid{M}\;\Varid{m}\bind \Varid{h}\mathrel{=}\bigOp\sum{\Varid{a}\mathbin{\in}\Conid{M}\!.\Varid{keys}\;\Varid{m}}{2.5}\;\Varid{m}\mathbin{!}\Varid{a}\cdot\Varid{h}\;\Varid{a}{}\<[E]%
\ColumnHook
\end{hscode}\resethooks
\vspace{-2ex}
\begin{theorem}[\provedIn{theorem:standard FunApp}]\thmlabel{standard FunApp}
The definitions of \ensuremath{\Varid{fmap}} and \ensuremath{\Varid{liftA}_{2}} on \ensuremath{(\leftarrow)\;\Varid{b}} in \figref{FunApp} satisfy the following standard equations for monads:
\begin{hscode}\SaveRestoreHook
\column{B}{@{}>{\hspre}l<{\hspost}@{}}%
\column{15}{@{}>{\hspre}l<{\hspost}@{}}%
\column{E}{@{}>{\hspre}l<{\hspost}@{}}%
\>[B]{}\Varid{fmap}\;\Varid{h}\;\Varid{p}\mathrel{=}\Varid{p}\bind \Varid{pure}\hsdot{\circ }{.\:}\Varid{h}{}\<[E]%
\\[\blanklineskip]%
\>[B]{}\Varid{liftA}_{2}\;\Varid{h}\;\Varid{p}\;\Varid{q}{}\<[15]%
\>[15]{}\mathrel{=}\Varid{p}\bind \lambda\, \Varid{u}\to \Varid{fmap}\;(\Varid{h}\;\Varid{u})\;\Varid{q}{}\<[E]%
\\
\>[15]{}\mathrel{=}\Varid{p}\bind \lambda\, \Varid{u}\to \Varid{q}\bind \lambda\, \Varid{v}\to \Varid{pure}\;(\Varid{h}\;\Varid{u}\;\Varid{v}){}\<[E]%
\ColumnHook
\end{hscode}\resethooks
\end{theorem}

\sectionl{Other Applications}

\subsectionl{Polynomials}

As is well known, univariate polynomials form a semiring and can be multiplied by convolving their coefficients.
Perhaps less known is that this trick extends naturally to power series and to multivariate polynomials.



Looking more closely, univariate polynomials (and even power series) can be represented by a collection of coefficients indexed by exponents, or conversely as a collection of exponents weighted by coefficients.
For a polynomial in a variable \ensuremath{\Varid{x}}, an association \ensuremath{\Varid{i}\mapsto\Varid{c}} of coefficient \ensuremath{\Varid{c}} with exponent \ensuremath{\Varid{i}} represents the monomial (polynomial term) \ensuremath{\Varid{c}\mathbin{*}\Varid{x}^\Varid{i}}.
One can use a variety of representations for these indexed collections.
We'll consider efficient representations below, but let's begin as \ensuremath{\Varid{b}\leftarrow\mathbb{N}} along with a denotation as a (polynomial) function of type \ensuremath{\Varid{b}\to \Varid{b}}:
\begin{hscode}\SaveRestoreHook
\column{B}{@{}>{\hspre}l<{\hspost}@{}}%
\column{32}{@{}>{\hspre}l<{\hspost}@{}}%
\column{E}{@{}>{\hspre}l<{\hspost}@{}}%
\>[B]{}\Varid{poly}_{1}\mathbin{::}\Conid{Semiring}\;\Varid{b}\Rightarrow (\Varid{b}\leftarrow\mathbb{N})\to (\Varid{b}\to \Varid{b}){}\<[E]%
\\
\>[B]{}\Varid{poly}_{1}\;(\Conid{F}\;\Varid{f})\mathrel{=}\lambda\, \Varid{x}\to \bigOp\sum{\Varid{i}}{0}\;{}\<[32]%
\>[32]{}\Varid{f}\;\Varid{i}\mathbin{*}\Varid{x}^\Varid{i}{}\<[E]%
\ColumnHook
\end{hscode}\resethooks
Polynomial multiplication via convolution follows from the following property:
\begin{theorem}[\provedIn{theorem:poly hom}]\thmlabel{poly hom}
The function \ensuremath{\Varid{poly}_{1}} is a semiring homomorphism when multiplication on \ensuremath{\Varid{b}} commutes.
\end{theorem}
Pragmatically, \thmref{poly hom} says that the \ensuremath{\Varid{b}\leftarrow\mathbb{N}} semiring (in which \ensuremath{(\mathbin{*})} is convolution) correctly implements arithmetic on univariate polynomials.
More usefully, we can adopt other representations of \ensuremath{\Varid{b}\leftarrow\mathbb{N}}, such as \ensuremath{\Conid{Map}\;\mathbb{N}\;\Varid{b}}.
For viewing results, wrap these representations in a new type, and provide a \ensuremath{\Conid{Show}} instance:
\begin{hscode}\SaveRestoreHook
\column{B}{@{}>{\hspre}l<{\hspost}@{}}%
\column{E}{@{}>{\hspre}l<{\hspost}@{}}%
\>[B]{}\mathbf{newtype}\;\Conid{Poly}_{1}\;\Varid{z}\mathrel{=}\Conid{Poly}_{1}\;\Varid{z}\;\mathbf{deriving}\;(\Conid{Additive},\Conid{Semiring},\Conid{Indexable}\;\Varid{n}\;\Varid{b},\Conid{HasSingle}\;\Varid{n}\;\Varid{b}){}\<[E]%
\\[\blanklineskip]%
\>[B]{}\mathbf{instance}\;(\ldots)\Rightarrow \Conid{Show}\;(\Conid{Poly}_{1}\;\Varid{z})\;\mathbf{where}\;{}\ldots{}\<[E]%
\ColumnHook
\end{hscode}\resethooks
Try it out (with prompts indicated by ``\ensuremath{\lambda\rangle\ }''):
\begin{hscode}\SaveRestoreHook
\column{B}{@{}>{\hspre}l<{\hspost}@{}}%
\column{E}{@{}>{\hspre}l<{\hspost}@{}}%
\>[B]{}\lambda\rangle\ \mathbf{let}\;\Varid{p}\mathrel{=}\Varid{single}\;\mathrm{1}\mathbin{+}\Varid{value}\;\mathrm{3}\mathbin{::}\Conid{Poly}_{1}\;(\Conid{Map}\;\mathbb{N}\;\mathbb Z){}\<[E]%
\\
\>[B]{}\lambda\rangle\ \Varid{p}{}\<[E]%
\\
\>[B]{}\Varid{x}\mathbin{+}\mathrm{3}{}\<[E]%
\\[1.5ex]\>[B]{}\lambda\rangle\ \Varid{p}^\mathrm{3}{}\<[E]%
\\
\>[B]{}\Varid{x}^\mathrm{3}\mathbin{+}\mathrm{9}{}\Varid{x}^\mathrm{2}\mathbin{+}\mathrm{27}{}\Varid{x}\mathbin{+}\mathrm{27}{}\<[E]%
\\[1.5ex]\>[B]{}\lambda\rangle\ \Varid{p}^\mathrm{7}{}\<[E]%
\\
\>[B]{}\mathrm{2187}\mathbin{+}\mathrm{5103}{}\Varid{x}\mathbin{+}\mathrm{5103}{}\Varid{x}^\mathrm{2}\mathbin{+}\mathrm{2835}{}\Varid{x}^\mathrm{3}\mathbin{+}\mathrm{945}{}\Varid{x}^\mathrm{4}\mathbin{+}\mathrm{189}{}\Varid{x}^\mathrm{5}\mathbin{+}\mathrm{21}{}\Varid{x}^\mathrm{6}\mathbin{+}\Varid{x}^\mathrm{7}{}\<[E]%
\\[1.5ex]\>[B]{}\lambda\rangle\ \Varid{poly}_{1}\;(\Varid{p}^\mathrm{5})\;\mathrm{17}\mathrel{=}(\Varid{poly}_{1}\;\Varid{p}\;\mathrm{17})^\mathrm{5}{}\<[E]%
\\
\>[B]{}\Conid{True}{}\<[E]%
\ColumnHook
\end{hscode}\resethooks

We can also use \ensuremath{[\mskip1.5mu \mskip1.5mu]} in place of \ensuremath{\Conid{Map}\;\mathbb{N}}.
The example above yields identical results.
Since lists are potentially infinite (unlike finite maps), however, this simple change enables power series%
.
Following \citet{McIlroy1999PSPS,McIlroy2001MS}, define integration and differentiation as follows:
\begin{hscode}\SaveRestoreHook
\column{B}{@{}>{\hspre}l<{\hspost}@{}}%
\column{3}{@{}>{\hspre}l<{\hspost}@{}}%
\column{5}{@{}>{\hspre}l<{\hspost}@{}}%
\column{19}{@{}>{\hspre}l<{\hspost}@{}}%
\column{20}{@{}>{\hspre}l<{\hspost}@{}}%
\column{24}{@{}>{\hspre}l<{\hspost}@{}}%
\column{28}{@{}>{\hspre}l<{\hspost}@{}}%
\column{34}{@{}>{\hspre}l<{\hspost}@{}}%
\column{E}{@{}>{\hspre}l<{\hspost}@{}}%
\>[B]{}\Varid{integral}\mathbin{::}\Conid{Fractional}\;\Varid{b}\Rightarrow \Conid{Poly}_{1}\;[\mskip1.5mu \Varid{b}\mskip1.5mu]\to \Conid{Poly}_{1}\;[\mskip1.5mu \Varid{b}\mskip1.5mu]{}\<[E]%
\\
\>[B]{}\Varid{integral}\;(\Conid{Poly}_{1}\;\Varid{bs}_{0})\mathrel{=}\Conid{Poly}_{1}\;(\mathrm{0}\mathbin{:}\Varid{go}\;\mathrm{1}\;\Varid{bs}_{0}){}\<[E]%
\\
\>[B]{}\hsindent{3}{}\<[3]%
\>[3]{}\mathbf{where}{}\<[E]%
\\
\>[3]{}\hsindent{2}{}\<[5]%
\>[5]{}\Varid{go}\;\anonymous \;[\mskip1.5mu \mskip1.5mu]{}\<[19]%
\>[19]{}\mathrel{=}[\mskip1.5mu \mskip1.5mu]{}\<[E]%
\\
\>[3]{}\hsindent{2}{}\<[5]%
\>[5]{}\Varid{go}\;\Varid{n}\;(\Varid{b}\mathbin{:}\Varid{d}){}\<[19]%
\>[19]{}\mathrel{=}\Varid{b}\mathbin{/}\Varid{n}\mathbin{:}\Varid{go}\;(\Varid{n}\mathbin{+}\mathrm{1})\;\Varid{d}{}\<[E]%
\\[1.5ex]\>[B]{}\Varid{derivative}\mathbin{::}(\Conid{Additive}\;\Varid{b},\Conid{Fractional}\;\Varid{b})\Rightarrow \Conid{Poly}_{1}\;[\mskip1.5mu \Varid{b}\mskip1.5mu]\to \Conid{Poly}_{1}\;[\mskip1.5mu \Varid{b}\mskip1.5mu]{}\<[E]%
\\
\>[B]{}\Varid{derivative}\;(\Conid{Poly}_{1}\;{}\<[24]%
\>[24]{}[\mskip1.5mu \mskip1.5mu]{}\<[34]%
\>[34]{})\mathrel{=}\mathrm{0}{}\<[E]%
\\
\>[B]{}\Varid{derivative}\;(\Conid{Poly}_{1}\;(\anonymous {}\<[24]%
\>[24]{}\mathbin{:}{}\<[28]%
\>[28]{}\Varid{bs}_{0}){}\<[34]%
\>[34]{})\mathrel{=}\Conid{Poly}_{1}\;(\Varid{go}\;\mathrm{1}\;\Varid{bs}_{0}){}\<[E]%
\\
\>[B]{}\hsindent{3}{}\<[3]%
\>[3]{}\mathbf{where}{}\<[E]%
\\
\>[3]{}\hsindent{2}{}\<[5]%
\>[5]{}\Varid{go}\;\anonymous \;[\mskip1.5mu \mskip1.5mu]{}\<[20]%
\>[20]{}\mathrel{=}[\mskip1.5mu \mskip1.5mu]{}\<[E]%
\\
\>[3]{}\hsindent{2}{}\<[5]%
\>[5]{}\Varid{go}\;\Varid{n}\;(\Varid{b}\mathbin{:}\Varid{bs}){}\<[20]%
\>[20]{}\mathrel{=}\Varid{n}\mathbin{*}\Varid{b}\mathbin{:}\Varid{go}\;(\Varid{n}\mathbin{+}\mathrm{1})\;\Varid{bs}{}\<[E]%
\ColumnHook
\end{hscode}\resethooks
Then define \ensuremath{\Varid{sin}}, \ensuremath{\Varid{cos}}, and \ensuremath{\Varid{exp}} via simple ordinary differential equations (ODEs):
\begin{hscode}\SaveRestoreHook
\column{B}{@{}>{\hspre}l<{\hspost}@{}}%
\column{E}{@{}>{\hspre}l<{\hspost}@{}}%
\>[B]{}\Varid{sin}_{\hspace{-1pt}p},\Varid{cos}_{\hspace{-1pt}p},\Varid{exp}_{\hspace{-1pt}p}\mathbin{::}\Conid{Poly}_{1}\;[\mskip1.5mu \Conid{Rational}\mskip1.5mu]{}\<[E]%
\\
\>[B]{}\Varid{sin}_{\hspace{-1pt}p}\mathrel{=}\Varid{integral}\;\Varid{cos}_{\hspace{-1pt}p}{}\<[E]%
\\
\>[B]{}\Varid{cos}_{\hspace{-1pt}p}\mathrel{=}\mathrm{1}\mathbin{-}\Varid{integral}\;\Varid{sin}_{\hspace{-1pt}p}{}\<[E]%
\\
\>[B]{}\Varid{exp}_{\hspace{-1pt}p}\mathrel{=}\mathrm{1}\mathbin{+}\Varid{integral}\;\Varid{exp}_{\hspace{-1pt}p}{}\<[E]%
\ColumnHook
\end{hscode}\resethooks
Try it out:
\begin{hscode}\SaveRestoreHook
\column{B}{@{}>{\hspre}l<{\hspost}@{}}%
\column{118}{@{}>{\hspre}c<{\hspost}@{}}%
\column{118E}{@{}l@{}}%
\column{E}{@{}>{\hspre}l<{\hspost}@{}}%
\>[B]{}\lambda\rangle\ \Varid{sin}_{\hspace{-1pt}p}{}\<[E]%
\\
\>[B]{}\Varid{x}\mathbin{-}\mathrm{1}/\mathrm{6}\mathbin{*}\Varid{x}^\mathrm{3}\mathbin{+}\mathrm{1}/\mathrm{120}\mathbin{*}\Varid{x}^\mathrm{5}\mathbin{-}\mathrm{1}/\mathrm{5040}\mathbin{*}\Varid{x}^\mathrm{7}\mathbin{+}\mathrm{1}/\mathrm{362880}\mathbin{*}\Varid{x}^\mathrm{9}\mathbin{-}\mathrm{1}/\mathrm{39916800}\mathbin{*}\Varid{x}^\mathrm{11}\mathbin{+}\mathrm{1}/\mathrm{6227020800}\mathbin{*}\Varid{x}^\mathrm{13}\mathbin{-}{}\<[118]%
\>[118]{}\ldots{}\<[118E]%
\\
\>[B]{}\lambda\rangle\ \Varid{cos}_{\hspace{-1pt}p}{}\<[E]%
\\
\>[B]{}\mathrm{1}/\mathrm{1}\mathbin{-}\mathrm{1}/\mathrm{2}\mathbin{*}\Varid{x}^\mathrm{2}\mathbin{+}\mathrm{1}/\mathrm{24}\mathbin{*}\Varid{x}^\mathrm{4}\mathbin{-}\mathrm{1}/\mathrm{720}\mathbin{*}\Varid{x}^\mathrm{6}\mathbin{+}\mathrm{1}/\mathrm{40320}\mathbin{*}\Varid{x}^\mathrm{8}\mathbin{-}\mathrm{1}/\mathrm{3628800}\mathbin{*}\Varid{x}^\mathrm{10}\mathbin{+}\mathrm{1}/\mathrm{479001600}\mathbin{*}\Varid{x}^\mathrm{12}\mathbin{-}\ldots{}\<[E]%
\\
\>[B]{}\lambda\rangle\ \Varid{exp}_{\hspace{-1pt}p}{}\<[E]%
\\
\>[B]{}\mathrm{1}/\mathrm{1}\mathbin{+}\Varid{x}\mathbin{+}\mathrm{1}/\mathrm{2}\mathbin{*}\Varid{x}^\mathrm{2}\mathbin{+}\mathrm{1}/\mathrm{6}\mathbin{*}\Varid{x}^\mathrm{3}\mathbin{+}\mathrm{1}/\mathrm{24}\mathbin{*}\Varid{x}^\mathrm{4}\mathbin{+}\mathrm{1}/\mathrm{120}\mathbin{*}\Varid{x}^\mathrm{5}\mathbin{+}\mathrm{1}/\mathrm{720}\mathbin{*}\Varid{x}^\mathrm{6}\mathbin{+}\mathrm{1}/\mathrm{5040}\mathbin{*}\Varid{x}^\mathrm{7}\mathbin{+}\mathrm{1}/\mathrm{40320}\mathbin{*}\Varid{x}^\mathrm{8}\ldots{}\<[E]%
\ColumnHook
\end{hscode}\resethooks
As expected, \ensuremath{\Varid{derivative}\;\Varid{sin}_{\hspace{-1pt}p}\mathrel{=}\Varid{cos}_{\hspace{-1pt}p}}, \ensuremath{\Varid{derivative}\;\Varid{cos}_{\hspace{-1pt}p}\mathrel{=}\mathbin{-}\Varid{sin}_{\hspace{-1pt}p}}, and \ensuremath{\Varid{derivative}\;\Varid{exp}_{\hspace{-1pt}p}\mathrel{=}\Varid{exp}_{\hspace{-1pt}p}}:
\begin{hscode}\SaveRestoreHook
\column{B}{@{}>{\hspre}l<{\hspost}@{}}%
\column{29}{@{}>{\hspre}l<{\hspost}@{}}%
\column{E}{@{}>{\hspre}l<{\hspost}@{}}%
\>[B]{}\lambda\rangle\ \Varid{derivative}\;\Varid{sin}_{\hspace{-1pt}p}{}\<[29]%
\>[29]{}\mbox{\onelinecomment  \ensuremath{\mathrel{=}\Varid{cos}_{\hspace{-1pt}p}}}{}\<[E]%
\\
\>[B]{}\mathrm{1}/\mathrm{1}\mathbin{-}\mathrm{1}/\mathrm{2}\mathbin{*}\Varid{x}^\mathrm{2}\mathbin{+}\mathrm{1}/\mathrm{24}\mathbin{*}\Varid{x}^\mathrm{4}\mathbin{-}\mathrm{1}/\mathrm{720}\mathbin{*}\Varid{x}^\mathrm{6}\mathbin{+}\mathrm{1}/\mathrm{40320}\mathbin{*}\Varid{x}^\mathrm{8}\mathbin{-}\mathrm{1}/\mathrm{3628800}\mathbin{*}\Varid{x}^\mathrm{10}\mathbin{+}\mathrm{1}/\mathrm{479001600}\mathbin{*}\Varid{x}^\mathrm{12}\mathbin{-}\ldots{}\<[E]%
\\
\>[B]{}\lambda\rangle\ \Varid{derivative}\;\Varid{cos}_{\hspace{-1pt}p}{}\<[29]%
\>[29]{}\mbox{\onelinecomment  \ensuremath{\mathrel{=}\mathbin{-}\Varid{sin}_{\hspace{-1pt}p}}}{}\<[E]%
\\
\>[B]{}(\mathbin{-}\mathrm{1})/\mathrm{1}\mathbin{*}\Varid{x}\mathbin{+}\mathrm{1}/\mathrm{6}\mathbin{*}\Varid{x}^\mathrm{3}\mathbin{-}\mathrm{1}/\mathrm{120}\mathbin{*}\Varid{x}^\mathrm{5}\mathbin{+}\mathrm{1}/\mathrm{5040}\mathbin{*}\Varid{x}^\mathrm{7}\mathbin{-}\mathrm{1}/\mathrm{362880}\mathbin{*}\Varid{x}^\mathrm{9}\mathbin{+}\mathrm{1}/\mathrm{39916800}\mathbin{*}\Varid{x}^\mathrm{11}\mathbin{-}\ldots{}\<[E]%
\\
\>[B]{}\lambda\rangle\ \Varid{derivative}\;\Varid{exp}_{\hspace{-1pt}p}{}\<[29]%
\>[29]{}\mbox{\onelinecomment  \ensuremath{\mathrel{=}\Varid{exp}_{\hspace{-1pt}p}}}{}\<[E]%
\\
\>[B]{}\mathrm{1}/\mathrm{1}\mathbin{+}\Varid{x}\mathbin{+}\mathrm{1}/\mathrm{2}\mathbin{*}\Varid{x}^\mathrm{2}\mathbin{+}\mathrm{1}/\mathrm{6}\mathbin{*}\Varid{x}^\mathrm{3}\mathbin{+}\mathrm{1}/\mathrm{24}\mathbin{*}\Varid{x}^\mathrm{4}\mathbin{+}\mathrm{1}/\mathrm{120}\mathbin{*}\Varid{x}^\mathrm{5}\mathbin{+}\mathrm{1}/\mathrm{720}\mathbin{*}\Varid{x}^\mathrm{6}\mathbin{+}\mathrm{1}/\mathrm{5040}\mathbin{*}\Varid{x}^\mathrm{7}\mathbin{+}\mathrm{1}/\mathrm{40320}\mathbin{*}\Varid{x}^\mathrm{8}\ldots{}\<[E]%
\ColumnHook
\end{hscode}\resethooks
Crucially for termination of ODEs such as these, \ensuremath{\Varid{integral}} is nonstrict, yielding its result's first coefficient before examining its argument.
In particular, the definition of \ensuremath{\Varid{integral}} does \emph{not} optimize for \ensuremath{\Conid{Poly}_{1}\;[\mskip1.5mu \mskip1.5mu]}.

\vspace{2ex}

What about multivariate polynomials, i.e., polynomial functions over higher-dimensional domains?
Consider a 2D domain:
\begin{hscode}\SaveRestoreHook
\column{B}{@{}>{\hspre}l<{\hspost}@{}}%
\column{E}{@{}>{\hspre}l<{\hspost}@{}}%
\>[B]{}\Varid{poly}_{2}\mathbin{::}\Conid{Semiring}\;\Varid{b}\Rightarrow (\Varid{b}\leftarrow\mathbb{N} \times \mathbb{N})\to (\Varid{b}\mathbin{*}\Varid{b}\to \Varid{b}){}\<[E]%
\\
\>[B]{}\Varid{poly}_{2}\;(\Conid{F}\;\Varid{f})\mathrel{=}\lambda\, (\Varid{x},\Varid{y})\to \bigOp\sum{\Varid{i},\Varid{j}}{0}\;\Varid{f}\;(\Varid{i},\Varid{j})\mathbin{*}\Varid{x}^\Varid{i}\mathbin{*}\Varid{y}^\Varid{j}{}\<[E]%
\ColumnHook
\end{hscode}\resethooks
Then
\begin{hscode}\SaveRestoreHook
\column{B}{@{}>{\hspre}c<{\hspost}@{}}%
\column{BE}{@{}l@{}}%
\column{5}{@{}>{\hspre}l<{\hspost}@{}}%
\column{50}{@{}>{\hspre}l<{\hspost}@{}}%
\column{E}{@{}>{\hspre}l<{\hspost}@{}}%
\>[5]{}\Varid{poly}_{2}\;(\Conid{F}\;\Varid{f})\;(\Varid{x},\Varid{y}){}\<[E]%
\\
\>[B]{}\mathrel{=}{}\<[BE]%
\>[5]{}\bigOp\sum{\Varid{i},\Varid{j}}{0}\;\Varid{f}\;(\Varid{i},\Varid{j})\mathbin{*}\Varid{x}^\Varid{i}\mathbin{*}\Varid{y}^\Varid{j}{}\<[50]%
\>[50]{}\mbox{\onelinecomment  \ensuremath{\Varid{poly}_{2}} definition}{}\<[E]%
\\
\>[B]{}\mathrel{=}{}\<[BE]%
\>[5]{}\bigOp\sum{\Varid{i},\Varid{j}}{0}\;\Varid{curry}\;\Varid{f}\;\Varid{i}\;\Varid{j}\mathbin{*}\Varid{x}^\Varid{i}\mathbin{*}\Varid{y}^\Varid{j}{}\<[50]%
\>[50]{}\mbox{\onelinecomment  \ensuremath{\Varid{curry}} definition}{}\<[E]%
\\
\>[B]{}\mathrel{=}{}\<[BE]%
\>[5]{}\bigOp\sum{\Varid{i}}{0}\;(\bigOp\sum{\Varid{j}}{0}\;\Varid{curry}\;\Varid{f}\;\Varid{i}\;\Varid{j}\mathbin{*}\Varid{y}^\Varid{j})\mathbin{*}\Varid{x}^\Varid{i}{}\<[50]%
\>[50]{}\mbox{\onelinecomment  linearity and commutativity assumption}{}\<[E]%
\\
\>[B]{}\mathrel{=}{}\<[BE]%
\>[5]{}\bigOp\sum{\Varid{i}}{0}\;\Varid{poly}\;(\Varid{curry}\;\Varid{f}\;\Varid{i})\;\Varid{y}\mathbin{*}\Varid{x}^\Varid{i}{}\<[50]%
\>[50]{}\mbox{\onelinecomment  \ensuremath{\Varid{poly}} definition}{}\<[E]%
\\
\>[B]{}\mathrel{=}{}\<[BE]%
\>[5]{}\Varid{poly}\;(\lambda\, \Varid{i}\to \Varid{poly}\;(\Varid{curry}\;\Varid{f}\;\Varid{i})\;\Varid{y})\;\Varid{x}{}\<[50]%
\>[50]{}\mbox{\onelinecomment  \ensuremath{\Varid{poly}} definition}{}\<[E]%
\ColumnHook
\end{hscode}\resethooks
The essential idea here is that a polynomial with a pair-valued domain can be viewed as a polynomial over polynomials.

We can do much better, however, generalizing from two dimensions to \ensuremath{\Varid{n}} dimensions for any \ensuremath{\Varid{n}}:
\twocol{0.5}{
\begin{hscode}\SaveRestoreHook
\column{B}{@{}>{\hspre}l<{\hspost}@{}}%
\column{44}{@{}>{\hspre}l<{\hspost}@{}}%
\column{E}{@{}>{\hspre}l<{\hspost}@{}}%
\>[B]{}\Varid{poly}\mathbin{::}(\Varid{b}\leftarrow\mathbb{N}^\Varid{n})\to (\Varid{b}^\Varid{n}\to \Varid{b}){}\<[E]%
\\
\>[B]{}\Varid{poly}\;(\Conid{F}\;\Varid{f})\;(\Varid{x}\mathbin{::}\Varid{b}^\Varid{n})\mathrel{=}\bigOp\sum{\Varid{p}\mathbin{::}\mathbb{N}^\Varid{n}}{0}\;{}\<[44]%
\>[44]{}\Varid{f}\;\Varid{p}\mathbin{*}\Varid{x}\!{\string^}^{\hspace{-1pt}\Varid{p}}{}\<[E]%
\ColumnHook
\end{hscode}\resethooks
}{0.45}{
\begin{hscode}\SaveRestoreHook
\column{B}{@{}>{\hspre}l<{\hspost}@{}}%
\column{E}{@{}>{\hspre}l<{\hspost}@{}}%
\>[B]{}\mathbf{infixr}\;\mathrm{8}\;{}\string^{}\<[E]%
\\
\>[B]{}(\string^)\mathbin{::}\Varid{b}^\Varid{n}\to \mathbb{N}^\Varid{n}\to \Varid{b}{}\<[E]%
\\
\>[B]{}\Varid{x}\!{\string^}^{\hspace{-1pt}\Varid{p}}\mathrel{=}\bigOp\prod{\Varid{i}\mathbin{<}\Varid{n}}{0}{\,}\Varid{x}_\Varid{i}^{\Varid{p}_\Varid{i}}{}\<[E]%
\ColumnHook
\end{hscode}\resethooks
}

\noindent
For instance, for \ensuremath{\Varid{n}\mathrel{=}\mathrm{3}}, \ensuremath{(\Varid{x},\Varid{y},\Varid{z})\!{\string^}^{\hspace{-1pt}(\Varid{i},\Varid{j},\Varid{k})}\mathrel{=}\Varid{x}^\Varid{i}\mathbin{*}\Varid{y}^\Varid{j}\mathbin{*}\Varid{z}^\Varid{k}}.
Generalizing further,\out{ rather than taking \ensuremath{\Varid{n}} here to be a natural number,} let \ensuremath{\Varid{n}} be any type\out{ with countable cardinality}, and interpret \ensuremath{\Varid{b}^\Varid{n}} and \ensuremath{\mathbb{N}^\Varid{n}} as \ensuremath{\Varid{n}\to \Varid{b}} and \ensuremath{\Varid{n}\to \mathbb{N}}:
\twocol{0.5}{
\begin{hscode}\SaveRestoreHook
\column{B}{@{}>{\hspre}l<{\hspost}@{}}%
\column{51}{@{}>{\hspre}l<{\hspost}@{}}%
\column{E}{@{}>{\hspre}l<{\hspost}@{}}%
\>[B]{}\Varid{poly}\mathbin{::}(\Varid{b}\leftarrow(\Varid{n}\to \mathbb{N}))\to ((\Varid{n}\to \Varid{b})\to \Varid{b}){}\<[E]%
\\
\>[B]{}\Varid{poly}\;(\Conid{F}\;\Varid{f})\;(\Varid{x}\mathbin{::}\Varid{n}\to \Varid{b})\mathrel{=}\bigOp\sum{\Varid{p}\mathbin{::}\Varid{n}\to \mathbb{N}}{1}\;{}\<[51]%
\>[51]{}\Varid{f}\;\Varid{p}\mathbin{*}\Varid{x}\!{\string^}^{\hspace{-1pt}\Varid{p}}{}\<[E]%
\ColumnHook
\end{hscode}\resethooks
}{0.45}{
\begin{hscode}\SaveRestoreHook
\column{B}{@{}>{\hspre}l<{\hspost}@{}}%
\column{E}{@{}>{\hspre}l<{\hspost}@{}}%
\>[B]{}\mathbf{infixr}\;\mathrm{8}\;{}^{}\<[E]%
\\
\>[B]{}(\string^)\mathbin{::}(\Varid{n}\to \Varid{b})\to (\Varid{n}\to \mathbb{N})\to \Varid{b}{}\<[E]%
\\
\>[B]{}\Varid{x}\!{\string^}^{\hspace{-1pt}\Varid{p}}\mathrel{=}\bigOp\prod{\Varid{i}}{0}{\,}(\Varid{x}\;\Varid{i})^{(\Varid{p}\;\Varid{i})}{}\<[E]%
\ColumnHook
\end{hscode}\resethooks
}

\begin{lemma}[\provedIn{lemma:pows hom}]\lemlabel{pows hom}
When \ensuremath{(\mathbin{*})} commutes, \ensuremath{(\string^)} satisfies the following exponentiation laws:
\notefoot{Maybe also \ensuremath{(\Varid{x}\!{\string^}^{\hspace{-1pt}\Varid{p}})\!{\string^}^{\hspace{-1pt}\Varid{q}}\mathrel{=}\Varid{x}\!{\string^}^{\hspace{-1pt}\Varid{p}\mathbin{*}\Varid{q}}}. Hunch: I'd have to generalize regular exponentiation and make \ensuremath{(\string^)} a special case.
Handily, I could then drop the carrot symbol.
I think I'll also need the \ensuremath{\Conid{Listable}} class (preferably with a better name).}
\begin{hscode}\SaveRestoreHook
\column{B}{@{}>{\hspre}l<{\hspost}@{}}%
\column{E}{@{}>{\hspre}l<{\hspost}@{}}%
\>[B]{}\Varid{x}\!{\string^}^{\hspace{-1pt}\mathrm{0}}\mathrel{=}\mathrm{1}{}\<[E]%
\\
\>[B]{}\Varid{x}\!{\string^}^{\hspace{-1pt}\Varid{p}\mathbin{+}\Varid{q}}\mathrel{=}\Varid{x}\!{\string^}^{\hspace{-1pt}\Varid{p}}\mathbin{*}\Varid{x}\!{\string^}^{\hspace{-1pt}\Varid{q}}{}\<[E]%
\ColumnHook
\end{hscode}\resethooks
In other words, \ensuremath{\Varid{x}\!{\string^}^{\hspace{-1pt}\cdot }} is a (commutative) monoid homomorphism from the sum monoid to the product monoid.
\end{lemma}

\begin{theorem}%
\thmlabel{generalized poly hom}
The generalized \ensuremath{\Varid{poly}} function is a semiring homomorphism when multiplication on \ensuremath{\Varid{b}} commutes.
\end{theorem}
\begin{proof}
Just like the proof of \thmref{poly hom}, given \lemref{pows hom}.
\end{proof}
\thmref{generalized poly hom} says that the \ensuremath{\Varid{b}\leftarrow(\Varid{n}\to \mathbb{N})} semiring (in which \ensuremath{(\mathbin{*})} is higher-dimensional convolution) correctly implements arithmetic on multivariate polynomials.
We can instead use \ensuremath{\Conid{Map}\;(\Varid{f}\;\mathbb{N})\;\Varid{b}} to denote \ensuremath{\Varid{b}\leftarrow(\Varid{n}\to \mathbb{N})}, where \ensuremath{\Varid{f}} is indexable with \ensuremath{\Conid{Key}\;\Varid{f}\mathrel{=}\Varid{n}}.
One convenient choice is to let \ensuremath{\Varid{n}\mathrel{=}\Conid{String}} for variable names, and \ensuremath{\Varid{f}\mathrel{=}\Conid{Map}\;\Conid{String}}.\footnote{Unfortunately, the \ensuremath{\Conid{Monoid}} instance for the standard \ensuremath{\Conid{Map}} type defines \ensuremath{\Varid{m} \diamond \Varid{m'}} so that keys present in \ensuremath{\Varid{m'}} \emph{replace} those in \ensuremath{\Varid{m}}.
This behavior is problematic for our use (and many others), so we must use a \ensuremath{\Conid{Map}} variant that wraps the standard type, changing the \ensuremath{\Conid{Monoid}} instance so that \ensuremath{\Varid{m} \diamond \Varid{m'}} \emph{combines} values (via \ensuremath{( \diamond )}) associated with the same keys in \ensuremath{\Varid{m}} and \ensuremath{\Varid{m'}}.}
As with \ensuremath{\Conid{Poly}_{1}}, wrap this representation in a new type, and add a \ensuremath{\Conid{Show}} instance:
\begin{hscode}\SaveRestoreHook
\column{B}{@{}>{\hspre}l<{\hspost}@{}}%
\column{3}{@{}>{\hspre}l<{\hspost}@{}}%
\column{E}{@{}>{\hspre}l<{\hspost}@{}}%
\>[B]{}\mathbf{newtype}\;\Conid{Poly}_{\!M}\;\Varid{b}\mathrel{=}\Conid{Poly}_{\!M}\;(\Conid{Map}\;(\Conid{Map}\;\Conid{String}\;\mathbb{N})\;\Varid{b}){}\<[E]%
\\
\>[B]{}\hsindent{3}{}\<[3]%
\>[3]{}\mathbf{deriving}\;(\Conid{Additive},\Conid{Semiring},\Conid{Indexable}\;\Varid{n},\Conid{HasSingle}\;\Varid{n},\Conid{Functor}){}\<[E]%
\\[1.5ex]\>[B]{}\mathbf{instance}\;(\ldots)\Rightarrow \Conid{Show}\;(\Conid{Poly}_{\!M}\;\Varid{b})\;\mathbf{where}\ldots{}\<[E]%
\\[1.5ex]\>[B]{}\Varid{var}\mathbin{::}\Conid{Semiring}\;\Varid{b}\Rightarrow \Conid{String}\to \Conid{Poly}_{\!M}\;\Varid{b}{}\<[E]%
\\
\>[B]{}\Varid{var}\mathrel{=}\Varid{single}\hsdot{\circ }{.\:}\Varid{single}{}\<[E]%
\ColumnHook
\end{hscode}\resethooks
Try it out:
\begin{hscode}\SaveRestoreHook
\column{B}{@{}>{\hspre}l<{\hspost}@{}}%
\column{E}{@{}>{\hspre}l<{\hspost}@{}}%
\>[B]{}\lambda\rangle\ \mathbf{let}\;\Varid{p}\mathrel{=}\Varid{var}\;\text{\tt \char34 x\char34}\mathbin{+}\Varid{var}\;\text{\tt \char34 y\char34}\mathbin{+}\Varid{var}\;\text{\tt \char34 z\char34}\mathbin{::}\Conid{Poly}_{\!M}\;\mathbb Z{}\<[E]%
\\
\>[B]{}\lambda\rangle\ \Varid{p}{}\<[E]%
\\
\>[B]{}\Varid{x}\mathbin{+}\Varid{y}\mathbin{+}\Varid{z}{}\<[E]%
\\[1.5ex]\>[B]{}\lambda\rangle\ \Varid{p}^\mathrm{2}{}\<[E]%
\\
\>[B]{}\Varid{x}^\mathrm{2}\mathbin{+}\mathrm{2}{}\Varid{x}{}\Varid{y}\mathbin{+}\mathrm{2}{}\Varid{x}{}\Varid{z}\mathbin{+}\Varid{y}^\mathrm{2}\mathbin{+}\mathrm{2}{}\Varid{y}{}\Varid{z}\mathbin{+}\Varid{z}^\mathrm{2}{}\<[E]%
\\[1.5ex]\>[B]{}\lambda\rangle\ \Varid{p}^\mathrm{3}{}\<[E]%
\\
\>[B]{}\Varid{x}^\mathrm{3}\mathbin{+}\mathrm{3}{}\Varid{x}^\mathrm{2}{}\Varid{y}\mathbin{+}\mathrm{3}{}\Varid{x}{}\Varid{y}^\mathrm{2}\mathbin{+}\mathrm{6}{}\Varid{x}{}\Varid{y}{}\Varid{z}\mathbin{+}\mathrm{3}{}\Varid{x}^\mathrm{2}{}\Varid{z}\mathbin{+}\mathrm{3}{}\Varid{x}{}\Varid{z}^\mathrm{2}\mathbin{+}\Varid{y}^\mathrm{3}\mathbin{+}\mathrm{3}{}\Varid{y}^\mathrm{2}{}\Varid{z}\mathbin{+}\mathrm{3}{}\Varid{y}{}\Varid{z}^\mathrm{2}\mathbin{+}\Varid{z}^\mathrm{3}{}\<[E]%
\ColumnHook
\end{hscode}\resethooks

\note{Next:
\begin{itemize}
\item Generalize \ensuremath{\Varid{a}^\Varid{b}} and \ensuremath{(\Varid{a}\string^\Varid{b})} via a class, as in the implementation.
\item Maybe generalize \ensuremath{\Conid{Poly}_{1}} from the start.
\item Power series (infinite polynomials).
      Maybe also \ensuremath{[\mskip1.5mu \Varid{a}\mskip1.5mu]}, representing \ensuremath{\Varid{a}\leftarrow\mathbb{N}}.
\item Should I move multidimensional convolution to \secref{Convolution}?
\item References on multivariate polynomial multiplication \href{https://www.google.com/search?q=algorithm+for+multiplying+multivariate+polynomials}{(starting here)}
\item Generalize to $m$-dimensional codomains (and maybe swap roles of $m$ and $n$)
\item Finite maps
\end{itemize}
}

\subsectionl{Image Convolution}

\nc\figO[1]{
\begin{minipage}[t]{0.23\textwidth}
\centering
\includegraphics[width=\textwidth]{test/wizard-#1.png}
\\\textit{#1}
\end{minipage}
}

\figrefdef{wizard}{Image convolution}{
\figO{original}\figO{blur}\figO{sharpen}\figO{edge-detect}
\vspace{-3ex}
} shows examples of image convolution with some commonly used kernels \citep{Petrick2016Kernels,Young95FIP}.
The source image (left) and convolution kernels are all represented as lists of lists of floating point grayscale values.
Because (semiring) multiplication on \ensuremath{[\mskip1.5mu \Varid{b}\mskip1.5mu]} is defined via multiplication on \ensuremath{\Varid{b}}, one can nest representations arbitrarily.
Other more efficient representations can work similarly.

\sectionl{Related Work}

This paper began with a desire to understand regular expression matching via ``derivatives'' by \citet{Brzozowski64} more fundamentally and generally.
Brzozowski's method spurred much follow-up investigation in recent years.
\citet{Owens2009RE} dusted off regular expression derivatives after years of neglect with a new exposition and experience report.
\citet{Might2010YaccID} considerably extended expressiveness to context-free grammars\iflong{ (recursively defined regular expressions)} as well as addressing some efficiency issues, including memoization, with further performance analysis given later \citep{Adams2016CPP}.
\citet{Fischer2010PRE} also extended regular language membership from boolean to ``weighted'' by an arbitrary semiring, relating them to weighted finite automata.
\citet{Piponi2015PF} investigated regular expressions and their relationship to the semiring of polynomial functors, as well as data type derivatives and dissections.
\citet{Radanne2018RLG} explored regular expressions extended to include intersection and complement (as did Brzozowski) with an emphasis on testing.

\citet{McIlroy1999PSPS,McIlroy2001MS} formulated power series as a small and beautiful collection of operations on infinite coefficient streams, including not only the arithmetic operations, but also inversion and composition, as well as differentiation and integration.
He also defined transcendental operations by simple recursion and integration, such as \ensuremath{\Varid{sin}\mathrel{=}\Varid{integral}\;\Varid{cos}} and \ensuremath{\Varid{cos}\mathrel{=}\mathrm{1}\mathbin{-}\Varid{integral}\;\Varid{sin}}.

\citet{Dongol2016CUC} investigated convolution in a general algebraic setting that includes formal language concatenation.
\citet{Kmett2015MfL} observed that Moore machines are a special case of the cofree comonad.
The connections between parsing and semirings have been explored deeply by \citet{Goodman1998PIO,Goodman1999SP} and by \citet{Liu2004}, building on the foundational work of \citet{Chomsky1959CFL}.
\citet{Kmett2011FreeModules} also explored some issues similar to those in the present paper, building on semirings and free semimodules, pointing out that the classic continuation monad can neatly represent linear functionals.

\citet{Kidney2016Semi,semiring-num} implemented a Haskell semiring library that helped with early implementations leading to the present paper, with a particular leaning toward convolution \citep{Kidney2017CS}.
Several of the class instances given above, though independently encountered, also appear in that library.

\note{To do: More fully describe connections between this paper and the work cited above.}


\note{
\begin{itemize}
\item More careful performance testing, analysis, and optimization.
\item Explore Brzozowski derivatives as actual derivatives of residual functions, as in my journal notes from 2019-02-08.
\item Generalization from lists to other data types.
\item Comonadic animation and imagery.
\end{itemize}
}

\appendix

\sectionl{Proofs}

\subsection{\thmref{curry additive}}\prooflabel{theorem:curry additive}

\begin{hscode}\SaveRestoreHook
\column{B}{@{}>{\hspre}c<{\hspost}@{}}%
\column{BE}{@{}l@{}}%
\column{5}{@{}>{\hspre}l<{\hspost}@{}}%
\column{51}{@{}>{\hspre}l<{\hspost}@{}}%
\column{E}{@{}>{\hspre}l<{\hspost}@{}}%
\>[5]{}\Varid{curry}\;\mathrm{0}{}\<[E]%
\\
\>[B]{}\mathrel{=}{}\<[BE]%
\>[5]{}\Varid{curry}\;(\lambda\, (\Varid{x},\Varid{y})\to \mathrm{0}){}\<[51]%
\>[51]{}\mbox{\onelinecomment  \ensuremath{\mathrm{0}} on functions}{}\<[E]%
\\
\>[B]{}\mathrel{=}{}\<[BE]%
\>[5]{}\lambda\, \Varid{x}\to \lambda\, \Varid{y}\to \mathrm{0}{}\<[51]%
\>[51]{}\mbox{\onelinecomment  \ensuremath{\Varid{curry}} definition}{}\<[E]%
\\
\>[B]{}\mathrel{=}{}\<[BE]%
\>[5]{}\lambda\, \Varid{x}\to \mathrm{0}{}\<[51]%
\>[51]{}\mbox{\onelinecomment  \ensuremath{\mathrm{0}} on functions}{}\<[E]%
\\
\>[B]{}\mathrel{=}{}\<[BE]%
\>[5]{}\mathrm{0}{}\<[51]%
\>[51]{}\mbox{\onelinecomment  \ensuremath{\mathrm{0}} on functions}{}\<[E]%
\\[1.5ex]\>[5]{}\Varid{curry}\;(\Varid{f}\mathbin{+}\Varid{g}){}\<[E]%
\\
\>[B]{}\mathrel{=}{}\<[BE]%
\>[5]{}\Varid{curry}\;(\lambda\, (\Varid{x},\Varid{y})\to \Varid{f}\;(\Varid{x},\Varid{y})\mathbin{+}\Varid{g}\;(\Varid{x},\Varid{y})){}\<[51]%
\>[51]{}\mbox{\onelinecomment  \ensuremath{(\mathbin{+})} on functions}{}\<[E]%
\\
\>[B]{}\mathrel{=}{}\<[BE]%
\>[5]{}\lambda\, \Varid{x}\to \lambda\, \Varid{y}\to \Varid{f}\;(\Varid{x},\Varid{y})\mathbin{+}\Varid{g}\;(\Varid{x},\Varid{y}){}\<[51]%
\>[51]{}\mbox{\onelinecomment  \ensuremath{\Varid{curry}} definition}{}\<[E]%
\\
\>[B]{}\mathrel{=}{}\<[BE]%
\>[5]{}\lambda\, \Varid{x}\to \lambda\, \Varid{y}\to \Varid{curry}\;\Varid{f}\;\Varid{x}\;\Varid{y}\mathbin{+}\Varid{curry}\;\Varid{g}\;\Varid{x}\;\Varid{y}{}\<[51]%
\>[51]{}\mbox{\onelinecomment  \ensuremath{\Varid{curry}} definition (twice)}{}\<[E]%
\\
\>[B]{}\mathrel{=}{}\<[BE]%
\>[5]{}\lambda\, \Varid{x}\to \Varid{curry}\;\Varid{f}\;\Varid{x}\mathbin{+}\Varid{curry}\;\Varid{g}\;\Varid{x}{}\<[51]%
\>[51]{}\mbox{\onelinecomment  \ensuremath{(\mathbin{+})} on functions}{}\<[E]%
\\
\>[B]{}\mathrel{=}{}\<[BE]%
\>[5]{}\Varid{curry}\;\Varid{f}\mathbin{+}\Varid{curry}\;\Varid{g}{}\<[51]%
\>[51]{}\mbox{\onelinecomment  \ensuremath{(\mathbin{+})} on functions}{}\<[E]%
\ColumnHook
\end{hscode}\resethooks
Likewise for \ensuremath{\Varid{uncurry}}, or because \ensuremath{\Varid{curry}} and \ensuremath{\Varid{uncurry}} are inverses.

\subsection{\thmref{curry semiring}}\prooflabel{theorem:curry semiring}

For \ensuremath{\mathrm{1}\mathbin{::}\Varid{u} \times \Varid{v}\to \Varid{b}},
\begin{hscode}\SaveRestoreHook
\column{B}{@{}>{\hspre}c<{\hspost}@{}}%
\column{BE}{@{}l@{}}%
\column{5}{@{}>{\hspre}l<{\hspost}@{}}%
\column{38}{@{}>{\hspre}l<{\hspost}@{}}%
\column{E}{@{}>{\hspre}l<{\hspost}@{}}%
\>[5]{}\Varid{curry}\;\mathrm{1}{}\<[E]%
\\
\>[B]{}\mathrel{=}{}\<[BE]%
\>[5]{}\Varid{curry}\;(\varepsilon\mapsto\mathrm{1}){}\<[38]%
\>[38]{}\mbox{\onelinecomment  \ensuremath{\mathrm{1}} on functions}{}\<[E]%
\\
\>[B]{}\mathrel{=}{}\<[BE]%
\>[5]{}\Varid{curry}\;((\varepsilon,\varepsilon)\mapsto\mathrm{1}){}\<[38]%
\>[38]{}\mbox{\onelinecomment  \ensuremath{\varepsilon} on pairs}{}\<[E]%
\\
\>[B]{}\mathrel{=}{}\<[BE]%
\>[5]{}\varepsilon\mapsto\varepsilon\mapsto\mathrm{1}{}\<[38]%
\>[38]{}\mbox{\onelinecomment  \lemref{curry +->}}{}\<[E]%
\\
\>[B]{}\mathrel{=}{}\<[BE]%
\>[5]{}\varepsilon\mapsto\mathrm{1}{}\<[38]%
\>[38]{}\mbox{\onelinecomment  \ensuremath{\mathrm{1}} on functions}{}\<[E]%
\\
\>[B]{}\mathrel{=}{}\<[BE]%
\>[5]{}\mathrm{1}{}\<[38]%
\>[38]{}\mbox{\onelinecomment  \ensuremath{\mathrm{1}} on functions}{}\<[E]%
\ColumnHook
\end{hscode}\resethooks

\noindent
For \ensuremath{\Varid{f},\Varid{g}\mathbin{::}\Varid{u} \times \Varid{v}\to \Varid{b}},
\begin{hscode}\SaveRestoreHook
\column{B}{@{}>{\hspre}c<{\hspost}@{}}%
\column{BE}{@{}l@{}}%
\column{5}{@{}>{\hspre}l<{\hspost}@{}}%
\column{82}{@{}>{\hspre}l<{\hspost}@{}}%
\column{E}{@{}>{\hspre}l<{\hspost}@{}}%
\>[5]{}\Varid{curry}\;(\Varid{f}\mathbin{*}\Varid{g}){}\<[E]%
\\
\>[B]{}\mathrel{=}{}\<[BE]%
\>[5]{}\Varid{curry}\;(\bigOp\sum{(\Varid{u},\Varid{v}),(\Varid{s},\Varid{t})}{1.4}\;(\Varid{u},\Varid{s}) \diamond (\Varid{v},\Varid{t})\mapsto\Varid{f}\;(\Varid{u},\Varid{s})\mathbin{*}\Varid{g}\;(\Varid{v},\Varid{t})){}\<[82]%
\>[82]{}\mbox{\onelinecomment  \ensuremath{(\mathbin{*})} on functions (monoid semiring)}{}\<[E]%
\\
\>[B]{}\mathrel{=}{}\<[BE]%
\>[5]{}\Varid{curry}\;(\bigOp\sum{(\Varid{u},\Varid{v}),(\Varid{s},\Varid{t})}{1.4}\;(\Varid{u} \diamond \Varid{v},\Varid{s} \diamond \Varid{t})\mapsto\Varid{f}\;(\Varid{u},\Varid{s})\mathbin{*}\Varid{g}\;(\Varid{v},\Varid{t})){}\<[82]%
\>[82]{}\mbox{\onelinecomment  \ensuremath{( \diamond )} on pairs}{}\<[E]%
\\
\>[B]{}\mathrel{=}{}\<[BE]%
\>[5]{}\bigOp\sum{(\Varid{u},\Varid{v}),(\Varid{s},\Varid{t})}{1.4}\;\Varid{u} \diamond \Varid{v}\mapsto\Varid{s} \diamond \Varid{t}\mapsto\Varid{f}\;(\Varid{u},\Varid{s})\mathbin{*}\Varid{g}\;(\Varid{v},\Varid{t}){}\<[82]%
\>[82]{}\mbox{\onelinecomment  \lemref{curry +->}}{}\<[E]%
\\
\>[B]{}\mathrel{=}{}\<[BE]%
\>[5]{}\bigOp\sum{\Varid{u},\Varid{v}}{0}{\,}\bigOp\sum{\Varid{s},\Varid{t}}{0}\;\Varid{u} \diamond \Varid{v}\mapsto\Varid{s} \diamond \Varid{t}\mapsto\Varid{f}\;(\Varid{u},\Varid{s})\mathbin{*}\Varid{g}\;(\Varid{v},\Varid{t}){}\<[82]%
\>[82]{}\mbox{\onelinecomment  summation mechanics}{}\<[E]%
\\
\>[B]{}\mathrel{=}{}\<[BE]%
\>[5]{}\bigOp\sum{\Varid{u},\Varid{v}}{0}\;\Varid{u} \diamond \Varid{v}\mapsto\bigOp\sum{\Varid{s},\Varid{t}}{0}\;\Varid{s} \diamond \Varid{t}\mapsto\Varid{f}\;(\Varid{u},\Varid{s})\mathbin{*}\Varid{g}\;(\Varid{v},\Varid{t}){}\<[82]%
\>[82]{}\mbox{\onelinecomment  \lemref{+-> homomorphism}}{}\<[E]%
\\
\>[B]{}\mathrel{=}{}\<[BE]%
\>[5]{}\bigOp\sum{\Varid{u},\Varid{v}}{0}\;\Varid{u} \diamond \Varid{v}\mapsto\bigOp\sum{\Varid{s},\Varid{t}}{0}\;\Varid{s} \diamond \Varid{t}\mapsto\Varid{curry}\;\Varid{f}\;\Varid{u}\;\Varid{s}\mathbin{*}\Varid{curry}\;\Varid{g}\;\Varid{v}\;\Varid{t}{}\<[82]%
\>[82]{}\mbox{\onelinecomment  \ensuremath{\Varid{curry}} definition}{}\<[E]%
\\
\>[B]{}\mathrel{=}{}\<[BE]%
\>[5]{}\bigOp\sum{\Varid{u},\Varid{v}}{0}\;\Varid{u} \diamond \Varid{v}\mapsto\Varid{curry}\;\Varid{f}\;\Varid{u}\mathbin{*}\Varid{curry}\;\Varid{g}\;\Varid{v}{}\<[82]%
\>[82]{}\mbox{\onelinecomment  \ensuremath{(\mapsto)} on functions}{}\<[E]%
\\
\>[B]{}\mathrel{=}{}\<[BE]%
\>[5]{}\Varid{curry}\;\Varid{f}\mathbin{*}\Varid{curry}\;\Varid{g}{}\<[82]%
\>[82]{}\mbox{\onelinecomment  \ensuremath{(\mapsto)} on functions}{}\<[E]%
\ColumnHook
\end{hscode}\resethooks

\subsection{\lemref{decomp +->}}\prooflabel{lemma:decomp +->}

\begin{hscode}\SaveRestoreHook
\column{B}{@{}>{\hspre}c<{\hspost}@{}}%
\column{BE}{@{}l@{}}%
\column{5}{@{}>{\hspre}l<{\hspost}@{}}%
\column{55}{@{}>{\hspre}l<{\hspost}@{}}%
\column{E}{@{}>{\hspre}l<{\hspost}@{}}%
\>[5]{}\bigOp\sum{\Varid{a}}{0}\;\Varid{a}\mapsto\Varid{f}\;\Varid{a}{}\<[E]%
\\
\>[B]{}\mathrel{=}{}\<[BE]%
\>[5]{}\bigOp\sum{\Varid{a}}{0}\;(\lambda\, \Varid{a'}\to \mathbf{if}\;\Varid{a'}\mathrel{=}\Varid{a}\;\mathbf{then}\;\Varid{f}\;\Varid{a}\;\mathbf{else}\;\mathrm{0}){}\<[55]%
\>[55]{}\mbox{\onelinecomment  \ensuremath{(\mapsto)} on \ensuremath{\Varid{a}\to \Varid{b}}}{}\<[E]%
\\
\>[B]{}\mathrel{=}{}\<[BE]%
\>[5]{}\lambda\, \Varid{a'}\to \bigOp\sum{\Varid{a}}{0}\;\mathbf{if}\;\Varid{a'}\mathrel{=}\Varid{a}\;\mathbf{then}\;\Varid{f}\;\Varid{a}\;\mathbf{else}\;\mathrm{0}{}\<[55]%
\>[55]{}\mbox{\onelinecomment  \ensuremath{(\mathbin{+})} on \ensuremath{\Varid{a}\to \Varid{b}}}{}\<[E]%
\\
\>[B]{}\mathrel{=}{}\<[BE]%
\>[5]{}\lambda\, \Varid{a'}\to \Varid{f}\;\Varid{a'}{}\<[55]%
\>[55]{}\mbox{\onelinecomment  other addends vanish}{}\<[E]%
\\
\>[B]{}\mathrel{=}{}\<[BE]%
\>[5]{}\Varid{f}{}\<[55]%
\>[55]{}\mbox{\onelinecomment  $\eta$ reduction}{}\<[E]%
\ColumnHook
\end{hscode}\resethooks

\subsection{\lemref{curry +->}}\prooflabel{lemma:curry +->}

\begin{hscode}\SaveRestoreHook
\column{B}{@{}>{\hspre}c<{\hspost}@{}}%
\column{BE}{@{}l@{}}%
\column{5}{@{}>{\hspre}l<{\hspost}@{}}%
\column{74}{@{}>{\hspre}l<{\hspost}@{}}%
\column{E}{@{}>{\hspre}l<{\hspost}@{}}%
\>[5]{}\Varid{curry}\;((\Varid{a},\Varid{b})\mapsto\Varid{c}){}\<[E]%
\\
\>[B]{}\mathrel{=}{}\<[BE]%
\>[5]{}\Varid{curry}\;(\lambda\, (\Varid{u},\Varid{v})\to \mathbf{if}\;(\Varid{u},\Varid{v})\mathrel{=}(\Varid{a},\Varid{b})\;\mathbf{then}\;\Varid{c}\;\mathbf{else}\;\mathrm{0}){}\<[74]%
\>[74]{}\mbox{\onelinecomment  \ensuremath{(\mapsto)} on functions}{}\<[E]%
\\
\>[B]{}\mathrel{=}{}\<[BE]%
\>[5]{}\Varid{curry}\;(\lambda\, (\Varid{u},\Varid{v})\to \mathbf{if}\;\Varid{u}\mathrel{=}\Varid{a}\mathrel{\wedge}\Varid{v}\mathrel{=}\Varid{b}\;\mathbf{then}\;\Varid{c}\;\mathbf{else}\;\mathrm{0}){}\<[74]%
\>[74]{}\mbox{\onelinecomment  pairing is injective}{}\<[E]%
\\
\>[B]{}\mathrel{=}{}\<[BE]%
\>[5]{}\lambda\, \Varid{u}\to \lambda\, \Varid{v}\to \mathbf{if}\;\Varid{u}\mathrel{=}\Varid{a}\mathrel{\wedge}\Varid{v}\mathrel{=}\Varid{b}\;\mathbf{then}\;\Varid{c}\;\mathbf{else}\;\mathrm{0}{}\<[74]%
\>[74]{}\mbox{\onelinecomment  \ensuremath{\Varid{curry}} definition}{}\<[E]%
\\
\>[B]{}\mathrel{=}{}\<[BE]%
\>[5]{}\lambda\, \Varid{u}\to \lambda\, \Varid{v}\to \mathbf{if}\;\Varid{u}\mathrel{=}\Varid{a}\;\mathbf{then}\;(\mathbf{if}\;\Varid{v}\mathrel{=}\Varid{b}\;\mathbf{then}\;\Varid{c}\;\mathbf{else}\;\mathrm{0})\;\mathbf{else}\;\mathrm{0}{}\<[74]%
\>[74]{}\mbox{\onelinecomment  property of \ensuremath{\mathbf{if}} and \ensuremath{(\mathrel{\wedge})}}{}\<[E]%
\\
\>[B]{}\mathrel{=}{}\<[BE]%
\>[5]{}\lambda\, \Varid{u}\to \mathbf{if}\;\Varid{u}\mathrel{=}\Varid{a}\;\mathbf{then}\;(\lambda\, \Varid{v}\to \mathbf{if}\;\Varid{v}\mathrel{=}\Varid{b}\;\mathbf{then}\;\Varid{c}\;\mathbf{else}\;\mathrm{0})\;\mathbf{else}\;\mathrm{0}{}\<[74]%
\>[74]{}\mbox{\onelinecomment  \ensuremath{(\Varid{u}\mathrel{=}\Varid{a})} is independent of \ensuremath{\Varid{v}}}{}\<[E]%
\\
\>[B]{}\mathrel{=}{}\<[BE]%
\>[5]{}\lambda\, \Varid{u}\to \mathbf{if}\;\Varid{u}\mathrel{=}\Varid{a}\;\mathbf{then}\;\Varid{b}\mapsto\Varid{c}\;\mathbf{else}\;\mathrm{0}{}\<[74]%
\>[74]{}\mbox{\onelinecomment  \ensuremath{(\mapsto)} on functions}{}\<[E]%
\\
\>[B]{}\mathrel{=}{}\<[BE]%
\>[5]{}\Varid{a}\mapsto\Varid{b}\mapsto\Varid{c}{}\<[74]%
\>[74]{}\mbox{\onelinecomment  \ensuremath{(\mapsto)} on functions}{}\<[E]%
\ColumnHook
\end{hscode}\resethooks

\subsection{\thmref{semiring hom ->}}\prooflabel{theorem:semiring hom ->}

\begin{hscode}\SaveRestoreHook
\column{B}{@{}>{\hspre}c<{\hspost}@{}}%
\column{BE}{@{}l@{}}%
\column{5}{@{}>{\hspre}l<{\hspost}@{}}%
\column{49}{@{}>{\hspre}l<{\hspost}@{}}%
\column{E}{@{}>{\hspre}l<{\hspost}@{}}%
\>[5]{}\Varid{pred}\;\mathrm{1}{}\<[E]%
\\
\>[B]{}\mathrel{=}{}\<[BE]%
\>[5]{}\Varid{pred}\;\set{\varepsilon}{}\<[49]%
\>[49]{}\mbox{\onelinecomment  \ensuremath{\mathrm{1}} on sets}{}\<[E]%
\\
\>[B]{}\mathrel{=}{}\<[BE]%
\>[5]{}\lambda\, \Varid{w}\to \Varid{w}\mathbin{\in}\set{\varepsilon}{}\<[49]%
\>[49]{}\mbox{\onelinecomment  \ensuremath{\Varid{pred}} definition}{}\<[E]%
\\
\>[B]{}\mathrel{=}{}\<[BE]%
\>[5]{}\lambda\, \Varid{w}\to \Varid{w}\mathrel{=}\varepsilon{}\<[49]%
\>[49]{}\mbox{\onelinecomment  property of sets}{}\<[E]%
\\
\>[B]{}\mathrel{=}{}\<[BE]%
\>[5]{}\lambda\, \Varid{w}\to \mathbf{if}\;\Varid{w}\mathrel{=}\varepsilon\;\mathbf{then}\;\Conid{True}\;\mathbf{else}\;\Conid{False}{}\<[49]%
\>[49]{}\mbox{\onelinecomment  property of \ensuremath{\mathbf{if}}}{}\<[E]%
\\
\>[B]{}\mathrel{=}{}\<[BE]%
\>[5]{}\lambda\, \Varid{w}\to \mathbf{if}\;\Varid{w}\mathrel{=}\varepsilon\;\mathbf{then}\;\mathrm{1}\;\mathbf{else}\;\mathrm{0}{}\<[49]%
\>[49]{}\mbox{\onelinecomment  \ensuremath{\mathrm{1}} and \ensuremath{\mathrm{0}} on \ensuremath{\Conid{Bool}}}{}\<[E]%
\\
\>[B]{}\mathrel{=}{}\<[BE]%
\>[5]{}\varepsilon\mapsto\mathrm{1}{}\<[49]%
\>[49]{}\mbox{\onelinecomment  \ensuremath{(\mapsto)} definition}{}\<[E]%
\\
\>[B]{}\mathrel{=}{}\<[BE]%
\>[5]{}\mathrm{1}{}\<[49]%
\>[49]{}\mbox{\onelinecomment  \ensuremath{\mathrm{1}} on functions}{}\<[E]%
\ColumnHook
\end{hscode}\resethooks

\begin{hscode}\SaveRestoreHook
\column{B}{@{}>{\hspre}c<{\hspost}@{}}%
\column{BE}{@{}l@{}}%
\column{5}{@{}>{\hspre}l<{\hspost}@{}}%
\column{78}{@{}>{\hspre}l<{\hspost}@{}}%
\column{E}{@{}>{\hspre}l<{\hspost}@{}}%
\>[5]{}\Varid{pred}^{-1}\;(\Varid{pred}\;\Varid{p}\mathbin{*}\Varid{pred}\;\Varid{q}){}\<[E]%
\\
\>[B]{}\mathrel{=}{}\<[BE]%
\>[5]{}\Varid{pred}^{-1}\;(\lambda\, \Varid{w}\to \bigOp\sum{\Varid{u},\Varid{v}\;\!\!\\\!\!\;\Varid{w}\mathrel{=}\Varid{u} \diamond \Varid{v}}{1}\;\Varid{pred}\;\Varid{p}\;\Varid{u}\mathbin{*}\Varid{pred}\;\Varid{q}\;\Varid{v}){}\<[78]%
\>[78]{}\mbox{\onelinecomment  \ensuremath{(\mathbin{*})} on functions}{}\<[E]%
\\
\>[B]{}\mathrel{=}{}\<[BE]%
\>[5]{}\Varid{pred}^{-1}\;(\lambda\, \Varid{w}\to \bigOp\sum{\Varid{u},\Varid{v}\;\!\!\\\!\!\;\Varid{w}\mathrel{=}\Varid{u} \diamond \Varid{v}}{1}\;(\Varid{u}\mathbin{\in}\Varid{p})\mathbin{*}(\Varid{v}\mathbin{\in}\Varid{q})){}\<[78]%
\>[78]{}\mbox{\onelinecomment  \ensuremath{\Varid{pred}} definition (twice)}{}\<[E]%
\\
\>[B]{}\mathrel{=}{}\<[BE]%
\>[5]{}\Varid{pred}^{-1}\;(\lambda\, \Varid{w}\to \bigOp\bigvee{\Varid{u},\Varid{v}\;\!\!\\\!\!\;\Varid{w}\mathrel{=}\Varid{u} \diamond \Varid{v}}{1.5}\;\Varid{u}\mathbin{\in}\Varid{p}\mathrel{\wedge}\Varid{v}\mathbin{\in}\Varid{q}){}\<[78]%
\>[78]{}\mbox{\onelinecomment  \ensuremath{(\mathbin{+})} and \ensuremath{(\mathbin{*})} on \ensuremath{\Conid{Bool}}}{}\<[E]%
\\
\>[B]{}\mathrel{=}{}\<[BE]%
\>[5]{}\set{\Varid{w}\mid\bigOp\bigvee{\Varid{u},\Varid{v}}{0}\;\Varid{w}\mathrel{=}\Varid{u} \diamond \Varid{v}\mathrel{\wedge}\Varid{u}\mathbin{\in}\Varid{p}\mathrel{\wedge}\Varid{v}\mathbin{\in}\Varid{q}}{}\<[78]%
\>[78]{}\mbox{\onelinecomment  \ensuremath{\Varid{pred}^{-1}} definition}{}\<[E]%
\\
\>[B]{}\mathrel{=}{}\<[BE]%
\>[5]{}\set{\Varid{u} \diamond \Varid{v}\mid\Varid{u}\mathbin{\in}\Varid{p}\mathrel{\wedge}\Varid{v}\mathbin{\in}\Varid{q}}{}\<[78]%
\>[78]{}\mbox{\onelinecomment  set notation}{}\<[E]%
\\
\>[B]{}\mathrel{=}{}\<[BE]%
\>[5]{}\Varid{p}\mathbin{*}\Varid{q}{}\<[78]%
\>[78]{}\mbox{\onelinecomment  \ensuremath{(\mathbin{*})} on sets}{}\<[E]%
\ColumnHook
\end{hscode}\resethooks

For \ensuremath{\Conid{StarSemiring}} the default recursive definition embodies the star semiring law.
\note{Hm. Assuming not bottom?}


\subsection{\lemref{decomp ([c] -> b)}}\prooflabel{lemma:decomp ([c] -> b)}

Any argument to \ensuremath{\Varid{f}} must be either \ensuremath{[\mskip1.5mu \mskip1.5mu]} or \ensuremath{\Varid{c}\mathbin{:}\Varid{cs}} for some value \ensuremath{\Varid{c}} and list \ensuremath{\Varid{cs}}.
Consider each case:
\begin{hscode}\SaveRestoreHook
\column{B}{@{}>{\hspre}c<{\hspost}@{}}%
\column{BE}{@{}l@{}}%
\column{5}{@{}>{\hspre}l<{\hspost}@{}}%
\column{34}{@{}>{\hspre}l<{\hspost}@{}}%
\column{E}{@{}>{\hspre}l<{\hspost}@{}}%
\>[5]{}(\Varid{at}_\epsilon\;\Varid{f}\mathrel\triangleleft\derivOp\;\Varid{f})\;[\mskip1.5mu \mskip1.5mu]{}\<[E]%
\\
\>[B]{}\mathrel{=}{}\<[BE]%
\>[5]{}\Varid{at}_\epsilon\;\Varid{f}\;[\mskip1.5mu \mskip1.5mu]{}\<[34]%
\>[34]{}\mbox{\onelinecomment  \ensuremath{\Varid{b}\mathrel\triangleleft\Varid{h}} definition}{}\<[E]%
\\
\>[B]{}\mathrel{=}{}\<[BE]%
\>[5]{}\Varid{f}\;[\mskip1.5mu \mskip1.5mu]{}\<[34]%
\>[34]{}\mbox{\onelinecomment  \ensuremath{\Varid{at}_\epsilon} definition}{}\<[E]%
\\[1.5ex]\>[5]{}(\Varid{at}_\epsilon\;\Varid{f}\mathrel\triangleleft\derivOp\;\Varid{f})\;(\Varid{c}\mathbin{:}\Varid{cs})\;{}\<[34]%
\>[34]{}{}{}\<[E]%
\\
\>[B]{}\mathrel{=}{}\<[BE]%
\>[5]{}\derivOp\;\Varid{f}\;(\Varid{c}\mathbin{:}\Varid{cs}){}\<[34]%
\>[34]{}\mbox{\onelinecomment  \ensuremath{\Varid{b}\mathrel\triangleleft\Varid{h}} definition}{}\<[E]%
\\
\>[B]{}\mathrel{=}{}\<[BE]%
\>[5]{}\Varid{f}\;(\Varid{c}\mathbin{:}\Varid{cs}){}\<[34]%
\>[34]{}\mbox{\onelinecomment  \ensuremath{\derivOp} definition}{}\<[E]%
\ColumnHook
\end{hscode}\resethooks
Thus, for \emph{all} \ensuremath{\Varid{w}\mathbin{::}[\mskip1.5mu \Varid{c}\mskip1.5mu]}, \ensuremath{\Varid{f}\;\Varid{w}\mathrel{=}(\Varid{at}_\epsilon\;\Varid{f}\mathrel\triangleleft\derivOp\;\Varid{f})\;\Varid{w}}, from which the lemma follows by extensionality.

For the other two equations:
\begin{hscode}\SaveRestoreHook
\column{B}{@{}>{\hspre}c<{\hspost}@{}}%
\column{BE}{@{}l@{}}%
\column{5}{@{}>{\hspre}l<{\hspost}@{}}%
\column{68}{@{}>{\hspre}l<{\hspost}@{}}%
\column{E}{@{}>{\hspre}l<{\hspost}@{}}%
\>[5]{}\Varid{at}_\epsilon\;(\Varid{b}\mathrel\triangleleft\Varid{h}){}\<[E]%
\\
\>[B]{}\mathrel{=}{}\<[BE]%
\>[5]{}\Varid{at}_\epsilon\;(\lambda\, {}\;\mathbf{case}\;\{\mskip1.5mu {}\;[\mskip1.5mu \mskip1.5mu]\to \Varid{b}\;{};{}\;\Varid{c}\mathbin{:}\Varid{cs}\to \Varid{h}\;\Varid{c}\;\Varid{cs}\;{}\mskip1.5mu\}){}\<[68]%
\>[68]{}\mbox{\onelinecomment  \ensuremath{(\mathrel\triangleleft)} definition}{}\<[E]%
\\
\>[B]{}\mathrel{=}{}\<[BE]%
\>[5]{}(\lambda\, {}\;\mathbf{case}\;\{\mskip1.5mu {}\;[\mskip1.5mu \mskip1.5mu]\to \Varid{b}\;{};{}\;\Varid{c}\mathbin{:}\Varid{cs}\to \Varid{h}\;\Varid{c}\;\Varid{cs}\;{}\mskip1.5mu\})\;[\mskip1.5mu \mskip1.5mu]{}\<[68]%
\>[68]{}\mbox{\onelinecomment  \ensuremath{\Varid{at}_\epsilon} definition}{}\<[E]%
\\
\>[B]{}\mathrel{=}{}\<[BE]%
\>[5]{}\Varid{b}{}\<[68]%
\>[68]{}\mbox{\onelinecomment  semantics of \ensuremath{\mathbf{case}}}{}\<[E]%
\ColumnHook
\end{hscode}\resethooks
\vspace{-4ex}
\begin{hscode}\SaveRestoreHook
\column{B}{@{}>{\hspre}c<{\hspost}@{}}%
\column{BE}{@{}l@{}}%
\column{5}{@{}>{\hspre}l<{\hspost}@{}}%
\column{84}{@{}>{\hspre}l<{\hspost}@{}}%
\column{E}{@{}>{\hspre}l<{\hspost}@{}}%
\>[5]{}\derivOp\;(\Varid{b}\mathrel\triangleleft\Varid{h}){}\<[E]%
\\
\>[B]{}\mathrel{=}{}\<[BE]%
\>[5]{}\derivOp\;(\lambda\, {}\;\mathbf{case}\;\{\mskip1.5mu {}\;[\mskip1.5mu \mskip1.5mu]\to \Varid{b}\;{};{}\;\Varid{c}\mathbin{:}\Varid{cs}\to \Varid{h}\;\Varid{c}\;\Varid{cs}\;{}\mskip1.5mu\}){}\<[84]%
\>[84]{}\mbox{\onelinecomment  \ensuremath{(\mathrel\triangleleft)} definition}{}\<[E]%
\\
\>[B]{}\mathrel{=}{}\<[BE]%
\>[5]{}\lambda\, \Varid{c}\to \lambda\, \Varid{cs}\to (\lambda\, {}\;\mathbf{case}\;\{\mskip1.5mu {}\;[\mskip1.5mu \mskip1.5mu]\to \Varid{b}\;{};{}\;\Varid{c}\mathbin{:}\Varid{cs}\to \Varid{h}\;\Varid{c}\;\Varid{cs}\;{}\mskip1.5mu\})\;(\Varid{c}\mathbin{:}\Varid{cs}){}\<[84]%
\>[84]{}\mbox{\onelinecomment  \ensuremath{\derivOp} definition}{}\<[E]%
\\
\>[B]{}\mathrel{=}{}\<[BE]%
\>[5]{}\lambda\, \Varid{c}\to \lambda\, \Varid{cs}\to \Varid{h}\;\Varid{c}\;\Varid{cs}{}\<[84]%
\>[84]{}\mbox{\onelinecomment  semantics of \ensuremath{\mathbf{case}}}{}\<[E]%
\\
\>[B]{}\mathrel{=}{}\<[BE]%
\>[5]{}\Varid{h}{}\<[84]%
\>[84]{}\mbox{\onelinecomment  $\eta$ reduction (twice)}{}\<[E]%
\ColumnHook
\end{hscode}\resethooks

\subsection{\lemref{atEps [c] -> b}}\prooflabel{lemma:atEps [c] -> b}

\begin{hscode}\SaveRestoreHook
\column{B}{@{}>{\hspre}c<{\hspost}@{}}%
\column{BE}{@{}l@{}}%
\column{5}{@{}>{\hspre}l<{\hspost}@{}}%
\column{26}{@{}>{\hspre}l<{\hspost}@{}}%
\column{E}{@{}>{\hspre}l<{\hspost}@{}}%
\>[5]{}\Varid{at}_\epsilon\;\mathrm{0}{}\<[E]%
\\
\>[B]{}\mathrel{=}{}\<[BE]%
\>[5]{}\Varid{at}_\epsilon\;(\lambda\, \Varid{a}\to \mathrm{0}){}\<[26]%
\>[26]{}\mbox{\onelinecomment  \ensuremath{\mathrm{0}} on functions}{}\<[E]%
\\
\>[B]{}\mathrel{=}{}\<[BE]%
\>[5]{}(\lambda\, \Varid{a}\to \mathrm{0})\;[\mskip1.5mu \mskip1.5mu]{}\<[26]%
\>[26]{}\mbox{\onelinecomment  \ensuremath{\Varid{at}_\epsilon} definition}{}\<[E]%
\\
\>[B]{}\mathrel{=}{}\<[BE]%
\>[5]{}\mathrm{0}{}\<[26]%
\>[26]{}\mbox{\onelinecomment  $\beta$ reduction}{}\<[E]%
\ColumnHook
\end{hscode}\resethooks

\begin{hscode}\SaveRestoreHook
\column{B}{@{}>{\hspre}c<{\hspost}@{}}%
\column{BE}{@{}l@{}}%
\column{5}{@{}>{\hspre}l<{\hspost}@{}}%
\column{57}{@{}>{\hspre}l<{\hspost}@{}}%
\column{E}{@{}>{\hspre}l<{\hspost}@{}}%
\>[5]{}\Varid{at}_\epsilon\;\mathrm{1}{}\<[E]%
\\
\>[B]{}\mathrel{=}{}\<[BE]%
\>[5]{}\Varid{at}_\epsilon\;(\varepsilon\mapsto\mathrm{1}){}\<[57]%
\>[57]{}\mbox{\onelinecomment  \ensuremath{\mathrm{1}} on functions}{}\<[E]%
\\
\>[B]{}\mathrel{=}{}\<[BE]%
\>[5]{}\Varid{at}_\epsilon\;(\lambda\, \Varid{b}\to \mathbf{if}\;\Varid{b}\mathrel{=}\varepsilon\;\mathbf{then}\;\mathrm{1}\;\mathbf{else}\;\mathrm{0}){}\<[57]%
\>[57]{}\mbox{\onelinecomment  \ensuremath{(\mapsto)} on functions}{}\<[E]%
\\
\>[B]{}\mathrel{=}{}\<[BE]%
\>[5]{}\Varid{at}_\epsilon\;(\lambda\, \Varid{b}\to \mathbf{if}\;\Varid{b}\mathrel{=}\varepsilon\;\mathbf{then}\;\Conid{True}\;\mathbf{else}\;\Conid{False}){}\<[57]%
\>[57]{}\mbox{\onelinecomment  \ensuremath{\mathrm{1}} and \ensuremath{\mathrm{0}} on \ensuremath{\Conid{Bool}}}{}\<[E]%
\\
\>[B]{}\mathrel{=}{}\<[BE]%
\>[5]{}\Varid{at}_\epsilon\;(\lambda\, \Varid{b}\to \Varid{b}\mathrel{=}\varepsilon){}\<[57]%
\>[57]{}\mbox{\onelinecomment  property of \ensuremath{\mathbf{if}}}{}\<[E]%
\\
\>[B]{}\mathrel{=}{}\<[BE]%
\>[5]{}\varepsilon\mathrel{=}\varepsilon{}\<[57]%
\>[57]{}\mbox{\onelinecomment  \ensuremath{\Varid{at}_\epsilon} definition            }{}\<[E]%
\\
\>[B]{}\mathrel{=}{}\<[BE]%
\>[5]{}\mathrm{1}{}\<[E]%
\ColumnHook
\end{hscode}\resethooks

\begin{hscode}\SaveRestoreHook
\column{B}{@{}>{\hspre}c<{\hspost}@{}}%
\column{BE}{@{}l@{}}%
\column{5}{@{}>{\hspre}l<{\hspost}@{}}%
\column{33}{@{}>{\hspre}l<{\hspost}@{}}%
\column{E}{@{}>{\hspre}l<{\hspost}@{}}%
\>[5]{}\Varid{at}_\epsilon\;(\Varid{f}\mathbin{+}\Varid{g}){}\<[E]%
\\
\>[B]{}\mathrel{=}{}\<[BE]%
\>[5]{}\Varid{at}_\epsilon\;(\lambda\, \Varid{a}\to \Varid{f}\;\Varid{a}\mathbin{+}\Varid{g}\;\Varid{a}){}\<[33]%
\>[33]{}\mbox{\onelinecomment  \ensuremath{(\mathbin{+})} on functions}{}\<[E]%
\\
\>[B]{}\mathrel{=}{}\<[BE]%
\>[5]{}(\lambda\, \Varid{a}\to \Varid{f}\;\Varid{a}\mathbin{+}\Varid{g}\;\Varid{a})\;[\mskip1.5mu \mskip1.5mu]{}\<[33]%
\>[33]{}\mbox{\onelinecomment  \ensuremath{\Varid{at}_\epsilon} definition}{}\<[E]%
\\
\>[B]{}\mathrel{=}{}\<[BE]%
\>[5]{}\Varid{f}\;[\mskip1.5mu \mskip1.5mu]\mathbin{+}\Varid{g}\;[\mskip1.5mu \mskip1.5mu]{}\<[33]%
\>[33]{}\mbox{\onelinecomment  $\beta$ reduction}{}\<[E]%
\\
\>[B]{}\mathrel{=}{}\<[BE]%
\>[5]{}\Varid{at}_\epsilon\;\Varid{f}\mathbin{+}\Varid{at}_\epsilon\;\Varid{g}{}\<[33]%
\>[33]{}\mbox{\onelinecomment  \ensuremath{\Varid{at}_\epsilon} definition}{}\<[E]%
\ColumnHook
\end{hscode}\resethooks

\begin{hscode}\SaveRestoreHook
\column{B}{@{}>{\hspre}c<{\hspost}@{}}%
\column{BE}{@{}l@{}}%
\column{5}{@{}>{\hspre}l<{\hspost}@{}}%
\column{63}{@{}>{\hspre}l<{\hspost}@{}}%
\column{E}{@{}>{\hspre}l<{\hspost}@{}}%
\>[5]{}\Varid{at}_\epsilon\;(\Varid{f}\mathbin{*}\Varid{g}){}\<[E]%
\\
\>[B]{}\mathrel{=}{}\<[BE]%
\>[5]{}\Varid{at}_\epsilon\;(\bigOp\sum{\Varid{u},\Varid{v}}{0}\;\Varid{u} \diamond \Varid{v}\mapsto\Varid{f}\;\Varid{u}\mathbin{*}\Varid{g}\;\Varid{v}){}\<[63]%
\>[63]{}\mbox{\onelinecomment  \ensuremath{(\mathbin{*})} on functions}{}\<[E]%
\\
\>[B]{}\mathrel{=}{}\<[BE]%
\>[5]{}\Varid{at}_\epsilon\;(\lambda\, \Varid{w}\to \bigOp\sum{\Varid{u},\Varid{v}\;\!\!\\\!\!\;\Varid{u} \diamond \Varid{v}\mathrel{=}[\mskip1.5mu \mskip1.5mu]}{1}\;\Varid{f}\;\Varid{u}\mathbin{*}\Varid{g}\;\Varid{v}){}\<[63]%
\>[63]{}\mbox{\onelinecomment  alternative definition from \figref{monoid semiring}}{}\<[E]%
\\
\>[B]{}\mathrel{=}{}\<[BE]%
\>[5]{}\bigOp\sum{\Varid{u},\Varid{v}\;\!\!\\\!\!\;\Varid{u}\mathrel{=}[\mskip1.5mu \mskip1.5mu]\mathrel{\wedge}\Varid{v}\mathrel{=}[\mskip1.5mu \mskip1.5mu]}{2.5}\;{}\;\Varid{f}\;\Varid{u}\mathbin{*}\Varid{g}\;\Varid{v}{}\<[63]%
\>[63]{}\mbox{\onelinecomment  \ensuremath{\Varid{u} \diamond \Varid{v}\mathrel{=}[\mskip1.5mu \mskip1.5mu]\Longleftrightarrow\Varid{u}\mathrel{=}[\mskip1.5mu \mskip1.5mu]\mathrel{\wedge}\Varid{v}\mathrel{=}[\mskip1.5mu \mskip1.5mu]} }{}\<[E]%
\\
\>[B]{}\mathrel{=}{}\<[BE]%
\>[5]{}\Varid{f}\;[\mskip1.5mu \mskip1.5mu]\mathbin{*}\Varid{g}\;[\mskip1.5mu \mskip1.5mu]{}\<[63]%
\>[63]{}\mbox{\onelinecomment  singleton sum}{}\<[E]%
\\
\>[B]{}\mathrel{=}{}\<[BE]%
\>[5]{}\Varid{at}_\epsilon\;\Varid{f}\mathbin{*}\Varid{at}_\epsilon\;\Varid{g}{}\<[63]%
\>[63]{}\mbox{\onelinecomment  \ensuremath{\Varid{at}_\epsilon} definition}{}\<[E]%
\ColumnHook
\end{hscode}\resethooks


\begin{hscode}\SaveRestoreHook
\column{B}{@{}>{\hspre}c<{\hspost}@{}}%
\column{BE}{@{}l@{}}%
\column{5}{@{}>{\hspre}l<{\hspost}@{}}%
\column{30}{@{}>{\hspre}l<{\hspost}@{}}%
\column{E}{@{}>{\hspre}l<{\hspost}@{}}%
\>[5]{}\Varid{at}_\epsilon\;(\closure{\Varid{p}}){}\<[E]%
\\
\>[B]{}\mathrel{=}{}\<[BE]%
\>[5]{}\Varid{at}_\epsilon\;(\bigOp\sum{\Varid{i}}{0}{\,}\Varid{p}^\Varid{i}){}\<[30]%
\>[30]{}\mbox{\onelinecomment  alternative \ensuremath{\closure{\Varid{p}}} formulation}{}\<[E]%
\\
\>[B]{}\mathrel{=}{}\<[BE]%
\>[5]{}\bigOp\sum{\Varid{i}}{0}{\,}(\Varid{at}_\epsilon\;\Varid{p})^\Varid{i}{}\<[30]%
\>[30]{}\mbox{\onelinecomment  \ensuremath{\Varid{at}_\epsilon} is a semiring homomorphism (above)}{}\<[E]%
\\
\>[B]{}\mathrel{=}{}\<[BE]%
\>[5]{}\closure{(\Varid{at}_\epsilon\;\Varid{p})}{}\<[30]%
\>[30]{}\mbox{\onelinecomment  defining property of \ensuremath{\closure{\cdot }}}{}\<[E]%
\ColumnHook
\end{hscode}\resethooks

\begin{hscode}\SaveRestoreHook
\column{B}{@{}>{\hspre}c<{\hspost}@{}}%
\column{BE}{@{}l@{}}%
\column{5}{@{}>{\hspre}l<{\hspost}@{}}%
\column{29}{@{}>{\hspre}l<{\hspost}@{}}%
\column{E}{@{}>{\hspre}l<{\hspost}@{}}%
\>[5]{}\Varid{at}_\epsilon\;(\Varid{s}\cdot\Varid{f}){}\<[E]%
\\
\>[B]{}\mathrel{=}{}\<[BE]%
\>[5]{}\Varid{at}_\epsilon\;(\lambda\, \Varid{a}\to \Varid{s}\mathbin{*}\Varid{f}\;\Varid{a}){}\<[29]%
\>[29]{}\mbox{\onelinecomment  \ensuremath{(\cdot)} on functions}{}\<[E]%
\\
\>[B]{}\mathrel{=}{}\<[BE]%
\>[5]{}(\lambda\, \Varid{a}\to \Varid{s}\mathbin{*}\Varid{f}\;\Varid{a})\;[\mskip1.5mu \mskip1.5mu]{}\<[29]%
\>[29]{}\mbox{\onelinecomment  \ensuremath{\Varid{at}_\epsilon} definition}{}\<[E]%
\\
\>[B]{}\mathrel{=}{}\<[BE]%
\>[5]{}\Varid{s}\mathbin{*}\Varid{f}\;[\mskip1.5mu \mskip1.5mu]{}\<[29]%
\>[29]{}\mbox{\onelinecomment  $\beta$ reduction}{}\<[E]%
\\
\>[B]{}\mathrel{=}{}\<[BE]%
\>[5]{}\Varid{s}\mathbin{*}\Varid{at}_\epsilon\;\Varid{f}{}\<[29]%
\>[29]{}\mbox{\onelinecomment  \ensuremath{\Varid{at}_\epsilon} definition}{}\<[E]%
\ColumnHook
\end{hscode}\resethooks

\begin{hscode}\SaveRestoreHook
\column{B}{@{}>{\hspre}c<{\hspost}@{}}%
\column{BE}{@{}l@{}}%
\column{5}{@{}>{\hspre}l<{\hspost}@{}}%
\column{53}{@{}>{\hspre}l<{\hspost}@{}}%
\column{E}{@{}>{\hspre}l<{\hspost}@{}}%
\>[5]{}\Varid{at}_\epsilon\;([\mskip1.5mu \mskip1.5mu]\mapsto\Varid{b}){}\<[E]%
\\
\>[B]{}\mathrel{=}{}\<[BE]%
\>[5]{}\Varid{at}_\epsilon\;(\lambda\, \Varid{w}\to \mathbf{if}\;\Varid{w}\mathrel{=}[\mskip1.5mu \mskip1.5mu]\;\mathbf{then}\;\Varid{b}\;\mathbf{else}\;\mathrm{0}){}\<[53]%
\>[53]{}\mbox{\onelinecomment  \ensuremath{(\mapsto)} on \ensuremath{[\mskip1.5mu \Varid{c}\mskip1.5mu]\to \Varid{b}}}{}\<[E]%
\\
\>[B]{}\mathrel{=}{}\<[BE]%
\>[5]{}(\lambda\, \Varid{w}\to \mathbf{if}\;\Varid{w}\mathrel{=}[\mskip1.5mu \mskip1.5mu]\;\mathbf{then}\;\Varid{b}\;\mathbf{else}\;\mathrm{0})\;[\mskip1.5mu \mskip1.5mu]{}\<[53]%
\>[53]{}\mbox{\onelinecomment  \ensuremath{\Varid{at}_\epsilon} definition}{}\<[E]%
\\
\>[B]{}\mathrel{=}{}\<[BE]%
\>[5]{}\mathbf{if}\;[\mskip1.5mu \mskip1.5mu]\mathrel{=}[\mskip1.5mu \mskip1.5mu]\;\mathbf{then}\;\Varid{b}\;\mathbf{else}\;\mathrm{0}{}\<[53]%
\>[53]{}\mbox{\onelinecomment  $\beta$ reduction}{}\<[E]%
\\
\>[B]{}\mathrel{=}{}\<[BE]%
\>[5]{}\Varid{b}{}\<[53]%
\>[53]{}\mbox{\onelinecomment  \ensuremath{\mathbf{if}\;\Conid{True}}}{}\<[E]%
\\[1.5ex]\>[5]{}\Varid{at}_\epsilon\;(\Varid{c'}\mathbin{:}\Varid{cs'}\mapsto\Varid{b}){}\<[E]%
\\
\>[B]{}\mathrel{=}{}\<[BE]%
\>[5]{}\Varid{at}_\epsilon\;(\lambda\, \Varid{w}\to \mathbf{if}\;\Varid{w}\mathrel{=}\Varid{c'}\mathbin{:}\Varid{cs'}\;\mathbf{then}\;\Varid{b}\;\mathbf{else}\;\mathrm{0}){}\<[53]%
\>[53]{}\mbox{\onelinecomment  \ensuremath{(\mapsto)} on \ensuremath{[\mskip1.5mu \Varid{c}\mskip1.5mu]\to \Varid{b}}}{}\<[E]%
\\
\>[B]{}\mathrel{=}{}\<[BE]%
\>[5]{}(\lambda\, \Varid{w}\to \mathbf{if}\;\Varid{w}\mathrel{=}\Varid{c'}\mathbin{:}\Varid{cs'}\;\mathbf{then}\;\Varid{b}\;\mathbf{else}\;\mathrm{0})\;[\mskip1.5mu \mskip1.5mu]{}\<[53]%
\>[53]{}\mbox{\onelinecomment  \ensuremath{\Varid{at}_\epsilon} definition}{}\<[E]%
\\
\>[B]{}\mathrel{=}{}\<[BE]%
\>[5]{}\mathbf{if}\;[\mskip1.5mu \mskip1.5mu]\mathrel{=}\Varid{c'}\mathbin{:}\Varid{cs'}\;\mathbf{then}\;\Varid{b}\;\mathbf{else}\;\mathrm{0}{}\<[53]%
\>[53]{}\mbox{\onelinecomment  $\beta$ reduction}{}\<[E]%
\\
\>[B]{}\mathrel{=}{}\<[BE]%
\>[5]{}\mathrm{0}{}\<[53]%
\>[53]{}\mbox{\onelinecomment  \ensuremath{\mathbf{if}\;\Conid{False}}}{}\<[E]%
\ColumnHook
\end{hscode}\resethooks

\subsection{\lemref{deriv [c] -> b}}\prooflabel{lemma:deriv [c] -> b}

\begin{hscode}\SaveRestoreHook
\column{B}{@{}>{\hspre}c<{\hspost}@{}}%
\column{BE}{@{}l@{}}%
\column{5}{@{}>{\hspre}l<{\hspost}@{}}%
\column{42}{@{}>{\hspre}l<{\hspost}@{}}%
\column{E}{@{}>{\hspre}l<{\hspost}@{}}%
\>[5]{}\derivOp\;\mathrm{0}{}\<[E]%
\\
\>[B]{}\mathrel{=}{}\<[BE]%
\>[5]{}\derivOp\;(\lambda\, \Varid{w}\to \mathrm{0}){}\<[42]%
\>[42]{}\mbox{\onelinecomment  \ensuremath{\mathrm{0}} on functions}{}\<[E]%
\\
\>[B]{}\mathrel{=}{}\<[BE]%
\>[5]{}\lambda\, \Varid{c}\to \lambda\, \Varid{cs}\to (\lambda\, \Varid{w}\to \mathrm{0})\;(\Varid{c}\mathbin{:}\Varid{cs}){}\<[42]%
\>[42]{}\mbox{\onelinecomment  \ensuremath{\derivOp} on functions}{}\<[E]%
\\
\>[B]{}\mathrel{=}{}\<[BE]%
\>[5]{}\lambda\, \Varid{c}\to \lambda\, \Varid{cs}\to \mathrm{0}{}\<[42]%
\>[42]{}\mbox{\onelinecomment  $\beta$ reduction}{}\<[E]%
\\
\>[B]{}\mathrel{=}{}\<[BE]%
\>[5]{}\lambda\, \Varid{c}\to \mathrm{0}{}\<[42]%
\>[42]{}\mbox{\onelinecomment  \ensuremath{\mathrm{0}} on functions}{}\<[E]%
\\
\>[B]{}\mathrel{=}{}\<[BE]%
\>[5]{}\mathrm{0}{}\<[42]%
\>[42]{}\mbox{\onelinecomment  \ensuremath{\mathrm{0}} on \ensuremath{\Varid{a}\to \Varid{b}}}{}\<[E]%
\ColumnHook
\end{hscode}\resethooks
\vspace{-3ex}

\begin{hscode}\SaveRestoreHook
\column{B}{@{}>{\hspre}c<{\hspost}@{}}%
\column{BE}{@{}l@{}}%
\column{5}{@{}>{\hspre}l<{\hspost}@{}}%
\column{42}{@{}>{\hspre}l<{\hspost}@{}}%
\column{E}{@{}>{\hspre}l<{\hspost}@{}}%
\>[5]{}\derivOp\;\mathrm{1}{}\<[E]%
\\
\>[B]{}\mathrel{=}{}\<[BE]%
\>[5]{}\derivOp\;(\Varid{single}\;\varepsilon){}\<[42]%
\>[42]{}\mbox{\onelinecomment  \ensuremath{\mathrm{1}} on functions}{}\<[E]%
\\
\>[B]{}\mathrel{=}{}\<[BE]%
\>[5]{}\lambda\, \Varid{c}\to \lambda\, \Varid{cs}\to \Varid{single}\;\varepsilon\;(\Varid{c}\mathbin{:}\Varid{cs}){}\<[42]%
\>[42]{}\mbox{\onelinecomment  \ensuremath{\derivOp} on functions}{}\<[E]%
\\
\>[B]{}\mathrel{=}{}\<[BE]%
\>[5]{}\lambda\, \Varid{c}\to \lambda\, \Varid{cs}\to \mathrm{0}{}\<[42]%
\>[42]{}\mbox{\onelinecomment  \ensuremath{\Varid{c}\mathbin{:}\Varid{cs}\not=\varepsilon}}{}\<[E]%
\\
\>[B]{}\mathrel{=}{}\<[BE]%
\>[5]{}\lambda\, \Varid{c}\to \mathrm{0}{}\<[42]%
\>[42]{}\mbox{\onelinecomment  \ensuremath{\mathrm{0}} on functions}{}\<[E]%
\\
\>[B]{}\mathrel{=}{}\<[BE]%
\>[5]{}\mathrm{0}{}\<[42]%
\>[42]{}\mbox{\onelinecomment  \ensuremath{\mathrm{0}} on \ensuremath{\Varid{a}\to \Varid{b}}}{}\<[E]%
\ColumnHook
\end{hscode}\resethooks
\vspace{-3ex}

\begin{hscode}\SaveRestoreHook
\column{B}{@{}>{\hspre}c<{\hspost}@{}}%
\column{BE}{@{}l@{}}%
\column{5}{@{}>{\hspre}l<{\hspost}@{}}%
\column{55}{@{}>{\hspre}l<{\hspost}@{}}%
\column{E}{@{}>{\hspre}l<{\hspost}@{}}%
\>[5]{}\derivOp\;(\Varid{f}\mathbin{+}\Varid{g}){}\<[E]%
\\
\>[B]{}\mathrel{=}{}\<[BE]%
\>[5]{}\derivOp\;(\lambda\, \Varid{w}\to \Varid{f}\;\Varid{w}\mathbin{+}\Varid{g}\;\Varid{w}){}\<[55]%
\>[55]{}\mbox{\onelinecomment  \ensuremath{(\mathbin{+})} on functions}{}\<[E]%
\\
\>[B]{}\mathrel{=}{}\<[BE]%
\>[5]{}\lambda\, \Varid{c}\to \lambda\, \Varid{cs}\to (\lambda\, \Varid{w}\to \Varid{f}\;\Varid{w}\mathbin{+}\Varid{g}\;\Varid{w})\;(\Varid{c}\mathbin{:}\Varid{cs}){}\<[55]%
\>[55]{}\mbox{\onelinecomment  \ensuremath{\derivOp} on functions}{}\<[E]%
\\
\>[B]{}\mathrel{=}{}\<[BE]%
\>[5]{}\lambda\, \Varid{c}\to \lambda\, \Varid{cs}\to \Varid{f}\;(\Varid{c}\mathbin{:}\Varid{cs})\mathbin{+}\Varid{g}\;(\Varid{c}\mathbin{:}\Varid{cs}){}\<[55]%
\>[55]{}\mbox{\onelinecomment  $\beta$ reduction}{}\<[E]%
\\
\>[B]{}\mathrel{=}{}\<[BE]%
\>[5]{}\lambda\, \Varid{c}\to (\lambda\, \Varid{cs}\to \Varid{f}\;(\Varid{c}\mathbin{:}\Varid{cs}))\mathbin{+}(\lambda\, \Varid{cs}\to \Varid{g}\;(\Varid{c}\mathbin{:}\Varid{cs})){}\<[55]%
\>[55]{}\mbox{\onelinecomment  \ensuremath{(\mathbin{+})} on functions}{}\<[E]%
\\
\>[B]{}\mathrel{=}{}\<[BE]%
\>[5]{}\lambda\, \Varid{c}\to \derivOp\;\Varid{f}\;\Varid{c}\mathbin{+}\derivOp\;\Varid{g}\;\Varid{c}{}\<[55]%
\>[55]{}\mbox{\onelinecomment  \ensuremath{\derivOp} on functions}{}\<[E]%
\\
\>[B]{}\mathrel{=}{}\<[BE]%
\>[5]{}\derivOp\;\Varid{f}\mathbin{+}\derivOp\;\Varid{g}{}\<[55]%
\>[55]{}\mbox{\onelinecomment  \ensuremath{(\mathbin{+})} on \ensuremath{\Varid{a}\to \Varid{b}}}{}\<[E]%
\ColumnHook
\end{hscode}\resethooks

\begin{hscode}\SaveRestoreHook
\column{B}{@{}>{\hspre}c<{\hspost}@{}}%
\column{BE}{@{}l@{}}%
\column{5}{@{}>{\hspre}l<{\hspost}@{}}%
\column{125}{@{}>{\hspre}l<{\hspost}@{}}%
\column{E}{@{}>{\hspre}l<{\hspost}@{}}%
\>[5]{}\derivOp\;(\Varid{f}\mathbin{*}\Varid{g}){}\<[E]%
\\
\>[B]{}\mathrel{=}{}\<[BE]%
\>[5]{}\derivOp\;(\bigOp\sum{\Varid{u},\Varid{v}}{0}\;\Varid{u} \diamond \Varid{v}\mapsto\Varid{f}\;\Varid{u}\mathbin{*}\Varid{g}\;\Varid{v}){}\<[125]%
\>[125]{}\mbox{\onelinecomment  \ensuremath{(\mathbin{*})} on functions}{}\<[E]%
\\
\>[B]{}\mathrel{=}{}\<[BE]%
\>[5]{}\derivOp\;(\bigOp\sum{\Varid{v}}{0}\;(\varepsilon \diamond \Varid{v}\mapsto\Varid{f}\;\varepsilon\mathbin{*}\Varid{g}\;\Varid{v})\mathbin{+}\bigOp\sum{\Varid{c'},\Varid{u'},\Varid{v}}{1}\;((\Varid{c'}\mathbin{:}\Varid{u'}) \diamond \Varid{v}\mapsto\Varid{f}\;(\Varid{c'}\mathbin{:}\Varid{u'})\mathbin{*}\Varid{g}\;\Varid{v})){}\<[125]%
\>[125]{}\mbox{\onelinecomment  empty vs nonempty \ensuremath{\Varid{u}}}{}\<[E]%
\\
\>[B]{}\mathrel{=}{}\<[BE]%
\>[5]{}\derivOp\;(\bigOp\sum{\Varid{v}}{0}\;(\varepsilon \diamond \Varid{v}\mapsto\Varid{f}\;\varepsilon\mathbin{*}\Varid{g}\;\Varid{v}))\mathbin{+}\derivOp\;(\bigOp\sum{\Varid{c'},\Varid{u'},\Varid{v}}{1}\;((\Varid{c'}\mathbin{:}\Varid{u'}) \diamond \Varid{v}\mapsto\Varid{f}\;(\Varid{c'}\mathbin{:}\Varid{u'})\mathbin{*}\Varid{g}\;\Varid{v})){}\<[125]%
\>[125]{}\mbox{\onelinecomment  additivity of \ensuremath{\derivOp} (above)}{}\<[E]%
\ColumnHook
\end{hscode}\resethooks
First addend:
\begin{hscode}\SaveRestoreHook
\column{B}{@{}>{\hspre}c<{\hspost}@{}}%
\column{BE}{@{}l@{}}%
\column{5}{@{}>{\hspre}l<{\hspost}@{}}%
\column{56}{@{}>{\hspre}l<{\hspost}@{}}%
\column{E}{@{}>{\hspre}l<{\hspost}@{}}%
\>[5]{}\derivOp\;(\bigOp\sum{\Varid{v}}{0}\;(\varepsilon \diamond \Varid{v}\mapsto\Varid{f}\;\varepsilon\mathbin{*}\Varid{g}\;\Varid{v})){}\<[E]%
\\
\>[B]{}\mathrel{=}{}\<[BE]%
\>[5]{}\derivOp\;(\bigOp\sum{\Varid{v}}{0}\;(\Varid{v}\mapsto\Varid{f}\;\varepsilon\mathbin{*}\Varid{g}\;\Varid{v})){}\<[56]%
\>[56]{}\mbox{\onelinecomment  monoid law}{}\<[E]%
\\
\>[B]{}\mathrel{=}{}\<[BE]%
\>[5]{}\derivOp\;(\Varid{f}\;\varepsilon\cdot\bigOp\sum{\Varid{v}}{0}\;(\Varid{v}\mapsto\Varid{g}\;\Varid{v})){}\<[56]%
\>[56]{}\mbox{\onelinecomment  distributivity (semiring law)}{}\<[E]%
\\
\>[B]{}\mathrel{=}{}\<[BE]%
\>[5]{}\lambda\, \Varid{c}\to \derivOp\;(\Varid{f}\;\varepsilon\cdot\bigOp\sum{\Varid{v}}{0}\;(\Varid{v}\mapsto\Varid{g}\;\Varid{v}))\;\Varid{c}{}\<[56]%
\>[56]{}\mbox{\onelinecomment  $\eta$ expansion}{}\<[E]%
\\
\>[B]{}\mathrel{=}{}\<[BE]%
\>[5]{}\lambda\, \Varid{c}\to \Varid{f}\;\varepsilon\cdot\derivOp\;(\bigOp\sum{\Varid{v}}{0}\;\Varid{v}\mapsto\Varid{g}\;\Varid{v})\;\Varid{c}{}\<[56]%
\>[56]{}\mbox{\onelinecomment  additivity of \ensuremath{\derivOp} (above)}{}\<[E]%
\\
\>[B]{}\mathrel{=}{}\<[BE]%
\>[5]{}\lambda\, \Varid{c}\to \Varid{f}\;\varepsilon\cdot\derivOp\;\Varid{g}\;\Varid{c}{}\<[56]%
\>[56]{}\mbox{\onelinecomment  \lemref{decomp +->}}{}\<[E]%
\\
\>[B]{}\mathrel{=}{}\<[BE]%
\>[5]{}\lambda\, \Varid{c}\to \Varid{at}_\epsilon\;\Varid{f}\cdot\derivOp\;\Varid{g}{}\<[56]%
\>[56]{}\mbox{\onelinecomment  \ensuremath{\Varid{at}_\epsilon} on functions}{}\<[E]%
\\
\>[B]{}\mathrel{=}{}\<[BE]%
\>[5]{}\Varid{fmap}\;(\Varid{at}_\epsilon\;\Varid{f}\;{}\cdot)\;(\derivOp\;\Varid{g}\;\Varid{c}){}\<[56]%
\>[56]{}\mbox{\onelinecomment  \ensuremath{\Varid{fmap}} on functions}{}\<[E]%
\ColumnHook
\end{hscode}\resethooks
Second addend:
\begin{hscode}\SaveRestoreHook
\column{B}{@{}>{\hspre}c<{\hspost}@{}}%
\column{BE}{@{}l@{}}%
\column{5}{@{}>{\hspre}l<{\hspost}@{}}%
\column{69}{@{}>{\hspre}l<{\hspost}@{}}%
\column{E}{@{}>{\hspre}l<{\hspost}@{}}%
\>[5]{}\derivOp\;(\bigOp\sum{\Varid{c'},\Varid{u'},\Varid{v}}{1}\;((\Varid{c'}\mathbin{:}\Varid{u'}) \diamond \Varid{v}\mapsto\Varid{f}\;(\Varid{c'}\mathbin{:}\Varid{u'})\mathbin{*}\Varid{g}\;\Varid{v})){}\<[E]%
\\
\>[B]{}\mathrel{=}{}\<[BE]%
\>[5]{}\bigOp\sum{\Varid{c'},\Varid{u'},\Varid{v}}{1}\;\derivOp\;((\Varid{c'}\mathbin{:}\Varid{u'}) \diamond \Varid{v}\mapsto\Varid{f}\;(\Varid{c'}\mathbin{:}\Varid{u'})\mathbin{*}\Varid{g}\;\Varid{v}){}\<[69]%
\>[69]{}\mbox{\onelinecomment  additivity of \ensuremath{\derivOp}}{}\<[E]%
\\
\>[B]{}\mathrel{=}{}\<[BE]%
\>[5]{}\bigOp\sum{\Varid{c'},\Varid{u'},\Varid{v}}{1}\;\derivOp\;(\Varid{c'}\mathbin{:}(\Varid{u'} \diamond \Varid{v})\mapsto\Varid{f}\;(\Varid{c'}\mathbin{:}\Varid{u'})\mathbin{*}\Varid{g}\;\Varid{v}){}\<[69]%
\>[69]{}\mbox{\onelinecomment  \ensuremath{( \diamond )} on lists}{}\<[E]%
\\
\>[B]{}\mathrel{=}{}\<[BE]%
\>[5]{}\lambda\, \Varid{c}\to \bigOp\sum{\Varid{u'},\Varid{v}}{0}\;\Varid{u'} \diamond \Varid{v}\mapsto\Varid{f}\;(\Varid{c}\mathbin{:}\Varid{u'})\mathbin{*}\Varid{g}\;\Varid{v}{}\<[69]%
\>[69]{}\mbox{\onelinecomment  \ensuremath{\derivOp} on \ensuremath{(\mapsto)} below}{}\<[E]%
\\
\>[B]{}\mathrel{=}{}\<[BE]%
\>[5]{}\lambda\, \Varid{c}\to \bigOp\sum{\Varid{u'},\Varid{v}}{0}\;\Varid{u'} \diamond \Varid{v}\mapsto(\lambda\, \Varid{cs}\to \Varid{f}\;(\Varid{c}\mathbin{:}\Varid{cs}))\;\Varid{u'}\mathbin{*}\Varid{g}\;\Varid{v}{}\<[69]%
\>[69]{}\mbox{\onelinecomment  $\beta$ expansion}{}\<[E]%
\\
\>[B]{}\mathrel{=}{}\<[BE]%
\>[5]{}\lambda\, \Varid{c}\to \lambda\, \Varid{cs}\to \Varid{f}\;(\Varid{c}\mathbin{:}\Varid{cs})\mathbin{*}\Varid{g}{}\<[69]%
\>[69]{}\mbox{\onelinecomment  \ensuremath{(\mathbin{*})} on functions}{}\<[E]%
\\
\>[B]{}\mathrel{=}{}\<[BE]%
\>[5]{}\lambda\, \Varid{c}\to \derivOp\;\Varid{f}\;\Varid{c}\mathbin{*}\Varid{g}{}\<[69]%
\>[69]{}\mbox{\onelinecomment  \ensuremath{\derivOp} on functions}{}\<[E]%
\\
\>[B]{}\mathrel{=}{}\<[BE]%
\>[5]{}\Varid{fmap}\;(\mathbin{*}{}\;\Varid{g})\;(\derivOp\;\Varid{f}){}\<[69]%
\>[69]{}\mbox{\onelinecomment  \ensuremath{\Varid{fmap}} on functions}{}\<[E]%
\ColumnHook
\end{hscode}\resethooks
Combining addends,
\begin{hscode}\SaveRestoreHook
\column{B}{@{}>{\hspre}l<{\hspost}@{}}%
\column{E}{@{}>{\hspre}l<{\hspost}@{}}%
\>[B]{}\derivOp\;(\Varid{f}\mathbin{*}\Varid{g})\mathrel{=}\Varid{fmap}\;(\Varid{at}_\epsilon\;\Varid{f}\;{})\;(\derivOp\;\Varid{g})\mathbin{+}\Varid{fmap}\;(\mathbin{*}{}\;\Varid{g})\;(\derivOp\;\Varid{f}){}\<[E]%
\ColumnHook
\end{hscode}\resethooks
\noindent
Continuing with the other equations in \lemref{deriv [c] -> b},
\begin{hscode}\SaveRestoreHook
\column{B}{@{}>{\hspre}c<{\hspost}@{}}%
\column{BE}{@{}l@{}}%
\column{5}{@{}>{\hspre}l<{\hspost}@{}}%
\column{66}{@{}>{\hspre}l<{\hspost}@{}}%
\column{E}{@{}>{\hspre}l<{\hspost}@{}}%
\>[5]{}\derivOp\;(\closure{\Varid{p}}){}\<[E]%
\\
\>[B]{}\mathrel{=}{}\<[BE]%
\>[5]{}\derivOp\;(\mathrm{1}\mathbin{+}\Varid{p}\mathbin{*}\closure{\Varid{p}}){}\<[66]%
\>[66]{}\mbox{\onelinecomment  star semiring law}{}\<[E]%
\\
\>[B]{}\mathrel{=}{}\<[BE]%
\>[5]{}\derivOp\;\mathrm{1}\mathbin{+}\derivOp\;(\Varid{p}\mathbin{*}\closure{\Varid{p}}){}\<[66]%
\>[66]{}\mbox{\onelinecomment  additivity of \ensuremath{\derivOp} (above)}{}\<[E]%
\\
\>[B]{}\mathrel{=}{}\<[BE]%
\>[5]{}\derivOp\;(\Varid{p}\mathbin{*}\closure{\Varid{p}}){}\<[66]%
\>[66]{}\mbox{\onelinecomment  \ensuremath{\derivOp\;\mathrm{1}\mathrel{=}\mathrm{0}} (above)}{}\<[E]%
\\
\>[B]{}\mathrel{=}{}\<[BE]%
\>[5]{}\lambda\, \Varid{c}\to \Varid{at}_\epsilon\;\Varid{p}\cdot\derivOp\;(\closure{\Varid{p}})\;\Varid{c}\mathbin{+}\derivOp\;\Varid{p}\;\Varid{c}\mathbin{*}\closure{\Varid{p}}{}\<[66]%
\>[66]{}\mbox{\onelinecomment  \ensuremath{\derivOp\;(\Varid{p}\mathbin{*}\Varid{q})} above}{}\<[E]%
\\
\>[B]{}\mathrel{=}{}\<[BE]%
\>[5]{}\lambda\, \Varid{c}\to \closure{(\Varid{at}_\epsilon\;\Varid{p})}\cdot\derivOp\;\Varid{p}\;\Varid{c}\mathbin{*}\closure{\Varid{p}}{}\<[66]%
\>[66]{}\mbox{\onelinecomment  \lemref{affine over semimodule}}{}\<[E]%
\\
\>[B]{}\mathrel{=}{}\<[BE]%
\>[5]{}\Varid{fmap}\;(\lambda\, \Varid{d}\to \closure{(\Varid{at}_\epsilon\;\Varid{p})}\cdot\Varid{d}\mathbin{*}\Conid{Star}\;\Varid{p})\;(\derivOp\;\Varid{p}){}\<[66]%
\>[66]{}\mbox{\onelinecomment  \ensuremath{\Varid{fmap}} on functions}{}\<[E]%
\ColumnHook
\end{hscode}\resethooks

\begin{hscode}\SaveRestoreHook
\column{B}{@{}>{\hspre}c<{\hspost}@{}}%
\column{BE}{@{}l@{}}%
\column{5}{@{}>{\hspre}l<{\hspost}@{}}%
\column{45}{@{}>{\hspre}l<{\hspost}@{}}%
\column{E}{@{}>{\hspre}l<{\hspost}@{}}%
\>[5]{}\derivOp\;(\Varid{s}\cdot\Varid{f}){}\<[E]%
\\
\>[B]{}\mathrel{=}{}\<[BE]%
\>[5]{}\derivOp\;(\lambda\, \Varid{w}\to \Varid{s}\mathbin{*}\Varid{f}\;\Varid{w}){}\<[45]%
\>[45]{}\mbox{\onelinecomment  \ensuremath{(\cdot)} on functions}{}\<[E]%
\\
\>[B]{}\mathrel{=}{}\<[BE]%
\>[5]{}\lambda\, \Varid{c}\to \lambda\, \Varid{cs}\to (\lambda\, \Varid{w}\to \Varid{s}\mathbin{*}\Varid{f}\;\Varid{w})\;(\Varid{c}\mathbin{:}\Varid{cs}){}\<[45]%
\>[45]{}\mbox{\onelinecomment  \ensuremath{\derivOp} definition}{}\<[E]%
\\
\>[B]{}\mathrel{=}{}\<[BE]%
\>[5]{}\lambda\, \Varid{c}\to \lambda\, \Varid{cs}\to \Varid{s}\mathbin{*}\Varid{f}\;(\Varid{c}\mathbin{:}\Varid{cs}){}\<[45]%
\>[45]{}\mbox{\onelinecomment  $\beta$ reduction}{}\<[E]%
\\
\>[B]{}\mathrel{=}{}\<[BE]%
\>[5]{}\lambda\, \Varid{c}\to \Varid{s}\cdot(\lambda\, \Varid{cs}\to \Varid{f}\;(\Varid{c}\mathbin{:}\Varid{cs})){}\<[45]%
\>[45]{}\mbox{\onelinecomment  \ensuremath{(\cdot)} on functions}{}\<[E]%
\\
\>[B]{}\mathrel{=}{}\<[BE]%
\>[5]{}\lambda\, \Varid{c}\to \Varid{s}\cdot\derivOp\;\Varid{f}\;\Varid{c}{}\<[45]%
\>[45]{}\mbox{\onelinecomment  \ensuremath{\derivOp} definition}{}\<[E]%
\\
\>[B]{}\mathrel{=}{}\<[BE]%
\>[5]{}\Varid{fmap}\;(\Varid{s}\;{}\cdot)\;(\derivOp\;\Varid{f}){}\<[45]%
\>[45]{}\mbox{\onelinecomment  \ensuremath{\Varid{fmap}} on functions}{}\<[E]%
\ColumnHook
\end{hscode}\resethooks

\begin{hscode}\SaveRestoreHook
\column{B}{@{}>{\hspre}c<{\hspost}@{}}%
\column{BE}{@{}l@{}}%
\column{5}{@{}>{\hspre}l<{\hspost}@{}}%
\column{58}{@{}>{\hspre}l<{\hspost}@{}}%
\column{79}{@{}>{\hspre}l<{\hspost}@{}}%
\column{E}{@{}>{\hspre}l<{\hspost}@{}}%
\>[5]{}\derivOp\;([\mskip1.5mu \mskip1.5mu]\mapsto\Varid{b})\;\Varid{c}{}\<[E]%
\\
\>[B]{}\mathrel{=}{}\<[BE]%
\>[5]{}\derivOp\;(\lambda\, \Varid{w}\to \mathbf{if}\;\Varid{w}\mathrel{=}[\mskip1.5mu \mskip1.5mu]\;\mathbf{then}\;\Varid{b}\;\mathbf{else}\;\mathrm{0}){}\<[58]%
\>[58]{}\mbox{\onelinecomment  \ensuremath{(\mapsto)} on functions}{}\<[E]%
\\
\>[B]{}\mathrel{=}{}\<[BE]%
\>[5]{}\lambda\, \Varid{cs}\to (\lambda\, \Varid{w}\to \mathbf{if}\;\Varid{w}\mathrel{=}[\mskip1.5mu \mskip1.5mu]\;\mathbf{then}\;\Varid{b}\;\mathbf{else}\;\mathrm{0})\;(\Varid{c}\mathbin{:}\Varid{cs}){}\<[58]%
\>[58]{}\mbox{\onelinecomment  \ensuremath{\derivOp} definition}{}\<[E]%
\\
\>[B]{}\mathrel{=}{}\<[BE]%
\>[5]{}\lambda\, \Varid{cs}\to \mathbf{if}\;\Varid{c}\mathbin{:}\Varid{cs}\mathrel{=}[\mskip1.5mu \mskip1.5mu]\;\mathbf{then}\;\Varid{b}\;\mathbf{else}\;\mathrm{0}{}\<[58]%
\>[58]{}\mbox{\onelinecomment  $\beta$ reduction}{}\<[E]%
\\
\>[B]{}\mathrel{=}{}\<[BE]%
\>[5]{}\lambda\, \Varid{cs}\to \mathrm{0}{}\<[58]%
\>[58]{}\mbox{\onelinecomment  \ensuremath{\Varid{c}\mathbin{:}\Varid{cs}\not=[\mskip1.5mu \mskip1.5mu]}}{}\<[E]%
\\
\>[B]{}\mathrel{=}{}\<[BE]%
\>[5]{}\mathrm{0}{}\<[58]%
\>[58]{}\mbox{\onelinecomment  \ensuremath{\mathrm{0}} on functions}{}\<[E]%
\\[1.5ex]\>[5]{}\derivOp\;(\Varid{c'}\mathbin{:}\Varid{cs'}\mapsto\Varid{b}){}\<[E]%
\\
\>[B]{}\mathrel{=}{}\<[BE]%
\>[5]{}\derivOp\;(\lambda\, \Varid{w}\to \mathbf{if}\;\Varid{w}\mathrel{=}\Varid{c'}\mathbin{:}\Varid{cs'}\;\mathbf{then}\;\Varid{b}\;\mathbf{else}\;\mathrm{0}){}\<[79]%
\>[79]{}\mbox{\onelinecomment  \ensuremath{(\mapsto)} on functions}{}\<[E]%
\\
\>[B]{}\mathrel{=}{}\<[BE]%
\>[5]{}\lambda\, \Varid{c}\to \lambda\, \Varid{cs}\to (\lambda\, \Varid{w}\to \mathbf{if}\;\Varid{w}\mathrel{=}\Varid{c'}\mathbin{:}\Varid{cs'}\;\mathbf{then}\;\Varid{b}\;\mathbf{else}\;\mathrm{0})\;(\Varid{c}\mathbin{:}\Varid{cs}){}\<[79]%
\>[79]{}\mbox{\onelinecomment  \ensuremath{(\mapsto)} on functions}{}\<[E]%
\\
\>[B]{}\mathrel{=}{}\<[BE]%
\>[5]{}\lambda\, \Varid{c}\to \lambda\, \Varid{cs}\to \mathbf{if}\;\Varid{c}\mathbin{:}\Varid{cs}\mathrel{=}\Varid{c'}\mathbin{:}\Varid{cs'}\;\mathbf{then}\;\Varid{b}\;\mathbf{else}\;\mathrm{0}{}\<[79]%
\>[79]{}\mbox{\onelinecomment  $\beta$ reduction}{}\<[E]%
\\
\>[B]{}\mathrel{=}{}\<[BE]%
\>[5]{}\lambda\, \Varid{c}\to \lambda\, \Varid{cs}\to \mathbf{if}\;\Varid{c}\mathrel{=}\Varid{c'}\mathrel{\wedge}\Varid{cs}\mathrel{=}\Varid{cs'}\;\mathbf{then}\;\Varid{b}\;\mathbf{else}\;\mathrm{0}{}\<[79]%
\>[79]{}\mbox{\onelinecomment  \ensuremath{(\mathbin{:})} injectivity}{}\<[E]%
\\
\>[B]{}\mathrel{=}{}\<[BE]%
\>[5]{}\lambda\, \Varid{c}\to \lambda\, \Varid{cs}\to \mathbf{if}\;\Varid{c}\mathrel{=}\Varid{c'}\;\mathbf{then}\;(\mathbf{if}\;\Varid{cs}\mathrel{=}\Varid{cs'}\;\mathbf{then}\;\Varid{b}\;\mathbf{else}\;\mathrm{0})\;\mathbf{else}\;\mathrm{0}{}\<[79]%
\>[79]{}\mbox{\onelinecomment  property of \ensuremath{\mathbf{if}} and \ensuremath{(\mathrel{\wedge})}}{}\<[E]%
\\
\>[B]{}\mathrel{=}{}\<[BE]%
\>[5]{}\lambda\, \Varid{c}\to \mathbf{if}\;\Varid{c}\mathrel{=}\Varid{c'}\;\mathbf{then}\;(\lambda\, \Varid{cs}\to \mathbf{if}\;\Varid{cs}\mathrel{=}\Varid{cs'}\;\mathbf{then}\;\Varid{b}\;\mathbf{else}\;\mathrm{0}\;\mathbf{else}\;\mathrm{0}){}\<[79]%
\>[79]{}\mbox{\onelinecomment  property of \ensuremath{\mathbf{if}}}{}\<[E]%
\\
\>[B]{}\mathrel{=}{}\<[BE]%
\>[5]{}\lambda\, \Varid{c}\to \mathbf{if}\;\Varid{c}\mathrel{=}\Varid{c'}\;\mathbf{then}\;\Varid{cs'}\mapsto\Varid{b}\;\mathbf{else}\;\mathrm{0}{}\<[79]%
\>[79]{}\mbox{\onelinecomment  \ensuremath{(\mapsto)} on functions}{}\<[E]%
\\
\>[B]{}\mathrel{=}{}\<[BE]%
\>[5]{}\Varid{c'}\mapsto\Varid{cs'}\mapsto\Varid{b}{}\<[79]%
\>[79]{}\mbox{\onelinecomment  \ensuremath{(\mapsto)} on \ensuremath{\Varid{s}\to \Varid{t}}}{}\<[E]%
\ColumnHook
\end{hscode}\resethooks

\subsection{\thmref{semiring decomp [c] -> b}}\prooflabel{theorem:semiring decomp [c] -> b}

\begin{hscode}\SaveRestoreHook
\column{B}{@{}>{\hspre}c<{\hspost}@{}}%
\column{BE}{@{}l@{}}%
\column{5}{@{}>{\hspre}l<{\hspost}@{}}%
\column{31}{@{}>{\hspre}l<{\hspost}@{}}%
\column{E}{@{}>{\hspre}l<{\hspost}@{}}%
\>[5]{}\mathrm{0}{}\<[E]%
\\
\>[B]{}\mathrel{=}{}\<[BE]%
\>[5]{}\Varid{at}_\epsilon\;\mathrm{0}\mathrel\triangleleft\derivOp\;\mathrm{0}{}\<[31]%
\>[31]{}\mbox{\onelinecomment  \lemref{decomp ([c] -> b)}}{}\<[E]%
\\
\>[B]{}\mathrel{=}{}\<[BE]%
\>[5]{}\mathrm{0}\mathrel\triangleleft\lambda\, \Varid{c}\to \mathrm{0}{}\<[31]%
\>[31]{}\mbox{\onelinecomment  \lemreftwo{atEps [c] -> b}{deriv [c] -> b}}{}\<[E]%
\\
\>[B]{}\mathrel{=}{}\<[BE]%
\>[5]{}\mathrm{0}\mathrel\triangleleft\mathrm{0}{}\<[31]%
\>[31]{}\mbox{\onelinecomment  \ensuremath{\mathrm{0}} on functions}{}\<[E]%
\ColumnHook
\end{hscode}\resethooks

\begin{hscode}\SaveRestoreHook
\column{B}{@{}>{\hspre}c<{\hspost}@{}}%
\column{BE}{@{}l@{}}%
\column{5}{@{}>{\hspre}l<{\hspost}@{}}%
\column{29}{@{}>{\hspre}l<{\hspost}@{}}%
\column{E}{@{}>{\hspre}l<{\hspost}@{}}%
\>[5]{}\mathrm{1}{}\<[E]%
\\
\>[B]{}\mathrel{=}{}\<[BE]%
\>[5]{}\Varid{at}_\epsilon\;\mathrm{1}\mathrel\triangleleft\derivOp\;\mathrm{1}{}\<[29]%
\>[29]{}\mbox{\onelinecomment  \lemref{decomp ([c] -> b)}}{}\<[E]%
\\
\>[B]{}\mathrel{=}{}\<[BE]%
\>[5]{}\mathrm{1}\mathrel\triangleleft\lambda\, \Varid{c}\to \mathrm{0}{}\<[29]%
\>[29]{}\mbox{\onelinecomment  \lemreftwo{atEps [c] -> b}{deriv [c] -> b}}{}\<[E]%
\\
\>[B]{}\mathrel{=}{}\<[BE]%
\>[5]{}\mathrm{1}\mathrel\triangleleft\mathrm{0}{}\<[29]%
\>[29]{}\mbox{\onelinecomment  \ensuremath{\mathrm{0}} on functions}{}\<[E]%
\ColumnHook
\end{hscode}\resethooks

\begin{hscode}\SaveRestoreHook
\column{B}{@{}>{\hspre}c<{\hspost}@{}}%
\column{BE}{@{}l@{}}%
\column{5}{@{}>{\hspre}l<{\hspost}@{}}%
\column{73}{@{}>{\hspre}l<{\hspost}@{}}%
\column{E}{@{}>{\hspre}l<{\hspost}@{}}%
\>[5]{}(\Varid{a}\mathrel\triangleleft\Varid{dp})\mathbin{+}(\Varid{b}\mathrel\triangleleft\Varid{dp}){}\<[E]%
\\
\>[B]{}\mathrel{=}{}\<[BE]%
\>[5]{}\Varid{at}_\epsilon\;((\Varid{a}\mathrel\triangleleft\Varid{dp})\mathbin{+}(\Varid{b}\mathrel\triangleleft\Varid{dq}))\mathrel\triangleleft\derivOp\;((\Varid{a}\mathrel\triangleleft\Varid{dp})\mathbin{+}(\Varid{b}\mathrel\triangleleft\Varid{dq})){}\<[73]%
\>[73]{}\mbox{\onelinecomment  \lemref{decomp ([c] -> b)}}{}\<[E]%
\\
\>[B]{}\mathrel{=}{}\<[BE]%
\>[5]{}\Varid{a}\mathbin{+}\Varid{b}\mathrel\triangleleft\Varid{dp}\mathbin{+}\Varid{dq}{}\<[73]%
\>[73]{}\mbox{\onelinecomment  \lemref{atEps and deriv via (<:)} below}{}\<[E]%
\ColumnHook
\end{hscode}\resethooks


\begin{hscode}\SaveRestoreHook
\column{B}{@{}>{\hspre}c<{\hspost}@{}}%
\column{BE}{@{}l@{}}%
\column{5}{@{}>{\hspre}l<{\hspost}@{}}%
\column{79}{@{}>{\hspre}l<{\hspost}@{}}%
\column{E}{@{}>{\hspre}l<{\hspost}@{}}%
\>[5]{}(\Varid{a}\mathrel\triangleleft\Varid{dp})\mathbin{*}(\Varid{b}\mathrel\triangleleft\Varid{dq}){}\<[E]%
\\
\>[B]{}\mathrel{=}{}\<[BE]%
\>[5]{}\Varid{at}_\epsilon\;((\Varid{a}\mathrel\triangleleft\Varid{dp})\mathbin{*}(\Varid{b}\mathrel\triangleleft\Varid{dq}))\mathrel\triangleleft\derivOp\;((\Varid{a}\mathrel\triangleleft\Varid{dp})\mathbin{*}(\Varid{b}\mathrel\triangleleft\Varid{dq})){}\<[79]%
\>[79]{}\mbox{\onelinecomment  \lemref{decomp ([c] -> b)}}{}\<[E]%
\\
\>[B]{}\mathrel{=}{}\<[BE]%
\>[5]{}\Varid{a}\mathbin{*}\Varid{b}\mathrel\triangleleft\lambda\, \Varid{c}\to \Varid{a}\cdot\Varid{dq}\;\Varid{c}\mathbin{+}\Varid{dp}\;\Varid{c}\mathbin{*}(\Varid{b}\mathrel\triangleleft\Varid{dq}){}\<[79]%
\>[79]{}\mbox{\onelinecomment  \lemref{atEps and deriv via (<:)} below}{}\<[E]%
\\
\>[B]{}\mathrel{=}{}\<[BE]%
\>[5]{}(\Varid{a}\mathbin{*}\Varid{b}\mathbin{+}\mathrm{0})\mathrel\triangleleft(\lambda\, \Varid{c}\to \Varid{a}\cdot\Varid{dq}\;\Varid{c})\mathbin{+}(\lambda\, \Varid{c}\to \Varid{dp}\;\Varid{c}\mathbin{*}(\Varid{b}\mathrel\triangleleft\Varid{dq})){}\<[79]%
\>[79]{}\mbox{\onelinecomment  additive identity; \ensuremath{(\mathbin{+})} on functions}{}\<[E]%
\\
\>[B]{}\mathrel{=}{}\<[BE]%
\>[5]{}(\Varid{a}\mathbin{*}\Varid{b}\mathrel\triangleleft\lambda\, \Varid{c}\to \Varid{a}\cdot\Varid{dq}\;\Varid{c})\mathbin{+}(\mathrm{0}\mathrel\triangleleft\lambda\, \Varid{c}\to \Varid{dp}\;\Varid{c}\mathbin{*}(\Varid{b}\mathrel\triangleleft\Varid{dq})){}\<[79]%
\>[79]{}\mbox{\onelinecomment  previous result}{}\<[E]%
\\
\>[B]{}\mathrel{=}{}\<[BE]%
\>[5]{}\Varid{a}\cdot(\Varid{b}\mathrel\triangleleft\Varid{dq})\mathbin{+}(\mathrm{0}\mathrel\triangleleft\lambda\, \Varid{c}\to \Varid{dp}\;\Varid{c}\mathbin{*}(\Varid{b}\mathrel\triangleleft\Varid{dq})){}\<[79]%
\>[79]{}\mbox{\onelinecomment  \ensuremath{(\cdot)} case below}{}\<[E]%
\\
\>[B]{}\mathrel{=}{}\<[BE]%
\>[5]{}\Varid{a}\cdot(\Varid{b}\mathrel\triangleleft\Varid{dq})\mathbin{+}(\mathrm{0}\mathrel\triangleleft(\mathbin{*}{}\;(\Varid{b}\mathrel\triangleleft\Varid{dq}))\hsdot{\circ }{.\:}\Varid{dp}){}\<[79]%
\>[79]{}\mbox{\onelinecomment  \ensuremath{(\hsdot{\circ }{.\:})} definition}{}\<[E]%
\\
\>[B]{}\mathrel{=}{}\<[BE]%
\>[5]{}\Varid{a}\cdot(\Varid{b}\mathrel\triangleleft\Varid{dq})\mathbin{+}(\mathrm{0}\mathrel\triangleleft\Varid{fmap}\;(\mathbin{*}{}\;(\Varid{b}\mathrel\triangleleft\Varid{dq}))\;\Varid{dp}){}\<[79]%
\>[79]{}\mbox{\onelinecomment  fmap on functions}{}\<[E]%
\ColumnHook
\end{hscode}\resethooks





\begin{hscode}\SaveRestoreHook
\column{B}{@{}>{\hspre}c<{\hspost}@{}}%
\column{BE}{@{}l@{}}%
\column{5}{@{}>{\hspre}l<{\hspost}@{}}%
\column{59}{@{}>{\hspre}l<{\hspost}@{}}%
\column{E}{@{}>{\hspre}l<{\hspost}@{}}%
\>[5]{}\closure{(\Varid{a}\mathrel\triangleleft\Varid{dp})}{}\<[E]%
\\
\>[B]{}\mathrel{=}{}\<[BE]%
\>[5]{}\Varid{at}_\epsilon\;(\closure{(\Varid{a}\mathrel\triangleleft\Varid{dp})})\mathrel\triangleleft\derivOp\;(\closure{(\Varid{a}\mathrel\triangleleft\Varid{dp})}){}\<[59]%
\>[59]{}\mbox{\onelinecomment  \lemref{decomp ([c] -> b)}}{}\<[E]%
\\
\>[B]{}\mathrel{=}{}\<[BE]%
\>[5]{}\closure{\Varid{a}}\mathrel\triangleleft\lambda\, \Varid{c}\to \closure{\Varid{a}}\cdot\Varid{dp}\;\Varid{c}\mathbin{*}\closure{(\Varid{a}\mathrel\triangleleft\Varid{dp})}{}\<[59]%
\>[59]{}\mbox{\onelinecomment  \lemref{atEps and deriv via (<:)} below}{}\<[E]%
\\
\>[B]{}\mathrel{=}{}\<[BE]%
\>[5]{}\closure{\Varid{a}}\cdot(\mathrm{1}\mathrel\triangleleft\lambda\, \Varid{c}\to \Varid{dp}\;\Varid{c}\mathbin{*}\closure{(\Varid{a}\mathrel\triangleleft\Varid{dp})}){}\<[59]%
\>[59]{}\mbox{\onelinecomment  \ensuremath{(\cdot)} case below}{}\<[E]%
\\
\>[B]{}\mathrel{=}{}\<[BE]%
\>[5]{}\closure{\Varid{a}}\cdot(\mathrm{1}\mathrel\triangleleft\Varid{fmap}\;(\mathbin{*}{}{\,}\closure{(\Varid{a}\mathrel\triangleleft\Varid{dp})})\;\Varid{dp}){}\<[59]%
\>[59]{}\mbox{\onelinecomment  \ensuremath{\Varid{fmap}} on functions}{}\<[E]%
\ColumnHook
\end{hscode}\resethooks

\begin{hscode}\SaveRestoreHook
\column{B}{@{}>{\hspre}c<{\hspost}@{}}%
\column{BE}{@{}l@{}}%
\column{5}{@{}>{\hspre}l<{\hspost}@{}}%
\column{53}{@{}>{\hspre}l<{\hspost}@{}}%
\column{E}{@{}>{\hspre}l<{\hspost}@{}}%
\>[5]{}\Varid{s}\cdot(\Varid{b}\mathrel\triangleleft\Varid{h}){}\<[E]%
\\
\>[B]{}\mathrel{=}{}\<[BE]%
\>[5]{}\Varid{at}_\epsilon\;(\Varid{s}\cdot(\Varid{b}\mathrel\triangleleft\Varid{h}))\mathrel\triangleleft\derivOp\;(\Varid{s}\cdot(\Varid{b}\mathrel\triangleleft\Varid{h})){}\<[53]%
\>[53]{}\mbox{\onelinecomment  \lemref{decomp ([c] -> b)}}{}\<[E]%
\\
\>[B]{}\mathrel{=}{}\<[BE]%
\>[5]{}\Varid{s}\mathbin{*}\Varid{b}\mathrel\triangleleft\lambda\, \Varid{c}\to \Varid{s}\cdot\Varid{dp}\;\Varid{c}{}\<[53]%
\>[53]{}\mbox{\onelinecomment  \lemref{atEps and deriv via (<:)} below}{}\<[E]%
\\
\>[B]{}\mathrel{=}{}\<[BE]%
\>[5]{}\Varid{s}\mathbin{*}\Varid{b}\mathrel\triangleleft(\Varid{s}\cdot)\hsdot{\circ }{.\:}\Varid{dp}{}\<[53]%
\>[53]{}\mbox{\onelinecomment  \ensuremath{(\hsdot{\circ }{.\:})} definition}{}\<[E]%
\\
\>[B]{}\mathrel{=}{}\<[BE]%
\>[5]{}\Varid{s}\mathbin{*}\Varid{b}\mathrel\triangleleft\Varid{fmap}\;(\Varid{s}\;{}\cdot)\;\Varid{dp}{}\<[53]%
\>[53]{}\mbox{\onelinecomment  \ensuremath{\Varid{fmap}} on functions}{}\<[E]%
\ColumnHook
\end{hscode}\resethooks

\begin{hscode}\SaveRestoreHook
\column{B}{@{}>{\hspre}c<{\hspost}@{}}%
\column{BE}{@{}l@{}}%
\column{5}{@{}>{\hspre}l<{\hspost}@{}}%
\column{56}{@{}>{\hspre}l<{\hspost}@{}}%
\column{E}{@{}>{\hspre}l<{\hspost}@{}}%
\>[5]{}[\mskip1.5mu \mskip1.5mu]\mapsto\Varid{b}{}\<[E]%
\\
\>[B]{}\mathrel{=}{}\<[BE]%
\>[5]{}\Varid{at}_\epsilon\;([\mskip1.5mu \mskip1.5mu]\mapsto\Varid{b})\mathrel\triangleleft\derivOp\;([\mskip1.5mu \mskip1.5mu]\mapsto\Varid{b}){}\<[56]%
\>[56]{}\mbox{\onelinecomment  \lemref{decomp ([c] -> b)}}{}\<[E]%
\\
\>[B]{}\mathrel{=}{}\<[BE]%
\>[5]{}\Varid{b}\mathrel\triangleleft\lambda\, \Varid{c}\to \mathrm{0}{}\<[56]%
\>[56]{}\mbox{\onelinecomment  \lemreftwo{atEps [c] -> b}{deriv [c] -> b}}{}\<[E]%
\\
\>[B]{}\mathrel{=}{}\<[BE]%
\>[5]{}\Varid{b}\mathrel\triangleleft\mathrm{0}{}\<[56]%
\>[56]{}\mbox{\onelinecomment  \ensuremath{\mathrm{0}} on functions}{}\<[E]%
\\[1.5ex]\>[5]{}\Varid{c'}\mathbin{:}\Varid{cs'}\mapsto\Varid{b}{}\<[E]%
\\
\>[B]{}\mathrel{=}{}\<[BE]%
\>[5]{}\Varid{at}_\epsilon\;(\Varid{c'}\mathbin{:}\Varid{cs'}\mapsto\Varid{b})\mathrel\triangleleft\derivOp\;(\Varid{c'}\mathbin{:}\Varid{cs'}\mapsto\Varid{b}){}\<[56]%
\>[56]{}\mbox{\onelinecomment  \lemref{decomp ([c] -> b)}}{}\<[E]%
\\
\>[B]{}\mathrel{=}{}\<[BE]%
\>[5]{}\mathrm{0}\mathrel\triangleleft\lambda\, \Varid{c}\to \mathbf{if}\;\Varid{c}\mathrel{=}\Varid{c'}\;\mathbf{then}\;\Varid{cs'}\mapsto\Varid{b}\;\mathbf{else}\;\mathrm{0}{}\<[56]%
\>[56]{}\mbox{\onelinecomment  \lemreftwo{atEps [c] -> b}{deriv [c] -> b}}{}\<[E]%
\\
\>[B]{}\mathrel{=}{}\<[BE]%
\>[5]{}\mathrm{0}\mathrel\triangleleft\Varid{c'}\mapsto\Varid{cs'}\mapsto\Varid{b}{}\<[56]%
\>[56]{}\mbox{\onelinecomment  \ensuremath{(\mapsto)} on functions}{}\<[E]%
\ColumnHook
\end{hscode}\resethooks
Expressed via \ensuremath{\Varid{foldr}},
\begin{hscode}\SaveRestoreHook
\column{B}{@{}>{\hspre}l<{\hspost}@{}}%
\column{3}{@{}>{\hspre}l<{\hspost}@{}}%
\column{E}{@{}>{\hspre}l<{\hspost}@{}}%
\>[3]{}\Varid{w}\mapsto\Varid{b}\mathrel{=}\Varid{foldr}\;(\lambda\, \Varid{c}\;\Varid{t}\to \mathrm{0}\mathrel\triangleleft\Varid{c}\mapsto\Varid{t})\;(\Varid{b}\mathrel\triangleleft\mathrm{0})\;\Varid{w}{}\<[E]%
\ColumnHook
\end{hscode}\resethooks
\begin{lemma}\lemlabel{atEps and deriv via (<:)}
The \ensuremath{\Varid{at}_\epsilon} and \ensuremath{\derivOp} functions satisfy the following properties in terms of \ensuremath{(\mathrel\triangleleft)}-decompositions:
\begin{spacing}{1.2}
\begin{hscode}\SaveRestoreHook
\column{B}{@{}>{\hspre}l<{\hspost}@{}}%
\column{19}{@{}>{\hspre}c<{\hspost}@{}}%
\column{19E}{@{}l@{}}%
\column{24}{@{}>{\hspre}l<{\hspost}@{}}%
\column{E}{@{}>{\hspre}l<{\hspost}@{}}%
\>[B]{}\Varid{at}_\epsilon\;((\Varid{a}\mathrel\triangleleft\Varid{dp}){}\<[19]%
\>[19]{}\mathbin{+}{}\<[19E]%
\>[24]{}(\Varid{b}\mathrel\triangleleft\Varid{dq}))\mathrel{=}\Varid{a}\mathbin{+}\Varid{b}{}\<[E]%
\\
\>[B]{}\Varid{at}_\epsilon\;((\Varid{a}\mathrel\triangleleft\Varid{dp}){}\<[19]%
\>[19]{}\mathbin{*}{}\<[19E]%
\>[24]{}(\Varid{b}\mathrel\triangleleft\Varid{dq}))\mathrel{=}\Varid{a}\mathbin{*}\Varid{b}{}\<[E]%
\\
\>[B]{}\Varid{at}_\epsilon\;(\closure{(\Varid{a}\mathrel\triangleleft\Varid{dp})})\mathrel{=}\closure{\Varid{a}}{}\<[E]%
\\
\>[B]{}\Varid{at}_\epsilon\;(\Varid{s}\cdot(\Varid{a}\mathrel\triangleleft\Varid{dp}))\mathrel{=}\Varid{s}\mathbin{*}\Varid{a}{}\<[E]%
\ColumnHook
\end{hscode}\resethooks
\begin{hscode}\SaveRestoreHook
\column{B}{@{}>{\hspre}l<{\hspost}@{}}%
\column{19}{@{}>{\hspre}c<{\hspost}@{}}%
\column{19E}{@{}l@{}}%
\column{24}{@{}>{\hspre}l<{\hspost}@{}}%
\column{E}{@{}>{\hspre}l<{\hspost}@{}}%
\>[B]{}\derivOp\;((\Varid{a}\mathrel\triangleleft\Varid{dp}){}\<[19]%
\>[19]{}\mathbin{+}{}\<[19E]%
\>[24]{}(\Varid{b}\mathrel\triangleleft\Varid{dq}))\;\Varid{c}\mathrel{=}\Varid{dp}\;\Varid{c}\mathbin{+}\Varid{dq}\;\Varid{c}{}\<[E]%
\\
\>[B]{}\derivOp\;((\Varid{a}\mathrel\triangleleft\Varid{dp}){}\<[19]%
\>[19]{}\mathbin{*}{}\<[19E]%
\>[24]{}(\Varid{b}\mathrel\triangleleft\Varid{dq}))\;\Varid{c}\mathrel{=}\Varid{a}\cdot\Varid{dq}\;\Varid{c}\mathbin{+}\Varid{dp}\;\Varid{c}\mathbin{*}(\Varid{b}\mathrel\triangleleft\Varid{dq}){}\<[E]%
\\
\>[B]{}\derivOp\;(\closure{(\Varid{a}\mathrel\triangleleft\Varid{dp})})\;\Varid{c}\mathrel{=}\closure{\Varid{a}}\cdot\Varid{dp}\;\Varid{c}\mathbin{*}\closure{(\Varid{a}\mathrel\triangleleft\Varid{dp})}{}\<[E]%
\\
\>[B]{}\derivOp\;(\Varid{s}\cdot(\Varid{a}\mathrel\triangleleft\Varid{dp}))\;\Varid{c}\mathrel{=}\Varid{s}\cdot\Varid{dp}\;\Varid{c}{}\<[E]%
\ColumnHook
\end{hscode}\resethooks
\end{spacing}
\vspace{-2ex}
\end{lemma}
\begin{proof}
Substitute into \lemreftwo{atEps [c] -> b}{deriv [c] -> b}, and simplify, using \lemref{decomp ([c] -> b)}.
\end{proof}


\subsection{\thmref{Cofree}}\prooflabel{theorem:Cofree}

The theorem follows from \thmref{semiring decomp [c] -> b}.
A few details:

\begin{spacing}{1.2}

\begin{hscode}\SaveRestoreHook
\column{B}{@{}>{\hspre}c<{\hspost}@{}}%
\column{BE}{@{}l@{}}%
\column{5}{@{}>{\hspre}l<{\hspost}@{}}%
\column{29}{@{}>{\hspre}l<{\hspost}@{}}%
\column{E}{@{}>{\hspre}l<{\hspost}@{}}%
\>[5]{}(\mathbin{!})\;\mathrm{0}{}\<[E]%
\\
\>[B]{}\mathrel{=}{}\<[BE]%
\>[5]{}(\mathbin{!})\;(\mathrm{0}\mathrel{\Varid{:\!\!\triangleleft}}\mathrm{0}){}\<[29]%
\>[29]{}\mbox{\onelinecomment  \ensuremath{\mathrm{0}} for \ensuremath{\Conid{Cofree}\;\Varid{c}\;\Varid{b}}}{}\<[E]%
\\
\>[B]{}\mathrel{=}{}\<[BE]%
\>[5]{}\mathrm{0}\mathrel\triangleleft(\mathbin{!})\hsdot{\circ }{.\:}(\mathbin{!})\;\mathrm{0}{}\<[29]%
\>[29]{}\mbox{\onelinecomment  \ensuremath{(\mathbin{!})} for \ensuremath{\Conid{Cofree}\;\Varid{c}\;\Varid{b}}}{}\<[E]%
\\
\>[B]{}\mathrel{=}{}\<[BE]%
\>[5]{}\mathrm{0}\mathrel\triangleleft(\mathbin{!})\hsdot{\circ }{.\:}\mathrm{0}{}\<[29]%
\>[29]{}\mbox{\onelinecomment  Additive functor assumption}{}\<[E]%
\\
\>[B]{}\mathrel{=}{}\<[BE]%
\>[5]{}\mathrm{0}\mathrel\triangleleft\mathrm{0}{}\<[29]%
\>[29]{}\mbox{\onelinecomment  coinduction}{}\<[E]%
\\
\>[B]{}\mathrel{=}{}\<[BE]%
\>[5]{}\mathrm{0}{}\<[29]%
\>[29]{}\mbox{\onelinecomment  \thmref{semiring decomp [c] -> b}}{}\<[E]%
\ColumnHook
\end{hscode}\resethooks

\begin{hscode}\SaveRestoreHook
\column{B}{@{}>{\hspre}c<{\hspost}@{}}%
\column{BE}{@{}l@{}}%
\column{5}{@{}>{\hspre}l<{\hspost}@{}}%
\column{61}{@{}>{\hspre}l<{\hspost}@{}}%
\column{E}{@{}>{\hspre}l<{\hspost}@{}}%
\>[5]{}(\mathbin{!})\;((\Varid{a}\mathrel{\Varid{:\!\!\triangleleft}}\Varid{dp})\mathbin{+}(\Varid{b}\mathrel{\Varid{:\!\!\triangleleft}}\Varid{dq})){}\<[E]%
\\
\>[B]{}\mathrel{=}{}\<[BE]%
\>[5]{}(\mathbin{!})\;(\Varid{a}\mathbin{+}\Varid{b}\mathrel{\Varid{:\!\!\triangleleft}}\Varid{dp}\mathbin{+}\Varid{dq}){}\<[61]%
\>[61]{}\mbox{\onelinecomment  \ensuremath{(\mathbin{+})} on \ensuremath{\Conid{Cofree}\;\Varid{c}\;\Varid{b}}}{}\<[E]%
\\
\>[B]{}\mathrel{=}{}\<[BE]%
\>[5]{}\Varid{a}\mathbin{+}\Varid{b}\mathrel\triangleleft(\mathbin{!})\hsdot{\circ }{.\:}(\mathbin{!})\;(\Varid{dp}\mathbin{+}\Varid{dq}){}\<[61]%
\>[61]{}\mbox{\onelinecomment  \ensuremath{(\mathbin{!})} on \ensuremath{\Conid{Cofree}\;\Varid{c}\;\Varid{b}}}{}\<[E]%
\\
\>[B]{}\mathrel{=}{}\<[BE]%
\>[5]{}\Varid{a}\mathbin{+}\Varid{b}\mathrel\triangleleft(\mathbin{!})\hsdot{\circ }{.\:}((\mathbin{!})\;\Varid{dp}\mathbin{+}(\mathbin{!})\;\Varid{dq}){}\<[61]%
\>[61]{}\mbox{\onelinecomment  \ensuremath{\Conid{Indexable}} law}{}\<[E]%
\\
\>[B]{}\mathrel{=}{}\<[BE]%
\>[5]{}\Varid{a}\mathbin{+}\Varid{b}\mathrel\triangleleft(\mathbin{!})\hsdot{\circ }{.\:}(\lambda\, \Varid{cs}\to \Varid{dp}\mathbin{!}\Varid{cs}\mathbin{+}\Varid{dq}\mathbin{!}\Varid{cs}){}\<[61]%
\>[61]{}\mbox{\onelinecomment  \ensuremath{(\mathbin{+})} on functions}{}\<[E]%
\\
\>[B]{}\mathrel{=}{}\<[BE]%
\>[5]{}\Varid{a}\mathbin{+}\Varid{b}\mathrel\triangleleft\lambda\, \Varid{cs}\to (\mathbin{!})\;(\Varid{dp}\mathbin{!}\Varid{cs}\mathbin{+}\Varid{dq}\mathbin{!}\Varid{cs}){}\<[61]%
\>[61]{}\mbox{\onelinecomment  \ensuremath{(\hsdot{\circ }{.\:})} definition}{}\<[E]%
\\
\>[B]{}\mathrel{=}{}\<[BE]%
\>[5]{}\Varid{a}\mathbin{+}\Varid{b}\mathrel\triangleleft\lambda\, \Varid{cs}\to (\mathbin{!})\;(\Varid{dp}\mathbin{!}\Varid{cs})\mathbin{+}(\mathbin{!})\;(\Varid{dq}\mathbin{!}\Varid{cs}){}\<[61]%
\>[61]{}\mbox{\onelinecomment  \ensuremath{\Conid{Indexable}} law}{}\<[E]%
\\
\>[B]{}\mathrel{=}{}\<[BE]%
\>[5]{}\Varid{a}\mathbin{+}\Varid{b}\mathrel\triangleleft\lambda\, \Varid{cs}\to ((\mathbin{!})\hsdot{\circ }{.\:}(\mathbin{!})\;\Varid{dp})\;\Varid{cs}\mathbin{+}((\mathbin{!})\hsdot{\circ }{.\:}(\mathbin{!})\;\Varid{dq})\;\Varid{cs}{}\<[61]%
\>[61]{}\mbox{\onelinecomment  \ensuremath{(\hsdot{\circ }{.\:})} definition}{}\<[E]%
\\
\>[B]{}\mathrel{=}{}\<[BE]%
\>[5]{}\Varid{a}\mathbin{+}\Varid{b}\mathrel\triangleleft((\mathbin{!})\hsdot{\circ }{.\:}(\mathbin{!})\;\Varid{dp})\mathbin{+}((\mathbin{!})\hsdot{\circ }{.\:}(\mathbin{!})\;\Varid{dq}){}\<[61]%
\>[61]{}\mbox{\onelinecomment  \ensuremath{(\mathbin{+})} on functions}{}\<[E]%
\\
\>[B]{}\mathrel{=}{}\<[BE]%
\>[5]{}(\Varid{a}\mathrel{\Varid{:\!\!\triangleleft}}(\mathbin{!})\hsdot{\circ }{.\:}(\mathbin{!})\;\Varid{dp})\mathbin{+}(\Varid{b}\mathrel{\Varid{:\!\!\triangleleft}}(\mathbin{!})\hsdot{\circ }{.\:}(\mathbin{!})\;\Varid{dq}){}\<[61]%
\>[61]{}\mbox{\onelinecomment  \ensuremath{(\mathbin{+})} on \ensuremath{\Conid{Cofree}\;\Varid{c}\;\Varid{b}}}{}\<[E]%
\\
\>[B]{}\mathrel{=}{}\<[BE]%
\>[5]{}(\mathbin{!})\;(\Varid{a}\mathrel{\Varid{:\!\!\triangleleft}}\Varid{dp})\mathbin{+}(\mathbin{!})\;(\Varid{b}\mathrel{\Varid{:\!\!\triangleleft}}\Varid{dq}){}\<[61]%
\>[61]{}\mbox{\onelinecomment  \ensuremath{(\mathbin{!})} on \ensuremath{\Conid{Cofree}\;\Varid{c}\;\Varid{b}}}{}\<[E]%
\ColumnHook
\end{hscode}\resethooks


\end{spacing}

\subsection{\thmref{Cofree hom}}\prooflabel{theorem:Cofree hom}

First show that \ensuremath{(\mathbin{!})} is natural (a functor homomorphism):
\begin{hscode}\SaveRestoreHook
\column{B}{@{}>{\hspre}l<{\hspost}@{}}%
\column{3}{@{}>{\hspre}l<{\hspost}@{}}%
\column{E}{@{}>{\hspre}l<{\hspost}@{}}%
\>[3]{}(\mathbin{!})\;(\Varid{fmap}\;\Varid{f}\;(\Varid{a}\mathrel{\Varid{:\!\!\triangleleft}}\Varid{ds}))\mathrel{=}\Varid{fmap}\;\Varid{f}\;((\mathbin{!})\;(\Varid{a}\mathrel{\Varid{:\!\!\triangleleft}}\Varid{ds})){}\<[E]%
\ColumnHook
\end{hscode}\resethooks
i.e.,
\begin{hscode}\SaveRestoreHook
\column{B}{@{}>{\hspre}l<{\hspost}@{}}%
\column{3}{@{}>{\hspre}l<{\hspost}@{}}%
\column{E}{@{}>{\hspre}l<{\hspost}@{}}%
\>[3]{}\Varid{fmap}\;\Varid{f}\;(\Varid{a}\mathrel{\Varid{:\!\!\triangleleft}}\Varid{ds})\mathbin{!}\Varid{cs}\mathrel{=}\Varid{fmap}\;\Varid{f}\;((\mathbin{!})\;(\Varid{a}\mathrel{\Varid{:\!\!\triangleleft}}\Varid{ds}))\;\Varid{cs}{}\<[E]%
\ColumnHook
\end{hscode}\resethooks
Consider cases for \ensuremath{\Varid{cs}}:

\begin{hscode}\SaveRestoreHook
\column{B}{@{}>{\hspre}c<{\hspost}@{}}%
\column{BE}{@{}l@{}}%
\column{5}{@{}>{\hspre}l<{\hspost}@{}}%
\column{37}{@{}>{\hspre}l<{\hspost}@{}}%
\column{40}{@{}>{\hspre}l<{\hspost}@{}}%
\column{E}{@{}>{\hspre}l<{\hspost}@{}}%
\>[5]{}\Varid{fmap}\;\Varid{f}\;(\Varid{a}\mathrel{\Varid{:\!\!\triangleleft}}\Varid{ds})\mathbin{!}[\mskip1.5mu \mskip1.5mu]{}\<[E]%
\\
\>[B]{}\mathrel{=}{}\<[BE]%
\>[5]{}\Varid{f}\;\Varid{a}\mathrel{\Varid{:\!\!\triangleleft}}(\Varid{fmap}\;(\Varid{fmap}\;\Varid{f})\;\Varid{ds})\mathbin{!}[\mskip1.5mu \mskip1.5mu]{}\<[37]%
\>[37]{}\mbox{\onelinecomment  \ensuremath{\Varid{fmap}} on \ensuremath{\Conid{Cofree}\;\Varid{h}}}{}\<[E]%
\\
\>[B]{}\mathrel{=}{}\<[BE]%
\>[5]{}\Varid{f}\;\Varid{a}{}\<[37]%
\>[37]{}\mbox{\onelinecomment  \ensuremath{(\mathbin{!})} on \ensuremath{\Conid{Cofree}\;\Varid{h}}}{}\<[E]%
\\[1.5ex]\>[5]{}\Varid{fmap}\;\Varid{f}\;((\mathbin{!})\;(\Varid{a}\mathrel{\Varid{:\!\!\triangleleft}}\Varid{ds}))\;[\mskip1.5mu \mskip1.5mu]{}\<[E]%
\\
\>[B]{}\mathrel{=}{}\<[BE]%
\>[5]{}(\Varid{f}\hsdot{\circ }{.\:}(\mathbin{!})\;(\Varid{a}\mathrel{\Varid{:\!\!\triangleleft}}\Varid{ds}))\;[\mskip1.5mu \mskip1.5mu]{}\<[37]%
\>[37]{}\mbox{\onelinecomment  \ensuremath{\Varid{fmap}} on functions}{}\<[E]%
\\
\>[B]{}\mathrel{=}{}\<[BE]%
\>[5]{}\Varid{f}\;((\Varid{a}\mathrel{\Varid{:\!\!\triangleleft}}\Varid{ds})\mathbin{!}[\mskip1.5mu \mskip1.5mu]){}\<[37]%
\>[37]{}\mbox{\onelinecomment  \ensuremath{(\hsdot{\circ }{.\:})} definition}{}\<[E]%
\\
\>[B]{}\mathrel{=}{}\<[BE]%
\>[5]{}\Varid{f}\;\Varid{a}{}\<[37]%
\>[37]{}\mbox{\onelinecomment  \ensuremath{(\mathbin{!})} on \ensuremath{\Conid{Cofree}\;\Varid{h}}}{}\<[E]%
\\[1.5ex]\>[5]{}\Varid{fmap}\;\Varid{f}\;(\Varid{a}\mathrel{\Varid{:\!\!\triangleleft}}\Varid{ds})\mathbin{!}(\Varid{c}\mathbin{:}\Varid{cs'}){}\<[E]%
\\
\>[B]{}\mathrel{=}{}\<[BE]%
\>[5]{}\Varid{f}\;\Varid{a}\mathrel{\Varid{:\!\!\triangleleft}}\Varid{fmap}\;(\Varid{fmap}\;\Varid{f})\;\Varid{ds}\mathbin{!}(\Varid{c}\mathbin{:}\Varid{cs'}){}\<[40]%
\>[40]{}\mbox{\onelinecomment  \ensuremath{\Varid{fmap}} on \ensuremath{\Conid{Cofree}\;\Varid{h}}}{}\<[E]%
\\
\>[B]{}\mathrel{=}{}\<[BE]%
\>[5]{}\Varid{fmap}\;(\Varid{fmap}\;\Varid{f})\;\Varid{ds}\mathbin{!}\Varid{c}\mathbin{!}\Varid{cs}{}\<[40]%
\>[40]{}\mbox{\onelinecomment  \ensuremath{(\mathbin{!})} on \ensuremath{\Conid{Cofree}\;\Varid{h}}}{}\<[E]%
\\
\>[B]{}\mathrel{=}{}\<[BE]%
\>[5]{}\Varid{fmap}\;\Varid{f}\;(\Varid{ds}\mathbin{!}\Varid{c})\mathbin{!}\Varid{cs}{}\<[40]%
\>[40]{}\mbox{\onelinecomment  \ensuremath{(\mathbin{!})} on \ensuremath{\Varid{h}} is natural}{}\<[E]%
\\
\>[B]{}\mathrel{=}{}\<[BE]%
\>[5]{}\Varid{f}\;(\Varid{ds}\mathbin{!}\Varid{c}\mathbin{!}\Varid{cs}){}\<[40]%
\>[40]{}\mbox{\onelinecomment  \ensuremath{(\mathbin{!})} on \ensuremath{\Varid{h}} is natural}{}\<[E]%
\\[1.5ex]\>[5]{}\Varid{fmap}\;\Varid{f}\;((\mathbin{!})\;(\Varid{a}\mathrel{\Varid{:\!\!\triangleleft}}\Varid{ds}))\;(\Varid{c}\mathbin{:}\Varid{cs'}){}\<[E]%
\\
\>[B]{}\mathrel{=}{}\<[BE]%
\>[5]{}(\Varid{f}\hsdot{\circ }{.\:}(\mathbin{!})\;(\Varid{a}\mathrel{\Varid{:\!\!\triangleleft}}\Varid{ds}))\;(\Varid{c}\mathbin{:}\Varid{cs'}){}\<[40]%
\>[40]{}\mbox{\onelinecomment  \ensuremath{\Varid{fmap}} on functions}{}\<[E]%
\\
\>[B]{}\mathrel{=}{}\<[BE]%
\>[5]{}\Varid{f}\;((\Varid{a}\mathrel{\Varid{:\!\!\triangleleft}}\Varid{ds})\mathbin{!}(\Varid{c}\mathbin{:}\Varid{cs'})){}\<[40]%
\>[40]{}\mbox{\onelinecomment  \ensuremath{(\hsdot{\circ }{.\:})} definition}{}\<[E]%
\\
\>[B]{}\mathrel{=}{}\<[BE]%
\>[5]{}\Varid{f}\;(\Varid{ds}\mathbin{!}\Varid{c}\mathbin{!}\Varid{cs'}){}\<[40]%
\>[40]{}\mbox{\onelinecomment  \ensuremath{(\mathbin{!})} on \ensuremath{\Conid{Cofree}\;\Varid{h}}}{}\<[E]%
\ColumnHook
\end{hscode}\resethooks
Next show that
\begin{hscode}\SaveRestoreHook
\column{B}{@{}>{\hspre}l<{\hspost}@{}}%
\column{3}{@{}>{\hspre}l<{\hspost}@{}}%
\column{E}{@{}>{\hspre}l<{\hspost}@{}}%
\>[3]{}\Varid{coreturn}\;(\Varid{a}\mathrel{\Varid{:\!\!\triangleleft}}\Varid{ds})\mathrel{=}\Varid{coreturn}\;((\mathbin{!})\;(\Varid{a}\mathrel{\Varid{:\!\!\triangleleft}}\Varid{ds})){}\<[E]%
\ColumnHook
\end{hscode}\resethooks
\begin{hscode}\SaveRestoreHook
\column{B}{@{}>{\hspre}c<{\hspost}@{}}%
\column{BE}{@{}l@{}}%
\column{5}{@{}>{\hspre}l<{\hspost}@{}}%
\column{E}{@{}>{\hspre}l<{\hspost}@{}}%
\>[5]{}\Varid{coreturn}\;((\mathbin{!})\;(\Varid{a}\mathrel{\Varid{:\!\!\triangleleft}}\Varid{ds})){}\<[E]%
\\
\>[B]{}\mathrel{=}{}\<[BE]%
\>[5]{}((\mathbin{!})\;(\Varid{a}\mathrel{\Varid{:\!\!\triangleleft}}\Varid{ds}))\;\varepsilon{}\<[E]%
\\
\>[B]{}\mathrel{=}{}\<[BE]%
\>[5]{}(\Varid{a}\mathrel{\Varid{:\!\!\triangleleft}}\Varid{ds})\mathbin{!}\varepsilon{}\<[E]%
\\
\>[B]{}\mathrel{=}{}\<[BE]%
\>[5]{}\Varid{a}{}\<[E]%
\\
\>[B]{}\mathrel{=}{}\<[BE]%
\>[5]{}\Varid{coreturn}\;(\Varid{a}\mathrel{\Varid{:\!\!\triangleleft}}\Varid{ds}){}\<[E]%
\ColumnHook
\end{hscode}\resethooks
Finally,
\begin{hscode}\SaveRestoreHook
\column{B}{@{}>{\hspre}l<{\hspost}@{}}%
\column{3}{@{}>{\hspre}l<{\hspost}@{}}%
\column{E}{@{}>{\hspre}l<{\hspost}@{}}%
\>[3]{}(\mathbin{!})\hsdot{\circ }{.\:}\Varid{fmap}\;(\mathbin{!})\hsdot{\circ }{.\:}\Varid{cojoin}\mathrel{=}\Varid{cojoin}\hsdot{\circ }{.\:}(\mathbin{!}){}\<[E]%
\ColumnHook
\end{hscode}\resethooks
i.e.,
\begin{hscode}\SaveRestoreHook
\column{B}{@{}>{\hspre}l<{\hspost}@{}}%
\column{3}{@{}>{\hspre}l<{\hspost}@{}}%
\column{E}{@{}>{\hspre}l<{\hspost}@{}}%
\>[3]{}\Varid{fmap}\;(\mathbin{!})\;(\Varid{cojoin}\;(\Varid{a}\mathrel{\Varid{:\!\!\triangleleft}}\Varid{ds}))\mathbin{!}\Varid{cs}\mathrel{=}\Varid{cojoin}\;((\mathbin{!})\;(\Varid{a}\mathrel{\Varid{:\!\!\triangleleft}}\Varid{ds}))\;\Varid{cs}{}\<[E]%
\ColumnHook
\end{hscode}\resethooks
\begin{hscode}\SaveRestoreHook
\column{B}{@{}>{\hspre}c<{\hspost}@{}}%
\column{BE}{@{}l@{}}%
\column{5}{@{}>{\hspre}l<{\hspost}@{}}%
\column{68}{@{}>{\hspre}l<{\hspost}@{}}%
\column{E}{@{}>{\hspre}l<{\hspost}@{}}%
\>[5]{}\Varid{fmap}\;(\mathbin{!})\;(\Varid{cojoin}\;(\Varid{a}\mathrel{\Varid{:\!\!\triangleleft}}\Varid{ds}))\mathbin{!}[\mskip1.5mu \mskip1.5mu]{}\<[E]%
\\
\>[B]{}\mathrel{=}{}\<[BE]%
\>[5]{}\Varid{fmap}\;(\mathbin{!})\;((\Varid{a}\mathrel{\Varid{:\!\!\triangleleft}}\Varid{ds})\mathrel{\Varid{:\!\!\triangleleft}}\Varid{fmap}\;\Varid{cojoin}\;\Varid{ds})\mathbin{!}[\mskip1.5mu \mskip1.5mu]{}\<[68]%
\>[68]{}\mbox{\onelinecomment  \ensuremath{\Varid{cojoin}} on \ensuremath{\Conid{Cofree}\;\Varid{h}}}{}\<[E]%
\\
\>[B]{}\mathrel{=}{}\<[BE]%
\>[5]{}((\mathbin{!})\;(\Varid{a}\mathrel{\Varid{:\!\!\triangleleft}}\Varid{ds})\mathrel{\Varid{:\!\!\triangleleft}}\Varid{fmap}\;(\Varid{fmap}\;(\mathbin{!}))\;(\Varid{fmap}\;\Varid{cojoin}\;\Varid{ds}))\mathbin{!}[\mskip1.5mu \mskip1.5mu]{}\<[68]%
\>[68]{}\mbox{\onelinecomment  \ensuremath{\Varid{fmap}} on \ensuremath{\Conid{Cofree}\;\Varid{h}}}{}\<[E]%
\\
\>[B]{}\mathrel{=}{}\<[BE]%
\>[5]{}(\mathbin{!})\;(\Varid{a}\mathrel{\Varid{:\!\!\triangleleft}}\Varid{ds}){}\<[68]%
\>[68]{}\mbox{\onelinecomment  \ensuremath{(\mathbin{!})} on \ensuremath{\Conid{Cofree}\;\Varid{h}}}{}\<[E]%
\\[1.5ex]\>[5]{}\Varid{cojoin}\;((\mathbin{!})\;(\Varid{a}\mathrel{\Varid{:\!\!\triangleleft}}\Varid{ds}))\;[\mskip1.5mu \mskip1.5mu]{}\<[E]%
\\
\>[B]{}\mathrel{=}{}\<[BE]%
\>[5]{}(\lambda\, \Varid{u}\to \lambda\, \Varid{v}\to (\Varid{a}\mathrel{\Varid{:\!\!\triangleleft}}\Varid{ds})\mathbin{!}(\Varid{u} \diamond \Varid{v}))\;[\mskip1.5mu \mskip1.5mu]{}\<[68]%
\>[68]{}\mbox{\onelinecomment  \ensuremath{\Varid{cojoin}} on functions}{}\<[E]%
\\
\>[B]{}\mathrel{=}{}\<[BE]%
\>[5]{}\lambda\, \Varid{v}\to (\Varid{a}\mathrel{\Varid{:\!\!\triangleleft}}\Varid{ds})\mathbin{!}([\mskip1.5mu \mskip1.5mu] \diamond \Varid{v}){}\<[68]%
\>[68]{}\mbox{\onelinecomment  $\beta$ reduction}{}\<[E]%
\\
\>[B]{}\mathrel{=}{}\<[BE]%
\>[5]{}\lambda\, \Varid{v}\to (\Varid{a}\mathrel{\Varid{:\!\!\triangleleft}}\Varid{ds})\mathbin{!}\Varid{v}{}\<[68]%
\>[68]{}\mbox{\onelinecomment  \ensuremath{\Conid{Monoid}} law (with \ensuremath{\varepsilon\mathrel{=}[\mskip1.5mu \mskip1.5mu]})}{}\<[E]%
\\
\>[B]{}\mathrel{=}{}\<[BE]%
\>[5]{}(\mathbin{!})\;(\Varid{a}\mathrel{\Varid{:\!\!\triangleleft}}\Varid{ds}){}\<[68]%
\>[68]{}\mbox{\onelinecomment  $\eta$ reduction}{}\<[E]%
\\[1.5ex]\>[5]{}\Varid{fmap}\;(\mathbin{!})\;(\Varid{cojoin}\;(\Varid{a}\mathrel{\Varid{:\!\!\triangleleft}}\Varid{ds}))\mathbin{!}(\Varid{c}\mathbin{:}\Varid{cs'}){}\<[E]%
\\
\>[B]{}\mathrel{=}{}\<[BE]%
\>[5]{}\Varid{fmap}\;(\mathbin{!})\;((\Varid{a}\mathrel{\Varid{:\!\!\triangleleft}}\Varid{ds})\mathrel{\Varid{:\!\!\triangleleft}}\Varid{fmap}\;\Varid{cojoin}\;\Varid{ds})\mathbin{!}(\Varid{c}\mathbin{:}\Varid{cs'}){}\<[68]%
\>[68]{}\mbox{\onelinecomment  \ensuremath{\Varid{cojoin}} on \ensuremath{\Conid{Cofree}\;\Varid{h}}}{}\<[E]%
\\
\>[B]{}\mathrel{=}{}\<[BE]%
\>[5]{}((\mathbin{!})\;(\Varid{a}\mathrel{\Varid{:\!\!\triangleleft}}\Varid{ds})\mathrel{\Varid{:\!\!\triangleleft}}\Varid{fmap}\;(\Varid{fmap}\;(\mathbin{!}))\;(\Varid{fmap}\;\Varid{cojoin}\;\Varid{ds}))\mathbin{!}(\Varid{c}\mathbin{:}\Varid{cs'}){}\<[68]%
\>[68]{}\mbox{\onelinecomment  \ensuremath{\Varid{fmap}} on \ensuremath{\Conid{Cofree}\;\Varid{h}}}{}\<[E]%
\\
\>[B]{}\mathrel{=}{}\<[BE]%
\>[5]{}\Varid{fmap}\;(\Varid{fmap}\;(\mathbin{!}))\;(\Varid{fmap}\;\Varid{cojoin}\;\Varid{ds})\mathbin{!}\Varid{c}\mathbin{!}\Varid{cs'}{}\<[68]%
\>[68]{}\mbox{\onelinecomment  \ensuremath{(\mathbin{!})} on \ensuremath{\Conid{Cofree}\;\Varid{h}}}{}\<[E]%
\\
\>[B]{}\mathrel{=}{}\<[BE]%
\>[5]{}\Varid{fmap}\;(\Varid{fmap}\;(\mathbin{!})\hsdot{\circ }{.\:}\Varid{cojoin})\;\Varid{ds}\mathbin{!}\Varid{c}\mathbin{!}\Varid{cs'}{}\<[68]%
\>[68]{}\mbox{\onelinecomment  \ensuremath{\Conid{Functor}} law}{}\<[E]%
\\
\>[B]{}\mathrel{=}{}\<[BE]%
\>[5]{}(\Varid{fmap}\;(\mathbin{!})\hsdot{\circ }{.\:}\Varid{cojoin})\;((\mathbin{!})\;\Varid{ds})\;\Varid{c}\mathbin{!}\Varid{cs'}{}\<[68]%
\>[68]{}\mbox{\onelinecomment  Naturality of \ensuremath{(\mathbin{!})}}{}\<[E]%
\\
\>[B]{}\mathrel{=}{}\<[BE]%
\>[5]{}\Varid{fmap}\;(\mathbin{!})\;(\Varid{cojoin}\;((\mathbin{!})\;\Varid{ds})\;\Varid{c})\mathbin{!}\Varid{cs'}{}\<[68]%
\>[68]{}\mbox{\onelinecomment  \ensuremath{(\hsdot{\circ }{.\:})} definition}{}\<[E]%
\\
\>[B]{}\mathrel{=}{}\<[BE]%
\>[5]{}\Varid{cojoin}\;((\mathbin{!})\;(\Varid{ds}\mathbin{!}\Varid{c}))\;\Varid{cs'}{}\<[68]%
\>[68]{}\mbox{\onelinecomment  coinduction}{}\<[E]%
\\
\>[B]{}\mathrel{=}{}\<[BE]%
\>[5]{}\lambda\, \Varid{v}\to \Varid{cojoin}\;((\mathbin{!})\;(\Varid{ds}\mathbin{!}\Varid{c}))\;\Varid{cs'}\;\Varid{v}{}\<[68]%
\>[68]{}\mbox{\onelinecomment  $\eta$ expansion}{}\<[E]%
\\
\>[B]{}\mathrel{=}{}\<[BE]%
\>[5]{}\lambda\, \Varid{v}\to (\mathbin{!})\;(\Varid{ds}\mathbin{!}\Varid{c})\;(\Varid{cs'} \diamond \Varid{v}){}\<[68]%
\>[68]{}\mbox{\onelinecomment  \ensuremath{\Varid{cojoin}} for functions}{}\<[E]%
\\
\>[B]{}\mathrel{=}{}\<[BE]%
\>[5]{}\lambda\, \Varid{v}\to \Varid{ds}\mathbin{!}\Varid{c}\mathbin{!}(\Varid{cs'} \diamond \Varid{v}){}\<[68]%
\>[68]{}\mbox{\onelinecomment  infix \ensuremath{(\mathbin{!})}}{}\<[E]%
\\
\>[B]{}\mathrel{=}{}\<[BE]%
\>[5]{}\lambda\, \Varid{v}\to (\Varid{a}\mathrel{\Varid{:\!\!\triangleleft}}\Varid{ds})\mathbin{!}(\Varid{c}\mathbin{:}(\Varid{cs'} \diamond \Varid{v})){}\<[68]%
\>[68]{}\mbox{\onelinecomment  \ensuremath{(\mathbin{!})} on \ensuremath{\Conid{Cofree}\;\Varid{h}}  }{}\<[E]%
\\
\>[B]{}\mathrel{=}{}\<[BE]%
\>[5]{}\lambda\, \Varid{v}\to (\Varid{a}\mathrel{\Varid{:\!\!\triangleleft}}\Varid{ds})\mathbin{!}((\Varid{c}\mathbin{:}\Varid{cs'}) \diamond \Varid{v}){}\<[68]%
\>[68]{}\mbox{\onelinecomment  \ensuremath{( \diamond )} on \ensuremath{[\mskip1.5mu \Varid{c}\mskip1.5mu]}      }{}\<[E]%
\\
\>[B]{}\mathrel{=}{}\<[BE]%
\>[5]{}(\lambda\, \Varid{u}\to \lambda\, \Varid{v}\to (\Varid{a}\mathrel{\Varid{:\!\!\triangleleft}}\Varid{ds})\mathbin{!}(\Varid{u} \diamond \Varid{v}))\;(\Varid{c}\mathbin{:}\Varid{cs'}){}\<[68]%
\>[68]{}\mbox{\onelinecomment  $\beta$ reduction    }{}\<[E]%
\\
\>[B]{}\mathrel{=}{}\<[BE]%
\>[5]{}\Varid{cojoin}\;((\mathbin{!})\;(\Varid{a}\mathrel{\Varid{:\!\!\triangleleft}}\Varid{ds}))\;(\Varid{c}\mathbin{:}\Varid{cs'}){}\<[68]%
\>[68]{}\mbox{\onelinecomment  \ensuremath{\Varid{cojoin}} on functions}{}\<[E]%
\ColumnHook
\end{hscode}\resethooks

\subsection{\thmref{pre hom}}\prooflabel{theorem:pre hom}

\begin{spacing}{1.2}
\begin{hscode}\SaveRestoreHook
\column{B}{@{}>{\hspre}l<{\hspost}@{}}%
\column{5}{@{}>{\hspre}l<{\hspost}@{}}%
\column{85}{@{}>{\hspre}l<{\hspost}@{}}%
\column{E}{@{}>{\hspre}l<{\hspost}@{}}%
\>[5]{}\Varid{pre}\;(\Varid{pure}\;\Varid{b}){}\<[E]%
\\
\>[B]{}\mathrel{=}{}\<[5]%
\>[5]{}\Varid{pre}\;(\lambda\, \Varid{a}\to \Varid{b}){}\<[85]%
\>[85]{}\mbox{\onelinecomment  \ensuremath{\Varid{pure}} on \ensuremath{\Varid{a}\to \Varid{b}}}{}\<[E]%
\\
\>[B]{}\mathrel{=}{}\<[5]%
\>[5]{}\Conid{F}\;(\lambda\, \Varid{b'}\to \set{\Varid{a}\mid\Varid{b}\mathrel{=}\Varid{b'}}){}\<[85]%
\>[85]{}\mbox{\onelinecomment  \ensuremath{\Varid{pre}} definition}{}\<[E]%
\\
\>[B]{}\mathrel{=}{}\<[5]%
\>[5]{}\Conid{F}\;(\lambda\, \Varid{b'}\to \mathbf{if}\;\Varid{b'}\mathrel{=}\Varid{b}\;\mathbf{then}\;\set{\Varid{a}\mid\Conid{True}}\;\mathbf{else}\;\set{\Varid{a}\mid\Conid{False}}){}\<[85]%
\>[85]{}\mbox{\onelinecomment  case split}{}\<[E]%
\\
\>[B]{}\mathrel{=}{}\<[5]%
\>[5]{}\Conid{F}\;(\lambda\, \Varid{b'}\to \mathbf{if}\;\Varid{b'}\mathrel{=}\Varid{b}\;\mathbf{then}\;\mathrm{1}\;\mathbf{else}\;\mathrm{0}){}\<[85]%
\>[85]{}\mbox{\onelinecomment  \ensuremath{\mathrm{1}} and \ensuremath{\mathrm{0}} for \ensuremath{\Pow\;\Varid{a}} (revised in \figref{-> and <-- semirings})}{}\<[E]%
\\
\>[B]{}\mathrel{=}{}\<[5]%
\>[5]{}\Varid{b}\mapsto\mathrm{1}{}\<[85]%
\>[85]{}\mbox{\onelinecomment  \ensuremath{(\mapsto)} definition}{}\<[E]%
\\
\>[B]{}\mathrel{=}{}\<[5]%
\>[5]{}\Varid{pure}\;\Varid{b}{}\<[85]%
\>[85]{}\mbox{\onelinecomment  \ensuremath{\Varid{pure}} for \ensuremath{\Pow\;\Varid{a}\leftarrow\Varid{b}}}{}\<[E]%
\\
\>[B]{}{}\;{}\<[E]%
\\[1.5ex]\>[B]{}\hsindent{5}{}\<[5]%
\>[5]{}\Varid{pre}\;(\Varid{fmap}\;\Varid{h}\;\Varid{f}){}\<[E]%
\\
\>[B]{}\mathrel{=}{}\<[5]%
\>[5]{}\Varid{pre}\;(\lambda\, \Varid{a}\to \Varid{h}\;(\Varid{f}\;\Varid{a})){}\<[85]%
\>[85]{}\mbox{\onelinecomment  \ensuremath{\Varid{fmap}} on \ensuremath{\Varid{a}\to \Varid{b}}}{}\<[E]%
\\
\>[B]{}\mathrel{=}{}\<[5]%
\>[5]{}\Conid{F}\;(\lambda\, \Varid{c}\to \set{\Varid{a}\mid\Varid{h}\;(\Varid{f}\;\Varid{a})\mathrel{=}\Varid{c}}){}\<[85]%
\>[85]{}\mbox{\onelinecomment  \ensuremath{\Varid{pre}} definition}{}\<[E]%
\\
\>[B]{}\mathrel{=}{}\<[5]%
\>[5]{}\Conid{F}\;(\lambda\, \Varid{c}\to \set{\Varid{a}\mid\exists \Varid{b}\hsforall \hsdot{\circ }{.\:}\Varid{f}\;\Varid{a}\mathrel{=}\Varid{b}\mathrel{\wedge}\Varid{h}\;\Varid{b}\mathrel{=}\Varid{c}}){}\<[85]%
\>[85]{}\mbox{\onelinecomment  intermediate variable}{}\<[E]%
\\
\>[B]{}\mathrel{=}{}\<[5]%
\>[5]{}\Conid{F}\;(\lambda\, \Varid{c}\to \bigOp\bigcup{\Varid{b}\;\!\!\\\!\!\;\Varid{h}\;\Varid{b}\mathrel{=}\Varid{c}}{0}{\,}\set{\Varid{a}\mid\Varid{f}\;\Varid{a}\mathrel{=}\Varid{b}}){}\<[85]%
\>[85]{}\mbox{\onelinecomment  logic/sets}{}\<[E]%
\\
\>[B]{}\mathrel{=}{}\<[5]%
\>[5]{}\Conid{F}\;(\lambda\, \Varid{c}\to \bigOp\bigcup{\Varid{b}\;\!\!\\\!\!\;\Varid{h}\;\Varid{b}\mathrel{=}\Varid{c}}{0}{\,}\Varid{pre}\;\Varid{f}\;\Varid{b}){}\<[85]%
\>[85]{}\mbox{\onelinecomment  \ensuremath{\Varid{pre}} definition}{}\<[E]%
\\
\>[B]{}\mathrel{=}{}\<[5]%
\>[5]{}\Conid{F}\;(\lambda\, \Varid{c}\to \bigOp\sum{\Varid{b}\;\!\!\\\!\!\;\Varid{h}\;\Varid{b}\mathrel{=}\Varid{c}}{0}{\,}\Varid{pre}\;\Varid{f}\;\Varid{b}){}\<[85]%
\>[85]{}\mbox{\onelinecomment  \ensuremath{(\mathbin{+})} on \ensuremath{\Pow\;\Varid{a}}}{}\<[E]%
\\
\>[B]{}\mathrel{=}{}\<[5]%
\>[5]{}\Varid{fmap}\;\Varid{h}\;(\Varid{pre}\;\Varid{f}){}\<[85]%
\>[85]{}\mbox{\onelinecomment  \ensuremath{\Varid{fmap}} on \ensuremath{\Pow\;\Varid{a}\leftarrow\Varid{b}}}{}\<[E]%
\\
\>[B]{}{}\;{}\<[E]%
\\[1.5ex]\>[B]{}\hsindent{5}{}\<[5]%
\>[5]{}\Varid{pre}\;(\Varid{liftA}_{2}\;\Varid{h}\;\Varid{f}\;\Varid{g}){}\<[E]%
\\
\>[B]{}\mathrel{=}{}\<[5]%
\>[5]{}\Varid{pre}\;(\lambda\, \Varid{a}\to \Varid{h}\;(\Varid{f}\;\Varid{a})\;(\Varid{g}\;\Varid{a})){}\<[85]%
\>[85]{}\mbox{\onelinecomment  \ensuremath{\Varid{liftA}_{2}} on \ensuremath{\Varid{a}\to \Varid{b}}}{}\<[E]%
\\
\>[B]{}\mathrel{=}{}\<[5]%
\>[5]{}\lambda\, \Varid{c}\to \{\mskip1.5mu \Varid{a}\mid \Varid{h}\;(\Varid{f}\;\Varid{a})\;(\Varid{g}\;\Varid{a})\mathrel{=}\Varid{c}\mskip1.5mu\}{}\<[85]%
\>[85]{}\mbox{\onelinecomment  \ensuremath{\Varid{pre}} definition}{}\<[E]%
\\
\>[B]{}\mathrel{=}{}\<[5]%
\>[5]{}\lambda\, \Varid{c}\to \{\mskip1.5mu \Varid{a}\mid \exists \Varid{x}\hsforall \;\Varid{y}\hsdot{\circ }{.\:}\Varid{x}\mathrel{=}\Varid{f}\;\Varid{a}\mathrel{\wedge}\Varid{y}\mathrel{=}\Varid{g}\;\Varid{a}\mathrel{\wedge}\Varid{h}\;\Varid{x}\;\Varid{y}\mathrel{=}\Varid{c}\mskip1.5mu\}{}\<[85]%
\>[85]{}\mbox{\onelinecomment  intermediate variables}{}\<[E]%
\\
\>[B]{}\mathrel{=}{}\<[5]%
\>[5]{}\lambda\, \Varid{c}\to \{\mskip1.5mu \Varid{a}\mid \exists \Varid{x}\hsforall \;\Varid{y}\hsdot{\circ }{.\:}\Varid{a}\mathbin{\in}\Varid{pre}\;\Varid{f}\;\Varid{x}\mathrel{\wedge}\Varid{a}\mathbin{\in}\Varid{pre}\;\Varid{g}\;\Varid{y}\mathrel{\wedge}\Varid{h}\;\Varid{x}\;\Varid{y}\mathrel{=}\Varid{c}\mskip1.5mu\}{}\<[85]%
\>[85]{}\mbox{\onelinecomment  \ensuremath{\Varid{pre}} definition (twice)}{}\<[E]%
\\
\>[B]{}\mathrel{=}{}\<[5]%
\>[5]{}\lambda\, \Varid{c}\to \{\mskip1.5mu \Varid{a}\mid \exists \Varid{x}\hsforall \;\Varid{y}\hsdot{\circ }{.\:}\Varid{a}\mathbin{\in}(\Varid{pre}\;\Varid{f}\;\Varid{x}\cap\Varid{pre}\;\Varid{g}\;\Varid{y})\mathrel{\wedge}\Varid{h}\;\Varid{x}\;\Varid{y}\mathrel{=}\Varid{c}\mskip1.5mu\}{}\<[85]%
\>[85]{}\mbox{\onelinecomment  \ensuremath{\cap} definition}{}\<[E]%
\\
\>[B]{}\mathrel{=}{}\<[5]%
\>[5]{}\bigOp\bigcup{\Varid{x},\Varid{y}}{0}\;\Varid{h}\;\Varid{x}\;\Varid{y}\mapsto\Varid{pre}\;\Varid{f}\;\Varid{x}\cap\Varid{pre}\;\Varid{g}\;\Varid{y}{}\<[85]%
\>[85]{}\mbox{\onelinecomment  logic/sets}{}\<[E]%
\\
\>[B]{}\mathrel{=}{}\<[5]%
\>[5]{}\bigOp\bigcup{\Varid{x},\Varid{y}}{0}\;\Varid{h}\;\Varid{x}\;\Varid{y}\mapsto\Varid{pre}\;\Varid{f}\;\Varid{x}\mathbin{*}\Varid{pre}\;\Varid{g}\;\Varid{y}{}\<[85]%
\>[85]{}\mbox{\onelinecomment  \ensuremath{(\mathbin{*})} on \ensuremath{\Pow\;\Varid{a}} (revised in \figref{-> and <-- semirings})}{}\<[E]%
\\
\>[B]{}\mathrel{=}{}\<[5]%
\>[5]{}\bigOp\sum{\Varid{x},\Varid{y}}{0}\;\Varid{h}\;\Varid{x}\;\Varid{y}\mapsto\Varid{pre}\;\Varid{f}\;\Varid{x}\mathbin{*}\Varid{pre}\;\Varid{g}\;\Varid{y}{}\<[85]%
\>[85]{}\mbox{\onelinecomment  \ensuremath{(\mathbin{+})} on \ensuremath{\Pow\;\Varid{a}\leftarrow\Varid{b}}}{}\<[E]%
\\
\>[B]{}\mathrel{=}{}\<[5]%
\>[5]{}\Varid{liftA}_{2}\;\Varid{h}\;(\Varid{pre}\;\Varid{f})\;(\Varid{pre}\;\Varid{g}){}\<[E]%
\ColumnHook
\end{hscode}\resethooks
\end{spacing}

\subsection{\thmref{standard FunApp}}\prooflabel{theorem:standard FunApp}

First consider \ensuremath{\Varid{fmap}}, as defined in \figref{FunApp}.
\begin{hscode}\SaveRestoreHook
\column{B}{@{}>{\hspre}c<{\hspost}@{}}%
\column{BE}{@{}l@{}}%
\column{5}{@{}>{\hspre}l<{\hspost}@{}}%
\column{35}{@{}>{\hspre}l<{\hspost}@{}}%
\column{E}{@{}>{\hspre}l<{\hspost}@{}}%
\>[5]{}\Varid{fmap}\;\Varid{h}\;(\Conid{F}\;\Varid{f}){}\<[E]%
\\
\>[B]{}\mathrel{=}{}\<[BE]%
\>[5]{}\bigOp\sum{\Varid{u}}{0}\;\Varid{h}\;\Varid{u}\mapsto\Varid{f}\;\Varid{u}{}\<[35]%
\>[35]{}\mbox{\onelinecomment  definition of \ensuremath{\Varid{fmap}} on \ensuremath{(\leftarrow)\;\Varid{b}}}{}\<[E]%
\\
\>[B]{}\mathrel{=}{}\<[BE]%
\>[5]{}\bigOp\sum{\Varid{u}}{0}\;\Varid{h}\;\Varid{u}\mapsto\Varid{f}\;\Varid{u}\mathbin{*}\mathrm{1}{}\<[35]%
\>[35]{}\mbox{\onelinecomment  multiplicative identity}{}\<[E]%
\\
\>[B]{}\mathrel{=}{}\<[BE]%
\>[5]{}\bigOp\sum{\Varid{u}}{0}\;\Varid{f}\;\Varid{u}\cdot(\Varid{h}\;\Varid{u}\mapsto\mathrm{1})\mbox{\onelinecomment  \lemref{+-> homomorphism}}{}\<[E]%
\\
\>[B]{}\mathrel{=}{}\<[BE]%
\>[5]{}\bigOp\sum{\Varid{u}}{0}\;\Varid{f}\;\Varid{u}\cdot\Varid{single}\;(\Varid{h}\;\Varid{u}){}\<[35]%
\>[35]{}\mbox{\onelinecomment  definition of \ensuremath{\Varid{single}}}{}\<[E]%
\\
\>[B]{}\mathrel{=}{}\<[BE]%
\>[5]{}\bigOp\sum{\Varid{u}}{0}\;\Varid{f}\;\Varid{u}\cdot\Varid{pure}\;(\Varid{h}\;\Varid{u}){}\<[35]%
\>[35]{}\mbox{\onelinecomment  \ensuremath{\Varid{single}\mathrel{=}\Varid{pure}}}{}\<[E]%
\\
\>[B]{}\mathrel{=}{}\<[BE]%
\>[5]{}\Conid{F}\;\Varid{f}\bind \Varid{pure}\hsdot{\circ }{.\:}\Varid{h}{}\<[35]%
\>[35]{}\mbox{\onelinecomment  definition of \ensuremath{(\bind )}}{}\<[E]%
\ColumnHook
\end{hscode}\resethooks
\noindent
Similarly for \ensuremath{\Varid{liftA}_{2}}:
\vspace{-1ex}
\begin{spacing}{1.5}
\begin{hscode}\SaveRestoreHook
\column{B}{@{}>{\hspre}c<{\hspost}@{}}%
\column{BE}{@{}l@{}}%
\column{5}{@{}>{\hspre}l<{\hspost}@{}}%
\column{53}{@{}>{\hspre}l<{\hspost}@{}}%
\column{E}{@{}>{\hspre}l<{\hspost}@{}}%
\>[5]{}\Varid{liftA}_{2}\;\Varid{h}\;(\Conid{F}\;\Varid{f})\;(\Conid{F}\;\Varid{g}){}\<[E]%
\\
\>[B]{}\mathrel{=}{}\<[BE]%
\>[5]{}\bigOp\sum{\Varid{u},\Varid{v}}{0}\;\Varid{h}\;\Varid{u}\;\Varid{v}\mapsto\Varid{f}\;\Varid{u}\mathbin{*}\Varid{g}\;\Varid{v}{}\<[53]%
\>[53]{}\mbox{\onelinecomment  definition of \ensuremath{\Varid{liftA}_{2}}}{}\<[E]%
\\
\>[B]{}\mathrel{=}{}\<[BE]%
\>[5]{}\bigOp\sum{\Varid{u},\Varid{v}}{0}\;(\Varid{f}\;\Varid{u}\mathbin{*}\Varid{g}\;\Varid{v})\cdot\Varid{single}\;(\Varid{h}\;\Varid{u}\;\Varid{v}){}\<[53]%
\>[53]{}\mbox{\onelinecomment  as above}{}\<[E]%
\\
\>[B]{}\mathrel{=}{}\<[BE]%
\>[5]{}\bigOp\sum{\Varid{u},\Varid{v}}{0}\;\Varid{f}\;\Varid{u}\cdot(\Varid{g}\;\Varid{v}\cdot\Varid{single}\;(\Varid{h}\;\Varid{u}\;\Varid{v})){}\<[53]%
\>[53]{}\mbox{\onelinecomment  associativity}{}\<[E]%
\\
\>[B]{}\mathrel{=}{}\<[BE]%
\>[5]{}\bigOp\sum{\Varid{u}}{0}\;\Varid{f}\;\Varid{u}\cdot\bigOp\sum{\Varid{v}}{0}\;\Varid{g}\;\Varid{v}\cdot\Varid{single}\;(\Varid{h}\;\Varid{u}\;\Varid{v}){}\<[53]%
\>[53]{}\mbox{\onelinecomment  linearity}{}\<[E]%
\\
\>[B]{}\mathrel{=}{}\<[BE]%
\>[5]{}\bigOp\sum{\Varid{u}}{0}\;\Varid{f}\;\Varid{u}\cdot\bigOp\sum{\Varid{v}}{0}\;\Varid{h}\;\Varid{u}\;\Varid{v}\mapsto\Varid{g}\;\Varid{v}{}\<[53]%
\>[53]{}\mbox{\onelinecomment  as above}{}\<[E]%
\\
\>[B]{}\mathrel{=}{}\<[BE]%
\>[5]{}\bigOp\sum{\Varid{u}}{0}\;\Varid{f}\;\Varid{u}\cdot\Varid{fmap}\;(\Varid{h}\;\Varid{u})\;(\Conid{F}\;\Varid{g}){}\<[53]%
\>[53]{}\mbox{\onelinecomment  definition of \ensuremath{\Varid{fmap}}}{}\<[E]%
\\
\>[B]{}\mathrel{=}{}\<[BE]%
\>[5]{}\Conid{F}\;\Varid{f}\bind \lambda\, \Varid{u}\to \Varid{fmap}\;(\Varid{h}\;\Varid{u})\;(\Conid{F}\;\Varid{g}){}\<[53]%
\>[53]{}\mbox{\onelinecomment  definition of \ensuremath{(\bind )}}{}\<[E]%
\\
\>[B]{}\mathrel{=}{}\<[BE]%
\>[5]{}\Conid{F}\;\Varid{f}\bind \lambda\, \Varid{u}\to \Conid{F}\;\Varid{g}\bind \lambda\, \Varid{v}\to \Varid{pure}\;(\Varid{h}\;\Varid{u}\;\Varid{v}){}\<[53]%
\>[53]{}\mbox{\onelinecomment  above}{}\<[E]%
\ColumnHook
\end{hscode}\resethooks
\end{spacing}

\subsection{\thmref{poly hom}}\prooflabel{theorem:poly hom}

\begin{hscode}\SaveRestoreHook
\column{B}{@{}>{\hspre}c<{\hspost}@{}}%
\column{BE}{@{}l@{}}%
\column{5}{@{}>{\hspre}l<{\hspost}@{}}%
\column{22}{@{}>{\hspre}l<{\hspost}@{}}%
\column{36}{@{}>{\hspre}l<{\hspost}@{}}%
\column{E}{@{}>{\hspre}l<{\hspost}@{}}%
\>[5]{}\Varid{poly}\;\mathrm{0}{}\<[E]%
\\
\>[B]{}\mathrel{=}{}\<[BE]%
\>[5]{}\Varid{poly}\;(\Conid{F}\;(\lambda\, \Varid{i}\to \mathrm{0})){}\<[36]%
\>[36]{}\mbox{\onelinecomment  \ensuremath{\mathrm{0}} on \ensuremath{\Varid{b}\leftarrow\Varid{a}} (derived)}{}\<[E]%
\\
\>[B]{}\mathrel{=}{}\<[BE]%
\>[5]{}\lambda\, \Varid{x}\to \bigOp\sum{\Varid{i}}{0}\;{}\<[22]%
\>[22]{}\mathrm{0}\mathbin{*}\Varid{x}^\Varid{i}{}\<[36]%
\>[36]{}\mbox{\onelinecomment  \ensuremath{\Varid{poly}} definition}{}\<[E]%
\\
\>[B]{}\mathrel{=}{}\<[BE]%
\>[5]{}\lambda\, \Varid{x}\to \bigOp\sum{\Varid{i}}{0}\;{}\<[22]%
\>[22]{}\mathrm{0}{}\<[36]%
\>[36]{}\mbox{\onelinecomment  \ensuremath{\mathrm{0}} as annihilator}{}\<[E]%
\\
\>[B]{}\mathrel{=}{}\<[BE]%
\>[5]{}\lambda\, \Varid{x}\to \mathrm{0}{}\<[36]%
\>[36]{}\mbox{\onelinecomment  \ensuremath{\mathrm{0}} as additive identity}{}\<[E]%
\\
\>[B]{}\mathrel{=}{}\<[BE]%
\>[5]{}\mathrm{0}{}\<[36]%
\>[36]{}\mbox{\onelinecomment  \ensuremath{\mathrm{0}} on functions}{}\<[E]%
\ColumnHook
\end{hscode}\resethooks

\begin{hscode}\SaveRestoreHook
\column{B}{@{}>{\hspre}c<{\hspost}@{}}%
\column{BE}{@{}l@{}}%
\column{5}{@{}>{\hspre}l<{\hspost}@{}}%
\column{65}{@{}>{\hspre}l<{\hspost}@{}}%
\column{E}{@{}>{\hspre}l<{\hspost}@{}}%
\>[5]{}\Varid{poly}\;\mathrm{1}{}\<[E]%
\\
\>[B]{}\mathrel{=}{}\<[BE]%
\>[5]{}\Varid{poly}\;(\Varid{pure}\;\varepsilon){}\<[65]%
\>[65]{}\mbox{\onelinecomment  \ensuremath{\mathrm{1}} on \ensuremath{\Varid{b}\leftarrow\Varid{a}}}{}\<[E]%
\\
\>[B]{}\mathrel{=}{}\<[BE]%
\>[5]{}\Varid{poly}\;(\Conid{F}\;(\lambda\, \Varid{i}\to \mathbf{if}\;\Varid{i}\mathrel{=}\varepsilon\;\mathbf{then}\;\mathrm{1}\;\mathbf{else}\;\mathrm{0})){}\<[65]%
\>[65]{}\mbox{\onelinecomment  \ensuremath{\Varid{pure}} on \ensuremath{(\leftarrow)\;\Varid{b}}}{}\<[E]%
\\
\>[B]{}\mathrel{=}{}\<[BE]%
\>[5]{}\Varid{poly}\;(\Conid{F}\;(\lambda\, \Varid{i}\to \mathbf{if}\;\Varid{i}\mathrel{=}\mathrm{0}\;\mathbf{then}\;\mathrm{1}\;\mathbf{else}\;\mathrm{0})){}\<[65]%
\>[65]{}\mbox{\onelinecomment  \ensuremath{\varepsilon} on \ensuremath{\mathbb{N}}}{}\<[E]%
\\
\>[B]{}\mathrel{=}{}\<[BE]%
\>[5]{}\lambda\, \Varid{x}\to \bigOp\sum{\Varid{i}}{0}\;(\mathbf{if}\;\Varid{i}\mathrel{=}\mathrm{0}\;\mathbf{then}\;\mathrm{1}\;\mathbf{else}\;\mathrm{0})\mathbin{*}\Varid{x}^\Varid{i}{}\<[65]%
\>[65]{}\mbox{\onelinecomment  \ensuremath{\Varid{poly}} definition}{}\<[E]%
\\
\>[B]{}\mathrel{=}{}\<[BE]%
\>[5]{}\lambda\, \Varid{x}\to \bigOp\sum{\Varid{i}}{0}\;(\mathbf{if}\;\Varid{i}\mathrel{=}\mathrm{0}\;\mathbf{then}\;\Varid{x}^\Varid{i}\;\mathbf{else}\;\mathrm{0}){}\<[65]%
\>[65]{}\mbox{\onelinecomment  simplify}{}\<[E]%
\\
\>[B]{}\mathrel{=}{}\<[BE]%
\>[5]{}\lambda\, \Varid{x}\to \Varid{x}^\mathrm{0}{}\<[65]%
\>[65]{}\mbox{\onelinecomment  other terms vanish}{}\<[E]%
\\
\>[B]{}\mathrel{=}{}\<[BE]%
\>[5]{}\lambda\, \Varid{x}\to \mathrm{1}{}\<[65]%
\>[65]{}\mbox{\onelinecomment  multiplicative identity}{}\<[E]%
\\
\>[B]{}\mathrel{=}{}\<[BE]%
\>[5]{}\mathrm{1}{}\<[65]%
\>[65]{}\mbox{\onelinecomment  \ensuremath{\mathrm{1}} on \ensuremath{\Varid{a}\to \Varid{b}}}{}\<[E]%
\ColumnHook
\end{hscode}\resethooks

\begin{hscode}\SaveRestoreHook
\column{B}{@{}>{\hspre}c<{\hspost}@{}}%
\column{BE}{@{}l@{}}%
\column{5}{@{}>{\hspre}l<{\hspost}@{}}%
\column{22}{@{}>{\hspre}l<{\hspost}@{}}%
\column{23}{@{}>{\hspre}l<{\hspost}@{}}%
\column{51}{@{}>{\hspre}l<{\hspost}@{}}%
\column{65}{@{}>{\hspre}l<{\hspost}@{}}%
\column{E}{@{}>{\hspre}l<{\hspost}@{}}%
\>[5]{}\Varid{poly}\;(\Conid{F}\;\Varid{f}\mathbin{+}\Conid{F}\;\Varid{g}){}\<[E]%
\\
\>[B]{}\mathrel{=}{}\<[BE]%
\>[5]{}\Varid{poly}\;(\Conid{F}\;(\lambda\, \Varid{i}\to \Varid{f}\;\Varid{i}\mathbin{+}\Varid{g}\;\Varid{i})){}\<[65]%
\>[65]{}\mbox{\onelinecomment  \ensuremath{(\mathbin{+})} on \ensuremath{\Varid{b}\leftarrow\Varid{a}} (derived)}{}\<[E]%
\\
\>[B]{}\mathrel{=}{}\<[BE]%
\>[5]{}\lambda\, \Varid{x}\to \bigOp\sum{\Varid{i}}{0}\;{}\<[22]%
\>[22]{}(\Varid{f}\;\Varid{i}\mathbin{+}\Varid{g}\;\Varid{i})\mathbin{*}\Varid{x}^\Varid{i}{}\<[65]%
\>[65]{}\mbox{\onelinecomment  \ensuremath{\Varid{poly}} definition}{}\<[E]%
\\
\>[B]{}\mathrel{=}{}\<[BE]%
\>[5]{}\lambda\, \Varid{x}\to \bigOp\sum{\Varid{i}}{0}\;{}\<[22]%
\>[22]{}\Varid{f}\;\Varid{i}\mathbin{*}\Varid{x}^\Varid{i}\mathbin{+}\Varid{g}\;\Varid{i}\mathbin{*}\Varid{x}^\Varid{i}{}\<[65]%
\>[65]{}\mbox{\onelinecomment  distributivity}{}\<[E]%
\\
\>[B]{}\mathrel{=}{}\<[BE]%
\>[5]{}\lambda\, \Varid{x}\to (\bigOp\sum{\Varid{i}}{0}\;{}\<[23]%
\>[23]{}\Varid{f}\;\Varid{i}\mathbin{*}\Varid{x}^\Varid{i})\mathbin{+}(\bigOp\sum{\Varid{i}}{0}\;{}\<[51]%
\>[51]{}\Varid{g}\;\Varid{i}\mathbin{*}\Varid{x}^\Varid{i}){}\<[65]%
\>[65]{}\mbox{\onelinecomment  summation property}{}\<[E]%
\\
\>[B]{}\mathrel{=}{}\<[BE]%
\>[5]{}\lambda\, \Varid{x}\to \Varid{poly}\;(\Conid{F}\;\Varid{f})\;\Varid{x}\mathbin{+}\Varid{poly}\;(\Conid{F}\;\Varid{g})\;\Varid{x}{}\<[65]%
\>[65]{}\mbox{\onelinecomment  \ensuremath{\Varid{poly}} definition}{}\<[E]%
\\
\>[B]{}\mathrel{=}{}\<[BE]%
\>[5]{}\Varid{poly}\;(\Conid{F}\;\Varid{f})\mathbin{+}\Varid{poly}\;(\Conid{F}\;\Varid{g}){}\<[65]%
\>[65]{}\mbox{\onelinecomment  \ensuremath{(\mathbin{+})} on \ensuremath{\Varid{a}\to \Varid{b}}}{}\<[E]%
\ColumnHook
\end{hscode}\resethooks

\begin{hscode}\SaveRestoreHook
\column{B}{@{}>{\hspre}c<{\hspost}@{}}%
\column{BE}{@{}l@{}}%
\column{5}{@{}>{\hspre}l<{\hspost}@{}}%
\column{19}{@{}>{\hspre}l<{\hspost}@{}}%
\column{23}{@{}>{\hspre}l<{\hspost}@{}}%
\column{25}{@{}>{\hspre}l<{\hspost}@{}}%
\column{51}{@{}>{\hspre}l<{\hspost}@{}}%
\column{65}{@{}>{\hspre}l<{\hspost}@{}}%
\column{E}{@{}>{\hspre}l<{\hspost}@{}}%
\>[5]{}\Varid{poly}\;(\Conid{F}\;\Varid{f}\mathbin{*}\Conid{F}\;\Varid{g}){}\<[E]%
\\
\>[B]{}\mathrel{=}{}\<[BE]%
\>[5]{}\Varid{poly}\;(\Varid{liftA}_{2}\;( \diamond )\;(\Conid{F}\;\Varid{f})\;(\Conid{F}\;\Varid{g})){}\<[65]%
\>[65]{}\mbox{\onelinecomment  \ensuremath{(\mathbin{*})} on \ensuremath{\Varid{b}\leftarrow\Varid{a}}}{}\<[E]%
\\
\>[B]{}\mathrel{=}{}\<[BE]%
\>[5]{}\Varid{poly}\;(\bigOp\sum{\Varid{i},\Varid{j}}{0}\;{}\<[25]%
\>[25]{}\Varid{i} \diamond \Varid{j}\mapsto\Varid{f}\;\Varid{i}\mathbin{*}\Varid{g}\;\Varid{j}){}\<[65]%
\>[65]{}\mbox{\onelinecomment  \ensuremath{\Varid{liftA}_{2}} on \ensuremath{\Varid{b}\leftarrow\Varid{a}}}{}\<[E]%
\\
\>[B]{}\mathrel{=}{}\<[BE]%
\>[5]{}\Varid{poly}\;(\bigOp\sum{\Varid{i},\Varid{j}}{0}\;{}\<[25]%
\>[25]{}\Varid{i}\mathbin{+}\Varid{j}\mapsto\Varid{f}\;\Varid{i}\mathbin{*}\Varid{g}\;\Varid{j}){}\<[65]%
\>[65]{}\mbox{\onelinecomment  \ensuremath{( \diamond )} on \ensuremath{\mathbb{N}}}{}\<[E]%
\\
\>[B]{}\mathrel{=}{}\<[BE]%
\>[5]{}\bigOp\sum{\Varid{i},\Varid{j}}{0}\;{}\<[19]%
\>[19]{}\Varid{poly}\;(\Varid{i}\mathbin{+}\Varid{j}\mapsto\Varid{f}\;\Varid{i}\mathbin{*}\Varid{g}\;\Varid{j}){}\<[65]%
\>[65]{}\mbox{\onelinecomment  additivity of \ensuremath{\Varid{poly}} (previous property)}{}\<[E]%
\\
\>[B]{}\mathrel{=}{}\<[BE]%
\>[5]{}\bigOp\sum{\Varid{i},\Varid{j}}{0}\;(\lambda\, \Varid{x}\to (\Varid{f}\;\Varid{i}\mathbin{*}\Varid{g}\;\Varid{j})\mathbin{*}\Varid{x}^{\Varid{i}\mathbin{+}\Varid{j}}){}\<[65]%
\>[65]{}\mbox{\onelinecomment  \lemref{poly +->} below}{}\<[E]%
\\
\>[B]{}\mathrel{=}{}\<[BE]%
\>[5]{}\lambda\, \Varid{x}\to \bigOp\sum{\Varid{i},\Varid{j}}{0}\;(\Varid{f}\;\Varid{i}\mathbin{*}\Varid{g}\;\Varid{j})\mathbin{*}\Varid{x}^{\Varid{i}\mathbin{+}\Varid{j}}{}\<[65]%
\>[65]{}\mbox{\onelinecomment  \ensuremath{(\mathbin{+})} on functions}{}\<[E]%
\\
\>[B]{}\mathrel{=}{}\<[BE]%
\>[5]{}\lambda\, \Varid{x}\to \bigOp\sum{\Varid{i},\Varid{j}}{0}\;(\Varid{f}\;\Varid{i}\mathbin{*}\Varid{g}\;\Varid{j})\mathbin{*}(\Varid{x}^\Varid{i}\mathbin{*}\Varid{x}^\Varid{j}){}\<[65]%
\>[65]{}\mbox{\onelinecomment  exponentiation property}{}\<[E]%
\\
\>[B]{}\mathrel{=}{}\<[BE]%
\>[5]{}\lambda\, \Varid{x}\to \bigOp\sum{\Varid{i},\Varid{j}}{0}\;(\Varid{f}\;\Varid{i}\mathbin{*}\Varid{x}^\Varid{i})\mathbin{*}(\Varid{g}\;\Varid{j}\mathbin{*}\Varid{x}^\Varid{j}){}\<[65]%
\>[65]{}\mbox{\onelinecomment  commutativity assumption}{}\<[E]%
\\
\>[B]{}\mathrel{=}{}\<[BE]%
\>[5]{}\lambda\, \Varid{x}\to (\bigOp\sum{\Varid{i}}{0}\;{}\<[23]%
\>[23]{}\Varid{f}\;\Varid{i}\mathbin{*}\Varid{x}^\Varid{i})\mathbin{*}(\bigOp\sum{\Varid{j}}{0}\;{}\<[51]%
\>[51]{}\Varid{g}\;\Varid{j}\mathbin{*}\Varid{x}^\Varid{j}){}\<[65]%
\>[65]{}\mbox{\onelinecomment  summation property}{}\<[E]%
\\
\>[B]{}\mathrel{=}{}\<[BE]%
\>[5]{}\lambda\, \Varid{x}\to \Varid{poly}\;(\Conid{F}\;\Varid{f})\;\Varid{x}\mathbin{*}\Varid{poly}\;\Conid{F}\;\Varid{g})\;\Varid{x}{}\<[65]%
\>[65]{}\mbox{\onelinecomment  \ensuremath{\Varid{poly}} definition}{}\<[E]%
\\
\>[B]{}\mathrel{=}{}\<[BE]%
\>[5]{}\Varid{poly}\;(\Conid{F}\;\Varid{f})\mathbin{*}\Varid{poly}\;\Conid{F}\;\Varid{g}){}\<[65]%
\>[65]{}\mbox{\onelinecomment  \ensuremath{(\mathbin{*})} on functions}{}\<[E]%
\ColumnHook
\end{hscode}\resethooks


\begin{lemma}\lemlabel{poly +->}~
\begin{hscode}\SaveRestoreHook
\column{B}{@{}>{\hspre}l<{\hspost}@{}}%
\column{E}{@{}>{\hspre}l<{\hspost}@{}}%
\>[B]{}\Varid{poly}\;(\Varid{n}\mapsto\Varid{b})\mathrel{=}\lambda\, \Varid{x}\to \Varid{b}\mathbin{*}\Varid{x}^\Varid{n}{}\<[E]%
\ColumnHook
\end{hscode}\resethooks
\end{lemma}
\begin{proof}~
\begin{hscode}\SaveRestoreHook
\column{B}{@{}>{\hspre}l<{\hspost}@{}}%
\column{55}{@{}>{\hspre}l<{\hspost}@{}}%
\column{E}{@{}>{\hspre}l<{\hspost}@{}}%
\>[B]{}\Varid{poly}\;(\Varid{n}\mapsto\Varid{b}){}\<[E]%
\\
\>[B]{}\Varid{poly}\;(\Conid{F}\;(\lambda\, \Varid{i}\to \mathbf{if}\;\Varid{i}\mathrel{=}\Varid{n}\;\mathbf{then}\;\Varid{b}\;\mathbf{else}\;\mathrm{0})){}\<[55]%
\>[55]{}\mbox{\onelinecomment  \ensuremath{(\mapsto)} on \ensuremath{\Varid{b}\leftarrow\Varid{a}} (derived)}{}\<[E]%
\\
\>[B]{}\lambda\, \Varid{x}\to \bigOp\sum{\Varid{i}}{0}\;(\mathbf{if}\;\Varid{i}\mathrel{=}\Varid{n}\;\mathbf{then}\;\Varid{b}\;\mathbf{else}\;\mathrm{0})\mathbin{*}\Varid{x}^\Varid{n}{}\<[55]%
\>[55]{}\mbox{\onelinecomment  \ensuremath{\Varid{poly}} definition}{}\<[E]%
\\
\>[B]{}\lambda\, \Varid{x}\to \Varid{b}\mathbin{*}\Varid{x}^\Varid{n}{}\<[55]%
\>[55]{}\mbox{\onelinecomment  other terms vanish}{}\<[E]%
\ColumnHook
\end{hscode}\resethooks
\end{proof}

\subsection{\lemref{pows hom}}\prooflabel{lemma:pows hom}

\begin{hscode}\SaveRestoreHook
\column{B}{@{}>{\hspre}c<{\hspost}@{}}%
\column{BE}{@{}l@{}}%
\column{5}{@{}>{\hspre}l<{\hspost}@{}}%
\column{38}{@{}>{\hspre}l<{\hspost}@{}}%
\column{E}{@{}>{\hspre}l<{\hspost}@{}}%
\>[5]{}\Varid{x}\!{\string^}^{\hspace{-1pt}\mathrm{0}}{}\<[E]%
\\
\>[B]{}\mathrel{=}{}\<[BE]%
\>[5]{}\bigOp\prod{\Varid{i}}{0}{\,}(\Varid{x}\;\Varid{i})^{\mathrm{0}\;\Varid{i}}{}\<[38]%
\>[38]{}\mbox{\onelinecomment  \ensuremath{(\string^)} definition}{}\<[E]%
\\
\>[B]{}\mathrel{=}{}\<[BE]%
\>[5]{}\bigOp\prod{\Varid{i}}{0}{\,}(\Varid{x}\;\Varid{i})^{\mathrm{0}}{}\<[38]%
\>[38]{}\mbox{\onelinecomment  \ensuremath{\mathrm{0}} on functions}{}\<[E]%
\\
\>[B]{}\mathrel{=}{}\<[BE]%
\>[5]{}\bigOp\prod{\Varid{i}}{0}\;\mathrm{1}{}\<[38]%
\>[38]{}\mbox{\onelinecomment  exponentiation law}{}\<[E]%
\\
\>[B]{}\mathrel{=}{}\<[BE]%
\>[5]{}\mathrm{1}{}\<[38]%
\>[38]{}\mbox{\onelinecomment  multiplicative identity}{}\<[E]%
\ColumnHook
\end{hscode}\resethooks

\begin{hscode}\SaveRestoreHook
\column{B}{@{}>{\hspre}c<{\hspost}@{}}%
\column{BE}{@{}l@{}}%
\column{5}{@{}>{\hspre}l<{\hspost}@{}}%
\column{82}{@{}>{\hspre}l<{\hspost}@{}}%
\column{E}{@{}>{\hspre}l<{\hspost}@{}}%
\>[5]{}\Varid{x}\!{\string^}^{\hspace{-1pt}\Varid{p}\mathbin{+}\Varid{q}}{}\<[E]%
\\
\>[B]{}\mathrel{=}{}\<[BE]%
\>[5]{}\bigOp\prod{\Varid{i}}{0}{\,}(\Varid{x}\;\Varid{i})^{(\Varid{p}\mathbin{+}\Varid{q})\;\Varid{i}}{}\<[82]%
\>[82]{}\mbox{\onelinecomment  \ensuremath{(\string^)} definition}{}\<[E]%
\\
\>[B]{}\mathrel{=}{}\<[BE]%
\>[5]{}\bigOp\prod{\Varid{i}}{0}{\,}(\Varid{x}\;\Varid{i})^{\Varid{p}\;\Varid{i}\mathbin{+}\Varid{q}\;\Varid{i}}{}\<[82]%
\>[82]{}\mbox{\onelinecomment  \ensuremath{(\mathbin{+})} on functions}{}\<[E]%
\\
\>[B]{}\mathrel{=}{}\<[BE]%
\>[5]{}\bigOp\prod{\Varid{i}}{0}{\,}((\Varid{x}\;\Varid{i})^{\Varid{p}\;\Varid{i}})\mathbin{*}((\Varid{x}\;\Varid{i})^{\Varid{q}\;\Varid{i}}){}\<[82]%
\>[82]{}\mbox{\onelinecomment  exponentiation law (with commutative \ensuremath{(\mathbin{*})})}{}\<[E]%
\\
\>[B]{}\mathrel{=}{}\<[BE]%
\>[5]{}\left(\bigOp\prod{\Varid{i}}{0}{\,}(\Varid{x}\;\Varid{i})^{\Varid{p}\;\Varid{i}}\right)\mathbin{*}\left(\bigOp\prod{\Varid{i}}{0}{\,}(\Varid{x}\;\Varid{i})^{\Varid{q}\;\Varid{i}}\right){}\<[82]%
\>[82]{}\mbox{\onelinecomment  product property (with commutative \ensuremath{(\mathbin{*})})}{}\<[E]%
\\
\>[B]{}\mathrel{=}{}\<[BE]%
\>[5]{}\Varid{x}\!{\string^}^{\hspace{-1pt}\Varid{p}}\mathbin{*}\Varid{x}\!{\string^}^{\hspace{-1pt}\Varid{q}}{}\<[82]%
\>[82]{}\mbox{\onelinecomment  \ensuremath{(\string^)} definition}{}\<[E]%
\ColumnHook
\end{hscode}\resethooks

\bibliography{bib}

\end{document}